\makeatletter
\@ifundefined{INSERTION} %{\foo}
  {%% https://texfaq.org/FAQ-isdef
\documentclass[12pt,a4paper]{article}

\usepackage[warn]{mathtext}          % russian letters in formulas, with warnings
\usepackage[T2A]{fontenc}            % internal TeX uncoding
\usepackage[russian, english]{babel} % локализация и переносы;  note the opts order: the last is applied when conflict

\usepackage{amsmath, amsthm, amssymb}
\DeclareSymbolFontAlphabet{\mathcalorig}   {symbols}

\newcommand\fakeslant[1]{%
  \pdfliteral{1 0 0.210 1 -0.4 0 cm}#1\pdfliteral{1 0 -0.210 1 0.4 0 cm}} %%167
\newcommand\mathbbsl[1]{\mathbb{\fakeslant{#1}}}
\if{
\usepackage{graphicx}

}\fi

\usepackage{cmap}          % русский поиск в pdf

\usepackage{epsfig}

\usepackage[normalem]{ulem}   %% 2023 different underlining

\newtheorem{theorem}{Theorem}[section]
\newtheorem{corollary}{Corollary}[theorem]
\newtheorem{lemma}[theorem]{Lemma}
 \theoremstyle{definition}

 \theoremstyle{remark}
\newtheorem*{remark}{Remark}

\usepackage{cancel} % https://jansoehlke.com/2010/06/strikethrough-in-latex/
\usepackage{empheq} %% for fanier equation arrays // 2019 add

\usepackage{pdfsync}  %% for forward/inverse synchronization with pdf file
\usepackage{verbatim}  %% for comment environment

\usepackage[svgnames]{xcolor}

\usepackage{accents}

\usepackage{pgf}
\usepackage{tikz}
\usetikzlibrary{arrows,automata}
\usetikzlibrary{positioning}
\usepackage{tcolorbox} % для управления цветом частей текста
\ifx\EuScript\undefined
    \usepackage{euscript}
\else \fi

\usepackage{hyperref}
\hypersetup{
    colorlinks=true,
    linkcolor=DarkBlue,
    citecolor=Navy,%britishracingred,%DarkGreen,
    urlcolor=DarkBlue,
    }

 \textheight 24.8 cm  \topmargin -1.5 cm   % a4paper
 \textwidth 17.1 cm  \oddsidemargin -0.45 cm   % a4paper

\newcommand{\newpar}{\vspace{2mm}}

\definecolor{repColor}{rgb}{0.0, 0.0, 0.25} %% Color for reparameterized quantities
\definecolor{gColor}{rgb}{0.10, 0.0, 0.20} %% Color for gauge quantities
\definecolor{sColor}{RGB}{0, 40, 0} %% Color for spatial quantities
\definecolor{mygray}{gray}{0.2}
\definecolor{igray}{rgb}{0.1,0.1,0.15}

\DeclareRobustCommand{\ColorGauge}[1]{\color{gColor}#1\normalcolor}  %Gauge Coloring
\DeclareRobustCommand{\ColorRepAux}[1]{\color{repColor}#1\normalcolor}  %Parameterization Coloring

\newcommand{\beq}[1]{\begin{equation} \label{#1}}
\newcommand{\eeq}{\end{equation}}

\newcommand{\bea}[1]{\begin{eqnarray} \label{#1}}
\newcommand{\eea}{\end{eqnarray}}

\newcommand{\TODO}[1]{{\color{red}\textit{To Do}\normalcolor: \textit{#1} }}
\newcommand{\TBC}[1]{{\color{red}\textit{To Be Checked}\normalcolor: \textit{#1} }}
\newcommand{\HIDE}[1]{}

  \newcommand{\tskip}{{\hspace{0.8pt}}}   % tiny skip

\newcommand{\LiB}{\Big(}
\newcommand{\RiB}{\Big)}

\newcommand{\Ddim}{\mathrm{d}} %{D}%{{\color{DarkBlue}\mathbf{d}}}  % dimension of space-time
\newcommand{\tx}{t}
\newcommand{\sx}{{\color{sColor}{x}}}
\newcommand{\dsx}{d\sx} %{{d^{^{\,\Ddim{-}1}}\!\!\!\!\sx}}

\newcommand{\const}{ const }

\newcommand{\sVol}{{\color{sColor}V}}
\newcommand{\sint}{{\mathchoice{{\mathop{\color{sColor}\textstyle{\int}}}}
                         {{\mathop{\color{sColor}\textstyle{\int}}}}
                         {{\mathop{\color{sColor}{\int}}}}
                         {{\mathop{\color{sColor}{\int}}}} }}

\newcommand{\GR}{GR}
\newcommand{\UMG}{UMG}
\newcommand{\GUMG}{GUMG}

\newcommand{\g}{\gamma}
\newcommand{\sqrg}{\sqrt{\g\hspace{0.5pt}}\hspace{0.5pt}\vphantom{|}}%{\sqrt{{\!\,\!{\shortmid}{\g}{\shortmid}\!\,\!}}}%{} %{{\sqrt{\g}}}

\newcommand{\lapse}{\!\tskip\scriptscriptstyle{\perp}\!}

\newcommand{\INs}[1]{{#1}}  %# shift
\newcommand{\INl}[1]{{#1}^{\lapse}}  %# lapse

\newcommand{\iNl}[1]{{#1}_{\lapse}}   %# lapse

\newcommand{\Ns}{\INs{N}}
\newcommand{\Nl}{\INl{N}}
\newcommand{\lmrNl}{\iNl{\lambda}}

\newcommand{\zi}{{i}}   %i
\newcommand{\zj}{{j}}   %j

\newcommand{\zk}{{k}}   %k
\newcommand{\zl}{{l}}   %l
\newcommand{\zm}{{m}}   %m
\newcommand{\zn}{{n}}   %n

\newcommand{\Zl}{{\lambda}}  %l
\newcommand{\Zm}{{\mu}}      %m
\newcommand{\Zn}{{\nu}}      %n

\DeclareMathOperator{\weq}{\approx}
\newcommand{\weqwrt}[1]{\overset{\footnotesize{#1}}{\approx}} %\simeq%{\genfrac{}{}{0pt}{}{\approxeq}{\tiny{#1}} }

\DeclareMathOperator{\eomeq}{{\,\accentset{{\color{gray}\sim}}{=}\,}}

\DeclareMathOperator{\defeq}{{\accentset{{\color{gray}\scriptsize{\,---\,}}}{\,=\,}}}

\newcommand{\teq}{{\,\hspace{1pt}=\:}} %% text mode =
\newcommand{\tneq}{{\,\neq\:}} %% text mode =
\newcommand{\tequiv}{{\,\equiv\:}} %% text mode \equiv
\newcommand{\tdefeq}{{\,\defeq\:}} %% text mode \equiv

\def\tdot{\dot}
\def\tauxdot{\dot} %{\mathring}

\DeclareMathOperator{\rank}{\,\mathrm{rank}\,}

\newcommand{\var}{\operatorname{\delta}\hspace{-1pt}} %{\hspace{0.8pt}\delta} % variation symbol
\newcommand{\gvar}[1][]{{\color{gColor}\var}_{#1}}  % gauge variation %(bad: отступы в дробях)
\newcommand{\gvarp}[1][]{{\color{gColor}\var'\!\!}_{#1}}  % gauge variation %(bad:
\newcommand{\ugvar}[1][]{{\color{gColor}\var}_{#1}}  % unrestricted gauge transfromations
\newcommand{\gvarGR}[1][]{\gvar[#1]^{^{\scriptscriptstyle{(G\hspace{-0.5pt}R)}}}\!}

\newcommand{\TDer}[2]{\frac{{d}\tskip{#1}} {{d}\tskip{#2}}}  % total derivative
\newcommand{\tTDer}[2]{\tfrac{{d}\tskip{#1}} {{d}\tskip{#2}}}  % total derivative

\newcommand{\VDer}[3][]{ \frac{{\var^{#1}}  {#2}} {{\var^{}} {#3}} }  % variational derivative
\newcommand{\tVDer}[3][]{ \tfrac{\var^{#1} {#2}} {\var {#3}} }  % variational derivative

\newcommand{\svar}{{\color{sColor}\var\hspace{1pt}}} %{{^{^x\!\!}\delta}}
\newcommand{\sVDer}[3][]{ \frac{{\svar^{#1}}  {#2}} {\svar {#3}} }  % spatial variational derivative %%Unique to BVReducible
\newcommand{\Dvg}[2][\@empty]{%
  \ifx#1\@empty %This is the case when the argument is empty.%
    (\partial #2)
  \else %This is the case when the argument is not empty: #1.%
    (\partial #2)
  \fi
}
\newcommand{\Dvgp}[2][\@empty]{% %protected indexed divergence for quantities
    #2^{#1}_{\,,#1}
}
\newcommand{\Dvgn}[1]{% %protected divergence with no indexes
    (\partial #1)
}
\newcommand{\Dvgi}[2][\@empty]{% %improved divergence for field with lower index
  \ifx#1\@empty %This is the case when the argument is empty.%
    (\partial #2)
  \else %This is the case when the argument is not empty: #1.%
    {#2^{\,#1}}\vphantom{|\xi}_{,#1}
  \fi
}
\newcommand{\LieB}[4][]{{\mathchoice{{\big[{#3},{#4}\big]^{#2}}}
                         {{\big[{#3},{#4}\big]^{#2}}}
                         {{[{#3},{#4}]^{#2}}}
                         {{[{#3},{#4}]^{#2}}} }}

\newcommand{\LieD}[1]{\mathcal{L}_{#1}}      %# Lie derivative

\newcommand{\isst}{{\scriptscriptstyle (t)}} %% transverse index
\newcommand{\tProj}{{P_{\!\!\isst}}} %% transverse projector

 \newcommand{\geps}{\ColorGauge{\varepsilon}} %additional grouping makes acccent (bar, dot) position incorrect (lost italics adjustment)
 \newcommand{\gxi}{\ColorGauge{\xi}} %additional grouping makes acccent (bar, dot) position incorrect (lost italics adjustment)

 \newcommand{\gzeta}{\ColorGauge{\zeta}}
 \newcommand{\gchi}{\ColorGauge{\chi}} %% Lagrange multiplier at constraint \pi^io_,\zn

\newcommand{\gzetat}{\gzeta_{{\color{mygray}(\hspace{-1pt}t\hspace{-1pt})}}}

\newcommand{\gchit}{\gchi_{{\color{mygray}(\hspace{-1pt}t\hspace{-1pt})}}}

 \newcommand{\gx}{\ColorGauge{\zeta}}
 \newcommand{\gxs}{\INs{\gx}}
 \newcommand{\gxl}{\INl{\gx}}

\newcommand{\SSS}{S} % Action
\newcommand{\LLd}{\mathcal{L}}  % Lagrange density

\newcommand{\qq}{q}  % phase space (coordinates)
\newcommand{\pp}{p}  % phase space (momenta)
\newcommand{\CC}{\psi} % constraint

\newcommand{\IC}{\gamma} %I-class constraints
\newcommand{\IIC}{\chi} %II-class constraints

\newcommand{\PB}[2]{{\mathchoice{\big\{{#1},{#2}\big\}} {\big\{{#1},{#2}\big\}} {\{{#1},{#2}\}}{\{#1,#2\}}}}

 \newcommand{\phI}{\vphantom{\Big|}}

\newcommand{\Bra}[1]{\langle #1| }   %% <|
\newcommand{\Ket}[1]{|#1 \rangle}   %% |>
\newcommand{\BraKet}[2]{\langle #1| #2 \rangle}   %% scalar product
\newcommand{\BraOKet}[3]{\langle #1| #2| #3 \rangle}   %% scalar product
\newcommand{\tf}{\rep{\tau}}
\newcommand{\ptf}{\rep{\pi}}  % canonically conjugated momenta for time-field
\newcommand{\replm}{\rep{\lambda}}  % parametrization Lagrange multiplier

\newcommand{\CT}{\mathcal{T}}
\newcommand{\CS}{\mathcal{S}}
\newcommand{\CP}{\mathcal{P}}

\newcommand{\CPI}{\CP_{\scriptscriptstyle{\!\!(I)}}}

\newcommand{\rK}{K} %{{\color{SaddleBrown} K}}  % reduced extrinsic curvature
\newcommand{\sR}{R} %{{\color{DarkGreen} ^{^{\sdim}\!\!}R}}  % spatial curvature
\newcommand{\GBulk}{g} %{\textbf{g}}

\newcommand{\pg}{\pi}%{p}             % 3-metric momenta
\newcommand{\trpg}{\hspace{0.5pt}\pg\hspace{0.5pt}} %{{\color{Brown}\pg}}%        % 3-metric momenta trace

\newcommand{\ecK}{K}  %% extrinsic curvature K
\newcommand{\trecK}{K}  %% trace of extrinsic curvature K

\newcommand{\Hs}{{H}}                        % Momentum constraint
\newcommand{\Hl}{{H}_{^{_{\!\perp}}\!}}      % Hamiltonian constraint

\newcommand{\stR}{{^{^{\Ddim}\!\!}R}}  % spatial curvature

  \newcommand{\FF}{F}
  \newcommand{\WW}{W} % Polytropic parameter = 2 d ln F/ d ln \g
\newcommand{\wwc}{\mathrm{\hspace{0.5pt}w}}

  \newcommand{\OOmega}{\Omega}
  \newcommand{\TTheta}{\Theta}
 \newcommand{\ader}{\hspace{1pt}\accentset{\leftrightarrow}{\partial}\hspace{-1pt}} %% alternating derivative

\newcommand{\Ho}{\HH_0}
\definecolor{bulgarianrose}{rgb}{0.28, 0.02, 0.03} %% de-macro does not see it
\DeclareRobustCommand{\ColorRepAux}[1]{\color{repColor}#1\normalcolor}  %Gauge Coloring

\newcommand{\ialt}{alt}

\newcommand{\iGR}{{\scriptscriptstyle G\hspace{-0.5pt}R}}
\newcommand{\iGUMG}{}
\newcommand{\iwwc}{{\scriptscriptstyle{(\!\wwc\!)}}}

\newcommand{\iE}{{\scriptscriptstyle E}}
\newcommand{\iP}{{\scriptscriptstyle P}}

\newcommand{\io}{{\scriptscriptstyle{0}}}

\newcommand{\ibo}{{\scriptscriptstyle{\color{igray}\mathbf{0}}}}  % bold index 0

\newcommand{\rep}[1]{\ColorRepAux{#1}} %%{\breve{#1}}

\newcommand{\taux}{\rep{t}} %$\tau$ time variable in parameterized action

\newcommand{\xo}{t}%{\rep{t}}%{x^0} %% Stuekelberg time
\newcommand{\po}{p}%{p_{_0}}        %% Stuekelberg time momenta
\newcommand{\lmo}{\lambda}  %% time reparametrization Lagrange multiplier

 \newcommand{\HTf}{\rep{\mathcal{V}}} %%{\rep{\mathcal{X}}}

 \newcommand{\CCf}{\rep{\varLambda}}  %{\varLambda\hspace{0.5pt}}
 \newcommand{\CCof}{\rep{\varLambda}_\ibo}  %{\varLambda\hspace{0.5pt}}
 \newcommand{\cco}{\mathrm{c}_\ibo}

\newcommand{\Gaux}{\rep{\GBulk}}%{\rep{\GG}}

\newcommand{\repNs}{\rep{\Ns}}%{\rep{\Ns}}
\newcommand{\repNl}{\rep{\Nl}}%{\rep{\Nl}}
\newcommand{\lmCT}{\mu} %% vector Lagrange multipliers at \CT_{,\zm} %% to cmd.GRAV.sty
\newcommand{\replmCT}{\rep{\lmCT}}%{\rep{\lmv}}
\newcommand{\replmCTn}{\replmCT'}
\newcommand{\lmptf}{\rep{\nu}}%{\rep{\lmv}}

\newcommand{\hmg}[1]{\overline{#1}} %{\bar{#1}}
\newcommand{\inh}[1]{\widetilde{#1}} %{\check{#1}} {\tilde}

\newcommand{\hmgw}[1]{\overline{\hspace{0.4mm}#1\hspace{0.4mm}}} %{\bar{#1}}
\newcommand{\inhw}[1]{\hspace{-0.6mm}\widetilde{\hspace{0.6mm}#1\hspace{0.6mm}}\hspace{-0.6mm}} %{\check{#1}}

\newcommand{\hmgs}[1]{{\hspace{0.9mm}\hmg{\hspace{-0.5mm}#1\hspace{-0.5mm}}\hspace{0.9mm}}} %{\bar{#1}}

\newcommand{\TSEt}{T}
\newcommand{\pSEt}{p}
\newcommand{\eSEt}{\varrho} %{\varepsilon}

\newcommand{\repTSEt}{\rep{T}}
\newcommand{\reppSEt}{\rep{p}}
\newcommand{\repeSEt}{\rep{\varrho}} %{\varepsilon}

\newcommand{\Lib}{\big(}
\newcommand{\Rib}{\big)}

\newcommand{\wGUMG}{$\wwc$-\GUMG}
\newcommand{\Fw}{F_{\!\iwwc}}
\newcommand{\CTw}{\CT_{\!\iwwc}}

\newcommand{\argsqrg}{{\scalebox{0.8}{\mbox{\!\ensuremath{\sqrg}}}}}

\DeclareRobustCommand{\op}[1]{\hspace{-0pt}\mathbbsl{#1}}

\newcommand{\Id}{{\hspace{1.0pt}\op{{I}\hspace{1pt}}\hspace{0.8pt}}}
\newcommand{\hmgId}{{\op{\hspace{-0.5pt}\hmg{\hspace{1.6pt}{I}\hspace{0.4pt}}}\hspace{2.0pt}}}
\newcommand{\inhId}{{\op{\hspace{-0.5pt}\inhw{\hspace{1.6pt}{I}\hspace{0.3pt}}}\hspace{2.1pt}}}
\newcommand{\opE}{\hspace{0pt}\op{E}\hspace{0.8pt}}
\newcommand{\opinvE}{\hspace{0pt}\op{E}^{\,-1\!}}
\newcommand{\ofa}{a}
\newcommand{\opEa}{\opE_{\!\ofa}}
\newcommand{\opEE}{\opE_{\hspace{-0.5pt}ab\hspace{-0.2pt}}}%{\op{\mathcal{E}}} %% generic \opE

\newcommand{\opG}{{\hspace{0.4pt}\inh{\hspace{-0.5pt}\op{G}\hspace{0.5pt}}\hspace{0.5pt}}}
\newcommand{\opGm}{\inh{\hspace{-0.5pt}\op{G}\hspace{0.5pt}}_{{\hspace{-1.3pt}\scriptscriptstyle{-}}}}
\newcommand{\opGp}{\inh{\hspace{-0.5pt}\op{G}\hspace{0.5pt}}_{{\hspace{-1.3pt}\scriptscriptstyle{+}}}}

\newcommand{\tFProj}[2]{\mathop{\underaccent{#1,#2}{\Pi}}} %\accentset
\newcommand{\lFProj}[2]{\mathop{\underaccent{#1,#2}{\textrm{L}}}} %\accentset

\newcommand{\Af}{a}
\newcommand{\Bf}{b}
\newcommand{\Xf}{x}
\newcommand{\Yf}{y}
\newcommand{\invBraKet}[2]{\tfrac{1}{\langle #1| #2 \rangle}}   %% scalar product
\begin{document}

%%\if{ %%%% ARTICLE (b1)
  \title{Alternative Action for Generalized Unimodular Gravity}
  %\title{\bf Parameterized Action for Generalized Unimodular Gravity Models}
  \author{Dmitry~Nesterov$^{\star}$ }
 
  \date{}

  \maketitle

\vspace{-7mm}
\begin{center}
  %\hspace{-0mm}{\em Theory Department, Lebedev Physical Institute, Leninsky Prospect 53, Moscow 119991, Russia}
  \hspace{-0mm}{\em $^{\star}$ \HIDE{I.E.Tamm }Theory Department, P.N. Lebedev Physical Institute of the Russian Academy of Sciences,  Leninsky Prospekt 53, Moscow, 119991, Russia}\\
%  \phantom{.}\\
%  \hspace{-0mm}{\em $^{\dag}$ Moscow Institute of Physics and Technology, \\
%   Institutskiy per. 9, Dolgoprudny, Moscow Region, 141700, Russia}
% \\
 \noindent
\end{center}

% \noindent
% $^{\star}$e-mail: nesterov@lpi.ru\\
%% $^{\dag}$e-mail: \\

%%}\fi %%% END ARTICLE

\if{ %%% REVTEX
%Title of paper
\title{Alternative Action for Generalized Unimodular Gravity
 %\\ \phantom{.}
 \vspace{0.5mm} }
%\vspace{1mm}

% repeat the \author .. \affiliation  etc. as needed
% \email, \thanks, \homepage, \altaffiliation all apply to the current
% author. Explanatory text should go in the []'s, actual e-mail
% address or url should go in the {}'s for \email and \homepage.
% Please use the appropriate macro foreach each type of information

% \affiliation command applies to all authors since the last
% \affiliation command. The \affiliation command should follow the
% other information
% \affiliation can be followed by \email, \homepage, \thanks as well.
  \author{Dmitry~Nesterov} %$^{\star}$  %D.V.~Nesterov
\email[]{nesterov@lpi.ru}
%\homepage[]{Your web page}
%\thanks{}
%\altaffiliation{}
  \affiliation{ \HIDE{I.E.Tamm }Theory Department, P.N. Lebedev Physical Institute of the Russian Academy of Sciences,  Leninsky Prospekt 53, Moscow, 119991, Russia} % $^{\star}$

% insert suggested keywords - APS authors don't need to do this
%\keywords{}

%\maketitle must follow title, authors, abstract, and keywords
\maketitle

}\fi %%% END REVTEX

%%%============================================================================%%%

\begin{abstract}
  \vspace{0.6mm}

   We present an alternative formulation of generalized unimodular gravity (\GUMG), a class of modifications to general relativity characterized by{ a special} partial breaking of general coordinate covariance.
   The action for this formulation is derived constructively through a sequence of equivalent representations, starting from the original {\GUMG} setup and extending the configuration space and gauge structure by introducing time parametrization. Our approach, based on a canonical formalism, parallels the method used by Henneaux and Teitelboim to covariantize the action of unimodular gravity ({\UMG}), which was\HIDE{ subsequently} generalized to consistently accommodate a parameterization via local fields for the entire {\GUMG} family.
   For completeness, we explore the dynamical structure of the theory and provide a detailed account of its gauge properties.
   A notable consequence of the consistent parametrization is the emergence of explicit spatial delocalization in the action, manifestations of which we carefully examine.
   Within the {\GUMG} family, we identify a subfamily of models described by a local action. While preserving the essential structural features, these models avoid the complications of spatial nonlocality allowing a clearer analysis.
   The dynamical and constraint structures of {\GUMG} family models are different from those of {\UMG}, however the latter appears as an exceptional special case within the family, providing valuable comparative insights.

  \vspace{0.7mm}
  \phantom{.}

\end{abstract}

 \noindent
 $^{\star}$e-mail: nesterov@lpi.ru\\
% $^{\dag}$e-mail: \\

{
  %\small
  \newpage
  \setcounter{tocdepth}{2}
  \tableofcontents
}

  }
  {%
    % \foo defined
  }%
\makeatother

%%%============================================================================%%%
%%%============================================================================%%%
%%%============================================================================%%%

\newpage
\section{Introduction}
 \hspace{\parindent}
 The exploration of gravity theories with broken general covariance dates back to Einstein, who examined such frameworks in 1919 \cite{Einstein1919UMG}. Thus Unimodular gravity ({\UMG}) --- a significant model that retains key features of general relativity ({\GR})\HIDE{, including the dynamical massless graviton}, emerged nearly at the dawn of modern gravitational physics and gained renewed interest in late eighties,   \cite{Weinberg:1988cp,Unruh:1988in,Ng:1990xz,Bousso:2007gp,Bufalo:2015wda}. A useful advantage of this theory is an alternative covariant formulation, proposed by Henneaux and Teitelboim, \cite{Henneaux:1989zc}. Over the past decade, the search for experimentally viable modifications of general relativity has been influenced by the Horava-Lifshitz models \cite{Horava:2009uw,Blas:2009qj,Clifton:2011jh}, where spacetime covariance is explicitly broken to only spatial covariance, allowing to improve ultraviolet behavior.

 Generalized unimodular gravity ({\GUMG}) \cite{Barvinsky:2017pmm,Bufalo:2017tms,Barvinsky:2019agh,Barvinsky:2019qzx} in combines specific aspects of both approaches. Like {\UMG}, it can be described by simply restricting the configuration space of general relativity. At the same time it breaks spacetime covariance in a substantial way, introducing a preferred time direction and a corresponding spatial slicing of spacetime. Remarkably, although spatial diffeomorphism invariance is also implicitly broken in {\GUMG}, this violation does not generate explicit spatial anisotropy, allowing the theory to retain some of the benefits of spatial covariance within the formalism.
%\\

\newpar

Generalized unimodular gravity models originally were introduced in \cite{Barvinsky:2017pmm}
% \if{ %% ARTICLE
  \bea{Action_GUMG_L} % (\ref{Action_GUMG_L})
    \SSS^{\iGUMG}[\GBulk,\lmrNl] =
     \int \!d\tx \hspace{1pt} \dsx
     \,\sqrt{|\GBulk|}\,
     \stR(\GBulk)
      %%-2\Lambda
     - \int \!d\tx \hspace{1pt} \dsx
     \,\lmrNl \big(\Nl {-} \FF(\argsqrg) \big)
    \,,
  \eea
%  }\fi %% ARTICLE
\if{ %% REVTEX
  \bea{Action_GUMG_L} % (\ref{Action_GUMG_L})
    \!
    \SSS^{\iGUMG}[\GBulk,\lmrNl]
    = \! \int \hspace{-1pt}\!d \tx \hspace{1pt} \dsx \,
   \LiB\!
     \sqrt{|\GBulk|}\,\stR(\GBulk)
     -
     \lmrNl\big(\Nl {-\,} \FF(\argsqrg) \big)
   \RiB \hspace{-1pt}
    ,\;\;\;\;
  \eea
}\fi %% REVTEX
\HIDE{with $\FF(\argsqrg)$ --- some arbitrary positive definite function,} motivated by search for dark energy candidates. Here $\GBulk_{\Zm\Zn}$ is an unrestricted metric of $\Ddim$-dimensional spacetime,  which can be decomposed into Arnowitt–Deser–Misner (ADM) metric components \cite{Arnowitt:1959ah,Misner:1973prb} as
 %$
 \beq{GBulk_metric_ADM}% was {Gaux_correspondence}
  \GBulk_{\Zm\Zn}d{X}^{\Zm}d{X}^{\Zn}
   =  \g_{\zm\zn} (d\sx^\zm {+} {\Ns}^\zm d\tx)(d\sx^\zn {+} {\Ns}^\zn d\tx) - {\Nl}^2 d\tx \hspace{1pt} d\tx
   \,,
 %$
 \eeq
 with lapse function $\Nl(\tx,\sx)$, shift functions $\Ns^\zm(\tx,\sx)$ and the induced metric $\g_{\zm\zn}(\tx,\sx)$ on $\tx \teq \const$ hypersurfaces. The first term in (\ref{Action_GUMG_L}) is the standard Einstein-Hilbert action of general relativity and the second term implements the generalized unimodular restriction condition,
 \beq{GUMG_restriction}
  \Nl \eomeq \FF(\argsqrg)
  \,,
 \eeq
which restricts the lapse function $\Nl$ to be equal to some positive-definite function $\FF(\argsqrg) {\,>\,} 0$
  %%\footnote{Condition $\FF{\,>\,}0$ is sufficient in the range of values of the argument, $\sqrg>0$. Certain regularity conditions will suggest additional bounds on monotonicity and convexity of the characteristic function $\FF$ on the domain.}
of the induced metric determinant, $\sqrg \defeq \sqrt{\det \g_{\zm\zn}}$. A choice of the characteristic function $\FF$ fixes the particular theory. In the particular case $\FF\teq\sqrg^{-1}$, condition (\ref{GUMG_restriction}) reproduces the restriction condition of unimodular gravity, which justifies the name of the {\GUMG} family and allows for certain checks to be carried out.

The observation, which attracted attention to {\GUMG} models, is that the restriction term in the action (\ref{Action_GUMG_L}) generates the energy-momentum tensor of the cosmological perfect fluid,
  \bea{GUMGrStressTensor}
    \TSEt^{\Zm\Zn}%\big|_{\Nl=\FF}
     = \eSEt \,u^\Zm u^\Zn + \pSEt \,(\GBulk^{\Zm\Zn}+u^\Zm u^\Zn)
     \,,
  \eea
on the right hand side of Einstein equations. In (\ref{GUMGrStressTensor}) $u^\Zm$ --- timelike unit vector field, normal to $\tx = \const$ hypersurfaces. Perfect fluid's energy is $\eSEt \defeq \frac{\lmrNl}{2\sqrg}$ and pressure is $\pSEt \defeq \frac{\lmrNl}{2\sqrg}\frac{\FF\WW}{\Nl}$,
 % \footnote{
 %   Where $u^\Za$ is unit timelike vector field orthogonal to constant time slices. Energy $ \eSEt=\frac{\mu}{2\sqrg}$ and pressure $\pSEt=\frac{\mu}{2\sqrg}\frac{\FF\WW}{\Nl}$ which on the (\ref{GUMG_Restriction}) leads to energy pressure relation $\pSEt=\WW\eSEt$, so that $\WW(\argsqrg)$ is identified with barotropic parameter of cosmological fluid.
 % }
which on the restriction surface $\Nl \eomeq \FF(\argsqrg)$ implies the equation of state\HIDE{ of the cosmological perfect fluid}
  \bea{PerfectFluid_Eq_of_state} % was {defGUMGrWW}
    \pSEt  \eomeq  \WW \eSEt
    \,,
  \eea
where the role of barotropic parameter is played by the derivative characteristic function\footnote{
 To avoid additional branching of the canonical constraint structure we assume $\WW(\argsqrg)$ to be sign-definite\HIDE{ function}, which implies the monotonicity of $\FF(\argsqrg)$ (these characteristic functions are defined on domain $\sqrg {\,>\,} 0$). However, it is not the\HIDE{ absolutely} strict condition. For example, the particular \emph{\HIDE{cold }dust} case $\WW \teq 0$ also fits the formalism.
}
  \bea{def_WW_GUMG} % (\ref{def_WW_GUMG}) % was {defGUMGrWW}
    \WW(\argsqrg) \,\defeq\, \TDer{\,\ln \FF(\argsqrg)}{\,\ln{\sqrg}}
    \,.
  \eea

The dependence of $\WW$ on the cosmological scale of spatial sections is an attractive feature, as the coefficient in the equation of state typically depends on the redshift parameter, according to phenomenology of the Universe evolution (for recent reviews\HIDE{ in this field} see \cite{Wang:2018fng,Avsajanishvili:2023jcl}\footnote{
 Authors of\HIDE{ the review} \cite{Avsajanishvili:2023jcl} refer to homogeneous cosmological models, characterized by the dependance of $\WW$ on the conformal factor of spatial sections, as {$\WW$}CDM models\HIDE{ and give reference to certain validity tests based on modern observational data in a near $\WW$-constant limit}. They note, however, that such parametrization ``has no physical motivation, but is commonly used as an ansatz in data analysis to quantify differences and distinguish between dark energy models''. {\GUMG} models\HIDE{, with their restriction mechanism that directly generates {$\WW$}CDM parametrization ansatz,} present a novel perspective on their potential physical origin.
 %%At the same time {\GUMG} approach allows treating such models as fundamental theories.
}).

Dynamical and constraint analysis of {\GUMG} models disproved some original hopes, but uncovered various interesting features \cite{Barvinsky:2019agh}. Notably, it was shown that {\GUMG}
suggests new inflation mechanism consistent with observations within a certain subclass of $\FF\HIDE{(\argsqrg)}$ \cite{Barvinsky:2019qzx}.

\newpar

This work continues the general analysis of the generalized unimodular gravity theories. Here we construct and analyze an alternative representation for {\GUMG} models  \cite{Barvinsky:2017pmm,Barvinsky:2019agh}, analogous to the Henneaux--Teitelboim representation \cite{Henneaux:1989zc} for unimodular gravity. We discuss the dynamic and gauge structure of the resulting representation as well as related technical\HIDE{ and conceptual} issues, including peculiarities of parametrization of the field theory models.

\newcommand{\opB}{\op{B}}

Formally, the key result is a constructive proof that the general {\GUMG} model can be equivalently described by the action
  \bea{Alt_Action_GUMG}
    \SSS^{\iGUMG}_{\ialt}[\Gaux,\CCf,\HTf]
    = \int \!d\taux \hspace{1pt} \dsx 
      \,\sqrt{|\Gaux|}\;
     \big( \stR(\Gaux) - \CCf \big)
     + \int \!d\taux \hspace{1pt} \dsx 
     \: \partial_\Zm \HTf^\Zm \opB \hspace{1pt} \CCf
     \,,
  \eea
where configuration-space variables are the components of the bulk auxiliary metric $\Gaux_{\Zm\Zn}(\taux,\sx)$, cosmological-constant field $\CCf(\taux,\sx)$ and the auxiliary vector field $\HTf^{\Zm}(\taux,\sx)$.
The operator $\opB(\argsqrg)$ is a local in time but in general spatially nonlocal linear invertible operator of the form
  \beq{def_opB}
   \opB \HIDE{(\argsqrg)} = \opE \HIDE{(\argsqrg)} \FF\HIDE{(\argsqrg)} \sqrg
   \,,
  \eeq
where the spatial delocalization is encoded in the invertible operator $\opE$,
the explicit definition of which is given in Section \ref{SSect:ConstraintBasisChange} and basic properties are\HIDE{ summarized and proved} discussed in the Appendix \ref{ASSect:opE_properties}. Here we just note that though the choice of $\opE$ is not entirely unique, there is a natural one --- $\opE(\WW)$,\HIDE{ (\ref{opE_kernel}),} (\ref{opE_left-right_action}), functionally dependent on $\WW(\argsqrg)$ only and expressible as a sum of two projectors. This choice is all the more convenient because for the subfamily of models with
  \beq{AMG_part_case}
    \WW \equiv \wwc = \const %\quad (\wwc \neq 0, -1)
    \qquad \Leftrightarrow \qquad
    \FF \propto {\sqrg\vphantom{|^l}}^{\wwc}
    \,, \qquad [\text{\wGUMG}]
  \eeq
operator $\opE(\WW)$ becomes the {identity} operator, which makes the action and its equations of motion manifestly local. We will henceforth refer to the subfamily (\ref{AMG_part_case}) as {\wGUMG} and use it to highlight key features of the gauge structure and dynamics of (\ref{Alt_Action_GUMG}) without technical complications associated with spatial nonlocality.
In the\HIDE{ exceptional} particular case of unimodular gravity, corresponding to
  \beq{UMG_part_case}
    \WW \equiv - 1 %\quad (\wwc \neq 0, -1)
    \qquad \Leftrightarrow \qquad
    \FF \propto {\sqrg\vphantom{|^|}}^{-1}
    \,, \qquad [\text{\UMG}]
    %% $\WW \equiv -1$, $\FF=\sqrg^{-1}$
  \eeq
the entire operator $\opB(\argsqrg)$, (\ref{def_opB}), becomes the {identity} operator, which reduce action (\ref{Alt_Action_GUMG}) to the Henneaux--Teitelboim covariant {\UMG} action \cite{Henneaux:1989zc}.

We present a detailed analysis of dynamical properties of the theory (\ref{Alt_Action_GUMG}). In Section \ref{SSect:DynamicalProperties_GUMG} we also discuss the consequences of spatial delocalization\HIDE{ in general {\GUMG} models}. In particular, the effect\HIDE{ of the spatial nonlocality} may be seen in the on-shell behaviour of the cosmological-constant field $\CCf$ (\ref{ccf_onshell_GUMG}), which feels the global properties of $\WW(\argsqrg)$,
\vspace{-1mm}
  \beq{ccf_onshell_GUMG_Intro}
    \CCf
    %\eomeq \sqrg^{-1} \FF^{-1} \opinvE \cco
    \,\propto\,
    \sqrg^{-1} \FF^{-1} \frac{\WW^{-1}}{\,\hmg{\WW^{-1}}\,\vphantom{I^{|^I}}} \HIDE{\, \cco}
    \, ,
 \vspace{-2mm}
  \eeq
where\HIDE{  $\cco$ is some invariant constant parameter \,and\,} $\hmg{\WW^{-1}} \teq \hmg{\WW^{-1}}(\taux)$ is the average\HIDE{ value} of the $\WW^{-1}(\taux,\sx)$ over the $\,\taux \teq \const\,$ spatial hypersurface\HIDE{: $\hmg{\!\WW^{-1}\!} \teq \sint \WW^{-1}/\sint 1$}.

%\newpar

The key step linking the original and alternative representations of {\GUMG} models is the consistent implementation of time reparametrization, which we carry out in the canonical formalism. While conceptually this is similar to the procedure of Henneaux and Teitelboim for unimodular gravity \cite{Henneaux:1989zc}, a direct generalization is applicable only for the {\wGUMG} subfamily of models. The naive transfer of the mechanical parametrization \cite{Henneaux:1992ig} to a field theory generally does not provide an equivalent theory with local constraints.
At the same time, due to the specific constraint structure of {\GUMG} models,
a functionally complete modified parametrization constraint can still be recovered. This may be achieved by a mechanical-like parameterization using local in time but spatially global fields, then by introducing a purely gauge sector of average-free\HIDE{ (``inhomogeneous'')} auxiliary fields, and finally making a linear transformation of constraints.
This approach localizes the canonical phase space, potentially keeping simple spatial nonlocality in the constraint structure encoded by the operator $\opE$.
We provide a detailed discussion\HIDE{ of the peculiarities} of the consistent parametrization\HIDE{ procedure, in particular,} analyzing the emergence and nature of this spatial nonlocality.

We carefully examine the gauge structure of the resulting action and show that it also exhibits spatial nonlocality. In particular, it manifests in the fact that time reparametrization acts homogeneously on the metric sector. However, the spatial nonlocality of the gauge structure arises for reasons only indirectly related to those that generate $\opE$: the nonlocality persists even in manifestly local {\wGUMG} models, reflecting the mixed-class functional nature of constraints. In an appropriate constraint basis, only the spatially homogeneous component of parametrization constraint is first-class, whereas the complementary average-free\HIDE{ (``inhomogeneous'')} component is second-class. Among spatial diffeomorphisms, only those with divergence-free parameters are first-class. In {\GUMG} models with nonconstant $\WW(\argsqrg)$, the representation of gauge transformations becomes more involved due to the spatial nonlocality in the action. However, the gauge structure across all nonexceptional {\GUMG} theories remains qualitatively similar.

\newpar

The article is organized as follows.

In Section \ref{Sect:GUMG}, we briefly review the canonical formulation of generalized unimodular gravity\HIDE{, performing concise hamiltonization,} and complete it with comments on\HIDE{ particular and} exceptional cases. This section establishes the canonical extended action and constraint structure as the starting point for\HIDE{ implementing time} parameterization.

In Section \ref{Sect:AltAction_wGUMG}, we focus on the \wGUMG\ subfamily  (\ref{AMG_part_case}), deriving its\HIDE{ local} Henneaux--Teitelboim-like action and analyzing its dynamical properties and gauge structure. We take the opportunity\HIDE{ to postpone the discussion of modified parametrization,} to highlight various \GUMG-specific dynamical features, and to examine the gauge-structure spatial nonlocality for a manifestly local model, free from the nonlocalities introduced by $\opE$.

In Section \ref{Sect:AltActionGenGUMG}, we extend the analysis to general {\GUMG} models, with particular emphasis on the spatial delocalization associated with $\opE$ and the consistent parameterization.

The Appendices provide additional technical and conceptual details. 
In Appendix \ref{ASect:opE_properties} we summarize the properties of the delocalization operator $\opE$ and related operators. 
In Appendix \ref{ASect:gauge_calcs} we take out some calculational details concerning gauge structures\HIDE{ and hidden details of the gauge invariance check}. 
In Appendix \ref{ASect:Homogeneous_Parameterization} we present the Lagrangian form of the theory when parameterized only by spatially-homogeneous fields.
In Appendix \ref{ASect:NonCompact} we discuss the applicability of the developed technique and results to the case of manifolds with noncompact spatial sections.
In Appendix \ref{ASect:PB_GUMG_Tables} we collected together canonical structure relations for {\wGUMG} and general {\GUMG} theories.

%%In Section \ref{ASect:parametrization}, we comment on\HIDE{ the issues of} the canonical implementation of time parametrization in field theory.

%%%============================================================================%%%
%%%============================================================================%%%
%%%============================================================================%%%

\newpage
\section{Canonical Form of Generalized Unimodular Gravity}
 \hspace{\parindent}
  \label{Sect:GUMG}
In this section, we discuss the canonical formalism for {\GUMG} theory (\ref{Action_GUMG_L}) and its constraint structure. This largely follows the observations of \cite{Barvinsky:2019agh}, but is rewritten here using concise hamiltonization, where canonically conjugate momenta are introduced only for the induced metric on constant-time hypersurfaces  $\g_{\zm\zn}$, the only dynamical component of the spacetime metric $\GBulk_{\Zn\Zm}$ (\ref{GBulk_metric_ADM}).

The {\GUMG} action (\ref{Action_GUMG_L}), reduced with respect to fields $\lmrNl$ and $\Nl$, which form a pair of auxiliary variables\footnote{By \emph{auxiliary variables} we mean a subset of variables that can be expressed algebraically in terms of other fields by solving a subsystem of variational equations with respect to only these variables. The full action and the action reduced by eliminating the auxiliary variables are physically equivalent.}, up to total derivatives reads
  \bea{ADM_Action_GUMG}
    \SSS^{\iGUMG}_{red}[\g,\Ns] &=&
    \!\int \!d\tx \hspace{1pt} \dsx %\int dt \,\dsx
    \,\sqrg \FF(\argsqrg) \big(
      \rK^\zm_{\,\zn}\rK^\zn_{\,\zm} - \rK^\zm_{\,\zm}\rK^\zn_{\,\zn} + \sR(\g)
    \big)
    \,.
  \eea
The configuration space of the reduced theory (\ref{ADM_Action_GUMG}) consists of the induced metric $\g_{\zm\zn}$ and shift functions $\Ns^\zn$, (\ref{GBulk_metric_ADM}). $\sR(\g)$ denotes the scalar curvature of the induced metric on constant-time sections and the ``restricted''
 %%\footnote{These are standard extrinsic curvatures \cite{Misner:1973prb} with lapse function $\Nl$ substituted with $\FF(\g)$ in denominator of the prefactor.}
extrinsic curvature is
    \begin{eqnarray}
     \label{K_GUMG}
     \rK_{\zm\zn}
     = \frac{1}{2\FF(\argsqrg)} \big( \tdot{\g}_{\zm\zn} - \Ns_{\zm; \zn} - \Ns_{\zn; \zm} \big)
     \,.
    \end{eqnarray}
Spatial indices, denoted by Latin letters, are raised and contracted by the inverse spatial metric $\g^{\zm\zn}$. A \emph{dot} denotes for time derivative, a \emph{semicolon} indicates spatial covariant derivatives, while a \emph{colon} represents partial derivatives with respect to the\HIDE{ corresponding} spatial coordinate.

%%-------------------------=%        ******         %=-------------------------%%
%%-------------------------=%        ******         %=-------------------------%%

  \subsection{Canonical extended action}
   \label{SSect:GUMG_Canon_Actions}
    \hspace{\parindent}
  \newcommand{\fProj}{\op{\Pi}}
  \newcommand{\sLapl}{\vartriangle}
Concise transition to the canonical theory requires\HIDE{ introducing} conjugated momenta only for $\g_{\zm\zn}$,
 %% which have velocities in (\ref{ADM_Action_GUMG})
  \bea{momenta_GUMG}
    \pg^{\zm\zn} &\!=\HIDE{\defeq}\!& \frac{\svar\LLd}{\svar \dot{\g}_{\zm\zn}}
    \:=\: \sqrg \big( \rK^{\zm\zn}-\g^{\zm\zn} \rK^{\zk}_{\,\zk} \big)
    \,.
  \eea
After the Legendre transform, the action (\ref{ADM_Action_GUMG}) acquires \emph{primary} canonical form
  \bea{primAction_GUMG}  % was {GUMGHamiltonianAction}
    &&\SSS_{\iP}[\g,\pg,\Ns]
    = \int \!d\tx \hspace{1pt} \dsx %\int dt\, \dsx
    \,\Lib \pg^{\zm\zn}\dot{\g}_{\zm\zn}
    - \FF(\argsqrg) \Hl(\g,\pg) - \Ns^\zn \Hs_\zn(\g,\pg)
    \Rib
    \,.
  \eea
The phase-space structures $\Hl(\g,\pg)$ and $\Hs_\zn(\g,\pg)$ are identical to the Hamiltonian and momentum constraints in canonical Einstein gravity \cite{Misner:1973prb},
 \bea{}
    {\Hl}(\g,\pg) &=& \frac{1}{\sqrg} \Big( \pg^{\zm\zn}\pg_{\zm\zn} - \frac{1}{\Ddim{-}2}\trpg^2 \Big) - \sqrg\,\sR(\g) %%%-2\Lambda
    \,,
    \label{BGUMGHamiltonianStructure}
    \\
    {\Hs}_\zn(\g,\pg)  &=& -2\, \g_{\zn\zm} {\pg^{\zm\zk}}_{\!;\zk} \,,
    \label{BGUMGMomentaConstraints}
\eea
where \,$\pg_{\zm\zn} \defeq \g_{\zm\zk}\pg^{\zk\zl}\g_{\zl\zn}$ \,and\, $\trpg \defeq \g_{\zm\zn}\pg^{\zm\zn}$. The presence of constraint structures from the parental gauge theory is generically expected in restricted theories\footnote{We refer to Einstein gravity as \emph{parental} to {\GUMG} in the sense of \cite{Barvinsky:2022guw} as the gauge theory from which the considered \emph{restricted} model is obtained by imposing algebraic restriction, which fixes some gauge and/or non-gauge degrees of freedom.}, though their roles may change. In (\ref{primAction_GUMG}), the second structure (\ref{BGUMGMomentaConstraints}) is still defines constraint
 \beq{Hs_constraint}
  \Hs_\zn (\g,\pg) = 0
  \,,
 \eeq
entering the canonical action with Lagrange multipliers $\Ns^\zm$, while the first structure (\ref{BGUMGHamiltonianStructure}) appears as a factor in the nontrivial Hamiltonian density $\FF(\argsqrg) \Hl(\g,\pg)$.
In the concise setup, the phase-space variables are ($\g_{\zm\zn}, \pg^{\zm\zn}$), and the Poisson bracket of\HIDE{ two phase-space} functions $A(\g, \pg)$, $B(\g, \pg)$ is given by
 $ % \bea{BriefPB}
   \PB{A}{B}
   = \sint \big(\sVDer{A}{\g_{\zm\zn}(\sx)}\sVDer{B}{\pg^{\zm\zn}(\sx)} - \sVDer{A}{\pg^{\zm\zn}(\sx)}\sVDer{B}{\g_{\zm\zn}(\sx)}\big)
   \,,
 $ % \eea
where $\sint ...$ denotes the \emph{spatial} integration over the $\tx = \const$ hypersurface\HIDE{ of spacetime}, $\sint ... \defeq \int\limits_{\tx = \const} \!\!\! \dsx \,...\,$.

%\newpar

The Dirac consistency procedure \cite{Bergmann:1949zz,Dirac:1950pj,Henneaux:1992ig}, which in its first step requires the preservation of primary constraints $\Hs_\zm=0$ in time,
  \beq{consistency_primary_GUMG}%{consistency_secondary_GUMGr}
    \dot{{\Hs}}_\zn
    \;\eomeq\;
    \big(\WW\FF\Hl\big)_{,\zn} \! + {\Ns^\zm}_{,\zn}\Hs_\zm + \big(\Ns^\zm\Hs_\zn\big)_{,\zm}
    \,\weq\, 0
    \;,
  \eeq
reveals secondary constraints
%%(for details see Section \ref{ASect:GUMG_Algebra_Details}). %% deprecated
  \beq{secondary_GUMG}%{tertiary_GUMGr}
  %  \boxed{
   \CT_{,\zm}(\g,\pg) \;=\; 0
   \,,
  \eeq
where we introduce the phase-space function
 \beq{def_CT}
    \CT(\g,\pg) \;\equiv\; \WW(\argsqrg) \FF(\argsqrg) \Hl(\g,\pg)
    \,.
  %  }%\nonumber\\
  \eeq
This is where the equation-of-state parameter of the cosmological perfect fluid
 $ % \bea{def_WW}
  \WW(\argsqrg) = \tfrac{d \ln{\FF}}{d \ln\!{\sqrg}}
  \,,
 $ % \eea
 (\ref{def_WW_GUMG}), is explicitly incorporated into the canonical formalism. If $\WW(\argsqrg) \teq 0$ for certain values of $\sqrg$, it would induce constraint branching at this step\HIDE{and complicate the constraint structure}. Though it does not affect the main result, to avoid further complication of the constraint structure\HIDE{ by default} we assume sign-definiteness of $\WW(\argsqrg)$ and thus monotonicity of $\FF(\argsqrg)$.

The new constraint (\ref{secondary_GUMG}) takes the form of a gradient, where the derivative is the partial spatial derivative. Since $\CT\HIDE{(\g,\pg)}$ is not a scalar function\footnote{
 \HIDE{In (\ref{def_CT})} $\Hl(\g,\pg)$ is a scalar density of weight $1$ (like $\sqrg$). Structures $\FF(\argsqrg)$, $\WW(\argsqrg)$ generically do not have a definite weight, unless they are homogeneous in $\sqrg$.
}
under spatial diffeomorphisms, this suggests a partial breaking of spatial covariance. The latter is confirmed by the analysis of the constraint structure in Section \ref{SSect:GUMG_Constraint_Structure}, which shows that the longitudinal part of the spatial diffeomorphisms is a second-class constraint.

Notably, the constraint (\ref{secondary_GUMG}) is equivalent to \emph{functionally incomplete} scalar constraint $\inh{\CT}=0$,
 \beq{secondary_GUMG_inh}
   \CT_{,\zn}\HIDE{(\g,\pg)} = 0
   \qquad\Leftrightarrow\qquad
   \inh{\CT}\HIDE{(\g,\pg)} = 0
   \,,
 \eeq
where $\inh{\CT}$ denotes the spatially average-free\HIDE{ (``{inhomogeneous}'')}  part of\HIDE{ the function} $\CT$\HIDE{, (\ref{def_CT})}. Any local function $f(\tx,\sx)$\HIDE{ dependent on spatial coordinates $\sx^\zm$} is uniquely decomposed into\HIDE{ the sum of} average (\HIDE{spatially }homogeneous) and average-free (inhomogeneous) components
 \beq{def_inh_hmg_sVol}
   f(\tx,\sx) = \hmg{f}(\tx) + \inhw{f}(\tx,\sx) \;:
   \qquad
   \begin{array}{|l}
     \; \hmg{f}(\tx) \defeq \frac{1}{\sVol} \sint f(\tx,\sx)
     \,, \vphantom{\big|^l}
     \quad \;
     \sVol \defeq \sint 1
     \,,
     \\
     \;
      \inhw{f}(\tx,\sx) \defeq
     % f(\tx,\sx) - \frac{1}{\sVol} \sint f(\tx,\sx)
     f(\tx,\sx) - \hmg{f}(\tx)
     \,, \vphantom{{\big{|}^l}^i}
     \\
   \end{array}
   \qquad
  %% \sVol = \sint 1 \,,
 \eeq
where integrals $\sint$\HIDE{ without an explicit measure} denote spatial integration and $\sVol$ is the background (coordinate) volume of the $\tx \teq \const$ \, hypersurface\footnote{
  In {\GUMG} models, the background spatial volume is gauge-invariant on the branch where the decomposition (\ref{def_inh_hmg_sVol}) applies. This is because only \emph{transverse} diffeomorphisms (\ref{transverse_diffeomorphisms}), associated with unit-Jacobian spatial coordinate transformations, remain as gauge symmetries, thereby preserving spatial integrals with a fixed measure.
 \if{
   Note that in\HIDE{ nonexceptional} {\GUMG} models, the background spatial volume is gauge-invariant on the branch, where the decomposition (\ref{def_inh_hmg_sVol}) is actually employed, since there only the transverse spatial diffeomorphisms are gauge symmetries.
   %%Transverse diffeomorphisms preserve spatial integrals for any fixed measure, as they correspond to the coordinate transformations with unit Jacobian.
   These\HIDE{Transversal diffeomorphisms} correspond to spatial\HIDE{ time-dependent} coordinate transformations with unit Jacobian and\HIDE{ thus} preserve spatial integrals for any fixed measure.
 }\fi
}.
The two subspaces are orthogonal, which implies $\sint f \, h \,\teq\, \sVol \hmg{f} \, \hmg{h} {\,+ } \sint \inh{f}\: \inh{h} $ and $ \sint \hmg{f}\, \inh{h} \,\teq \hmg{f} \sint \inh{h} \,\teq\, 0$.
While the explicit form of definitions on the right of (\ref{def_inh_hmg_sVol}) is straightforward for compact spatial sections, a similar projection decomposition is applicable for noncompact spatial sections as well (see Appendix \ref{ASect:NonCompact}).
\if{
 \footnote{
  The accurate formulation for the noncompact case needs specifying the asymptotic conditions of fields, entering the decomposed combinations.
  The important actually used features of the decomposition (\ref{def_inh_hmg_sVol}) are that the average-free quantities vanish upon either spatial integration or averaging\HIDE{ and is orthogonal with homogeneous part} and, if needed, can be represented as the divergence of a spatial\HIDE{ tangent} vector field.
  %
  %% We extensively use such decomposition in what follows.
  %% Here the functionally incompleteness of the secondary constraint is hidden by redundancy of the gradient-type  constraint. However, this would no longer be convenient in the context of the alternative formulation.
  %% We put the definition here to stress that this source of the spatial nonlocality is present already in the original formulation.
  %% We comment on this more in the Appendix \ref{ASect:Noncompact}.
 }
}\fi 

\newpar

Preservation of the secondary constraint $\CT_{,\zn}\weq0$\HIDE{, (\ref{secondary_GUMG}),}  in time leads to the condition
  \bea{consistency_secondary_GUMG}%{consistency_secondary_GUMG}
    \dot{\CT}_{,\zn}(\g,\pg)
    \;\eomeq\;
     \partial_\zn \big(\CT\,\CS\big)
        + \partial_\zn \big(\WW\partial_\zm(\FF^2 \g^{\zm\zk} \Hs_\zk)\big)
        + \partial_\zn \big(\Ns^\zm \CT_{,\zm}\big)
    \,\weq\, 0
    \,,
  %%%\ref{dot_CTa_GUMGr}
  \eea
where
\vspace{-2mm}
  \bea{def_GUMG_CS} % (\ref{def_GUMG_CS})
    & \CS(\g,\pg,\Ns)
    \;\defeq\;
    \OOmega(\argsqrg)\,\Dvg[\zn]{\Ns} - \TTheta(\argsqrg)\,\trpg\;,
    \vphantom{\Big{|_{\big|}}}
    \\
    %\nonumber\\
 %  \eea
%%with $\trpg=\g_{\zm\zn}\pg^{\zm\zn}$
% and
 % \bea{def_OOmega_TTheta}
    & \displaystyle{
    \OOmega(\argsqrg) \defeq  \TDer{\ln{\WW}}{\ln\!\sqrg} + \WW + 1\;,
    \qquad
    \TTheta(\argsqrg) \defeq  \frac{1}{\Ddim{-}2}\, \TDer{\ln{\WW}}{\ln\!{\sqrg}}\,\frac{\FF}{\sqrg}
      }
    \,.
    \label{def_OOmega_TTheta}
    %\nonumber\\
  \eea
The two latter terms in\HIDE{ the middle side of the consistency equation} (\ref{consistency_secondary_GUMG}) vanish on constraint surface $\Hs_\zn \teq 0$, $\CT_{,\zn} \teq 0$. So, given $(\CT\hspace{1pt}\CS)_{,\zn} = \CT\hspace{1pt}\CS_{,\zn} {+\,}\CT_{,\zn}\CS$, the  consistency condition (\ref{consistency_secondary_GUMG}) generates the constraint
  \bea{tertiary_GUMG}%{quaternary_GUMGr_T}
   % \boxed{
     \CT\HIDE{{\cdot}}\,\CS_{,\zn} = 0
     \,.
   % }%\nonumber\\
  \eea
Importantly, $\CT$\HIDE{ in the left-hand side} does not vanish on the constraint surface --- only its average-free part, $\inh{\CT}$, does.

At this point a peculiar feature of generalized unimodular gravity emerges. The consistency condition for the secondary constraint takes a {factorizable} form, signaling a bifurcation of the constraint set into two dynamical branches,
 \beq{branches_GUMG}
  \begin{array}{lll|lll}
     & \text{\emph{Non-GR branch}} &&& \quad \text{\emph{GR branch}} &\!  \\
    \hline
     & \Hs_\zn=0 \,, &&& \Hs_\zn=0 \,, &\! \\
     & \CT_{,\zn}=0 \,, \;\;(\CT\neq0) &&& \CT =0 \;\;(\Leftrightarrow\,\Hl=0) \,. &\! \\
  \end{array}
 \eeq
The constraint set of the second branch coincides with the constraint set of the parental theory --- Einstein's general relativity.
  %%\footnote{The restoration of the ``parental'' branch with dynamic and gauge structure of the parental theory is the common behaviour in the restricted theories which are obtained from some parental gauge theory.}
For this reason, we refer to the second branch as the \emph{GR branch} and the\HIDE{ physically more interesting} first as \emph{non-GR branch}.
 %%\HIDE{\cite{Barvinsky:2022guw}.}

\newpar

Within the concise hamiltonization approach, the Dirac consistency procedure ends at this step, as preservation in time of the constraint $\CT\,\CS_{,\zn} = 0$ does not introduce new restrictions on the constraint surface. But the underlying mechanisms differ between the two branches.

The GR branch ($\Hs_\zn \teq 0$, $\CT\HIDE{ \sim\Hl} \teq 0$) is preserved in time since
  \bea{dot_CT_GUMGr} %%\ref{dot_CT_GUMGr}
    \dot{\CT}(\g,\pg)
   % &=& \PB{\CT}{\sint \FF\Hl } + \PB{\CT}{\sint \Ns^\zm \Hs_\zm} \;,
   % \\
    &=& \CT\,\CS
        + \WW\partial_\zn(\FF^2 \g^{\zn\zm} \Hs_\zm)
        + \Ns^\zn \CT_{,\zn}
    \weqwrt{\Hs_\zn,\CT}  0
    \,.
  \eea
The preservation is predictable, because the constraints are equivalent to the constraint set of the parental gauge theory with closed algebra ($\Hs_\zn {=\,} 0$, $\Hl \teq 0$), and the Hamiltonian $\FF\Hl$ on this branch weakly vanishes.

In contrast, the condition $\CS_{,\zn} {=\,} 0$, (\ref{def_GUMG_CS}), of the non-GR branch ($\Hs_\zn {=\,} 0$, $\CT_{,\zn} {=\,} 0$, $\CT \neq 0$)
does not impose additional constraints on the phase space since it restricts the divergence of the Lagrange multipliers\footnote{
 The solution (\ref{DvgNs_from_CS_GUMG}) \HIDE{for the Lagrange multipliers }implicitly assumes that the derivative characteristic function $\OOmega(\argsqrg)$, (\ref{def_OOmega_TTheta}), should be sign-definite. The assumption $\OOmega(\argsqrg)\neq0$ imposes a convexity inequality on the characteristic function $\FF(\argsqrg)$ and prevents further branching of the canonical constraints structure. As discussed in Section \ref{SSect:Except_Partic_Cases}, {\GUMG} models from the exceptional subfamily $\OOmega(\argsqrg)\tequiv0$ are dynamically too restrictive, if not to say pathological, except important regular {\UMG} case (\ref{UMG_part_case}), for which $\TTheta\tequiv0$ and $\CS$ does not appear.
}
\vspace{-2mm}
  \bea{DvgNs_from_CS_GUMG} %{DvgNs_from_CS_LRG} \COPY
   \Dvg[\zn]{\Ns}
    \,\eomeq\, \inh{U}_0(\g,\pg)
    \,\defeq\, \HIDE{+} \OOmega^{-1}\TTheta \,\trpg
      - \OOmega^{-1}\frac{\sint \OOmega^{-1}\TTheta \,\trpg}{\sint \OOmega^{-1}}
    \,.
  \eea
For {\wGUMG} subfamily $\WW \tequiv \wwc \teq \const$, (\ref{AMG_part_case}), the average-free structure $\inh{U}_0(\g,\pg)$ vanishes, as  $\OOmega=\wwc{+}1$ and $\TTheta \propto \TDer{\WW}{\sqrg} \equiv 0$. %% (\ref{def_OOmega_TTheta}).
In the general {\GUMG} case, $\inh{U}_0(\g,\pg)$ remains nontrivial.
  %% and can be expressed as the spatial divergence of a vector quantity $\inh{U}_0(\g,\pg)=\partial_\zn {U}^\zn_0(\g,\pg)$. The latter structure will play a role in defining the first-class subalgebra of the parameterized theory.

\newpar

Finally, the \emph{extended} canonical action of the {\GUMG} theory may be written as
 \bea{extAction_GUMG} % (\ref{extAction_GUMG}) %was (\ref{extAction_BLRG})
    \SSS_{\iE}[\g,\pg,\Ns,\lmCT]
    &=& \!\int \!d\tx \hspace{1pt} \dsx %\int dt\, \dsx
    \,\Lib
     \pg^{\zm\zn}\dot{\g}_{\zm\zn}
     - \FF \Hl
     - \Ns^\zn \Hs_\zn
     - \lmCT^\zn \CT_{,\zn}
    \Rib
    \,.
  \eea
While it looks as the extended action for the\HIDE{ physically interesting} non-GR branch, it \HIDE{encapsulates}describes the dynamics on both branches, explicitly incorporating their total constraint surface.
%% The formal uniqueness of the form of the extended action for both branches is the benefit of the concise hamiltonization.

\if{
  \footnote{
   This form of the action is granted by using brief hamiltonization, which in contrast to the standard Dirac one used in \cite{Barvinsky:2019agh}, allows to describe the complete constraint surface (including both branches) by the set of conditions (\ref{Hs_constraint}) and (\ref{secondary_GUMG_inh}). In the standard approach with different number of constraints on branches one could also write down the action, but generically it is impossible without explicit introduction of factorizable (nonconstant-rank) constraints which has its own shortcomings. %Strictly speaking this is the extended action for the non-GR branch since it explicitly contains the complete set of its canonical constraints. For the GR branch with the constraint $\Hl=0$ this is formally not the extended form. However, the subspace of the constraint surface correspondent to GR branch is contained in the constraint surface $\Hs_\zm=0$, $\CT_{,\zm}\!=0$ of (\ref{extAction_GUMG}) and this is enough for the purposes of the discussion.}
  }
}\fi

\if{
  %\footnote{
  Here is the point where imposition of the third bound on the type of the characteristic function $\FF(\argsqrg)$ would be in place. The constraint structure avoids additional branching and specific submanifolds if on assume $\OOmega(\argsqrg)\neq0$ on the domain $\sqrg>0$. This allows to completely solve $\CS_{,\zm} = 0$ in terms of $\Dvg{\Ns}$, (\ref{DvgNs_from_CS_GUMG}), which avoids new constraints on this branch. Later the same assumption will save from the additional complication of the rank irregularity of the Poisson bracket matrix of constraints. However, the exceptional unimodular case, which falls into $\OOmega\equiv0$ subfamily, will also be described by the alternative action as the particular case\HIDE{, though its constraint structure and physical content differ from generic {\GUMG} models}.
  %}
}\fi

\if{
  %\footnote{
   %% The stability of the parental branch is predictable since the Hamiltonian $\FF\Hl$ vanishes and the system of constraints repeats that of the parental gauge theory.

  Formally for the GR branch there appears the new constraint $\Hl=0$. Which we will not indicate in the action explicitly. Thus, the canonical action (\ref{extAction_GUMG}) with explicitly specified constraints $\Hs_\zm=0$, $\CT_{,\zm}=0$, which we in what follows refer as the \emph{extended} action, is truly extended for the non-GR branch, but not for the GR branch. However, we prefer to deal with the single action functional for all configuration space. And this will not lead to inconsistencies in what follows: the two-branch structure of physical space will be preserved, which we periodically check.
  %}
}\fi

%%-------------------------=%        ******         %=-------------------------%%
%%-------------------------=%        ******         %=-------------------------%%

%\newpage
     %\subsection{Poisson brackets of {\GUMG} constraints (2020)}
     \subsection{Constraint structure}
     \label{SSect:GUMG_Constraint_Structure}
      \hspace{\parindent}
 \if{ %% INIT ver
A review of canonical {\GUMG} would be incomplete without mentioning the constraint structure. In the basis $\Hs_\zm,\CT_{,\zm}$ the Poisson brackets of the constraints are
  \bea{GUMG_Algebra_offshell} %From GUMGr_NDV file \ref{PB_GUMGr_B1}
   \begin{array}{lll}
     \PB{\sint {\xi}^\zn\Hs_\zn}{\sint {\eta}^\zm\Hs_\zm}
     & = & \sint \LieB{\zn}{\xi}{\eta} \Hs_\zn
    \,, \vphantom{\big|^I}\\
     \PB{\sint {\xi}^\zn\Hs_\zn}{\sint \eta^{\zm}\CT_{,\zm}}
     & = &
     \sint {\xi}^\zn \Dvg{\eta}\, \CT_{,\zn} + \sint \Dvg{\xi} \Dvg{\eta} \,\OOmega\, \CT
    \,, \vphantom{\big|^I}\\
     \PB{\sint {\xi}^{\zn}\CT_{,\zn}}{\sint {\eta}^{\zm}\CT_{,\zm}}
     & = & \sint \big(\Dvg{\xi} {\Dvg{\eta}}_{,\zn} {-} {\Dvg{\eta}}{\Dvg{\xi}}_{,\zn}\big) \g^{\zn\zm} \FF^2 \WW^2 \Hs_\zm
    \,, \vphantom{\big|^I}\\
   \end{array}
  \eea
where\HIDE{ structures} $\WW(\argsqrg)$ and $\OOmega(\argsqrg)$ are defined in (\ref{def_WW_GUMG}) and (\ref{def_OOmega_TTheta})\HIDE{ correspondingly}, and $\xi^\zn(\sx)$, $\eta^\zm(\sx)$ are arbitrary test functions with $\Dvg{\xi}\defeq \Dvg[\zk]{\xi}$ --- a compact notation for the divergence of a spatial vector.
 %%Integral form allows to represent the field-theoretic Poisson bracket relations in compact form.
 %%\footnote{Remind that integrals without specifying the measure are spatial integrals over  constant-time hypersurfaces.}

At the complete constraint surface of the theory ($\Hs_\zn=0$, $\CT_{,\zm}\!=0$) one gets
  \beq{GUMG_Algebra_onshell} %From GUMGr_NDV file \ref{PB_GUMGr_B1}
     \begin{array}{lll}
       \PB{\sint {\xi}^\zn\Hs_\zn}{\sint {\eta}^\zm\Hs_\zm}
       & \weq & 0
       \, , \vphantom{\big|^I}\\
       \PB{\sint {\xi}^\zn\Hs_\zn}{\sint \eta^{\zm}\CT_{,\zm}}
       & \weq & \sint \Dvg{\xi} \Dvg{\eta} \,\OOmega\, \CT
       \, , \vphantom{\big|^I}\\
       \PB{\sint {\xi}^{\zn}\CT_{,\zn}}{\sint {\eta}^{\zm}\CT_{,\zm}}
       & \weq & 0
       \, , \vphantom{\big|^I} \\
     \end{array}
  \eeq
which shows that rank of Poisson bracket matrix is \emph{nonconstant} on the constraint surface.
}\fi %% END INIT ver
A review of canonical {\GUMG} would be incomplete without mentioning the constraint structure. In the basis $\Hs_\zm$, $\CT_{,\zm}$, the Poisson brackets of the constraints are
  \bea{GUMG_Algebra_offshell} %From GUMGr_NDV file \ref{PB_GUMGr_B1}
   \begin{array}{lll}
     \PB{\sint {\xi}^\zn\Hs_\zn}{\sint {\eta}^\zm\Hs_\zm}
     & {\!=\!} & \sint \LieB{\zn}{\xi}{\eta} \Hs_\zn
     \;\; \weq \;\; 0
    \,, \vphantom{\big|^I}\\
     \PB{\sint {\xi}^\zn\Hs_\zn}{\sint \eta^{\zm}\CT_{,\zm}}
     & {\!=\!} &
     \sint {\xi}^\zn \Dvg{\eta}\, \CT_{,\zn} + \sint \Dvg{\xi} \Dvg{\eta} \,\OOmega\, \CT
     \;\; \weq \;\; \sint \Dvg{\xi} \Dvg{\eta} \,\OOmega\, \hmg{\CT}
    \,, \vphantom{\big|^I}\\
     \PB{\sint {\xi}^{\zn}\CT_{,\zn}}{\sint {\eta}^{\zm}\CT_{,\zm}}
     & {\!=\!} &
   %%  \sint \big(\Dvg{\xi} {\Dvg{\eta}}_{,\zn} \!{-} {\Dvg{\eta}}{\Dvg{\xi}}_{,\zn}\big) \FF^2 \WW^2 \g^{\zn\zm} \Hs_\zm
     \sint \big(\Dvg{\xi} \ader_{\zn} {\Dvg{\eta}} \big) \FF^2 \WW^2 \g^{\zn\zm} \Hs_\zm
     \;\; \weq \;\; 0
    \,, \vphantom{\big|^I}\\
   \end{array}
   \label{GUMG_Algebra_onshell}
  \eea
where $\WW(\argsqrg)$ and $\OOmega(\argsqrg)$ are defined in (\ref{def_WW_GUMG}) and (\ref{def_OOmega_TTheta}), and the vectors $\xi^\zn(\sx)$, $\eta^\zm(\sx)$ are arbitrary local\HIDE{ test} functions.\HIDE{ For readability,} We use the notation $\,\Dvgn{\xi}\defeq \Dvgp[\zn]{\xi}\,$ for the\HIDE{ spatial} vector divergence,
$(f \ader_{\zn} g) \defeq (f g_{,\zn} {-} f_{,\zn} g)$ for the antisymmetric derivative\HIDE{ generalized Wronskian}, and $\LieB{\zn}{\xi}{\eta} \teq {\xi}^\zm\partial_\zm {\eta}^\zn {-} {\eta}^\zm\partial_\zm {\xi}^\zn $ for the Lie bracket.

The right-hand side, revealing nontrivial contributions on the complete constraint surface ($\Hs_\zn \teq 0$, $\CT_{,\zn} \teq 0$), demonstrates that the rank of the Poisson bracket matrix is \emph{nonconstant}\HIDE{ on this surface}. On the \emph{non-GR branch} ($\Hs_\zm \teq 0$, $\CT_{,\zn} \teq 0$, $\CT \tneq 0$), the constraints $\CT_{,\zn}$ and the longitudinal part of $\Hs_\zm$ form a second-class\HIDE{ conjugate} pair, while the \emph{transverse}\footnote{
  We refer to the diffeomorphisms (\ref{transverse_diffeomorphisms}) with divergence-free gauge parameters as \emph{transverse}.
  It is worth clarifying that this decomposition is covariant with respect to a background metric of unit determinant, rather than with respect to the physical metric $\g_{\zm\zn}$.
  %It is worth clarifying that this decomposition is not covariant with respect to the physical metric $\g_{\zm\zn}$, but instead is with respect to a background metric with unit determinant.
  %To avoid confusion, note that this is not a covariant decomposition with respect to the physical metric $\g_{\zm\zn}$, but rather with respect to a background metric of unit determinant.
  % To avoid confusion, note that this does not imply a covariantly divergence-free vector with respect to the physical metric.
  % Rather, the gauge parameter can be viewed as covariantly divergence-free with respect to a background metric of unit determinant.
  % The divergence is covariant rather with respect to a background metric of unit determinant.
} 
momentum constraint,
  \beq{transverse_diffeomorphisms}
    \sint{\xi}_{\isst}^\zn \Hs_\zn\,,
    \qquad \Dvgn{\xi_{\isst}}\equiv \Dvgi[\zn]{\xi_{\isst}} = 0
    \,,
  \eeq
is first-class.
This follows from the fact that form factors\HIDE{ $\xi^\zn$} of $\Hs_\zn$ appear on the right side of (\ref{GUMG_Algebra_onshell}) only through divergence (form factors at $\CT_{,\zn}$ always form divergence combination due to the constraint's gradient form).
On the \emph{GR branch} ($\Hs_\zn \teq 0$, $\CT \teq 0$), the rank of the constraint matrix (\ref{GUMG_Algebra_onshell}) is zero implying that all constraints are first-class.
This offers an alternative viewpoint on the irregular nature of the {\GUMG} constraint structure, as having a Poisson bracket matrix with non-constant rank.

%\newpar

 %% v.2024-03
 %% We will not use the explicit decomposition of the spatial diffeomorphisms constraints into longitudinal and transverse components, but some remarks are in order. 
 %% -= SIDE COMMENT =-
Transverse spatial diffeomorphisms form a gauge subalgebra in the full algebra of spatial diffeomorphisms
 $ % \beq{}
  \PB{\sint {\xi}_{\isst}^\zm\Hs_\zm}{\sint {\eta}_{\isst}^\zn\Hs_\zn} {\;=\;} \sint {\lambda}_{\isst}^\zm \Hs_\zm,
 $ % \eeq
where ${\lambda}_{\isst}^\zn \teq \LieB{\zn}{{\xi}_{\isst}}{{\eta}_{\isst}} $, and $\Dvgi{\lambda_{\isst}} \teq 0$ for $\Dvgi{\xi_{\isst}} \teq 0$ and $\Dvgi {\eta_{\isst}} \teq 0$.
The transverse components of a vector field, ${\xi}_{\isst}^\zn \teq \xi^\zm \tProj_\zm^{\,\zn}$, can be extracted using the projector $\tProj_\zn^{\,\zm} \defeq \delta_\zn^{\,\zm} - \partial_\zn \tfrac{1}{\sLapl_0} \partial_\zk \g_0^{\zk\zm}$, where $\frac{1}{\sLapl_0}$ is the inverse of the symmetric second-order spatial differential operator $\sLapl_0 \!\defeq \partial_\zm \g_0^{\zm\zn} \partial_\zn$, constructed using an auxiliary nondegenerate symmetric form\footnote{
 If the operator $\sLapl_0$ is degenerate, as for compact manifolds, its inverse is\HIDE{ uniquely} defined to be symmetric, with its\HIDE{ left and right} kernels coinciding with those of $\sLapl_0$ (the Moore-Penrose inverse \cite{MoorePenrose}).
 Note that\HIDE{ an operator of the form} $\sLapl_0\!\! \tdefeq \partial_\zm \g_0^{\zm\zn} \partial_\zn$ is covariant under coordinate transformations with the unit Jacobian, corresponding to the gauge freedom of the model.
} $\g_0^{\zm\zn}$.
Transverse diffeomorphisms
 $ % \beq{transverse_diff_projector}
    \sint{\xi}_{\isst}^\zn \Hs_\zn
    \teq
     \sint \xi^\zn
     \tProj_\zn^{\,\zm}
     \Hs_\zm
   % \,.
 $ % \eeq
correspond to spatial coordinate transformations with a unit determinant of the Jacobian.
The second-class longitudinal part of\HIDE{ the spatial diffeomorphism constraints} $\Hs_\zn \teq 0$, can be expressed as
 $ % \beq{diff_longit}
  \tfrac{1}{\sLapl_0} \partial_\zm (\g_0^{\zm\zn} \Hs_\zn)
  {\;=\;} 0
  \,.
 $ % \eeq
 Consequently, the constraint term in the action can be decomposed into the first- and second-class contributions
 $ % \beq{}
  \sint \Ns^\zn \Hs_\zn
  \,\teq\,
  \sint \Ns_{\isst}^\zn \Hs_\zn {\,-\,} \sint \Dvg{\Ns} \tfrac{1}{\sLapl_0} \partial_\zm (\g_0^{\zm\zn} \Hs_\zn)
  \,,
 $ % \eeq
where $\Ns_{\isst}^\zn \tdefeq \tProj_\zn^{\,\zm} \Ns^\zm$ and $\Dvg{\Ns} \tdefeq \Dvgp[\zn]{\Ns}$.
 The parametric dependence of such longitudinal-transverse decomposition on $\g_0^{\zm\zn}$ does not compromise consistency.
 Change of $\g_0^{\zm\zn}$ induces a shift by a transverse combination,
  $\var_{\!\g_0} \tProj_\zn^{\,\zm} =  - \partial_\zn \tfrac{1}{\sLapl_0} \partial_\zk (\var \g_0^{\zk\zl}) \tProj_\zk^{\,\zm} $,
 which\HIDE{ obviously} preserves divergence-free property of the transverse vectors. Likewise,
  $\var_{\!\g_0} \tfrac{1}{\sLapl_0} \partial_\zl \g_0^{\zl\zn}  \Hs_\zn
  =  \tfrac{1}{\sLapl_0} \partial_\zl (\var \g_0^{\zl\zk}) \tProj_\zk^{\,\zn}  \Hs_\zn$,
 which indicates that the ambiguity in the longitudinal second-class component\HIDE{ of spatial diffeomorphisms} stems from arbitrary shifts by first-class transverse diffeomorphisms.

%\newpar

The presence of second-class constraints on the non-GR branch imposes on-shell conditions on the corresponding Lagrange multipliers\HIDE{ (\ref{extAction_GUMG})} \cite{Henneaux:1992ig}. The\HIDE{ linear} system of\HIDE{ equations on the Lagrange multipliers from} consistency equations for constraints yields a particular solution
\if{
\footnote{
 The first relation (\ref{def_Uo}) is the consequence of the conservation equation for the secondary constraint $\CT_{,\zm}=0$, so is directly related to the solution (\ref{DvgNs_from_CS_GUMG})\HIDE{ with trivial transverse part}. In is defined up to arbitrary shifts of the transverse vectors, which corresponds to the shift of the second-class constraint with the first-class constraints. The second condition (\ref{lmv_sol}) comes from the conservation relation of the primary constraint within the extended setup.
}
}\fi
\vspace{-1mm}
 \bea{lm_solutions_nonGR}
  \begin{array}{|llllll}
   \;\Ns^\zn
    &\! \eomeq \!&
    U_0^\zn %%(\g_0,\g,\pg)
    \;\defeq\;
     \g_0^{\zn\zm}\partial_\zm \frac{1}{\sLapl_0}
     \Big( \OOmega^{-1}\TTheta \trpg
      - \OOmega^{-1}\frac{\sint \OOmega^{-1}\TTheta \trpg}{\sint \OOmega^{-1}}
      \Big)
    \,,
    \label{def_Uo}
   %& \quad
   \\
   \;{\lmCT}^\zn
    &\!\eomeq\!&
    0
    \,,
    \label{lmv_sol}
  \end{array}
 \eea
where $\Dvgi[\zn]{U_0}$ coincides with $\inh{U}_0$ from (\ref{DvgNs_from_CS_GUMG}). Since the consistency conditions constrain only divergences\HIDE{ of these vector fields, } $\Dvg[\zn]{\Ns}$ and $\Dvg[\zn]{\lmCT}$, the general solution is obtained by adding arbitrary transverse vectors on the right-hand sides of (\ref{lm_solutions_nonGR}).

As noted in \cite{Barvinsky:2019agh}, the GR and non-GR branches (\ref{branches_GUMG}) describe different numbers of degrees of freedom. While the functional dimensions of the configuration space and constraint surface are\HIDE{ almost} the same, the difference lies in the number of first-class constraints and the corresponding gauge-equivalence classes. The GR branch has ${\Ddim(\Ddim{-}3)}/{2}$ local physical degrees of freedom, which is one less than ${(\Ddim{-}1)(\Ddim{-}2)}/{2}$ on the non-GR branch.

     \subsection{Exceptional cases}
     \label{SSect:Except_Partic_Cases}
      \hspace{\parindent}
To complete the discussion of the standard formulation of {\GUMG}, we note important exceptional cases within the family (\ref{ADM_Action_GUMG}) based to its canonical description.

%% Functional characteristic parameter $\WW(\argsqrg)$ (\ref{def_WW_GUMG}) is the factor in the secondary canonical constraint $\CT_{,\zm}$ (\ref{secondary_GUMG}, \ref{def_CT}), whereas parameter $\OOmega(\argsqrg)$ (\ref{def_OOmega_TTheta}) is the factor at the only on-shell nonvanishing term in the right-hand side of the Poisson bracket matrix (\ref{GUMG_Algebra_onshell}). Their exceptional values characterize the exceptional cases of {\GUMG}.

%% Note also that values of parameter $\WW$, entering the equation of state (\ref{}) of the cosmological perfect fluid differ physically different cosmological scenarios.

\begin{description}
  \item[\HIDE{$\bullet\;$} Unimodular case $\WW \tequiv \HIDE{ \wwc \teq} {-}1$.]

  This is the case of \emph{unimodular} gravity corresponding to $\FF(\argsqrg)  \propto \sqrg^{-1}$, (\ref{UMG_part_case}).
   Its exceptional status in the\HIDE{ \GUMG} family arises from $\OOmega(\argsqrg) \tequiv0$. Relations (\ref{GUMG_Algebra_onshell}) imply that the rank of the Poisson brackets matrix\HIDE{ of the complete set} of constraints is regular and equals zero. Thus, all constraints ($\Hs_\zm$, $\CT_{,\zm}$) are first-class and there is no division of the physical space into two dynamical branches.

  This theory describes $\Ddim(\Ddim{-}3)/2$ local gravitational degrees of freedom, similar to {\GR}\HIDE{ general relativity}.
  \if{
    \footnote{Strictly speaking \UMG\ has one additional spacetime-global parameter --- nonfixed cosmological constant. At the same time it has fixed spacetime-volume global degree of freedom, which is free in general relativity \cite{Barvinsky:2022guw}.}
  }\fi

  \item[\HIDE{$\bullet\;$} Other cases with $\OOmega \tequiv 0$.]

  There are non-unimodular cases where $\OOmega(\argsqrg) \tequiv 0$.
  These cases correspond to restriction functions
    $ %\beq{FF_for_OOmega_0}
     \FF(\argsqrg) \propto  |\, b {\,+} \sqrg^{-1}|
    $ % \eeq
  with some overall positive constant and another arbitrary constant $b\neq0$.
  %% \footnote{Case $b=0$ corresponds to unimodular gravity which is among $\OOmega=0$ family. At the same time limit $b\to \infty$, $b c \to \const $ corresponds to $\WW=0$ exceptional case which corresponds to $\OOmega=1$.}
 %
  Due to $\OOmega=0$ for $\WW \neq \const $ the consistency procedure starting from the secondary constraint (\ref{consistency_secondary_GUMG})\HIDE{,(\ref{def_GUMG_CS})} generate a chain of further constraints:
  $\inh{\frac{\trpg}\sqrg}$, etc. , making these cases exceptional and potentially  pathological.

  %% $b=0$ corresponds to exceptional case of unimodular gravity ($\WW \tequiv -1$). However, the family (\ref{FF_for_OOmega_0}) in the limit $b\to \infty$, $b c \to \const $ contain another exceptional case $\WW \tequiv 0$. It formally corresponds to $\OOmega\tequiv1$, but in this case structure $\OOmega$ does not appear in Poisson bracket matrix and Dirac consistency procedure.

%\end{description}

%\begin{description}
  \item[\HIDE{$\bullet\;$} Dust case $\WW \HIDE{\tequiv \wwc} \teq 0$.]

  This case corresponds to a pressureless perfect fluid with
    $ %\beq{}
     \FF(\argsqrg) \propto 1
     \,.
    $ % \eeq
  %% with some positive constant.
  %% in the restriction conditions,  where $c$ is the arbitrary positive constant.% (which in what follows we let equal $1$.)
   %
  In is notable because, unlike generic {\GUMG}, it possesses a regular canonical constraint structure. The regularity occurs because the consistency condition of the primary constraints\HIDE{ ($\Hs_\zm=0$)} (\ref{consistency_primary_GUMG}) does not generate secondary constraints $\CT_{,\zm}$ (\ref{secondary_GUMG}).
  All\HIDE{ $(\Ddim{-}1)$} primary constraints $\Hs_\zm$ are first-class, exhibiting an unbroken spatial-diffeomorphism gauge invariance.%
  \if{
    \footnote{Somewhat analogous situation occurs for the relativistic point particle when in the Polyakov form one restricts the one-dimensional reparametrization parameter to constant value. Then such theory describes the particle moving along geodesics equation without the constraint equation which restricts the square of 4-velocities.% Which one has when applying correspondent gauge after the complete set of the equations of motion are obtained.
  }
  }\fi
  \if{
    \footnote{There is no reason to isolate the branch with $\Hl=0$ to restore the complete analogy with generic {\GUMG}. However, when external interactions are spacetime-covariant this is formally possible.}
  }\fi

  However, the extended action (\ref{extAction_GUMG}) also describes the \HIDE{cold }dust case. The manually added constraint $\CT_{,\zm} \to \alpha {\Hl}_{,\zm}$, with some constant $\alpha$ (instead of $\WW\FF)$, serves as a partial gauge-fixing condition for the longitudinal part of\HIDE{ spatial diffeomorphisms constraints} $\Hs_\zn$.
  This\HIDE{ theory} describes ${(\Ddim{-}1)(\Ddim{-}2)}/{2}$ local degrees of freedom, as in a general {\GUMG} theory\HIDE{ on the non-GR branch}.

\end{description}

As we have already mentioned, theories in which $\WW(\argsqrg)$ or $\OOmega(\argsqrg)$ vanish at certain values of $\sqrg$ exhibit additional branching\HIDE{ of the physical space} and can therefore also be considered exceptional.

%%-------------------------=%        ******         %=-------------------------%%
%%-------------------------=%        ******         %=-------------------------%%
%%-------------------------=%        ******         %=-------------------------%%

\newpage
 \section{Alternative Action for $\WW = \const$ {\GUMG} Theories}
  \label{Sect:AltAction_wGUMG}
  \hspace{\parindent}
In this section we derive and discuss an alternative Henneaux--Teitelboim-like form of the action for the {\wGUMG} subfamily defined by (\ref{AMG_part_case})
\HIDE{of generalized unimodular gravity theories}
 \beq{AMG_part_case_COPY}
    \WW \,\equiv\, \wwc = \const %\quad (\wwc \neq 0, -1)
    \qquad \Leftrightarrow \qquad
    \Fw \propto \sqrg^{\wwc}
    \,,
  \eeq
keeping in mind the exceptional status of the cases $\wwc \teq {-}1$ and $0$. The constancy of\HIDE{ the equation-of-state parameter} $\WW(\argsqrg)$ implies that the derivative\HIDE{ characteristic} functions (\ref{def_OOmega_TTheta}) also take the simple forms
 \bea{def_OOmega_TTheta_AMG}
  %\OOmega_{\scriptscriptstyle \wwc}(\argsqrg)
  \OOmega(\argsqrg) \,\equiv\, \wwc + 1 \,,\;
  \quad
  %\TTheta_{\!\scriptscriptstyle \wwc}(\argsqrg)
  \TTheta(\argsqrg) \,\equiv\, 0 \,.
 \eea
Following the reasoning that led Henneaux and Teitelboim to derive the so-called covariant action for unimodular gravity \cite{Henneaux:1989zc}, we obtain an alternative action\HIDE{ (\ref{Alt_Action_GUMG})} that extends their approach to {\wGUMG} models. The locality of this procedure within this subfamily makes it a valuable intermediate step toward analyzing the general {\GUMG} case.

%%-------------------------=%        ******         %=-------------------------%%
%%-------------------------=%        ******         %=-------------------------%%

  \subsection{Parameterized canonical action}
   \label{SSect:ParameterizedAction_AMG}
   \hspace{\parindent}
The first step is to parameterize the extended action (\ref{extAction_GUMG}) restricted to {\wGUMG},
  \bea{extAction_AMG} % (\ref{extAction_AMG}) %was (\ref{extAction_BLRG})
    \SSS^{\iwwc}_{\iE} [\g,\pg,\Ns,\lmCT]
    &=& \!\int \!d\tx \hspace{1pt} \dsx %\int dt\, \dsx
    \,\Lib
     \pg^{\zm\zn}\dot{\g}_{\zm\zn}
     - \Fw \Hl
     - \Ns^\zn \Hs_\zn
     - \lmCT^\zn \CTw{}_{,\zn}
    \Rib \,,
  \eea
where\HIDE{ the limitation to \wGUMG\ case manifests in} $\FF \mapsto \Fw \defeq \sqrg^\wwc$ and $\CT \mapsto \CTw \defeq \wwc\Fw\Hl$.
The secondary constraint structure $\CTw$ contains $\Hl$, the Hamiltonian constraint\HIDE{ structure} of general relativity. However, the constraint $\CTw{}_{,\zn}{=\,}0$ is functionally incomplete, as it lacks a spatially homogeneous mode.
 %%As a result, the {\GR} Hamiltonian constraint abandoned by the restriction (\ref{GUMG_restriction}) is not fully restored.
Parameterizing the action with the Hamiltonian $\FF\Hl$ should naturally restore this missing component.

Motivated by \cite{Henneaux:1989zc}, in this section we adopt a naive, local parametrization prescription,
   \bea{paramAction_naive_AMG} %(\ref{paramAction_naive_AMG})
    \SSS^{\iwwc loc}_{par}[\g,\pg,\tf^\io\!,\ptf_\io,\replm^\io\!,\repNs,\replmCTn]
    \!&=&\!\!\! \int \!d\taux \hspace{1pt} \dsx %\int d\taux\, \dsx
    \,\Lib
     \pg^{\zm\zn}\tauxdot{\g}_{\zm\zn}
     \!+ \ptf_\io\,\tauxdot{\tf}^\io\!
     - \replm^\io ( \ptf_\io {+} {\Fw \Hl}\! )
     - \repNs^\zn \Hs_\zn
     \!- \replmCTn^\zn \CTw{}_{,\zn}
     \Rib
     \,,
   \qquad% \nonumber\\
  \eea
which introduces time parametrization into the canonical action
by adding an auxiliary canonical field pair, ${\tf}^\io(\taux,\sx)$ and $\ptf_\io(\taux,\sx)$, and converting the Hamiltonian (shifted by $\ptf_\io$) into a new constraint with the Lagrange multiplier $\replm^\io(\taux,\sx)$.
 %% While this prescription is always valid in mechanics, field theory presents additional challenges, urging to verify the physical equivalence.
While this prescription is always valid in mechanics, it generally fails in field theory, making it necessary to verify physical equivalence.

The extended theory (\ref{extAction_AMG}), with the constraint set $\big(\Hs_\zn, \CTw{}_{,\zn}\big)$, and the parameterized theory (\ref{paramAction_naive_AMG}), with the enlarged constraint set $\big( (\ptf_\io {+} {\Fw \Hl} ), \Hs_\zn, \CTw{}_{,\zn}\big)$, are physically equivalent if the following two conditions are satisfied. First, the rank of the Poisson bracket matrices for both sets must be preserved, meaning the enlarged constraint space introduces a new \emph{first-class} constraint. This constraint may be either the newly added one or a linear combination of the new and original constraints, provided the coefficient of the new constraint is full-rank.
%%\footnote{
 %%This new first-class constraint may be either the newly introduced constraint itself or a linear combination involving the new constraint (with a full-rank coefficient) and the original constraints\HIDE{ from initial system  $\big(\Hs_\zm,\CTw{}_{,\zm}\big)$}.
 %%%%The coefficient at the parametrization constraint in the first-class linear combination should be full-rank\HIDE{ at least on the constraint surface}.
%}.
Second, an admissible (partial) gauge fixing must exist such that, when imposed on the parameterized action (\ref{paramAction_naive_AMG}), it recovers the original extended action (\ref{extAction_AMG}) upon reduction.
%%\footnote{By means of (partial) reduction with respect to two second-class constraints, one of which is the ``correspondence'' gauge-fixing constraint.}
%%(See more detailed discussion in the Appendix \ref{ASect:parametrization}.)

From the Poisson-bracket relations %% which complete (\ref{GUMG_Algebra_offshell})
  \bea{parAMG_Algebra_offshell} %From GUMGr_NDV file \ref{PB_GUMGr_B1}
   \begin{array}{lll}
     \PB{\sint f(\ptf_\io {+} {\Fw \Hl})}{\sint g(\ptf_\io {+} {\Fw \Hl})}
     \; = \; % & {\!=\!} &
     \sint \big( f {g}_{,\zn} \!{-} {g}{f}_{,\zn}\big) \Fw^2 \g^{\zn\zm} \Hs_\zm
     \; \weq \; 0
    \,, \vphantom{\big|^I}
    \\
     \PB{\sint f(\ptf_\io {+} {\Fw \Hl})}{\sint {\eta}^\zn \Hs_\zn}
    \; = \; % & {\!=\!} &
     (\wwc{+}1) \sint f \Dvg{\eta} \, \Fw\Hl
     + \sint f \eta^\zn \, {(\Fw\Hl)}_{,\zn}
     \; \weq \; \frac{\wwc{+}1}{\wwc} \sint f \Dvg{\eta} \, \hmg{\CTw}
    \,, \vphantom{\big|^I}
    \\
     \PB{\sint f(\ptf_\io {+} {\Fw \Hl})}{\sint \eta^{\zn} \CTw{}_{,\zn}}
     \; = \; % & {\!=\!} &
    {-} \wwc \sint \big(f\Dvg{\eta}_{,\zn} {-} \Dvg{\eta} f_{,\zn}\big) \Fw^2 \g^{\zn\zm} \Hs_\zm
     \; \weq \; 0
    \,, \vphantom{\big|^I}
    \\
\end{array}
   \label{parwGUMG_Algebra_onshell}
  \eea
where $f$, $g$, and $\eta^\zn$ are arbitrary local form factors, and $\Dvg{\eta}$ stands for divergence $\Dvg[\zn]{\eta}$, it can be noticed that the newly introduced constraint $(\ptf_\io {+} {\FF \Hl} )$ is not first-class, because it does not commute with the longitudinal part of the spatial diffeomorphism constraint. Therefore, it is necessary to verify whether a linear combination $\CPI \defeq (\ptf_\io {+} \Fw\Hl) {\,+\,} \alpha^\zn\Hs_\zn {\,+\,} \beta\,\inhw{\CTw}$, exists\footnote{
  Here $\alpha^\zm$ and $\beta$ generally are some linear operators. Also we used the equivalence (\ref{secondary_GUMG_inh}).
},
which is in involution with all the constraints in (\ref{paramAction_naive_AMG}).
\if{ %% TWO-COLUMN VERSION
Thus one should check if there is a linear combination of the form \footnote{
 We take into account (\ref{secondary_GUMG_inh}), implying that the constraint $\int \eta^\zm \CTw{}_{,\zm}$ is equivalent to $\sint g \inhw{\CTw}=\sint \inhw{g} \CTw=\sint \inhw{g} \inhw{\CTw}$ with $\inhw{g}=-\Dvg{\eta}$.
}
  \beq{lin_comb_CPI}
    \CPI = (\ptf_\io {+} \Fw\Hl) + \alpha^\zm\Hs_\zm + \beta\inhw{\CTw}
    \,,
  \eeq
where  $\alpha^\zm$, $\beta$ --- generically are some linear operators, so that $\CPI$ is first-class
 %% To avoid dealing with series in powers of $\tf^\io$, $\tf^\io_{,\zm}$ we find particular combination (\ref{lin_comb_CPI}) which solves the first-class conditions
 \beq{}
  \left\{
  \begin{array}{l}
    \PB{\CPI}{(\ptf_\io {+} \Fw\Hl)} \weq 0 \,,\\
    \PB{\CPI}{\Hs_\zm} \weq 0 \,,\\
    \PB{\CPI}{\inhw{\CTw}} \weq 0 \,.\\
  \end{array}
  \right.
 \eeq
}\fi
The general solution involves $\alpha^\zn$ --- an arbitrary transverse vector field $\alpha_{\isst}^{\zn}$ satisfying $\Dvgi[\zn]{\alpha_{\isst}} \teq 0$, and  $\beta$ --- multiplication by $-\tfrac{1}{\wwc}$. The arbitrariness of the transverse combinations of the spatial diffeomorphisms reflects the possibility of adding other first-class constraints.
The simplest particular solution
\if{
  \footnote{
   Since any spatially inhomogeneous function can be represented as divergence of some vector field here\HIDE{ (in the case of secondary constraint being inhomogeneous part of the initial Hamiltonian density)} the expression (\ref{CPI_AMG}) could also be obtained just by subtracting spatially nonlocal total divergence $ \partial_\zm (\g_0^{\zm\zn}\partial_\zn \tfrac{1}{\sLapl_0}\Fw\Hl)$ from the Hamiltonian density without involving the constraints (operators are defined in Section \ref{SSect:GUMG_Constraint_Structure}).
   This is the case of spatially nonlocal minimal time parametrization which is always possible (see discussion in Section \ref{ASect:parametrization}).
  }
}\fi
with $\alpha_{\isst}^{\zn}=0$ implies
  \beq{CPI_AMG}
     \CPI
     \;=\; \ptf_\io + \Fw \Hl - \frac{1}{\wwc}\inhw{\CTw}
     \;=\; \ptf_\io + \hmg{\Fw \Hl}
     \,.
  \eeq
The existence of such first-class constraint ensures that the rank of the Poisson bracket matrix of constraints is preserved while parameterizing.

On the right-hand side of (\ref{CPI_AMG}), only the homogeneous part of the Hamiltonian, $\Fw \Hl$, survives. This is how spatial nonlocality arises in the gauge structure of the local {\wGUMG} parameterized theory (this issue is further addressed in Section \ref{SSect:GaugeStructure_AMG}). The new first-class  parametrization constraint generates spatially local gauge transformations in the auxiliary variables sector and nonlocal (spatially homogeneous) transformations\HIDE{ in the sector} of the metric fields.

The gauge-fixing condition that transforms\! (\ref{paramAction_naive_AMG})\! into\! (\ref{extAction_AMG})\HIDE{ upon reduction},
 %% referred to as
called the\! \emph{correspondence gauge},\! is
 %%can be chosen as
 \beq{correspondence_gauge_AMG}
  \tf^\io(\taux,\sx) - \taux = 0
  \,.
 \eeq
The new first-class constraint (\ref{CPI_AMG}) is transversal to the gauge-fixing condition (\ref{correspondence_gauge_AMG}), which supports the admissibility of this gauge. It is clear how the reduction of the partially gauge-fixed parameterized system with respect to the second-class constraints  $\ptf_\io {+} {\FF \Hl}$ and $\tf^{\io} {-} \taux$ leads to (\ref{extAction_AMG}): the reduction eliminates these two constraints while restoring the Hamiltonian $\FF\Hl$ on the reduction surface from the term $\ptf_\io\,\tauxdot{\tf}^\io$ (where $\ptf_\io \to -\FF\Hl$, $\tauxdot{\tf}^\io \to 1$). This competes the check of the physical equivalence between the parameterized and original theories.

\if{ %% v.2024-10
  By the same reason after implementation of $\tf^{\io}{-}\taux$ constraint into the parameterized action (\ref{paramAction_naive_AMG})\HIDE{ (with some new Lagrange multiplier)} one may perform a reduction with respect to two conjugated constraints $(\ptf_\io {+} {\FF \Hl})$ and $(\tf^\io{-}\taux)$ \footnote{Such reduction is an equivalent reduction since $\PB{\big(\ptf_\io {+} {\FF \Hl}\big)(\taux,\sx)\,}{\tf^{\io}(\taux,\sx'){-}\taux}=-\delta(\sx-\sx')$, so $\tf^\io$, $\ptf_{\io}$, $\replm^{\io}$ and the Lagrange multiplier for gauge-fixing condition $(\tf^\io{-}\taux)$ form the set of auxiliary variables, expressible from the system of variational equations of motion with respect to these variables. Which guarantees that the reduction in the action leads to the physically equivalent model.} so that the reduced action will coincide with\HIDE{ extended action} (\ref{extAction_AMG}). Since the needed gauge-fixing is predefined in the naive parametrization the only thing to be verified the rank condition or, equivalently, the presence of the new independent first-class constraint.
}\fi

\newpar

The success of the \emph{local} implementation of time parametrization in {\wGUMG}, (\ref{paramAction_naive_AMG}),
stems from the fact that the average-free part of the local Hamiltonian, $\Fw\Hl$, is constrained to zero by the secondary constraint (\ref{secondary_GUMG_inh}), which is equivalent to $\wwc \inh{\Fw\Hl} \teq 0$. This is nontrivial, as an analogous local\HIDE{ time} parametrization is generally not possible in field theory.

To emphasize the role of the\HIDE{ secondary} constraint (\ref{secondary_GUMG_inh}) in ensuring the locality of the parameterized action (\ref{paramAction_naive_AMG}), consider an alternative stepwise approach to parameterization, which provides a foundation for extending the procedure to general {\GUMG} models. Starting from the extended action (\ref{extAction_AMG}), one can always introduce a set of canonically conjugate spatially homogeneous fields, $\hmg{\tf}^\io(\taux)$ and $\hmg{\ptf}_\io(\taux)$, along with the corresponding Lagrange multiplier $\hmg{\replm}^\io(\taux)$, which convert the Hamiltonian into a homogeneous constraint.
Applied to (\ref{extAction_AMG}), this yields an equivalent homogeneously parameterized action
  \beq{hmgparAction_AMG}
   \SSS^{\iwwc}_{\,\hmg{\!par\!}} [\g,\pg,\hmg{\tf}^\io\!,\hmg{\ptf}_\io,\hmg{\replm}^\io\!,\repNs,\replmCT]
    = \!
   \int \! d\taux\hspace{1pt} \dsx \, \big(
     \pg^{\zm\zn}\tauxdot{\g}_{\zm\zn}
     \!+  \hmg{\ptf}_\io\,\tauxdot{\hmg{\tf}}^\io
     -   \hmg{\replm}^\io ( \hmg{\ptf}_\io {+} \Fw \Hl )
     - \repNs^\zn \Hs_\zn
     - \replmCT^\zn \CTw{}_{,\zn}
 \big) \,.
 \eeq
 % (\label{hmgparAction_GUMG})
This action is equivalent to (\ref{extAction_AMG}) because $( \hmg{\ptf}_\io {+} \hmg{\FF \Hl} )$ is first-class and the correspondence gauge $\hmg{\tf}^\io {-\,} \taux \teq 0$ directly reduces it to the original (\ref{extAction_AMG}). This property is universal\footnote{
 To be precise, the first-class property of the homogeneous parametrization constraint holds for so-called first-class Hamiltonians. However, any local canonical Hamiltonian can be made first-class by shifting it by a certain combination of second-class constraints \cite{Henneaux:1992ig}, ensuring the existence of a linear combination of first class.
} 
following from the consistency of the dynamical and gauge canonical structures \cite{Henneaux:1992ig}.
In the second step one introduces to a theory a pair of spatially average-free canonical fields, $\inh{\tf}^\io(\taux,\sx)$, $\inh{\ptf}_\io(\taux,\sx)$, along with a corresponding Lagrange multiplier $\inh{\replm}^\io(\taux,\sx)$, via the trivial, pure-gauge combination
   \beq{inh_trivial_sector}
    \sint \Lib
     \inh{\ptf}_\io\,\tauxdot{\inh{\tf}}^\io
     {\,-\,} \inh{\replm}^\io \,\inh{\ptf}_\io
     \Rib \,.
   \eeq
This step localizes the auxiliary sector forming the local field combinations,
  $ % \bea{}
    \sint \Lib
     \hmg{\ptf}_\io\,\tauxdot{\hmg{\tf}}^\io
     {\,-\,} \hmg{\replm}^\io \,\HIDE{(} \hmg{\ptf}_\io %% + \FF \Hl )
     \Rib
     +
   \sint \Lib
     \inh{\ptf}_\io\,\tauxdot{\inh{\tf}}^\io
     {\,-\,} \inh{\replm}^\io \,\inh{\ptf}_\io
     \Rib
   \;=\;
    \sint \Lib
     \ptf_\io \tauxdot{\tf}^\io
     {\,-\,} \replm^\io \HIDE{(} \ptf_\io  %%+ \hmg{\FF \Hl} )
     \Rib \,,
  $ % \eea
and leads, in view of $\sint\, \hmg{\replm}^\io \, \FF \Hl  \teq \sint\, \hmg{\replm}^\io \, \hmg{\FF \Hl} = \sint\, {\replm}^\io \, \hmg{\FF \Hl} $, to the following equivalent representation
  \beq{parAction_cons_AMG}
   \SSS^{\iwwc}_{par} [\g,\pg,\tf^\io\!,\ptf_\io,\replm^\io\!,\repNs,\replmCT]
    = \!
   \int \! d\taux\hspace{1pt} \dsx \, \big(
     \pg^{\zm\zn}\tauxdot{\g}_{\zm\zn}
     \!+ {\ptf}_\io\,\tauxdot{{\tf}}^\io
     - {\replm}^\io ( {\ptf}_\io {\hspace{0.5pt}+\hspace{0.5pt}} \hmg{\Fw \Hl} )
     - \repNs^\zn \Hs_\zn
     - \replmCT^\zn \CTw{}_{,\zn}
 \big) \,.
 \eeq
\vspace{-2mm}

So far, the scheme\HIDE{ universally} applies to any field theory, but the resulting action is not local, as it still contains the nonlocal term $\sint \,{\replm}^\io \,\hmg{\FF \Hl}$. A distinctive feature of {\wGUMG} is the secondary constraint term, $ - \sint \replmCT^\zn (\wwc\Fw\Hl)_{,\zn} \teq \sint \Dvg[\zn]{\replmCT}\, (\wwc\inhw{\Fw\Hl})$, which through redefining Lagrange multipliers as $\Dvg[\zn]{\replmCT} \to \Dvg[\zn]{\replmCTn} \,\teq\, \Dvg[\zn]{\replmCT} {\,+\,} \wwc^{-1}\inh{\replm}^\io$\HIDE{ in the action}, incorporates the missing average-free part of the Hamiltonian into
 %%the otherwise nonlocal parametrization constraint
$\sint \replm^\io( {\ptf}_\io {\hspace{0.5pt}+\hspace{0.5pt}} \hmg{\Fw \Hl} )$, finally achieving localization (\ref{paramAction_naive_AMG}). Without this constraint, localization of the {\wGUMG} action fails. As discussed later, for general {\GUMG} models, localization of %% the analog of the parameterized action
(\ref{parAction_cons_AMG}) cannot be completed in the same way, since for nonconstant $\WW(\argsqrg)$, the secondary constraint is no longer an average-free part of the local Hamiltonian.

We end this excursus with the note that the local constraint $\ptf_\io {\hspace{0.5pt}+\hspace{0.5pt}} {\Fw \Hl}$ in (\ref{paramAction_naive_AMG}) is not first-class, as it decomposes into three constraints:
  \beq{parConstraint_decomposition}
    \ptf_\io {\hspace{0.5pt}+\hspace{0.5pt}} \Fw\Hl = \hmg{\ptf_\io {\hspace{0.5pt}+\hspace{0.5pt}} \Fw\Hl} + \inh{\ptf}_\io + \inh{\Fw\Hl}
    \,,
  \eeq
of which the first two are first-class and form $\CPI$, (\ref{CPI_AMG}), while the third, proportional to the secondary constraint, is second-class. A similar discussion applies to the {\UMG} case \cite{Henneaux:1989zc}; however, it is redundant there because the secondary constraint and the whole (\ref{parConstraint_decomposition}) are first-class.

%%-------------------------=%        ******         %=-------------------------%%
%%-------------------------=%        ******         %=-------------------------%%

  \subsection{Henneaux--Teitelboim-like action}
   \label{SSect:AltAction_AMG}
    \hspace{\parindent}
 As the next step, we change the basis of constraints in the parameterized action (\ref{paramAction_naive_AMG}). The nearest goal is to isolate the dependence on\HIDE{ the Hamiltonian structure} $\Hl$ into a single constraint, ensuring that this constraint, along with the momentum constraint and the\HIDE{ symplectic} term $\pg^{\zm\zn}\tauxdot{\g}_{\zm\zn}$, form the canonical counterpart of the Einstein-Hilbert term in the Lagrangian action. Other terms should be independent of $\pg^{\zm\zn}$ to facilitate a {\GR}-like reduction with respect to gravitational momenta\HIDE{ in the action}.
 %\footnote{ The latter is preferable because it guarantees, that after passing to Lagrangian action (by excluding the momenta $\pg^{\zm\zn}$)\HIDE{ in the canonical action} all time derivatives of metric components are assembled in spacetime scalar curvature term.} %%(modulo total derivative)

\if{ v.2022-10
The first step is to rewrite the action (\ref{extAction_GUMG}) in \emph{parameterized} form.  The extended action (\ref{}) acquired the constraint term proportional to $\Hl$ which is needed to restore canonical action. However, this constraint term is functionally incomplete due to absence of spatially homogeneous part. Moreover Hamiltonian is also proportional to $\Hl$ structure which should be turned into constraint term to complete the constraint to form a part of the canonical spacetime curvature expression.

Being universal in mechanics \cite{Henneaux:1992ig} parametrization of the theory acquiring  spacetime locality becomes a nontrivial problem in field theory. \TODO{Brief discussion of this can be found in Appendix (\ref{Sect:parametrization})}.

For particular family of {\GUMG} theories the action and so the theory and equations of motions preserve spacetime locality.
 This is the case of anisotropic unimodularity  with
 \beq{defAMG}
  \Fw(\argsqrg) = \sqrg^{\wwc} \;\;
  (\wwc\neq-1);
  \quad
  \WW(\argsqrg) = \wwc = const;
  \quad
  \OOmega(\argsqrg) = 1+\wwc.
 \eeq

In this case operator $\opE=\Id$ and ...
}\fi

It is straightforward to see that the following pairs of constraint conditions are equivalent:
  \bea{}
   \left\{
    \begin{array}{l}
      \ptf_\io + {\Fw \Hl} =0 \,, \\
      \CT_{,\zn} \:\equiv\: \wwc (\Fw \Hl)_{,\zn} =0  \,. \\
    \end{array}
   \right.
   \quad \Leftrightarrow \quad
   \left\{
    \begin{array}{l}
      \ptf_\io + {\Fw \Hl} =0 \,, \\
      {\ptf_\io}_{,\zn} =0 \,. \\
    \end{array}
   \right.
  \eea
In (\ref{paramAction_naive_AMG}) such basis change is achieved through a nondegenerate linear transformation of the Lagrange multipliers $ ( \replm^\io , \replmCTn^\zn )  \,\to\,  ( \replm , \lmptf^\zn )$:
  \beq{redef_lms_AMG}
    \begin{array}{|lcl}
      \replm^\io \!& \!=& \! \! \replm - \Dvg[\zn]{\lmptf} \,,\\
      \replmCTn^\zn \!& \!=&\! \! -\wwc^{-1} \lmptf^\zn \,,\\
    \end{array}
    %% checked 2024-11
  \eeq
leading to the action in the form\footnote{
 At this step, the overall $\wwc$ coefficient in one of the constraints disappears, eliminating the need for special reservations in the exceptional $\wwc =0$ case. (See the remark on this case in Section \ref{SSect:Except_Partic_Cases}.)
}
   \beq{paramAction_AMG'} %(\ref{parAction_BLRG})
    \SSS^{\iwwc}_{par'}[\g,\pg,\tf^\io\!,\ptf_\io,\replm,\repNs,\lmptf]
    \;\,=\, %\!\!&=& \!\!\!
    \int \!d\taux \hspace{1pt} \dsx  %\int \!d\taux \hspace{1pt} \dsx
    \Lib
     \pg^{\zm\zn}\tauxdot{\g}_{\zm\zn}
     + \ptf_\io \tauxdot{\tf}^\io
     - \replm ( \ptf_{\io} {+} {\Fw \Hl} )
     - \repNs^\zn \Hs_\zn
     - \lmptf^\zn {\ptf_\io}_{,\zn}
     \Rib \,.
   %\qquad% \nonumber\\
  \eeq
In (\ref{paramAction_AMG'}), all dependence on gravitational variables is confined  in local structures, analogous to those in Einstein's general relativity.
To establish an accurate correspondence, we rescale Lagrange multiplier ${\replm} \to \repNl$,\HIDE{ absorbing the $\Fw$ factor,} and rename the cosmological-constant field term $\ptf_\io \to \CCof$:
 \beq{replm_ptf_io_rescaling_AMG}
   \replm = \repNl \Fw^{-1}\,,
   \qquad
   \ptf_\io =  \CCof \,, %%\Fw \sqrg \CCf,
   %\label{ccf_0_AMG}
 \eeq
yielding the yet canonical action
   \bea{paramAction_AMG_fin} %(\ref{parAction_BLRG})
    \SSS^{\iwwc}_{par''}\![\g,\pg,\tf^\io\!,\CCof,\repNl\!,\repNs,\lmptf]
    &\!=\!\!&
    \int \!d\taux \hspace{1pt} \dsx  % \int d\taux\, \dsx
    \,\LiB
     \pg^{\zm\zn}\tauxdot{\g}_{\zm\zn}
     \!- \repNl ( \Hl {+} \Fw^{-1}\! \CCof  )
     - \repNs^\zn \Hs_\zn
     + \big(\tauxdot{\tf}^\io {+\,} \Dvg[\zn]{\lmptf}\big) \CCof
     \RiB
     .
   %\qquad
   \nonumber\\
  \eea
The factor multiplying\HIDE{ Lagrange multiplier} $\repNl$ is similar to the Hamiltonian constraint (\ref{BGUMGHamiltonianStructure}) of general relativity with the\HIDE{ dynamical} combination $\CCf \defeq \sqrg^{-1} \Fw^{-1} \CCof$\, playing the role of the dynamical cosmological-constant term.
The factor multiplying  $\repNs^\zn$ is the standard {\GR} momentum constraint (\ref{BGUMGMomentaConstraints}).

\newpar

Since all dependence on\HIDE{ gravitational momenta} $\pg^{\zm\zn}$ is confined within the {\GR}-like canonical structures, it can be expressed through their own variational equations of motion,
  \bea{momenta_ecK_GR}
    \pg^{\zm\zn} \eomeq {\sqrg}(\ecK^{\zm\zn}-\g^{\zm\zn} \trecK)
    \,, \qquad
    \ecK_{\zm\zn} = \tfrac{1}{2\repNl} ( \tauxdot{\g}_{\zm\zn} - \repNs_{\zm ; \zn} - \repNs_{\zn ; \zm})
    \,,
  \eea
and reduced from the action in the same manner as in general relativity. Upon reduction, the first three terms on the right-hand side of (\ref{paramAction_AMG_fin}) transform into Einstein-Hilbert Lagrangian with a modified dynamical cosmological term
 \bea{ActionHTlike0_AMG} %(\ref{ActionHTlike0_AMG})
    \SSS^{\iwwc}_{\ialt_0}[\Gaux,\CCof,\HTf]
    &=& \int \!d\taux \hspace{1pt} \dsx %\int d\taux\, \dsx \,
     \sqrt{|\Gaux|} \Big( \stR (\Gaux) - \sqrg^{-\wwc-1}\CCof\Big)
    + \int \!d\taux \hspace{1pt} \dsx %\int d\taux\, \dsx
     \,\partial_\Zm \HTf^\Zm\CCof
    \,,
    %\quad
    %\nonumber
  \eea
where $\Gaux_{\Zm\Zn}$ is the\HIDE{ covariant} spacetime metric corresponding to the ADM variables $(\g_{\zm\zn},\repNs^\zn,\repNl)$  via
 $ % \beq{GBulk_rep_metric_ADM}% was {Gaux_correspondence}
  \Gaux_{\Zm\Zn}d{X}^{\Zm}d{X}^{\Zn}
   \teq  \g_{\zm\zn} (d\sx^\zm {+} {\repNs}^\zm d\taux)(d\sx^\zn {+} {\repNs}^\zn d\taux) - {\repNl}^2 d\taux \hspace{1pt} d\taux
   \,,
 $ % \eeq
 %% as in (\ref{GBulk_metric_ADM}),
\HIDE{with }volume density $ \sqrt{|\Gaux|} = \repNl\! \sqrg$, and $\Fw$ is explicated as $\sqrg^{\wwc}$. Also, we form the Henneaux--Teitelboim\HIDE{ auxiliary} spacetime vector field $\HTf^\Zm$:
  \bea{def_tef}
    \HTf^\Zm \equiv \big(\HTf^0, \HTf^\zm \big)  = \big(\tf^\io, \lmptf^\zm \big)
    %\qquad \text{ so that } \quad
    \qquad \Rightarrow \quad
    \tauxdot{\tf}^\io + \Dvg[\zm]{\lmptf} = \partial_\Zm \HTf^\Zm.
  \eea
The terms in\HIDE{ actions} (\ref{paramAction_AMG_fin}) and (\ref{ActionHTlike0_AMG}) proportional to $\CCof$  are independent of $\pg^{\zm\zn}$ and thus remain unaffected by the reduction. In general, these terms introduce a perfect-fluid degree of freedom and break spacetime covariance.
In (\ref{ActionHTlike0_AMG}) the\HIDE{ familiar Henneaux--Teitelboim} structure of the last term ensures that $\CCof$ is an \emph{on-shell constant}, while the apparent noncovariance arises from the modified cosmological-field term\HIDE{ in the first integral}.

All action transformations leading to (\ref{ActionHTlike0_AMG}) for the {\wGUMG} subfamily preserved classical equivalence.
 %\footnote{In the sense that all transformations of the actions we made guaranteed that classical equations of motion has the equivalent spaces of solutions of variational equations of motion.}
Thus, this alternative representation\HIDE{ (\ref{ActionHTlike_AMG})} is equivalent to the original formulation by construction. It is also important to note that the action (\ref{ActionHTlike0_AMG}) describes models with all $\wwc\in\mathbb{R}$, including the exceptional cases $\wwc \teq 0$ and $\wwc \teq {-}1$.

\newpar

Another convenient form of the alternative representation is obtained by rescaling $\CCof \to \CCf$:
  \beq{ccf_AMG} %% was {ccf_0_AMG}
    \CCof \defeq \sqrg^{\wwc+1} \CCf
    \,,
  \eeq
so that the alternative action takes the form
 \bea{ActionHTlike_AMG} %(\ref{ActionHTlike_AMG})
    \SSS^{\iwwc}_{\ialt}[\Gaux,\CCf,\HTf] %%\tf^\io,\replmCTn
    &=& \!\int \!d\taux \hspace{1pt} \dsx % \int d\taux\, \dsx \,
    \sqrt{|\Gaux|} \Big( \stR (\Gaux) - \CCf\Big)
    +  \int \!d\taux \hspace{1pt} \dsx  %\int d\taux\, \dsx
    \: \partial_\Zm \HTf^\Zm \!\sqrg^{\wwc+1} \! \CCf
    \,,
   % \nonumber\\
  \eea
where the first integral acquires\HIDE{ covariant} the form of the Einstein-Hilbert action with the dynamical cosmological constant, while the apparent noncovariance shifts to the latter term\HIDE{ with the divergence of the auxiliary field}\footnote{
 If the metric $\Gaux_{\Zm\Zn}$ transforms as a tensor and $\CCf$ transforms as a scalar, the latter term in (\ref{ActionHTlike_AMG}) is noncovariant.
 Alternatively, if $\CCof \teq \sqrg^{\wwc+1} \! \CCf$ is assumed to transform as a scalar, diffeomorphism invariance is instead broken by the cosmological-field term in the first integral.
 In any case, covariance\HIDE{ (the full diffeomorphism invariance)} is restored for $\wwc \teq {-}1$, corresponding to the Henneaux--Teitelboim action \cite{Henneaux:1989zc} of \emph{unimodular} gravity.
}.

%%-------------------------=%        ******         %=-------------------------%%
%%-------------------------=%        ******         %=-------------------------%%

\newcommand{\vareqmapsto}{\mapsto}

\subsection{Dynamical properties}
 \label{SSect:Dynamical properties_AMG}
  \hspace{\parindent}
The fields  $\CCf$ and $\HTf^\Zm$ enter the Lagrangian action  (\ref{ActionHTlike_AMG}) linearly. They are intertwined, so their variational equations determine each other's on-shell behavior. The variational equation for the auxiliary fields  $\HTf^\Zm$,
 \beq{EoM_tef_AMG}
  \VDer{\SSS^{\iwwc}_{\ialt}}{\HTf^\Zm}
  = 0
  \qquad\vareqmapsto\qquad
  \partial_\Zm ( \sqrg^{\wwc+1} \CCf) \eomeq 0
  \,,
 \eeq
generalizes behavior of the cosmological term, implying that the on-shell spacetime constant is \beq{}
  \CCof = \sqrg^{\wwc+1}\CCf \eomeq \cco
  \,.
 \eeq
The variation with respect to the cosmological-constant field $\CCf$,
 \beq{EoM_Lambda_AMG}
   \VDer{\SSS^{\iwwc}_{\ialt}}{ \CCf }
   = 0
  \qquad\vareqmapsto\qquad
    \sqrt{|\Gaux|} =  \partial_\Zm \HTf^\Zm  \sqrg^{\wwc+1}
    \,,
 \eeq
relates the spacetime metric determinant to the divergence of the\HIDE{ Henneaux-Tritelboim} auxiliary vector field.

Since $\CCf$ enters the action linearly, its dynamical nature ensures that the  cosmological term compensates the divergence term on shell. Consequently, the principal function is determined by the values of the Einstein Lagrangian term on classical configurations\HIDE{ which is typical for restricted theories}.
Since $\HTf^\Zm$ appears in the action only through the divergence term $\partial_\Zm \HTf^\Zm$, it is defined by the equations of motion only up to an arbitrary spacetime transverse vector field. This reflects an additional gauge ambiguity in the auxiliary sector.

From the canonical structure of this action, it is also clear that $\CCof= \ptf_\io$ is a gauge-invariant homogeneous constant of motion. It obviously commutes with all constraints of (\ref{paramAction_AMG'}), ensuring that it is gauge-invariant and preserved in time. Furthermore, due to the constraint ${\ptf_\io}_{,\zm}=0$, on shell it does not depend on $\sx$. The physical sector of the model naturally foliates into leaves parameterized by constant $\cco$, which, in view of the first constraint in (\ref{paramAction_AMG'}), corresponds to the on-shell values of $\CTw \eomeq - \wwc \,\cco$. Consequently, the GR branch corresponds to the leaf $\cco=0$, while the non-GR branch consists of leaves with $\cco\neq0$.

\newpar

In the ADM parametrization of the metric field, relation (\ref{EoM_Lambda_AMG})
can be interpreted as relating\HIDE{ the divergence of the auxiliary vector field} $\partial_\Zm \HTf^\Zm$ to another nondynamical field --- the lapse function --- via
$\repNl \eomeq \partial_\Zm \HTf^\Zm  \sqrg^{\wwc}$.
At the same time, on the non-GR branch, the lapse function $\repNl$ satisfies additional on-shell conditions, since in the canonical framework it serves as a Lagrange multiplier for a mixed-class constraint.
From the Poisson bracket relations (\ref{GUMG_Algebra_onshell}), (\ref{parwGUMG_Algebra_onshell}) in the constraint basis of the {\wGUMG} parameterized action (\ref{paramAction_AMG'}), the second-class constraints are given by $\inh{(\ptf_\io {+} {\Fw\Hl})}$,  where $\Fw \teq \sqrg^{\wwc}$, and the longitudinal part of $\Hs_\zn$. For $\wwc \tneq {-}1$\HIDE{ and $\cco\neq0$} the corresponding Lagrange multipliers, $\inh{\replm} \teq \inh{\repNl\sqrg^{-\wwc}}$ and $\Dvg[\zk]{\repNs}$, must vanish due to the consistency equations. This,
 %\footnote{The first-class constraints are disentangled implicitly upon simple, phase-space independent decomposition of the Hamiltonian reparametrization constraint into spatially homogeneous and inhomogeneous parts and decomposition of the momenta constraint into transverse (divergence-free) and longitudinal (gradient) parts (see the discussion at the end of the Section \ref{SSect:GUMG_Constraint_Structure}).}
together with (\ref{EoM_Lambda_AMG})\footnote
{
  In the context of the canonical action (\ref{paramAction_AMG'}), equation (\ref{EoM_Lambda_AMG}) arises from variation with respect to $\ptf_\io$, implying the extremal condition:  $\replm \,\tequiv\, \repNl\sqrg^{-\wwc} \eomeq \tauxdot{\tf}^\io {+} \Dvg{\lmptf} \,\tequiv\, \partial_\Zm \HTf^\Zm$. By separating the average and average-free components and taking into account consistency conditions, one arrives at (\ref{DVG_tef}).
}, leads to
 \beq{DVG_tef}
   \begin{array}{|l}
    \inh {\partial_\Zm \HTf^\Zm} \,=\, \tauxdot{\inh{\tf}}^\io {+\,} \Dvg[\zm]{\lmptf} \,\eomeq\, 0 \,, \vphantom{\big|}
   \\
    \partial_\Zm \HTf^\Zm \,\eomeq\, \hmg{\partial_\Zm \HTf^\Zm} \,=\, \tauxdot{\hmg{\tf}}^\io \, \;\;\; \text{---\; arbitrary ,} \vphantom{\big|^l}
   \end{array}
   \qquad ( \wwc\neq-1\,, \;\; \text{non-GR branch:} \;\; \cco\neq0 \HIDE{\CCf\neq0} )
   \,.
  \eeq
Summarizing the on-shell behaviour of the parameterized lapse and shift functions, one obtains
 \beq{LM_rep_eom_AMG}
   \begin{array}{|l}
    \repNl \,\eomeq\, \tauxdot{\hmg{\tf}}^\io  \sqrg^{\wwc}
    \quad\;\; ( \,\Rightarrow \;\; \inh{\repNl\Fw^{-1} \HIDE{\sqrg^{-\wwc}}} \,\eomeq\, 0 \, )
    \,, \vphantom{\big|}
    \\ %\quad
    \Dvg[\zk]{\repNs} \,\eomeq\, 0
    \,, \vphantom{\big|^l}
   \end{array}
   \qquad ( \wwc\neq-1\,, \;\; \text{non-GR branch:} \;\; \cco\neq0 \HIDE{\CCf\neq0} )
   \,.
  \eeq
Relations (\ref{LM_rep_eom_AMG}) indicate that the gauge freedom of $\repNl$ is functionally incomplete, allowing for rescaling by an arbitrary nonzero function of time, $\tauxdot{\hmg{\tf}}^\io(\taux)$, and the transverse components of $\repNs^\zn$ remain undetermined by the equations of motion, reflecting the gauge freedom associated with transverse spatial diffeomorphisms.
Notably, the first condition in (\ref{LM_rep_eom_AMG}) reinstates the restriction condition (\ref{GUMG_restriction}) from the original theory, now with an extra factor accounting for homogeneous time reparametrization\footnote{The homogeneous reparametrization factor $\tauxdot{\hmg{\tf}}^\io$ reduces to unity in the correspondence gauge (\ref{correspondence_gauge_AMG}), verifying the result by comparison with the equations of motion of non-parameterized theory.}.

On the GR branch and in the exceptional {\UMG} case, where all constraints are first class, none of the Lagrange multipliers are fixed by the equations of motion
  \beq{LM_rep_eom_AMG_GR_branch}
   \begin{array}{|l}
    \repNl \,\eomeq\,  \partial_\Zm \HTf^\Zm  \sqrg^{\wwc}  \; \text{ --- arbitrary}\,,
    \\ %\quad
    \repNs^\zn \; \text{ --- arbitrary}
    \,,
   \end{array}
   \qquad (\text{for\;} \wwc =-1 \; \text{or\;  on the GR branch:} \; \cco = 0 \HIDE{\CCf\neq0} )
   \,.
  \eeq

\newpar

The dynamics of gravitational degrees of freedom in general is governed by the variational equations for the spacetime metric $\Gaux_{\Zm\Zn}$, which take the form of the Einstein-Hilbert equation
 \beq{EoM_GG_AMG}
   \VDer{\SSS^{\iwwc}_{\ialt}}{\Gaux_{\Zm\Zn}}
  \,=\, 0
  \qquad\vareqmapsto\qquad
    \stR^{\Zm\Zn} \!- \frac12 \stR \,\Gaux^{\Zm\Zn} 
    \:\eomeq\:
    \frac12 \repTSEt_{\iwwc}^{\Zm\Zn} 
    \,,
 \eeq
where all terms involving $\CCf$ are attributed to the cosmological \emph{matter} sector, contributing to the energy-momentum tensor $\repTSEt_{\iwwc}^{\Zm\Zn}$. Using  (\ref{EoM_Lambda_AMG}), the on-shell form of the energy-momentum tensor is
  \beq{Stress-Energy_Tensor_AMG}
    \repTSEt_{\iwwc}^{\Zm\Zn}
    \,\defeq\,
    \frac{2}{\sqrt{|\Gaux|}} \VDer{\SSS^{\iwwc}_{mat}}{\Gaux_{\Zm\Zn}}
    \:\eomeq\:
     \CCf \, n^{\Zm}n^{\Zn}
     + \wwc
     \CCf \big(\Gaux^{\Zm\Zn} \!+ n^{\Zm}n^{\Zn}\big)
     \,,
  \eeq
where $n^\Zm$ is the unit vector field orthogonal to constant-time hypersurfaces. Here the energy density is given by $\repeSEt \eomeq  \CCf $, and the pressure is $\reppSEt \eomeq \wwc \CCf$, reinstating the perfect fluid equation of state with the constant parameter $\wwc$:
\if{
  \footnote{
   The cosmological-constant term $\sqrt{|\Gaux|}\CCf$ generates $-\CCf\Gaux^{\Zm\Zn}$ contribution to $\repTSEt_{\wwc}^{\Zm\Zn}$ giving equal contribution to the energy and the pressure. Whereas gauge-breaking term from the second integral (\ref{ActionHTlike_AMG}), which feels only spatial induced metric, contributes only to the pressure component: $(\wwc{+}1)\CCf \big(\Gaux^{\Zm\Zn} {+} n^{\Zm}n^{\Zn}\big)$, in which by virtue of (\ref{EoM_Lambda_AMG}) we equated\HIDE{ the overall} factor $\sqrg^{\wwc+1} \partial_\Zm \HTf^\Zm /\sqrt{|\Gaux|}$ to unity. Extracting contribution to energy $\sim -n^\Zm n^\Zn$ from the first contribution and transferring the rest to pressure $\sim \Gaux^{\Zm\Zn} {+} n^\Zm n^\Zn$ gives the claimed
    balance (\ref{PF_EoS_AMG}).
   \TODO{hide?}
  }
}\fi
  \beq{PF_EoS_AMG}
    \reppSEt \eomeq \wwc \,\repeSEt
    \,.
  \eeq
The dynamical behavior of the energy and pressure becomes more explicit when $\CCf$ is expressed in terms of the spatially homogeneous constant of motion $\CCof \eomeq \cco$, (\ref{replm_ptf_io_rescaling_AMG}):
  \bea{energy_pressure_AMG}
   \begin{array}{|lcl}
    \eSEt
     &\eomeq&
     \sqrt{|\Gaux|}^{-1} \, \tauxdot{\hmg{\tf}}^\io \, \cco
     \,\eomeq\,
     \sqrg^{-\wwc-1} \, \cco
    \,,
   \\
    \pSEt &\eomeq&
     \wwc \sqrt{|\Gaux|}^{-1} \, \tauxdot{\hmg{\tf}}^\io\, \cco
     \,\eomeq\,
     \wwc \sqrg^{-\wwc-1} \, \cco
    \,.
    \end{array}
  \eea
 %%This behaviour is counterintuitive when compared to the cosmological constant in Einstein gravity, but it is not unusual when one reminds behavior of the cosmological constant's stress-energy in unimodular gravity.

Analogous to Einstein's general relativity in the presence of cosmological perfect fluid, the trace of the Einstein equations (\ref{EoM_GG_AMG})\HIDE{ in terms of the ADM metric fields} takes the form
 \beq{eom_tr_Einstein_AMG}
   %\left(
   \stR(\Gaux)
   \;\tequiv\;
   2\nabla_\Zm(n^\Zm \trecK {-\,} n^\Zn\nabla_\Zn n^\Zm)
   +\ecK^\zm_{\,\zn}\ecK^\zn_{\,\zm} - \trecK\trecK +\, \sR(\g)
   %\right)
   \eomeq    \frac1{\Ddim{-}2} (\wwc{+}1 -\Ddim\wwc ) \CCf
   \,.
 \eeq
At the same time, the Lagrangian counterpart of the canonical constraint $\Hl {+} \sqrg \CCf \eomeq 0$, (\ref{BGUMGHamiltonianStructure}), in ADM components is
 \beq{eom_Hl_AMG}
   %\left(
   \ecK^\zm_{\,\zn}\ecK^\zn_{\,\zm} - \trecK\trecK -\, \sR(\g)
   %\right)
   \eomeq - \CCf
   \,,
 \eeq
while the momentum constraint $\Hs_\zn \HIDE{= -2\g_{\zm\zn} {\pg^{\zn\zk}}_{;\zk}} \eomeq 0$, (\ref{BGUMGMomentaConstraints}), imposes a {\GR}-like condition on the extrinsic curvature (\ref{momenta_ecK_GR}):
 %The Lagrangian version of the momentum constraint $\Hs_\zn \HIDE{= -2\g_{\zm\zn} {\pg^{\zn\zk}}_{;\zk}} \eomeq 0$, (\ref{BGUMGMomentaConstraints}),  gives a {\GR}-like condition on the extrinsic curvature (\ref{momenta_ecK_GR}),
  \beq{}
    \ecK_{\zm\,;\zn}^{\;\zn} - \trecK_{,\zm} \eomeq 0
    \,.
  \eeq

%%-------------------------=%        ******         %=-------------------------%%
%%-------------------------=%        ******         %=-------------------------%%

  \subsection{Canonical gauge structure\HIDE{ of the non-GR branch}}
   \label{SSect:GaugeStructure_AMG}
    \hspace{\parindent}
     %% Major update and final reduction 2024-10
     %% Calculation in section \ref{ASSect:GaugeInvAMG'}
   %
In this subsection, we analyze the gauge symmetry on the non-GR branch. The simplifications in the {\wGUMG} case provide clearer expressions for the spatially nonlocal gauge structure and simpler derivation of Lagrangian gauge transformations from canonical ones.

As noted in Section \ref{SSect:ParameterizedAction_AMG}, in the parameterized theory (\ref{paramAction_naive_AMG}), the independent first-class canonical generators correspond to the constraint (\ref{CPI_AMG}) and the transverse spatial diffeomorphisms:
  \beq{I_class_constraints_AMG}
   \left|
    \begin{array}{l}
      \HIDE{\sint f \CPI \equiv}
      \sint \geps (\ptf_\io {+} \hmg{\Fw \Hl})
      = \sint \geps \ptf_\io +\sint{\hmg{\geps}\Fw \Hl}
      = \sint{\hmg{\geps}\, (\hmg{\ptf_\io {+} \Fw \Hl})}
        + \sint \,\inh{\geps} \,\inhw{\ptf}_\io
    \,, \vphantom{\big|} \\
      \sint \gzetat^{\zm} \Hs_\zm, \qquad \,\Dvgi{\gzetat} = 0\,
    \,,  \vphantom{\big|^l} \\
    \end{array}
   \right.
  \eeq
which follows from the {\wGUMG} Poisson bracket relations (\ref{GUMG_Algebra_onshell}) and (\ref{parwGUMG_Algebra_onshell}). In the constraint basis of the action (\ref{paramAction_AMG'}), the first generator is a linear combination of the homogeneous part of\HIDE{ the constraint} $\ptf_\io {+} {\Fw\Hl}$ and the constraint $\inh{\ptf}_{\io}$ (equivalent to a constraint, defined by ${{\ptf}_{\io}}_{,\zm}$).
On the configuration space of\HIDE{ the action} (\ref{paramAction_AMG'}), the canonical generators (\ref{I_class_constraints_AMG}) induce the following gauge transformations:
 % \footnote{
   % We do not explicate here lengthy gauge transformation for $\pg^{\zm\zn}$,\, $\gvar[\geps,\gzeta] \pg^{\zm\zn} = \PB{\pg^{\zm\zn}}{\sint  {\Fw \Hl}} \hmg{\geps} + \PB{\pg^{\zm\zn}}{ \sint \Hs_\zm} \gzetat^{\zm} $, since this field will be reduced when passing to Lagrangian formulation.
 % }
\beq{gauge_transfs_ham_AMG'}
 \left|
  \begin{array}{lll}
   \gvar[\geps,\gzeta] \g_{\zn\zm}
   &\!=&
   \tfrac{2\Fw}{\sqrg} \left(\pg_{\zn\zm} -\frac{1}{\Ddim{-}2}\,\trpg\,\g_{\zn\zm} \right)\!
   \hmg{\geps}
  +
   \big( \g_{\zk\zm} {\gzeta_{\isst}^{\zk}}_{,\zn}
   + \g_{\zn\zk}  {\gzeta_{\isst}^{\zk}}_{,\zm}
   + \g_{\zn\zm,\zk} \gzeta_{\isst}^{\zk}
   \big)
   \,, \\
   \gvar[\geps,\gzeta]  \tf^\io
   &\!=&
   \geps
      \,, \qquad \quad \;\;
      \gvar[\geps,\gzeta]  \ptf_\io \;=\; 0
  % \,, \\
  % \gvar[\geps,\gzeta]  \ptf_\io &=& 0
   \,, \vphantom{\big|} \\
   \gvar[\geps,\gzeta]  \repNs^\zn
     &\!=&
       \Fw^2 \g^{\zn\zm} \replm_{,\zm} \hmg{\geps}
      + \TDer{}{\taux} \gzetat^{\zn}
      - \LieB{\zn}{\repNs}{\gzetat}
   \,, \vphantom{\big|} \\
   \gvar[\geps,\gzeta]  \hmg{\replm}
    &\!=& \TDer{}{\taux} \hmg{\geps}
      \,, \qquad \;\;
      \gvar[\geps,\gzeta]  \inhw{\replm}
      \;=\;
      \wwc\, \Dvg[\zk]{\repNs}\hmg{\geps}
      + \replm_{,\zk} \gzetat^\zk
   \,, \vphantom{\big|} \\
   \gvar[\geps,\gzeta]  \Dvg[\zk]{\lmptf}
   \!\! &\!=&
    - \TDer{}{\taux} \inhw{\geps}
    + \wwc \,\Dvg[\zk]{\repNs} \hmg{\geps}
    + \replm_{,\zk} \gzetat^\zk
   \,, \vphantom{\big|} \\
  \end{array}
  \right.
 %% Note that in 2024 sign convention for $\gchi$ changed to $-gchi$
 %% Checked by calcs in the Section \ref{ASSect:GaugeInvAMG'}
 \eeq
provided ${\hmg{\geps}}_{,\zm}=0$ and $\Dvgi[\zn]{\gzetat}=0$.
In particular, transformations (\ref{gauge_transfs_ham_AMG'}) imply
 \beq{gvar_sqrt} %% NonCANON
  \begin{array}{lll}
   \gvar[\geps,\gzeta] \sqrg
   &=&   (\sqrg)_{,\zk} \gzeta_{\isst}^\zk
    - \tfrac{1}{\Ddim{-}2} \Fw \trpg \,\hmg{\geps}
    \,,
  \end{array}
  % GUMG.tex 236, 257 Sect 7.2 2023-05
 \eeq
which leads to the following transformation law for the characteristic function $\Fw(\argsqrg)\equiv\sqrg^{\wwc}$: \;
 $ % \beq{gvar_f_sqrt} %% NonCANON
   \gvar[\geps,\gzeta] \Fw
   \teq
    \wwc \tfrac{\Fw}{\sqrg} \sqrg_{,\zk} \gzeta_{\isst}^\zk
   - \wwc \tfrac{1}{\Ddim{-}2} \tfrac{\Fw^2}{\sqrg} \trpg \,\hmg{\geps}
   \,.
 $ % \eeq
Note that under transverse spatial diffeomorphisms, $\sqrg$ and its algebraic functions transform as scalars. The lengthy explicit expression for the gauge transformation of $\pg^{\zm\zn}$: \, $\gvar[\geps,\gzeta] \pg^{\zm\zn} = \PB{\pg^{\zm\zn}}{\sint  {\Fw \Hl}} \hmg{\geps} + \PB{\pg^{\zm\zn}}{ \sint \Hs_\zm} \gzetat^{\zm} $, is omitted here, as this field is reduced when passing to the Lagrangian formulation.

% \subsubsection*{The hidden trivial sector}
%  \hspace{\parindent}

In the canonical gauge transformations (\ref{gauge_transfs_ham_AMG'}), the metric variables\footnote{
  In the parameterized\HIDE{ canonical} treatment, metric variables include $\g_{\zm\zn}$, $\repNs^\zn$, and also $\replm$, from which the lapse function $\repNl$ will be reinstated. The conjugate momenta $\pg^{\zm\zn}$ also transforms with the homogeneous parameter $\hmg{\geps}$.
}
respond only to the \emph{homogeneous} part of the\HIDE{ time-reparametrization} parameter ${\geps}(\taux)$, undergoing spatially homogeneous time reparameterizations. This functional incompleteness is expected due to the averaged nature of\HIDE{ the non-auxiliary term} $\hmg{\Fw\Hl}$ in $\CPI$ (\ref{CPI_AMG}).
Only the auxiliary parameterization variables, $\tf^\io$ and $\replm^\io=\replm {\,-\,} \Dvg[\zk]{\lmptf}$, transform under the full gauge parameter $\geps(\taux,\sx)$, while another auxiliary variable, $\ptf_\io$, is gauge-invariant.

To highlight the distinction in gauge transformations, we decompose $\replm$ in (\ref{gauge_transfs_ham_AMG'}) into its\HIDE{ spatially} homogeneous and average-free parts, $\replm = \hmg{\replm} + \inh{\replm}$. This shows that only the homogeneous part, $\hmg{\replm}$, transforms as the Lagrange multiplier for a first-class constraint, whereas the transformation of the average-free component, $\inh{\replm}$, is on-shell trivial, indicating its second-class nature.
\if{
  \footnote{
   It is relevant since only the homogeneous part is the Lagrange multiplier for the first-class constraint and the $\inh{\replm}$ vanishes on shell.\\
    % One may note that the gauge transformations of the Lagrangian multipliers at the second-class constraints are trivial (proportional to on-shell vanishing quantities), as should be. Namely, transformations of $\Dvg[\zk]{\Ns}$ and $\inhw{\replm}$ are proportional to these two quantities (or ${\replm_{,\zm}}$) which all vanish on shell (see discussion before (\ref{DVG_tef})). %\\
  }
}\fi

To better understand the origin of the mixed-class gauge behavior of $\replm$ it is useful to examine the transformations of the Lagrange multipliers in the minimally parameterized action (\ref{parAction_cons_AMG}).
 There, $\replm^\io$ acts as the\HIDE{ Lagrange} multiplier for a first-class constraint, transforming as $\gvar[\geps,\gzeta] \replm^\io \teq \tauxdot{\geps}$, while $\Dvg[\zk]{\replmCT}$ is the multiplier for the second-class constraint $-\inh{\CTw}$, transforming as $\gvar[\geps,\gzeta]\Dvg[\zk]{\replmCT} \teq {-} \wwc\, \Dvg[\zk]{\repNs}\hmg{\geps} - \replm_{,\zk} \gzetat^\zk$.
\if{ %% v2024-08
  This translates to the Lagrange multipliers of (\ref{paramAction_AMG'}) as
   %% $\gvar[\geps,\gzeta] \hmg{\replm}^\io \,\,\mapsto\,\, $
  $\gvar[\geps,\gzeta] \hmg{\replm} \teq \tauxdot{\hmg{\geps}}\,$ and
   %% $\,\gvar[\geps,\gzeta] \inh{\replm}^\io \,\,\mapsto\,\, \gvar[\geps,\gzeta] (\inh{\replm} {\,-\,} \Dvg[\zk]{\lmptf}) \teq \tauxdot{\inh{\geps}}$.
  $\gvar[\geps,\gzeta] \inh{\replm} \teq \gvar[\geps,\gzeta] \inh{\replm}^\io {\,-\,} \wwc \,\gvar[\geps,\gzeta]\Dvg[\zk]{\replmCTn}$, where $\Dvg[\zk]{\replmCTn}$ is the\HIDE{ Lagrange} multiplier for the second-class constraint $\inh{\CTw}$.
}\fi
These correspond to the Lagrange multipliers of (\ref{paramAction_AMG'}) via
$\hmg{\replm} \teq \hmg{\replm}^\io$ and $ \inh{\replm} \teq {-} \wwc \Dvg[\zk]{\replmCT}$, endowing $\hmg{\replm}$ with first-class and $\inh{\replm}$ with second-class transformation properties.

\newpar

This structure aligns with the Henneaux--Teitelboim unimodular theory \cite{Henneaux:1989zc}, the special case (\ref{UMG_part_case}) with $\wwc \teq {-}1$ where the longitudinal spatial diffeomorphisms and inhomogeneous constraint $\inhw{\CTw}$ become first-class. In this case, the first-class constraints $\sint {\geps} \CPI$ and $\sint \inh{\geps}' \wwc^{-1} \inh{\CTw}$ generate a local time-reparametrization symmetry via $\ptf_{\io} {+} \sqrg^{-1}\Hl$, with the fully functional parameter $\hmg{\geps} {\,+\,} \inh{\geps}'$. Meanwhile, the average-free components of $\tf^\io$ and $\replm {\,-\,} \Dvg[\zk]{\lmptf}$ transform under their own average-free gauge parameter $\inh{\geps}$. Notably, the transformations of the metric and auxiliary sectors are coupled only through the homogeneous gauge parameter $\hmg{\geps}$, reflecting the decoupling of the gauge-trivial sector of the average-free auxiliary fields: $\inh{\tf}^\io$, $\inh{\ptf}\io$, and $\inh{\replm} {\,-\,} \Dvg[\zk]{\lmptf}$.

\if{
The latter effect may be seen directly in the\HIDE{ $\WW=\wwc$} {\wGUMG} naively parameterized action (\ref{paramAction_naive_AMG}) by noticing that when the secondary constraint is appropriately subtracted from the time-reparametrization constraint to kill the inhomogeneous part of $\Fw\Hl$, so that the hamiltonian constraint acquires the first-class form $\ptf_\io {+} \hmg{\Fw\Hl}$ (\ref{CPI_AMG}), then the inhomogeneous sector of the auxiliary variables decouples.

To state this from another perspective one can start not from the naive parameterized action (\ref{paramAction_naive_AMG}) but with equivalent \emph{minimal parameterized action}
     % \if{
        \footnote{
         The minimal parameterized action of the form (\ref{paramAction_min_AMG}) is always equivalent to the initial extended action in general field theory (without extending the Hamiltonian and initial constraints with terms proportional to $\tf^\io_{,\zm}$) contrary to naive reparameterized form of the action (\ref{paramAction_naive_AMG}), which in \UMG\ and \wGUMG\ is equivalent to non-parameterized action by chance. See the discussion in Appendix \ref{ASect:parametrization}.
         %
         %% This becomes obvious when the canonical action first is made parameterized w.r.t. homogeneous time reparametrization. Which goes in the field theory the same way as in mechanics and needs introduction of only homogeneous variables $\hmg{\tf^\io}$, $\hmg{\ptf_\io}$, $\hmg{\replm^\io}$. And then these homogeneous quantities may be localized by adding the pure gauge action of their inhomogeneous completions (the second integral in the r.h.s. of (\ref{min_param_decoupled}).
        %% (For the general discussion on this issue see Section (\TODO{}))
        }
     % }\fi
   \bea{paramAction_min_AMG} %(\ref{parAction_BLRG})
    \SSS^{\iwwc}_{par}[\g,\pg,\tf^\io\!,\ptf_\io,\replm^\io\!,\repNs,\replmCTn]
    = \!\int \!d\taux \hspace{1pt} \dsx  %\int d\taux\, \dsx
     \,\LiB
     \pg^{\zm\zn}\tauxdot{\g}_{\zm\zn}\!
     + \ptf_{\io}\tauxdot{\tf}^\io
     - \replm^\io ( \ptf_\io {+} \hmg{\Fw \Hl} )
     - \repNs^\zm \Hs_\zm
     - \replmCTn^\zm \CTw{}_{,\zm}
     \RiB
     \,, \;
   % \nonumber\\
  \eea
in which it is clear that the auxiliary sector enters the action in the form where homogeneous and inhomogeneous parts of variables $\tf^\io$, $\ptf_\io$, $\replm^\io$  decouple
 \beq{min_param_decoupled}
   \sint % \int d\taux\, \dsx \,
   \big( \ptf_\io\,\tauxdot{\tf}^\io  - \replm^\io ( \ptf_\io + \hmg{\Fw \Hl} ) \big)
   \,=
   \sint % \int d\taux\, \dsx \,
   \big( \hmg{\ptf_\io}\,\tauxdot{\hmg{\tf^\io}}  - \hmg{\replm^\io} ( \hmg{\ptf_\io} + \hmg{\Fw \Hl} ) \big)
   +
   \sint % \int d\taux\, \dsx \,
   \big(  \inhw{\ptf_\io} \tauxdot{\inhw{\tf^\io}}  - \inhw{\replm^\io}  \HIDE{(} \inhw{\ptf_\io} \HIDE{ + 0 )} \big) \,.
  \eeq

Such mechanism is not explicit in the Henneaux--Teitelboim action for the unimodular gravity and analogously could stay implicit in {\wGUMG}. But it is worth discussing this before parameterizing the general {\GUMG} model in the next section.
}\fi

    \subsubsection*{Lagrangian gauge symmetries}
     \hspace{\parindent}
Given the restricted nature of generalized unimodular theory, we expect Lagrangian gauge transformations of the metric fields to take the form of restricted diffeomorphisms of Einstein's general relativity. A direct way to verify\HIDE{ and achieve} this is by tracking field redefinitions leading to the configuration space of the action (\ref{paramAction_AMG_fin}) and adjusting gauge parameters to match those generating standard canonical transformations in {\GR}.
This can be done\footnote{We comment on this derivation in Appendix \ref{ASSect:gauge_calcs_AMG}.}, and ultimately,  such {\GR}-like gauge transformations are induced by the canonical generator
  \beq{canon_generator_AMG''_fin}
    \gvarp[\geps,\gzeta] {(\,.\,)}
    = \PB{\,.\,}{\sint ( \Hl {+} \Fw^{-1} \CCof )} \Fw\hmg{\geps}
    + \PB{\,.\,}{\sint \Hs_\zn} \gzetat^\zn
    + \PB{\,.\,}{\sint \inh{\CCof} } \inh{\geps}
    %% + \PB{\,.\,}{\sint \ptf_{\io,\zm}} \gchi^\zm
    \,.
  \eeq
It complies with the constraint basis of (\ref{paramAction_AMG_fin}) and for the ADM metric fields imply
 \beq{gauge_transfs_ham_AMG''_GRbasis}
 \left|
  \begin{array}{lll}
   \gvarp[\geps,\gzeta] \g_{\zn\zm}
   &=& \big( \g_{\zk\zm} {\gzeta_{\isst}^{\zk}}_{,\zn}
   + \g_{\zn\zk}  {\gzeta_{\isst}^{\zk}}_{,\zm}
   + \g_{\zn\zm,\zk} \gzeta_{\isst}^{\zk}
   \big)
   + \tfrac{2}{\sqrg}\left(\pg_{\zm\zn} -\frac{1}{\Ddim{-}2}\,\trpg\,\g_{\zm\zn} \right)
   \Fw\hmg{\geps}
   \,,
   \\
  % \gvar \pg^{\zn\zm}  %% fin
  %        &=&
  %    \gvar[\Fw\hmg{\geps}]^{^{(\GR)}} \pg^{\zn\zm}
  %    +\PB{\pg^{\zn\zm}}{\sint \Fw^{-1}\ptf_\io} \Fw \hmg{\geps}
  % \\
   \gvarp[\geps,\gzeta]{\repNl} %% fin
     &=& \tTDer{}{\taux} (\Fw \hmg{\geps})
     + \repNl\!_{,\zn} \gzetat^\zn %% {inh.}
     - \repNs^\zn (\Fw \hmg{\geps})_{,\zn}
   \,, \vphantom{\big|^l}
   \\
   \gvarp[\geps,\gzeta] \repNs^\zn %% fin
    &=& \tTDer{}{\taux} \gzetat^{\zn}
       - \LieB{\zn}{\repNs}{\gzetat}
       - \g^{\zn\zm} (\repNl \ader_\zm \Fw\hmg{\geps})
    \,.
   \\
  \end{array}
  \right.
  \eeq
These\HIDE{ transformations} are equivalent to the\HIDE{ restricted} {\GR} canonical diffeomorphism transformations
 \bea{GR_canon_gauge_transfs} %% (\ref{GR_canon_gauge_transfs}) %% from GR.tex
 %% GR gauge transfs
  \begin{array}{|lcl}
    \gvarGR \g_{\zn\zm}
    &=& \big( \g_{\zk\zm} \gxs^{\zk}_{,\zn}
    + \g_{\zn\zk}  \gxs^{\zk}_{,\zm}
    + \g_{\zn\zm,\zk} \gxs^{\zk}
    \big)
    + \tfrac{2}{\sqrt{\g}} \big( \pg_{\zn\zm} - \tfrac{1}{\Ddim-2}\,\trpg\, \g_{\zn\zm}\big)\gxl
    \,,
    \\
    \gvarGR \Nl
    &=& \tTDer{}{\taux} {\gxl} + \Nl_{,\zn} \gxs^\zn - \Ns^\zn \gxl_{,\zn}
    \,,
    %\nonumber\\
    \\
    \gvarGR \Ns^{\zn}
    &=& \tTDer{}{\taux} {\gxs}^{\zn} - \LieB{\zn}{\Ns}{\gxs} - \g^{\zn\zm} (\Nl {\ader_\zm} \gxl)
    \,,
\end{array}
 \eea
with the gauge parameters subject to the following constraints:
 \beq{wGUMG_gauge_restriction_from_GR}
  \begin{array}{|lcl}
   \gxl &\to& \Fw \hmg{\geps}
   \,,
   \\
   \gxs^\zn &\to& \gzetat^\zn:
   \quad \; \partial_\zn \gzetat^\zn=0
   \,.
  \end{array}
 \eeq
The minor difference with the {\GR}\HIDE{ canonical} gauge transformations generated by\HIDE{ the Hamiltonian vector fields of} (\ref{BGUMGHamiltonianStructure}, \ref{BGUMGMomentaConstraints}), arises from the\HIDE{ $\sqrg$-dependent} term $\Fw^{-1}(\argsqrg)\, \ptf_{\io}$ \,in\HIDE{ the generator} (\ref{canon_generator_AMG''_fin}), and in the metric sector affects only the transformations of\HIDE{ gravitational} the momenta $\pg^{\zn\zm}$.
\if{
     $\gvarp[\geps,\gzeta] \pg^{\zn\zm}
          \,=\,
      \gvar^{^{(\GR)}}\!\pg^{\zn\zm} \big|_{\genfrac{}{}{0pt}{}{\gxl=\Fw\hmg{\geps}}{\gxs^\zn =\gzetat^\zn}}
      +\PB{\pg^{\zn\zm}}{\sint \Fw^{-1}\ptf_\io} \Fw \hmg{\geps}$,

which we do not expand explicitly\HIDE{ here since this field will be reduced}.
}\fi
Transformations (\ref{gauge_transfs_ham_AMG'}) and (\ref{gauge_transfs_ham_AMG''_GRbasis}) are equivalent modulo a trivial gauge transformation arising from a field-dependent redefinition of gauge parameters (see Appendix \ref{ASSect:gauge_calcs_AMG}).

\newpar
% \subsubsection*{Lagrangian gauge symmetries}
%  \hspace{\parindent}

The transformations (\ref{gauge_transfs_ham_AMG''_GRbasis}) of the ADM metric components match the restricted {\GR} canonical gauge transformations (\ref{GR_canon_gauge_transfs}, \ref{wGUMG_gauge_restriction_from_GR}), and the reduction to the Lagrangian theory:
  %%(\ref{momenta_ecK_GR})
  $\pg^{\zm\zn} \to {\sqrg}(\ecK^{\zm\zn} {\,-\,} \g^{\zm\zn} \trecK)$,
where $ \ecK_{\zm\zn} = \tfrac{1}{2\repNl} ( \tauxdot{\g}_{\zm\zn} {\,-\,} \repNs_{\zm ; \zn} {\,-\,} \repNs_{\zn ; \zm})$,
occurs in the same way as in Einstein's general relativity. Thus, the spacetime metric in the Lagrangian actions (\ref{ActionHTlike0_AMG}) and (\ref{ActionHTlike_AMG}) transforms under the restricted covariant gauge transformations
%% of the spacetime metric
  \bea{Lagrangian_metric_diffeo_transfs_AMG}
   \gvar[\breve{\gxi}] \,\Gaux_{\Zn\Zm}
  % &=& \big( \Gaux_{\Zl\Zm} \partial_\Zn
  % + \Gaux_{\Zn\Zl} \partial_\Zm
  % + (\partial_\Zl \Gaux_{\Zn\Zm})
  % \big)\gxi^{\Zl}
  % \nonumber\\
   &\!=\!& \Gaux_{\Zl\Zm} \breve{\gxi}^{\Zl}_{\,,\Zn}
   + \Gaux_{\Zn\Zl} \breve{\gxi}^{\Zl}_{\,,\Zm}
   + \Gaux_{\Zn\Zm,\Zl} \breve{\gxi}^{\Zl}
  \eea
with the constrained covariant gauge parameters \,$ %%\gxi^{\Zn}\equiv \big(\gxi^0,\gxi^\zn \big)  \rightarrow
\breve{\gxi}^{\Zn}(\hmg{\geps},\gzetat^\zn)$:
  \bea{Lagrangian_diffeo_restricted_params_AMG}
    %% see Eq \ref{EGr_gauge_diffeo_correspondence_AMG}
    \left| \begin{array}{lclclcl}
              \gxi^\zn
              &\!=\!&
              \gxs^\zn - \frac{\repNs^\zn}{\repNl} \gxl
              &\;\rightarrow\;&
              \breve{\gxi}^\zn &\!\defeq\!& \gzetat^\zn - \frac{\repNs^\zn}{\repNl} \Fw \hmg{\geps} \;, \\
              \gxi^0
              &\!=\!&
              \frac{1}{\repNl} \gxl
              &\;\rightarrow\;&
              \breve{\gxi}^0 &\!\defeq\!& \frac{1}{\repNl} \Fw \hmg{\geps} \;.
            \end{array}
    \right.
  \eea

\newpar

The gauge transformations of the auxiliary sector variables $\CCof$ and $\HTf^\Zm$ follow from (\ref{gauge_transfs_ham_AMG'}) and has the same form in the canonical and Lagrangian theories.
The cosmological-constant field $\CCof \tequiv \ptf_\io$
is gauge-invariant,
  \beq{gvar_ccf0}
   \gvarp[\geps,\gzeta] \CCof \,=\, 0
   \, .
  \eeq
The canonical gauge transformation of the effective cosmological-constant field $\CCf = \sqrg^{-1}\!\Fw^{-1}\! \CCof$, (\ref{ccf_AMG}), depends on momenta $\pg^{\zm\zn}$. After restricting to the reduction surface %% (\ref{momenta_ecK_GR})
it takes the\HIDE{ Lagrangian} form
  \bea{gvar_ccf}
   \gvarp[\geps,\gzeta] \CCf
    \,=\,
    - (\wwc{+}1) \,\CCf\,
     \Big( \tfrac1{\sqrg}
           \big( \tTDer{}{\taux} { \sqrg }
              {-}  (\sqrg \repNs^\zn)_{,\zn}
           \big) \tfrac{1}{\repNl} \Fw \hmg{\geps}
          \,+\, \tfrac1{\sqrg}  (\sqrg\gzetat^\zn)_{,\zn}
     \Big)
    \,.
  \eea

The components of the Henneaux--Teitelboim auxiliary field $\HTf^\Zm \tequiv \big(\HTf^0, \HTf^\zm \big) \teq \big(\tf^\io, \lmptf^\zm \big)$, (\ref{def_tef}),
transform as
  \beq{gvar_tef}
   \begin{array}{|lcl}
    \gvarp[\geps,\gzeta]  \tf^\io  &=& \geps
    \;=\; \hmg{\geps} - \Dvg[\zk]{\gchi}
    \,, \vphantom{\big|}
    \\
    \gvarp[\geps,\gzeta,\gchi] {\lmptf^\zm}\!\!
    &=&
     \tTDer{}{\taux} \gchi^\zm +   \HIDE{\replm\to}\repNl \Fw^{-1} \gzetat^\zm + \wwc \repNs^\zm \hmg{\geps}
    \,,  \vphantom{\big|^I}
   \end{array}
    \qquad \big(\,\Dvg[\zk]{\gchi} = - \inh{\geps}\,\big)
  \eeq
where we introduced $\gchi^\zm$ --- an arbitrary infinitesimal vector gauge parameter with divergence\HIDE{ equal to} $-\inh{\geps}$.
The arbitrariness of the transverse (divergence-free) part, $\gchit^\zm$, reflects an ambiguity in $\lmptf^\zm$ when its divergence is fixed. The canonical action
(\ref{paramAction_AMG'})
and the Lagrangian actions (\ref{ActionHTlike0_AMG}) and (\ref{ActionHTlike_AMG}) depend only on its spatial divergence, $\Dvg[\zk]{\lmptf}$,  meaning the divergence-free part is redundant\footnote{
This redundancy, introduced to maintain locality, originates from the {\GUMG} extended action (\ref{extAction_GUMG}), where the secondary constraint (\ref{secondary_GUMG_inh}) appears in gradient form through the term $-\sint \lmCT^\zn \CT_{,\zn}$ which is equivalent to $ \sint (\Dvg[\zk]{\lmCT}) \inh{\CT}$.
The same discussion is applicable to the Henneaux--Teitelboim action for unimodular gravity \cite{Henneaux:1989zc}.
}.

 % \subsubsection*{Final notes on the symmetry structure}
 %  \hspace{\parindent}
  %
\newcommand{\gChiT}{{\color{gColor}\mathcal{X}}_{\scriptscriptstyle(T)}}

Moreover, in the Lagrangian actions  (\ref{ActionHTlike0_AMG}) and (\ref{ActionHTlike_AMG}), the Henneaux--Teitelboim auxiliary field $\HTf^\Zm$ appears only through $\partial_\Zm \HTf^\Zm$, which transforms as
 \beq{gvar_DVG_tef}
   \gvarp[\geps,\gzeta] \, {\partial_\Zm \HTf^\Zm}
   \;=\;
   \tTDer{}{\taux} \hmg{\geps} + \HIDE{\replm\to}(\repNl \Fw^{-1})_{,\zm} \gzetat^\zm + \wwc \Dvg[\zk]{\repNs} \,\hmg{\geps}
   \,.
 \eeq
In the right side of (\ref{gvar_DVG_tef}) the gauge parameter $\inh{\geps}$ cancels. Since $\inh{\geps}$ does not appear in the transformations of other fields (\ref{gauge_transfs_ham_AMG''_GRbasis}, \ref{gvar_ccf}), the\HIDE{ whole} associated
gauge symmetry is also the internal\HIDE{ gauge} symmetry of the
auxiliary field $\HTf^\Zm$. Alternatively, the spacetime divergence
implies the gauge redundancy
  \beq{}
   \HTf^\Zm \to \HTf^\Zm + \gChiT^\Zm \,,
   \qquad \big(\,\partial_\Zm \gChiT^\Zm = 0\,\big) \,,
  \eeq
where 
the\HIDE{ divergence-free covariant} condition $\partial_\Zm \gChiT^\Zm = 0$ implies that components of $\gChiT^\Zm$ are parameterized by a spatial vector $\gchi^\zm$: $\inh{\gChiT^0 \!\!\!}\, \teq {-}\Dvg[\zk]{\gchi}$ and $\gChiT^\zm \teq \tauxdot{\gchi}^\zm $, inducing the internal symmetry $\inh{\HTf^0\!} \to \inh{\HTf^0\!}\,  {\,-\,}\Dvg[\zk]{\gchi}$ and $\HTf^\zm\! \to \HTf^\zm {+\,} \tTDer{}{\taux} {\gchi}^\zm $.
Notably, this covariantly isolated noninfinitesimal\HIDE{ auxiliary} symmetry, parameterized by $\gchi^\zm$, incorporates two symmetries: the $\Dvg[\zm]{\gchi} \teq {-}\inh{\geps}$ symmetry of the gauge-trivial sector (\ref{inh_trivial_sector}) and the noncanonical
transverse $\gchit^\zm$ symmetry\HIDE{ originated in the extended action} due to the redundancy when incorporating the secondary constraint.
The rearrangement of the full gauge basis
%of the parameterized action
$(\geps,\gzetat^\zm,\gchit^\zm) \to (\hmg{\geps},\gzetat^\zm, \gchi^\zm)$ disentangles the internal auxiliary-sector symmetry from that affecting the metric fields.

In this section we discussed the symmetry on the physically interesting non-GR branch, $\CCf\propto\CCof \tneq 0$, which describes gravity in the presence of a cosmological perfect fluid. The GR branch,\HIDE{ tied to subspace} $\CCf\propto\CCof \teq 0$, retains full diffeomorphism invariance\HIDE{ of the general relativity} in the metric sector \cite{Barvinsky:2019agh}. In the canonical treatment this follows from the Poisson bracket matrix of constraints (\ref{GUMG_Algebra_onshell}, \ref{parAMG_Algebra_offshell}),
 %% (see table \ref{table:Par_PB_ww}), matrix of
which becomes zero-rank on shell for ${\CTw}\propto \Hl =0$, equivalent to $\CCof \teq \ptf_\io \teq 0$.
Since $\CCof$ is gauge-invariant (\ref{gvar_ccf0}), the branch division is itself gauge-invariant.

\if{
\footnote{
 We are not going to discuss the physical validity of the branch interpretation of the rank discontinuity here. Only note that the additional invariance of the GR branch may be easily killed by the interaction, which respects only the restricted gauge symmetry, which would restore the regularity of the dynamical and gauge content of the theory. However, if the coupled external matter and interaction with it respects the general coordinate covariance the two branch interpretation keeps valid.
}
}\fi

%%-------------------------=%        ******         %=-------------------------%%
%%-------------------------=%        ******         %=-------------------------%%

 \subsection{Gauge invariance check}
 \label{SSect:Gauge_Inv_Check_AMG}
  \hspace{\parindent}
The gauge transformations
(\ref{Lagrangian_metric_diffeo_transfs_AMG}\,--\,\ref{gvar_tef})
derived from the canonical treatment
can be verified\HIDE{ directly} by directly checking the invariance of the alternative action. We choose the representation $\SSS^{\iwwc}_{\ialt_0}[\Gaux,\CCof,\HTf]$, (\ref{ActionHTlike0_AMG}), which is more convenient due to the gauge invariance of $\CCof$. The {\wGUMG} action splits into diffeomorphism-invariant and gauge-breaking parts:
 \beq{S0=SGR+Smat0}
  \SSS^{\iwwc}_{\ialt_0} = \SSS^{\iGR} + \SSS^{\iwwc}_{mat_0}
  =
   \int \!d\taux \hspace{1pt} \dsx %\int d\taux\, \dsx \,
    \sqrt{|\Gaux|}\,  \stR (\Gaux)
   +
   \int \!d\taux \hspace{1pt} \dsx  %\int d\taux\, \dsx
   \,\big( {-} \sqrt{|\Gaux|}\sqrg^{-1}\!\Fw^{-1} \!+ \partial_\Zm \HTf^\Zm \big) \CCof
   \,,
 \eeq
where $\SSS^{\iGR}[\Gaux_{\Zm\Zn}]$ is the Einstein-Hilbert action of general relativity,
and the\HIDE{ symmetry breaking} restriction term $\SSS^{\iwwc}_{mat_0} [\sqrg,\repNl\!,\HTf, \CCof]$ can be interpreted as the cosmological matter action.
The canonical treatment implies that metric fields transform under the restricted diffeomorphism symmetry
 (\ref{Lagrangian_metric_diffeo_transfs_AMG}, \ref{Lagrangian_diffeo_restricted_params_AMG}).
Since the Einstein action
is fully diffeomorphism-invariant\HIDE{ with respect to the unrestricted and any restricted diffeomorphism transformations}, the only remaining check is the invariance of the matter part. As the local field $\CCof$ is gauge-invariant, (\ref{gvar_ccf0}), and acts as an overall factor in the integrand of $\SSS^{\iwwc}_{mat_0}$, it suffices to show that the gauge variation of the local expression in parentheses vanishes. This holds because the restricted gauge transformations (\ref{gauge_transfs_ham_AMG''_GRbasis})
in the canonical gauge basis $(\geps,\gzeta^\zn)$ imply
 \bea{}
   &&\gvarp[\geps,\gzeta] (\sqrt{|\Gaux|} \sqrg^{-1}\! \Fw^{-1})
   \,=\,   \gvarp[\geps,\gzeta] (\repNl \Fw^{-1})
   %\,=\,   \gvar\repNl \cdot\Fw^{-1}+\Nl\gvar \Fw^{-1}
%   \nonumber\\
 \if{
   &&\,=\,
    \uwave{\Fw^{-1}\tTDer{}{\taux} (\Fw \hmg{\geps})}
     + (\Fw^{-1}\repNl)_{,\zn} \gzetat^\zn
     - \cancel{\Fw^{-1}\repNs^\zm \partial_\zm \Fw \hmg{\geps}}
       + \uwave{\tTDer{}{\taux} \Fw^{-1} {\cdot} {\Fw\hmg{\geps}}}
          + \wwc \Dvg{ \repNs} \hmg{\geps}
          - \cancel{\repNs^\zn \partial_\zn \Fw^{-1}\! {\cdot}  {\Fw\hmg{\geps}}}
   \,.
   \quad
 }\fi
 %   &&
 \,=\,
    \tTDer{}{\taux} \hmg{\geps}
     + (\repNl\Fw^{-1})_{,\zn} \gzetat^\zn
        + \wwc \,\Dvg[\zk]{\repNs}\, \hmg{\geps}
   \,,
   \quad
 \eea
which exactly cancels\footnote{
 This cancellation is more transparent in the canonical representation (\ref{paramAction_AMG'}), where
  $%%\sqrt{|\Gaux|} \sqrg^{-1}\! \Fw^{-1} \tequiv
  \repNl \Fw^{-1}$
  is simply $\replm$, (\ref{replm_ptf_io_rescaling_AMG}), and transformations (\ref{gauge_transfs_ham_AMG'}) directly imply $- \gvar[\geps,\gzeta] {\replm} + \gvar[\geps,\gzeta] \big( \tauxdot{\tf}^\io {+\,} \Dvg[\zk]{\lmptf} \big) \teq 0$.
}
$ \gvarp[\geps,\gzeta\,] \partial_\Zm \HTf^\Zm$, (\ref{gvar_DVG_tef}).

\newpar

It is also instructive to examine the {completeness} of the obtained gauge symmetries from the perspective of the restricted-theory approach \cite{Barvinsky:2022guw}.
For a regular restriction of the gauge theory\HIDE{ with the closed gauge algebra}, one expects the restricted theory's gauge algebra to be a subalgebra of the parental theory's gauge algebra. The residual symmetries can be identified by solving the explicit linear constraints on the infinitesimal gauge parameters. For the representation (\ref{Action_GUMG_L}), this was performed in \cite{Barvinsky:2019agh} and led to a discussion of an apparent paradox: a mismatch between naive gauge restriction counting and actual number of canonical symmetries\footnote{
 In the original {\GUMG} approach (\ref{Action_GUMG_L}), imposing the one-dimensional algebraic restriction (\ref{GUMG_restriction}) on a parental theory with a closed gauge algebra could be expected to break at most one local gauge degree of freedom. However, the canonical treatment suggests that on the non-GR branch, two gauge symmetries are broken.
}.
This\HIDE{ problem} was resolved by the observation that one of the naive residual gauge symmetries violates spacetime locality.

In the alternative representation, a similar investigation of residual symmetries is possible, though the restriction mechanism is a bit subtler. One can vary the alternative {\wGUMG} action (\ref{ActionHTlike0_AMG}) under the unrestricted {\GR} gauge transformations (\ref{GR_canon_gauge_transfs}), adjusting auxiliary field variations to cancel as many structures as possible while preserving time locality. The anomaly in this variation imposes constraints on the gauge parameters of the parental symmetry.

Due to the gauge invariance of the first term, $\SSS^{\iGR}$, in\HIDE{ decomposition} (\ref{S0=SGR+Smat0}) under the unrestricted transformations (\ref{GR_canon_gauge_transfs}), a gauge anomaly $A_\iwwc$:\:
   $
    \int dt\, A_\iwwc
    {\defeq}
      \gvarGR \SSS^{\iwwc}_{\ialt_0}
    =
      \HIDE{ 0 +} \gvarGR \SSS^{\iwwc}_{mat_0}
     \,,
   $
arises solely from the variation of the matter term.
The metric gauge variation\HIDE{ (\ref{GR_canon_gauge_transfs})} in the first term $ \sqrt{|\Gaux|}\sqrg^{-1}\!\Fw^{-1} \teq \repNl\Fw^{-1}$,
  \beq{ugvar_repNl_invFw}
   %\hspace{-4mm}
   \gvarGR (\repNl\Fw^{-1}\hspace{-1pt})
   \,=\,
   \if{
    \tTDer{ }{\taux} (\Fw^{-1} {\gxl})
    -  ( \repNs^\zn \Fw^{-1}\gxl)_{,\zn}
    \! +  (\wwc{+}1) \Fw^{-1} \Dvg{\repNs}  \gxl
    \!+ (\repNl\Fw^{-1}\gxs^\zn)_{,\zn}
    \!- (\wwc{+}1) \repNl\Fw^{-1}  \Dvg{ \gxs}
    \nonumber\\
    &=&
    }\fi
    \tTDer{ }{\taux} (\Fw^{-1} \! {\gxl})
    -  ( \repNs^\zn \Fw^{-1} \! \gxl)_{,\zn}
    \!+ (\repNl\Fw^{-1} \! \gxs^\zn)_{,\zn}
    \! +  (\wwc{+}1) \big( \Fw^{-1} \! \Dvg[\zk]{\repNs} \gxl
          \!-  \repNl \Fw^{-1} \! \Dvg[\zk]{\gxs} \big)
    \,,
  \eeq
combined with compensating transformations of the auxiliary fields,
  \beq{gvar_dvgHTf}
    \begin{array}{|lcl}
     \gvarGR \tauxdot{\tf}^\io \!\!\!
     &=& \!\! \tTDer{ }{\taux} (\Fw^{-1} {\gxl}) - \tTDer{ }{\taux} \Dvg[\zk]{\gchi}
     \,,\\
     \gvarGR \Dvg[\zk]{\lmptf} \!\!\!
     &=& \!\! \tTDer{ }{\taux} \Dvg[\zk]{\gchi}
      - (\repNs^\zn  \Fw^{-1}\gxl)_{,\zn}
      \!+ (\repNl \Fw^{-1} \gxs^\zn)_{,\zn}
     %\\
     %&&
      \!+ (\wwc{+}1)
            \big( \inh{\Dvg[\zk]{\repNs} \Fw^{-1}\gxl}
            {-}  \inh{\repNl\Fw^{-1} \Dvg[\zk]{\gxs}} \big)
      \,, \\
    \end{array}
  \eeq
produces the anomaly
  \beq{Anomaly_AMG} % (\ref{Anomaly_AMG})
   A_\iwwc
    =
    (\wwc{+}1)
    \HIDE{\int d\taux\, \dsx }\, %\LiB
      \sint
      \big( \repNl \Fw^{-1} \Dvg[\zk]{\gxs} - \Dvg[\zk]{\repNs} \Fw^{-1}\gxl \big)\, \hmgs{\CCof}
    \,.
  \eeq
In the right-hand side, the overall factor of the anomaly is the spatially homogeneous part of ${\CCof}$,
because\footnote{
 To obtain\HIDE{ the form} (\ref{Anomaly_AMG}), the identity $\sint f \,\hmg{h} = \sint \hmg{f} \,h$ was used.
} the spatially homogeneous part of the terms proportional to $(\wwc{+}1)$ remain uncompensated by $ \gvarGR \partial_\Zm \HTf^\Zm $.

Transformation laws of $\HTf^\Zm$ are not determined by the parental theory, so their form is arbitrary. However, the divergence structure $\partial_\Zm \HTf^\Zm$ imposes limitations, as its gauge variation must preserve time locality and the average-free nature of the spatial divergence.
The transformations (\ref{gvar_dvgHTf}) were chosen to cancel terms from $\gvarGR (\repNl\Fw^{-1})$ as much as possible. The anomaly (\ref{Anomaly_AMG}) would be eliminated only if the time-nonlocal term\footnote{
 Or, equivalently, one may assume\HIDE{ transformations}
   $\gvarGR \tauxdot{\tf}^\io \teq \tTDer{ }{\taux} (\Fw^{-1} {\gxl}) - \tTDer{ }{\taux} \Dvg[\zk]{\gchi} + (\wwc{+}1) \int^{\taux}\!\! d\taux' \big( { \repNl \Fw^{-1} \Dvg[\zk]{\gxs} - \Dvg[\zk]{\repNs} \Fw^{-1}\gxl} \big)(\taux')$
   and
   $\gvarGR\! \Dvg[\zk]{\lmptf} \teq \tTDer{ }{\taux} \Dvg[\zk]{\gchi}
      {\,-\,} (\repNs^\zn  \Fw^{-1}\gxl)_{,\zn}
      {\,+\,} (\repNl \Fw^{-1} \gxs^\zn)_{,\zn}\,
   $. However, the initially given formulation more clearly highlights that canceling the anomaly requires modifying only the\HIDE{ canonical} transformation of the homogeneous mode $\hmg{\tf}^\io$.
}
$(\wwc{+}1) \int^{\taux}\! d\taux'\, \big( \hmg{ \repNl \Fw^{-1} \Dvg[\zk]{\gxs} - \Dvg[\zk]{\repNs} \Fw^{-1}\gxl} \big)(\taux')$ is added to $ \gvarGR {\tf}^\io $.
Local in time improvements to the auxiliary field's gauge transformations
cannot\HIDE{ eliminate} cure the anomaly in the {\wGUMG} models with $\wwc \tneq {-}1$.
Changing the metric diffeomorphism transformations also fails to eliminate the anomaly in the matter term. The two anomaly terms in (\ref{Anomaly_AMG}) are proportional to the {\GUMG}-specific consistency conditions, which cannot be expressed as linear combinations of the Einstein tensor ${\var \SSS^{\iGR}}\!/{\var \Gaux_{\Zm\Zn}}$.
 \HIDE{Thus this would inevitably invoke structures proportional to the new equations of motion in the anomaly.}

\if{ %% version 2024-12
  Improving the metric diffeomorphism transformations will also not eliminate the anomaly in the gauge-breaking term. One argument is that an arbitrary variation of the metric in the {\GR} term of (\ref{S0=SGR+Smat0}) will be proportional to the left-hand sides of Einstein's gravity equations of motion, $\VDer{\SSS^{\iGR}}{\Gaux_{\Zm\Zn}}$. However, none of the linear combinations of these equation-of-motion structures can impose a constraint on the canonical Lagrange multipliers $\inh{\replm}\equiv\inh{\repNl \Fw^{-1}}$ and $\Dvg{\repNs}$.\footnote{
   These Lagrange multipliers are factors at the first-class constraints in canonical general relativity, so they are not defined within variational equations of motion of this theory, which implies that they cannot be expressed as linear combinations of $\VDer{\SSS^{\iGR}}{\Gaux_{\Zm\Zn}}$.
  } At the same time the two anomaly terms (\ref{Anomaly_AMG}) are proportional to these combinations, which are the differential consequences of the {\wGUMG} equations of motion, (\ref{LM_rep_eom_AMG}), on the non-GR branch of this theory.\footnote{
   Note that the first term in the anomaly (\ref{Anomaly_AMG}) does not depend on the on-shell nonvanishing spatially homogeneous part of ${\replm}={\repNl \Fw^{-1}}$, since $\sint \hmg{\repNl \Fw^{-1}} \Dvg{\gxs} \hmgs{\CCof} \equiv 0$ due to the average-free property of the divergence $\Dvg{\gxs}$. It is important that the anomaly reduced to the form proportional only to the spatially homogeneous part of $\CCof$.
  }
  Thus, no improvement of the canonical transformations of $\Gaux_{\Zm\Zn}$ from the $\SSS^{\iGR}$ part of the {\wGUMG} alternative action (\ref{S0=SGR+Smat0}) can compensate contributions proportional to $\inh{\repNl \Fw^{-1}}$ and $\Dvg{\Ns}$ from the gauge variation of $\SSS^{\iwwc}_{mat_0}$. An arbitrary (either gauge or anomalous) variation of the action (\ref{S0=SGR+Smat0}) is proportional to a linear combination of the equations of motion of the theory, and thus terms proportional to the two new equations of motion of the restricted theory (\wGUMG) generally appear in the variation of the gauge-breaking term $\SSS^{\iwwc}_{mat_0}$ (since, as we noted, they cannot appear in the nonsingular variation of the parental (\GR) part $\SSS^{\iGR}$.
   %% \TBC{If none of these on-shell-trivial\HIDE{ equation-of-motion} terms appear then $\SSS^{\iGR} + \SSS^{\iwwc}_{mat_0}$ can be made gauge-invariant w.r.t. to improved parental gauge transformations.}
   %% Conversion w.r.t. to auxiliary canonical phase-space variables?
   %% What if one of theme disappear? What if linear combintaion disappear& is that possible?
  %\TODO{edit the previous paragraph}

 %% %DRAFT
 %% Due to this reason the anomaly cannot be cancelled via improving  the gauge transformations for auxiliary fields or via extending the diffeomorphism generators, acting on metric, with terms proportional to $\CCof$  without violation of the locality in time.
}\fi

\if{ %%% DRAFT 2024-09

  Here we skip the details of calculations, which can be found in the Appendix \ref{ASSect:SymmetryCheckUnrestricted} where we perform the canonical derivation. Finally, the result of varying the action
    $ % \bea{paramAction_AMG_fin_COPY} %(\ref{parAction_BLRG})
      \SSS^{\iwwc}_{par''} [\g,\pg,\tf^\io\!,\ptf_\io,\repNl,\repNs,\lmptf]
      \, = %\!\!&=& \!\!\!
      \int d\taux\, \dsx \, \LiB
       \pg^{\zm\zn}\tauxdot{\g}_{\zm\zn}
       + \ptf_\io \tauxdot{\tf}^\io
       - \repNl ( \Hl {+} \Fw^{-1}\ptf_{\io})
       - \repNs^\zm \Hs_\zm
       - \lmptf^\zm {\ptf_\io}_{,\zm}
       \RiB %.\;
    %\nonumber
   $, % \eea
    (\ref{paramAction_AMG_fin}),
    %
    %% Constraint set $\big(\,( \Fw^{-1}\ptf_{\io} {+} \Hl ),\, \Hs_\zm,\, {\ptf_\io}_{,\zm} \big) $ generate gauge transformations of the metric components in the form close to the canonical \GR\ transformations. Also since this action contains all ADM metric variables and the Legendre reduction to the Lagrangian action is performed via \emph{the same} equation $\pg^{\zn\zm}=\pg^{\zn\zm}(\g,\repNs,\repNl)$ one immediately gets the \GR-like covariant gauge transformations for space-time metric components.
    %
   % from
   %\subsection{Gauge invariance (\AMG'') in ``GR'' constraint basis}
   %\label{ASSect:GaugeInvAMG''}
  with \emph{unrestricted} canonical transformations of phase space fields generated by all constraints\footnote{
  Since the I-class constraint set is just the subset of these constraints then the genuine canonical gauge symmetry obtained by the restriction of the parameters. The generality of this consideration is based on the fact that these are the symmetries of the parental gauge theory (permissibly extended to be canonical) before adding the gauge-breaking term.}
  %(\ref{gauge_transfs_ham_AMG'})
    \beq{canon_transf_AMG''}
      \gvar {(\,.\,)} = \PB{\,.\,}{\sint ( \Hl {+} \Fw^{-1} \ptf_{\io} )} \gxl + \PB{\,.\,}{\sint \Hs_\zm} \gxs^\zm + \PB{\,.\,}{\sint \ptf_{\io,\zm}} \gchi^\zm
    \eeq
  and correspondent compensatory gauge transformations of the canonical Lagrange multipliers \TODO{ add ref to Appendix}

}\fi

Finally, since $\inh{\repNl \Fw^{-1}}$ and $\Dvg[\zk]{\repNs}\HIDE{\equiv \partial_\zm \repNs^\zm}$ appear as the left-hand sides of independent equations of motion on the non-GR branch, (\ref{LM_rep_eom_AMG}), anomaly cancellation for $\wwc \tneq {-}1$ theories imposes two constraints on the gauge parameters\footnote{
 Note that the first term in the anomaly (\ref{Anomaly_AMG}) is independent of the on-shell nonvanishing spatially homogeneous part of ${\replm} \teq {\repNl \Fw^{-1}}$, since $\sint \hmg{\repNl \Fw^{-1}}\, \Dvg[\zk]{\gxs} \,\hmgs{\CCof} = 0$ due to the average-free property of the divergence $\Dvg{\gxs}$. Thus, the anomaly reduces solely to terms proportional to the spatially homogeneous part of $\CCof$.
}:
  \beq{gauge_param_restrict}
    \begin{array}{|l}
     \,\Dvg[\zk]{\gxs} =\, 0 \,,
     \quad %
     \Rightarrow\;\;\gxs^\zm =\, \gzetat^\zm,
     \\
      \,\gxl = \Fw \hmg{\geps} \,.
    \end{array}
  \eeq
These correspond to the canonical gauge transformations (\ref{canon_generator_AMG''_fin}): the metric fields transform with the spatially homogeneous parameter $\hmg{\geps}$, while spatial diffeomorphisms are restricted to transverse ones, associated with coordinate transformations of unit Jacobians.
Besides these residual transformations
the auxiliary sector retains the internal symmetry parameterized by $\gchi^\zn$, which incorporates $\inh{\geps}$ from (\ref{canon_generator_AMG''_fin}) via $\inh{\geps} \teq {-}\Dvg[\zk]\gchi$.

The constraints (\ref{gauge_param_restrict}) on the parameters of the parental gauge symmetry confirm that all local (in time) symmetries in the physical sector are exhausted by the canonical symmetries of the theory, supporting the\HIDE{ analogous} conclusions of \cite{Barvinsky:2022guw} regarding the reasons for the  additional reduction of the gauge symmetry on the non-GR branch.

On the non-GR branch, where $\hmgs{\CCof}$ vanishes, the anomaly disappears, indicating the restoration of full diffeomorphism invariance. In the {\UMG} case, (\ref{UMG_part_case}), where $\wwc \teq {-}1$, the anomaly (\ref{Anomaly_AMG}) does not arise, consistent with the fact that the Henneaux--Teitelboim action for unimodular gravity retains full diffeomorphism invariance \cite{Henneaux:1989zc}.

\if{ %% FOOTNOTES
\footnote{
 Unrestricted transformations in the form (\ref{gvar_GR_aux_sector})  can be obtained from canonical action, in which after varying canonical variables with
  \( %\beq{canon_transf_AMG''}
    \gvar {(\,.\,)} = \PB{\,.\,}{\sint ( \Hl {+} \Fw^{-1} \ptf_{\io} )} \gxl + \PB{\,.\,}{\sint \Hs_\zm} \gxs^\zm + \PB{\,.\,}{\sint \ptf_{\io,\zm}} \gchi^\zm
    \,,
  %  \nonumber
 \) %\eeq
variation laws for the canonical Lagrange multipliers are defined from the maximal possible compensation of the non-cancelled terms in the action variation.}

\footnote{
 These Lagrange multipliers are factors at the first-class constraints in canonical general relativity, so they are not defined within variational equations of motion of this theory, which implies that they cannot be expressed as linear combinations of $\VDer{\SSS^{\iGR}}{\Gaux_{\Zm\Zn}}$.
}
}\fi

\if{ %% 2023-06
  \subsection{Exceptional cases}

  Section \ref{SSect:Except_Partic_Cases}

  The obtained alternative action will contain the Henneaux--Teitelboim action as a regular limiting case of the action with $\wwc =-1$. It may be seen by direct comparison with the Henneaux--Teitelboim action for unimodular gravity \TODO{Add ref}
  \\

  The cold pressureless dust exceptional case with $\wwc =0$ is also described by the alternative action (\ref{ActionHTlike_AMG}) with $\Fw=1$.

  This needs comments because for this exceptional case Dirac procedure does not generate secondary constraints in the canonical action and all $\Hs_\zk$ are first class (see Sect. \ref{SSect:Except_Partic_Cases}). However, one can partially break the gauge invariance (namely --- longitudinal spatial diffeomorphisms) with the $ (B\Hl)_{,\zk}$ canonical gauge condition where $B$ --- generically some nonnull function of phase space variables.
  %%(full-rank linear operator?).
   We choose it too be any nonnull function\footnote{
   However, the simplest natural choice here is $B=1$.
  }
  of $\sqrg$: $B(\argsqrg) >0$, so that $\TDer{\ln B}{\ln \sqrg} + 1 \neq 0$. The latter guarantees transversality of the partial gauge-fixing constraint surface $ (B\Hl)_{,\zk}=0$ to the longitudinal spatial diffeomorphisms (cf. the middle equation of (\ref{GUMG_Algebra_onshell}) and condition $\OOmega\neq0$ (\ref{def_OOmega_TTheta}) for $\Fw=1$, $\WW=B$.) After such partial gauge fixing the action for the theory acquires the form similar to extended action of generic {\GUMG} model. Namely for $B=1$ this exactly corresponds to the desired particular model for $\Fw=1$, $\WW=0$. (But note that in this particular case the choice of $B$ is actually generic, since it is not a secondary constraint, predefined by the canonical action of primary constraints on Hamiltonian, but a partial gauge-fixing which is a matter of our choice.

  %%\footnote{For $B=B(\argsqrg)$ one gets the simple condition of transversality of the partial gauge-fixing constraint surface to the transverse spatial diffeomorphism action: one just needs that $\TDer{\ln B}{\ln \sqrg} + 1 \neq 0$. (Cf. the middle equation of (\ref{GUMG_Algebra_onshell}) and condition $\OOmega\neq0$ (\ref{def_OOmega_TTheta}) for $B=\Fw\WW$.) However, the simplest and natural choice here is $B=1$.}

  This is not a regular gauge breaking. Gauge-breaking condition $\inh{\Hl}=0$ is transversal to $\Hs_\zk=0$ everywhere except the surface $\Hl=0$. At the complete constraint surface this subspace is defined by additional condition $\hmg{\Hl}=0$ where the rank of the matrix of constraints vanish (cf. (\ref{def_OOmega_TTheta}) for $\Fw\WW=B$).

  However, this fact does not prevent one from parameterizing the partially-gauge-fixed theory. Since we proceed from the extended-like action (\ref{extAction_GUMG}) (with $\Fw=1$ but arbitrary $\WW\to B$) parametrization goes the same way as for {\wGUMG}.\TBC{\footnote{The $B(\argsqrg)$ will enter as the parameter for operator $\opEE'$ instead of $\WW$ (here $\Fw=1$). When changing the basis of constraints one gets the same $\ptf_{,\zn}=0$ (which follows from Corollary \ref{Corollary:inhIWE-1f=0}). The footprint of $B$ is left only in internal structure of operator $\opEE'$.}}

  Thus, the ordinary situation may be restored in this case by the additional unfortunate rank-irregular gauge fixing. After that the parametrization and further procedures can be continued as for the generic {\wGUMG} case.

}\fi

%%-------------------------=%        ******         %=-------------------------%%
%%-------------------------=%        ******         %=-------------------------%%
%%-------------------------=%        ******         %=-------------------------%%

\newpage
 \section{Alternative Action for General {\GUMG} Theories}
  \label{Sect:AltActionGenGUMG}
   \hspace{\parindent}
In this section, we extend the derivation of the alternative parameterized action\HIDE{ from the previous section} to general {\GUMG} theories with nonconstant $\WW(\argsqrg)$. However, certain complications deserve special attention. The local parameterization procedure is generally inapplicable, so we analyze its failure
and use a consistent parameterization scheme, also discussed in Section \ref{SSect:ParameterizedAction_AMG}.
This clarifies the origin of \emph{spatial} nonlocality,
which affects the dynamics and further complicates the gauge structure. Note that the description remains local in time, so all nonlocalities discussed in this section are strictly \emph{spatial} nonlocalities.

  \subsection{Parameterized canonical action\HIDE{ (nonlocal minimal parametrization)} }
   \label{SSect:ParameterizedAction}
   \hspace{\parindent}
As in {\wGUMG} theory, we first seek a parameterized model that is physically equivalent to the initial extended canonical theory (\ref{extAction_GUMG}). However, for a generic {\GUMG} model, a common field-theoretic issue arises: the naive parameterized action with the local constraints,
 $ %  \bea{paramAction_naive} %(\ref{paramAction_naive_AMG})
    \SSS^{loc}_{par} %%[\g,\pg,\tf^\io\!,\ptf_\io,\replm^\io,\repNs,\replmCTn]
    = \! \int \!d\taux \hspace{1pt} \dsx %\int d\taux\, \dsx
    \,\Lib
     \pg^{\zm\zn}\tauxdot{\g}_{\zm\zn}
     \!+ \ptf_\io\,\tauxdot{\tf}^\io
     - \replm^\io ( \ptf_\io {+} {\FF \Hl} )
     - \repNs^\zm \Hs_\zm
     \!- \replmCTn^\zm {\CT}_{,\zm}
     \Rib
     \,,
 %  \qquad% \nonumber\\
 $ % \eea
fails to be preserve this equivalence.
To analyse the problem, consider the Poisson bracket relations for the new constraint,
  \bea{parGUMG_Algebra_offshell} %From GUMGr_NDV file \ref{PB_GUMGr_B1}
   \begin{array}{lll}
     \PB{\sint f(\ptf_\io {+} {\FF \Hl})}{\sint g(\ptf_\io {+} {\FF \Hl})}
     \; = \; % & {\!=\!} &
     \sint ( f \ader_{\zn} {g} )\, \FF^2 \g^{\zn\zm} \Hs_\zm
     \;\; \weq \;\; 0
    \,, \vphantom{\big|^I}
    \\
     \PB{\sint {\xi}^\zm\Hs_\zm}{\sint g(\ptf_\io {+} {\FF \Hl})}
    \; = \; % & {\!=\!} &
       \HIDE{+} \sint g_{,\zn} \xi^\zn \tfrac1\WW\, \CT
       - \sint g \Dvg{\xi} \,\CT
     \;\; \weq \;\;
      \sint g_{,\zn} \xi^\zn \frac{\WW{+}1}{\WW}\, \hmg{\CT}
    \,, \vphantom{\big|^I}
    \\
     \PB{\sint \eta^{\zm}{\CT}_{,\zm}}{\sint g (\ptf_\io {+} {\FF \Hl})}
     \; = \; % & {\!=\!} &
      \HIDE{+} \sint f \Dvg{\eta} \,\TTheta \,\trpg\,\CT
      + \sint \big(g \ader_{\zn} \Dvg{\eta} \WW \big) \FF^2 \g^{\zn\zm} \Hs_\zm
  %    \\
  %    \phantom{\PB{\sint f(\ptf_\io {+} {\FF \Hl})}{\sint \eta^{\zm}{\CT}_{,\zm}}}
     \;\; \weq \;\;
      \sint g \Dvg{\eta} \,\TTheta \,\trpg\,\hmg{\CT}
    \,, \vphantom{\big|^I}
    \hspace{-10mm}
    \\
   \end{array}
   \label{parGUMG_Algebra_onshell}
  \eea
where the structure $\TTheta(\argsqrg)$ is defined in (\ref{def_OOmega_TTheta}), as before, $\,\Dvgn{\xi}\defeq \Dvgp[\zk]{\xi}\,$\HIDE{ is the spatial vector divergence},
$(f \ader_{\zn} g) \defeq (f g_{,\zn} {-} f_{,\zn} g)$\HIDE{ is the antisymmetric derivative}, weak equality $\:\weq\:$ denotes the equality on the constraint surface, and $f(\sx), g(\sx), \xi^\zm(\sx), \eta^\zm(\sx)$ are arbitrary test functions. These relations extend (\ref{GUMG_Algebra_onshell}) to the full set of structure relations for the theory, described by the naive local parameterized action $\SSS^{loc}_{par}$.
The non-equivalence can be noticed for spatial diffeomorphisms $\sint {\xi}^\zm\Hs_\zm$.
\HIDE{In the original theory, }The right-hand side of (\ref{GUMG_Algebra_onshell}) depends only on the divergence $\Dvg{\xi}$, demanding that a gauge part of  spatial diffeomorphisms should be transverse. However, in the naive parameterized theory $\SSS^{loc}_{par}$, the right-hand side of (\ref{parGUMG_Algebra_onshell}) introduces dependence on additional components of $\xi^\zm$, implying that not all transverse diffeomorphisms are remain first-class. This signals a modification of the gauge structure and an altered number of degrees of freedom.

A more rigorous and direct proof of the inequivalence between the original {\GUMG} theory $\SSS_{\iE}$,  (\ref{extAction_GUMG}), and the naive parameterized model $\SSS^{loc}_{par}$ follows from the violation of the equivalence conditions\HIDE{ extensively} discussed in Section \ref{SSect:ParameterizedAction_AMG}. \HIDE{Specifically, }The constraint rank condition and the requirement for an\HIDE{ existence of the} accessible correspondence gauge for $\SSS^{loc}_{par}$ reduce to the question of whether a \emph{first-class} constraint of the form $\CPI^{loc} \defeq (\ptf_\io {+} \FF\Hl) {\,+\,} \alpha^\zn\Hs_\zn {\,+\,} \beta\,\inhw{\CT}$ exists, where $\alpha^\zn$ and $\beta$ are\HIDE{ some} linear operators acting to the right on the\HIDE{ respective} constraints.
The involution of $\CPI^{loc}$ with all constraints
\big($\sint f (\ptf_\io {+} {\Fw \Hl} ), \sint \xi^\zm\Hs_\zm, \sint \eta^\zm {\CT}_{,\zm}$\big) imposes the following on-shell conditions on $\alpha^\zm$ and $\beta$:
  \bea{CPI_naive_coef_eqs_left}
  \left\{
   \begin{array}{lll}
     \HIDE{\vec}{\alpha}^\zn (\tfrac{\WW+1}{\WW} \HIDE{\partial_\zn} f_{,\zn} ) \, - \,
     \HIDE{\vec}{\beta} ( \inh{\TTheta\trpg f} )
      \;\;{=}\;\; %\;\;\weq\;\;
      0
    \,,
    \\
     \HIDE{\vec}{\beta} \big( \inh{\OOmega \,\Dvg{\xi}} \big)
      \;\;{=}\;\; %\;\;\weq\;\;
     - \HIDE{\partial_\zn} (\tfrac{\WW+1}{\WW} \xi^\zn )_{,\zn}
    \,,
    \\
     \HIDE{\vec}{\alpha}^\zn \HIDE{\partial_\zn} \big( \OOmega \,\Dvg{\eta} \big)_{\!,\zn}
      \;\;{=}\;\; %\;\;\weq\;\;
     - \TTheta\,\trpg\, \Dvg{\eta}
     \,. \vphantom{\phI}
    \\
   \end{array}
   \right.
   %\label{parGUMG_Algebra_onshell_COPY_add}
  \eea
which must hold for \emph{arbitrary} local form factors $f(\sx), \xi^\zm(\sx), \eta^\zm(\sx)$.
However, these conditions fail in general. For instance, for transverse $\xi^\zm$: $\Dvgn{\xi} \HIDE{ \equiv \Dvgp[\zk]{\xi} } \teq 0$, the second relation imposes a to nontrivial constraint on the phase space variables: $\partial_\zn\WW(\argsqrg) \teq 0$. Likewise, the first relation imposes additional constraints\footnote{
  The system (\ref{CPI_naive_coef_eqs_left}) admits a solution on the\HIDE{ functional} subspace of the constraint surface, where the spatial volume density and the conjugate trace of the\HIDE{ gravitational} momenta are\HIDE{ spatially} homogeneous: $\sqrg \teq \hmg{\sqrg}(\taux)$,  $\trpg \teq \hmg{\trpg}(\taux)$. On this subspace, the naive parameterization looks consistent, providing the correspondence of the solutions and the gauge structures. However, the Poisson bracket matrix of constraints exhibits a rank discontinuity at this subspace. So, the two theories are inequivalent outside this subspace the theories describing different numbers of degrees of freedom.
}.
Since no independent first-class constraint $\CPI^{loc}$ exists to compensate for the additional phase space variables introduced by parameterization, the number of degrees of freedom changes, and the original nonparameterized action (\ref{extAction_GUMG}) cannot be recovered upon the consistent reduction. This confirms the inconsistency of the naive local parameterization $\SSS^{loc}_{par}$.

%\newpar

A possible way to make the local canonical parameterization scheme consistent is to extend the constraints with $O(\tf^\io_{,\zm})$ terms, which vanish in the correspondence gauge (\ref{correspondence_gauge_AMG}). These terms can be expressed as a series in powers of $\tf^\io$, $\ptf_\io$, and their spatial derivatives. However, in nonlinear gauge field theories with open algebras, such an approach typically leads to infinite series involving auxiliary variables and their derivatives.
Additionally, one could modify the initial Hamiltonian by adding total spatial derivatives\HIDE{ or other average-free combinations, if we allow spatial nonlocality}. While these do not affect the local properties of the Hamiltonian, they become nontrivial bulk components of a local constraint after the parameterization. However, this does not restore equivalence to the original theory while preserving spacetime locality.

\if{
  \footnote{
    A possible way to cure the non-first-class property of the naive parametrization constraint is to extend\HIDE{ (possibly all)} constraints with $O(\tf^\io_{,\zm})$ terms\HIDE{ with higher powers of $\tf^\io$, $\ptf_\io$ and their derivatives}, which vanish in the correspondence gauge (\ref{correspondence_gauge_AMG}).
     %%(See Section \ref{ASect:parametrization} for details.)
    These terms can be represented as series in powers of $\tf^\io$, $\ptf_\io$ and their spatial derivatives, however in nonlinear gauge field theories with open algebras one generically gets infinite series\HIDE{ in the auxiliary phase space variables}.
    Also one can add total spatial derivatives to initial Hamiltonian (which do not affect local properties of Hamiltonian, but after ``localization'' when parameterizing it become nontrivial bulk component of a local constraint).
  } %
}\fi

    \subsubsection*{Consistently parameterized action}
     \hspace{\parindent}
The discussion in the context of {\wGUMG} suggests that the key ingredients leading to the action in the Henneaux--Teitelboim-like form are the parameterization of the canonical action with local fields and the separation of a new constraint with a functionally complete $\Hl$ structure. These goals can be achieved by a consistent parameterization, which, however, picks up only the homogeneous component of the Hamiltonian into the new constraint, and then mixing it with the specific secondary constraint, which encodes information of the inhomogeneous (average free) component of the Hamiltonian.
As\HIDE{ was already} mentioned in Section \ref{SSect:ParameterizedAction_AMG}, one can always perform a homogeneous parameterization
  $\SSS_{\,\hmg{\!par\!}} = \int d\taux\hspace{1pt} \dsx \, \Lib
     \pg^{\zm\zn}\tauxdot{\g}_{\zm\zn}
     +  \hmg{\ptf}_\io\,\tauxdot{\hmg{\tf}}^\io
     -   \hmg{\replm}^\io ( \hmg{\ptf}_\io {+} \hmg{\FF \Hl} )
     - \repNs^\zm \Hs_\zm
     - \replmCT^\zm \CT_{,\zm}
     \Rib
  $
of the canonical theory (\ref{extAction_GUMG}), which implements\HIDE{ only} homogeneous time-reparametrization gauge freedom $\taux\to\taux'(\taux)$ and guarantees physical equivalence of the theories\footnote{
  The homogeneously parameterized theory $\SSS_{\,\hmg{\!par\!}}$ is equivalent to the original one, as it satisfies the rank and correspondence criteria. From general considerations \cite{Henneaux:1992ig}, there exist a first-class Hamiltonian of the form $\sint (\FF\Hl {+} U_0^\zn\Hs_\zn)$, with  $U_0^\zn \teq U_0^\zn(\g,\pg)$ (\ref{lm_solutions_nonGR}), providing a new first-class homogeneous constraint $\hmg{\replm}^\io\sint (\ptf_\io {+} \FF\Hl {+} U_0^\zn\Hs_\zn)$ in $\SSS_{\,\hmg{\!par\!}}$. This constraint is evidently transversal to the correspondence gauge $\hmg{\tf}^\io {-\,} \taux \teq 0$, guaranteeing equivalence with the original theory. The theory based on $\SSS_{\,\hmg{\!par\!}}$ and its Lagrangian version are briefly discussed in Appendix \ref{ASect:Homogeneous_Parameterization}.
}.
Furthermore, one can localize the sector of the auxiliary fields by adding to the action the\HIDE{ initially disentangled} pure gauge\HIDE{ canonical} combination of spatially inhomogeneous fields,
  $ % \bea{}
   \sint \Lib
     \inh{\ptf}_\io\,\tauxdot{\inh{\tf}}^\io
     {\,-\,} \inh{\replm}^\io \,\inh{\ptf}_\io
     \Rib
     \,,
  $ % \eea
 which remains decoupled until the basis of constraints is changed.
By aggregating the average and\HIDE{ complementary} average-free components of the\HIDE{ corresponding} auxiliary fields into local fields $\tf^\io(\taux,\sx)$, $\ptf_\io(\taux,\sx)$, we obtain the consistently parameterized action
  \bea{paramAction_cons_GUMG} %% (\ref{paramAction_cons_GUMG})
    \SSS_{par}[\g,\pg,\tf^\io,\ptf_\io,\replm^\io,\repNs,\replmCT]
    \!&=&\!\! \int \!d\taux \hspace{1pt} \dsx  %\int\! d\taux\, \dsx
    \,\LiB
     \pg^{\zm\zn}\tauxdot{\g}_{\zm\zn}
     \!+ \ptf_\io \,\tauxdot{\tf}^\io
     - \replm^\io ( \ptf_\io {+} \hmg{\FF \Hl} )
     - \repNs^\zm \Hs_\zm
     - \replmCT^\zm \CT_{,\zm}
     \RiB , \;
   \quad% \nonumber\\
  \eea
which is physically equivalent to the {\GUMG} extended action (\ref{extAction_GUMG}). The new parametrization constraint contains only the spatially average component $\hmg{\FF \Hl}(\taux)$ of the original Hamiltonian density term\HIDE{, and thus the action is not spatially local}.
\if{ %% v. 2025-02
However, there is a consistent way to parameterize the field-theoretical with local auxiliary fields, discussed in Section \ref{SSect:ParameterizedAction_AMG}, which however parameterizes the Hamiltonian term homogeneously. On the other hand, in {\GUMG} the specific structure of the secondary constraint term
 %%, according to discussion in Section \ref{SSect:ParameterizedAction_AMG},
allows to construct the new constraint with the functionally full $\Hl$ structure after the universally possible homogeneous time parameterization of the action and localizing the auxiliary sector fields by adding a gauge-trivial combination.\footnote{
 % To give another perspective of the above construction of nonlocal parameterized action we note that it
As was noted in Section \ref{SSect:ParameterizedAction_AMG}\HIDE{ and discussed in Appendix \ref{ASect:parametrization},} minimal parametrization (\ref{paramAction_cons_GUMG})
can be achieved from homogeneously parameterized action
$\int d\taux\hspace{1pt} \dsx \, \LiB
     \pg^{\zm\zn}\tauxdot{\g}_{\zm\zn}
     +  \hmg{\ptf}_\io\,\tauxdot{\hmg{\tf}}^\io
     -   \hmg{\replm}^\io ( \hmg{\ptf}_\io {+} \hmg{\FF \Hl} )
     - \repNs^\zm \Hs_\zm
     - \replmCT^\zm \CT_{,\zm}
     \RiB$
 % (\label{hmgparAction_GUMG})
which is a ``mechanical'' parametrization of the action (\ref{extAction_GUMG}), which implements only homogeneous time-reparametrization gauge freedom $\taux\to\taux'(\taux)$,
by adding a trivial (purely gauge) pair of spatially inhomogeneous fields
 \bea{}
    \sint \Lib
     \ptf_\io \tauxdot{\tf}^\io
     - \replm^\io ( \ptf_\io + \hmg{\FF \Hl} )
     \Rib
   \;=\;
    \sVol \Lib
     \hmg{\ptf}_\io\,\tauxdot{\hmg{\tf}}^\io
     - \hmg{\replm}^\io \,( \hmg{\ptf}_\io + \hmg{\FF \Hl} )
     \Rib
     +
   \sint \Lib
     \inh{\ptf}_\io\,\tauxdot{\inh{\tf}}^\io
     - \inh{\replm}^\io \,\inh{\ptf}_\io
     \Rib .
  \eea
Here \emph{spatially inhomogeneous} part $\inh{\phi}(\taux,\sx)$ of a field $\phi(\taux,\sx)$ is defined as
 $ % \beq{}
  \inh{\phi} \defeq \phi - \hmg{\phi},
 $ % \eeq
so that
$\sint \inh{\phi} =0$.
Before change of the constraint basis, the inhomogeneous sector of the new auxiliary fields is completely decoupled and does not interact with fields of the homogeneously parameterized action\HIDE{ (\ref{hmgparAction_GUMG})}.
}

 This procedure would allow to satisfy the rank and correspondence criteria of valid parameterization.
}\fi
\if{ %%% v.2024
  As in \wGUMG\ theory, we first rewrite the extended action (\ref{extAction_GUMG}) in \emph{time-parameterized} form. However, for generic {\GUMG} model one faces the difficulty that is typical for field theory. The standard mechanical time-parametrization prescription \cite{Henneaux:1992ig} when applied to field theory may contradict space-time locality and covariance, as well as it may be inconsistent with its local gauge structure. One is forced either to introduce additional auxiliary variables and constraints
  %%to parameterize all spacetime coordinates
  or the time-parameterized action appears as (generically infinite) expansion in powers of spatial derivatives of the parameterized canonical time field. Alternatively one can explicitly construct spatially nonlocal action with parameterized time directly.
  %%% Being universal in mechanics \cite{Henneaux:1992ig} parametrization of the theory acquiring spacetime locality becomes a nontrivial problem in field theory.
  ({Brief discussion of these issues is summarized in Appendix \ref{ASect:parametrization}}.)
  Below we use the latter approach to parameterize the generic {\GUMG} model. The naive parametrization prescription, used for parameterizing {\wGUMG} in Section \ref{SSect:ParameterizedAction_AMG} fails because ranks of the matrix of the
  Poisson brackets for the constraint set $\big(\Hs_\zm, {\CT}_{,\zm}\big)$
  of the extended theory (\ref{extAction_GUMG}) and that of the enlarged $\big( (\ptf_\io {+} {\FF \Hl} ), \Hs_\zm, {\CT}_{,\zm}\big)$ constraint set of naively-parameterized theory, analogous to (\ref{paramAction_naive_AMG}), do not coincide.\footnote{Remind the notation (\ref{def_CT}): $\CT\equiv\WW\FF\Hl$, $\CT_{,\zm}\equiv \partial_{\zm} \CT$.} Along the lines of the discussion in Section \ref{SSect:ParameterizedAction_AMG} this means that for generic {\GUMG} theory with $\WW\neq \wwc \teq \const$ naive parametrization does not introduce a new first-class constraint into the parameterized constraint set.
}\fi
\if{ %% 2023
  %\paragraph{Homogeneous reparametrization}\vphantom{.}\\
  \vspace{5mm}

  In field theory one can always use the direct analog of the quantum mechanical reparametrization by introducing only \emph{homogeneous}\footnote{In what follows we denote $\sx$-independent (only time-dependent) quantities \emph{homogeneous}.} degrees of freedom, which account for $\sx$-independent --- spatially-homogeneous --- time reparametrization
   \bea{hmgparAction_GUMG} %(\ref{parAction_BLRG})
      \SSS_{\,\hmg{\!par\!}}[\g,\pg,\hmg{\xo},\hmg{\po},\hmg{\replm},\repNs,\replmCT]
      &=& \int d\taux\, \dsx \, \LiB
       \pg^{\zm\zn}\tauxdot{\g}_{\zm\zn}
       + \hmg{\po}\,\tauxdot{\hmg{\xo}}
       - \hmg{\replm} ( \hmg{\po} + {\FF \Hl} )
       - \repNs^\zm \Hs_\zm
       - \replmCT^\zm \CT_{,\zm}
      \RiB
   %   \nonumber\\
   %   &=& \int d\taux\,  \LiB
   %    \BraKet{\pg^{\zm\zn}}{\tauxdot{\g}_{\zm\zn}}
   %    + \hmg{\po}\sVol\tauxdot{\hmg{\xo}}
   %    - \hmg{\replm} \sVol (\hmg{\po} + \hmg{\FF \Hl})%\BraKet{1}{\FF \Hl})
   %    - \BraKet{\repNs^\zm}{\Hs_\zm}
   %    - \BraKet{\replmCT^\zm}{\CT_{,\zm}}
   %   \RiB
    \eea
  where $\hmg{\xo}=\hmg{\xo}(\taux)$, $\hmg{\po}=\hmg{\po}(\taux)$, $\hmg{\replm} =\hmg{\replm} (\taux)>0$ --- new spatially homogeneous fields which depend only on time and $\sVol\equiv \sint 1$ --- is the coordinate volume of $\taux = \const$ hypersurfaces. All other fields are reexpressed as functions of new time time variable $\taux$ (and spatial coordinates $\sx$.\footnote{One could embed the original theory in parameterized theory in different ways. For example one can leave the time coordinate of the base spacetime manifold unchanged and just introduce new extended set of fields. This way does not assume reexpression of fields since the time coordinate is not changed. Also one can rescale Lagrange multipliers by $\replm$, or by $\tauxdot{\xo}$, or keep them as in original theory. \TBC{These all are physically equivalent ways of parameterizing the theory}.}

  Classical equivalence with the theory (\ref{extAction_GUMG}) is guaranteed by the fact that there appear the new first-class homogeneous ``Hamiltonian'' constraint $(\hmg{\po} + \hmg{\FF \Hl} + \hmg{U_0^\zm\Hs_\zm})$ which may be gauge-fixed with condition $\hmg{\xo}(\taux)-\taux=0$ so that the reduction with respect to these (now the second-class) constraints recovers the extended action (\ref{extAction_GUMG}).

  The configuration space of the parameterized theory (\ref{hmgparAction_GUMG}) contains fields which are elements of different functional spaces\HIDE{ (i.e. sections of different bundles)}. This introduces (spatial) nonlocality to the theory. However, it is more convenient to work with configuration space of the same functional space local in spacetime manifold. And it is always possible to reformulate (\ref{hmgparAction_GUMG}) in terms of only local fields by adding to homogeneous fields the correspondent complementary sector which is a pure gauge.\footnote{In the context of {\GUMG} the local parametrization will keep us closer to Henneaux--Teitelboim line of reasoning for {\UMG} theory \cite{Henneaux:1989zc}.}
}\fi
%
%\if{ %% Good 2025-02 version
However, this form of the constraint keeps the equivalence transparent. The criterion for equivalent parameterization is the introduction of a new first-class constraint\HIDE{ to the theory}, which is transversal to the correspondence gauge condition.
From general arguments \cite{Henneaux:1992ig}, the original theory (\ref{extAction_GUMG}) admits a first-class Hamiltonian of the form $\sint (\FF\Hl {+} U_0^\zn\Hs_\zn)$ with  $U_0^\zn=U_0^\zn(\g,\pg)$ given by (\ref{lm_solutions_nonGR})\footnote{
  These $U_0^\zn(\g,\pg)$ correspond to on-shell solutions for the Lagrange multipliers for the primary constraints.
}.
This promotes a new first-class constraint in the consistently parameterized theory (\ref{paramAction_cons_GUMG}):
  \beq{first_class_param_constraint_GUMG}
    \CPI \defeq \HIDE{\sint f(\sx)}
    \HIDE{(} \ptf_\io + \hmg{\FF\Hl {+} U_0^\zn\Hs_\zn} \HIDE{)}
    \,.
  \eeq
This constraint is clearly transversal to the correspondence gauge condition,
  \beq{correspondence_gauge}
    \tf^\io(\taux,\sx) - \taux = 0
    \,.
  \eeq
However, its transversality to a local constraint is ensured by the locality of the $\ptf_\io$ term in (\ref{first_class_param_constraint_GUMG}), meaning that the (non)locality of the Hamiltonian term, which depends on non-auxiliary metric fields, does not affect this property.

\if{ %% v 2025-02 good
One either can think that $\taux$ as the new time variable, which flows independently in different points of fixed spatial coordinates $\sx^\zn$, that dynamical field $\tf^\io$ characterize the values of time parameter $t$ of (\ref{extAction_GUMG}) in the new coordinate parametrization of spacetime points of (\ref{paramAction_cons_GUMG}), and that all fields are now redefined to be the fields, dependent on new time variable and spatial coordinates\footnote{
 In this interpretation the canonical Lagrange multipliers in the parameterized action formally relates to counterparts in the extended action as
 $ %  \bea{def_repNs_repvaux_}
  \repNs^\zm(\taux,\sx) \;\leftrightarrow\; \tauxdot{\tf}^\io(\taux,\sx) \Ns^\zm(\tf^\io(\taux,\sx),\sx),
  \;\; %\qquad
  \replmCT^\zm(\taux,\sx) \;\leftrightarrow\; \tauxdot{\tf}^\io(\taux,\sx) \lmCT^\zm(\tf^\io(\taux,\sx),\sx) \,.
   \nonumber
 $ %\eea
 Factors $\tauxdot{\tf}^\io(\taux)$ compensate the change of the integration measure.
 %% \TODO{change notation for l.m. in parameterized action or equality sign in (\ref{def_repNs_repvaux_})}
},
%% Though this is not important in classics and in quantum applications.
%\TODO{add ref|discussion why first class}
%
%% without extending constraint with terms $O(\inh{\xo})$.
%
 %\if{ %% version 2023
 % \bea{def_repNs_repvaux_}
 %  \repNs^\zm = \tauxdot{\xo}\Ns^\zm, \quad
 %  \replmCT^\zm = \tauxdot{\xo}{\lmv}^\zm.
 % \eea
 % (note that $\tauxdot{\xo} \eomeq \replm$)
 %% to preserve their gauge transformation law?
 %}\fi
with addition of the new fields $\tf^\io(\taux,\sx)$, $\ptf_\io(\taux,\sx)$  $\replm^\io(\taux,\sx)$. Or one may think of transition from description (\ref{extAction_GUMG}) to (\ref{paramAction_cons_GUMG}) as simple enlarging the configuration space of the extended action on the same spacetime manifold $(t,\sx^\zn)$, so that times $\taux$ in (\ref{paramAction_cons_GUMG}) and $t$ in (\ref{extAction_GUMG}) are identical.
Both interpretations are valid since the minimally parameterized action (\ref{paramAction_cons_GUMG}) satisfies equivalence conditions\HIDE{ (\ref{cond1}),(\ref{cond2}),(\ref{cond3})} discussed in Section \ref{SSect:ParameterizedAction_AMG}. Namely, there is a new\HIDE{ ``local'' (with respect to new auxiliary fields)} first-class constraint\footnote{
 The spatially nonlocal function $U_0^\zm(\g,\pg)$ is defined in (\ref{lm_solutions_nonGR}). In any case, we will discuss the constraint and gauge structures of the parameterized theory in more detail later.
 }
$\ptf_\io {+} \hmg{\FF \Hl} {+} \hmg{U_0^\zn\Hs_\zn}$
and this new gauge symmetry may be gauge-fixed with the accessible gauge-fixing condition --- the \emph{correspondence gauge}
  \beq{correspondence_gauge}
    \tf^\io(\taux,\sx) - \taux = 0
    \,,
  \eeq
so that the reduced gauge-fixed action recovers the extended action (\ref{extAction_GUMG}). This proves the physical equivalence of the theories defined by the actions (\ref{extAction_GUMG}) and (\ref{paramAction_cons_GUMG}) in any interpretation of the latter.
}\fi

\if{ %% version 2023-06
Local time reparametrization is acquired by introducing \emph{auxiliary time} $\taux$ which can flow in different spatial points independently, and adding new auxiliary fields: time parameter $t$ turns into dynamical field $\xo(\taux,\sx)$, $\po(\taux,\sx)$ is its conjugated momenta, $\replm(\taux,\sx)$ is an auxiliary Lagrange multiplier for the parameterized ``Hamiltonian'' constraint $\po(\taux,\sx) + \hmg{\FF \Hl}(\taux)$ where $\hmg{\FF \Hl}(\taux)\equiv \sVol^{-1}\sint \FF \Hl$ is the spatially homogeneous part of Hamiltonian density $\big(\FF \Hl\big)(\taux,\sx)$. The parameterized action reads
 \bea{paramAction_cons_GUMG} %(\ref{parAction_BLRG})
    \SSS_{par}[\g,\pg,\xo,\po,\replm,\repNs,\replmCT]
    &=& \int d\taux\, \dsx \, \LiB
     \pg^{\zm\zn}\tauxdot{\g}_{\zm\zn}
     + \po\,\tauxdot{\xo}
     - \replm ( \po + \hmg{\FF \Hl} )
     - \repNs^\zm \Hs_\zm
     - \replmCT^\zm \CT_{,\zm}
     \RiB
   \qquad% \nonumber\\
  \eea
which satisfies equivalence conditions (\ref{cond1}),(\ref{cond2}),(\ref{cond3}). There is a new spatially local first-class constraint $(\po + \hmg{\FF \Hl} + \hmg{U_0^\zm\Hs_\zm})$ (with nonlocal internal structure) which may be completely fixed with the gauge-fixing condition $\xo(\taux,\sx) - \taux = 0$ so that the reduced gauge-fixed action will recover the extended action (\ref{extAction_GUMG}). Which proves the equivalence.
}\fi

    \subsubsection*{Constraint structure}
     \hspace{\parindent}
\newcommand{\lmv}{{\color{DarkGreen}\mu}}
The constraint basis of the parameterized action (\ref{paramAction_cons_GUMG}), in addition to the original constraints $(\Hs_\zn,\CT_{,\zn} \HIDE{\equiv(\WW\FF\Hl)_{,\zn}} )$ of (\ref{extAction_GUMG}), includes the new constraint $\ptf_\io {\,+\,} \hmg{\FF\Hl}$. Consequently, the original structure relations (\ref{GUMG_Algebra_offshell}) are supplemented by
  \bea{GUMG_addpar_offshell}
   \begin{array}{lll}
     & \PB{\sint f (\ptf_\io{+}\hmg{\FF\Hl})}{\sint g (\ptf_\io{+}\hmg{\FF\Hl})}
     \;=\; 0
     \,, \vphantom{\big|^I} \\
     & \PB{\sint {\xi}^\zn\Hs_\zn}{\sint g (\ptf_\io{+}\hmg{\FF\Hl})}
     \;=\; - \sint \Dvg{\xi} \hmg{g} \,\CT
     \;\weq\; 0
     \,, \vphantom{\big|^I} \\
     & \PB{\sint \eta^{\zm}\CT_{,\zm}}{\sint g (\ptf_\io{+}\hmg{\FF\Hl})}
     \;=\; \sint \Dvg{\eta} \hmg{g} \,\TTheta \,\trpg\,\CT
        + \sint \big(\Dvg{\eta}\WW \big)_{\!,\zn} \hmg{g} \,\FF^2 \g^{\zn\zk} \Hs_\zk
     \;\weq\;
     \sint \Dvg{\eta} \hmg{g} \,\TTheta \,\trpg\,\hmg{\CT}
     \,, \vphantom{\big|^I}
     \hspace{-5mm} \\
   \end{array}
   \label{GUMG_addpar_onshell}
  \eea
where $\TTheta(\argsqrg)$ is defined in (\ref{def_OOmega_TTheta}). As expected, only the homogeneous component of the test function $g$ appears on the right-hand sides.
Due to the latter Poisson bracket, the new constraint $\ptf_\io{\,+\,}\hmg{\FF\Hl}$ is not first-class. However, structure relations (\ref{GUMG_Algebra_offshell}) and (\ref{GUMG_addpar_onshell}) imply that the system of involution conditions with basis constrains for the coefficients $\alpha$ and $\beta$ in the combination $\CPI \defeq (\ptf_\io {+} \FF\Hl) {\,+\,} \alpha^\zm\Hs_\zm {\,+\,} \beta\,\inhw{\CT}$ now has a solution given by (\ref{first_class_param_constraint_GUMG}).
The basis of the first-class constraint algebra for consistently parameterized theory (\ref{paramAction_cons_GUMG}) is spanned by $\sint f \CPI$ and transverse spatial diffeomorphisms $\,\sint{\xi}_{\isst}^\zm \Hs_\zm\,$ with $\Dvgi[\zm]{\xi_{\isst}} \teq 0$, (\ref{transverse_diffeomorphisms}). As in the original theory, the longitudinal component of spatial diffeomorphisms and the secondary constraint can be treated as second-class basis elements.

%%-------------------------=%        ******         %=-------------------------%%
%%-------------------------=%        ******         %=-------------------------%%

  \subsection{Henneaux--Teitelboim-like action and delocalization}
   \label{SSect:AltAction_GUMG}
    \subsubsection*{Delocalization operator}
     \label{SSect:ConstraintBasisChange}
     \hspace{\parindent}
Starting from the consistently parameterized action (\ref{paramAction_cons_GUMG}) we modify the constraint basis to consolidate the $\Hl$ structure in a single constraint. The set of constraint equations from representation (\ref{paramAction_cons_GUMG}): $\ptf_\io {+} \hmg{\FF \Hl} \teq 0$, \, ${\CT}_{,\zm} \equiv (\WW \FF \Hl)_{,\zm} \teq 0$,
describes the same constraint surface as the set
$\ptf_\io {+} \hmg{\FF \Hl} {+} \WW^{-1} \inh{\WW \FF \Hl} \teq 0$, \, ${\CT}_{,\zm} \teq 0$.
The first constraint in the new set can be reformulated as $\ptf_\io {+} \opE \FF \Hl \teq 0$, where the operator $\opE$ acts\HIDE{ to the right} on the function $\FF \Hl$ with kernel
  \beq{opE_kernel}
    E(\taux|\sx;\sx')
    =
    \delta^{\Ddim{-}1}(\sx;\sx')
    - \WW^{-1} (\taux,\sx) \frac1{\sVol{(\taux)}}  \WW (\taux,\sx')
    + \frac1{\sVol{(\taux)}}
  \,,
  \eeq
where $\WW(\taux,\sx) \tequiv \WW(\sqrg {\scalebox{0.8}{\mbox{\ensuremath{(\taux,\sx)}}}})$ and the volume normalization factor is\footnote{
 %This kernel representation is directly applicable to\HIDE{ the case of the} compact spatial sections with finite volumes. 
 For the case of noncompact spatial sections, the volume normalization factors can be understood through a\HIDE{ finite-volume} limiting procedure. See Appendix \ref{ASect:NonCompact}.
}
$\sVol(\taux) \defeq \sint_{\taux = \const}1$. To clarify the structure, note that in terms of the complementary symmetric projectors $\hmgId$ and $\inhId$, which extract the spatially homogeneous (average) and inhomogeneous (average-free) parts of a function respectively, the operator $\opE$ is represented as
 $ % \beq{def_opE}
  \opE \;{=}\; \hmgId + \WW^{-1} \inhId \WW
         \;{=}\; \Id + \hmgId - \WW^{-1} \hmgId \WW
         \,,
 $ % \eeq
where $\Id \teq \hmgId + \inhId$ is the identity operator.
Consequently, the action of $\opE$ on local function $f(\sx)$, $g(\sx)$ to the left and right can be expressed as
 \beq{opE_left-right_action}
  \begin{array}{|l}
    f\opE \:=\: f - \hmg {f \WW^{-1}}\,\WW + \hmg{f} \,, \\
    \opE g \:=\: g - \WW^{-1}\,\hmg{\WW g} + \hmg{g} \,, \\
  \end{array}
  \quad \Leftrightarrow \quad
  \sint f \opE g
    \:=\:
     \sint \big(
    f g
    -  \hmg{f\WW^{-1}\!}\,\, \hmg{\WW g}
    +  \hmg{f}\,\hmg{g}
    \big)
    \,.
 \eeq
Relevant properties of the operator $\opE$ are detailed in Appendix \ref{ASSect:opE_properties}. The crucial property is that $\opE$ is a \emph{nondegenerate} (full-rank) operator. This nondegeneracy ensures the existence of the inverse operator $\opinvE$ (see lemma \ref{Lemma:InvE}) with the kernel
 \beq{invE_kernel}
  E^{-1}(\taux|\sx;\sx')
  \;=\; \delta(\sx;\sx')
    \,-\, \frac{\WW^{-1}(\taux,\sx) }{\sint {\WW^{-1}} (\taux)\vphantom{\big{|}}}
    \,-\, \frac{\WW(\taux,\sx')}{\sint{\WW}(\taux)\vphantom{\big{|}}\,}
    \,+\, 2
    \frac{\WW^{-1}(\taux,\sx) }{\sint {\WW^{-1}} (\taux)\vphantom{\big{|}}}
    \sVol(\taux)
    \frac{\WW(\taux,\sx')}{\sint{\WW}(\taux)\vphantom{\big{|}}\,}
  \if{
   = \delta^{\Ddim{-}1}(\sx;\sx')
    - \WW^{-1}(\taux,\sx) \tfrac{1}{\sint {\WW^{-1}}(\taux)}
    -\tfrac{1}{\sint{\WW}(\taux)} \WW(\taux,\sx')
    + 2 \WW^{-1}(\taux,\sx)
        \tfrac{\sVol(\taux)}{\sint{\WW^{-1}}(\taux)  \sint{\WW}(\taux) }
        \WW(\taux,\sx')
  }\fi
  \,.
 \eeq
The existence of\HIDE{ the inverse operator} $\opinvE$ allows expressing $\Hl$ in terms of $\ptf_\io$ on the constraint surface: $\Hl = - \FF^{-1}\opinvE \ptf_\io$, eliminating it from the secondary\HIDE{ functionally incomplete} constraint (\ref{secondary_GUMG}).

To form a functionally complete term with $\Hl$ in the parametrization constraint, one could handle a more general linear combination with the secondary constraint. This choice does not affect the physical results. The unit coefficient at the parametrization constraint $\ptf_\io{+}\hmg{\FF\Hl}$ can be fixed by redefining\HIDE{ the Lagrange multiplier} $\replm^\io$, so the resulting constraint,  $\ptf_\io{+}\opEa{\FF\Hl}$, is parameterized by the nondegenerate operator\footnote{
 Properties of more generic $\opE$-type operators are outlined in Appendix \ref{SSect:generic_op_E}. We will further comment on the independence of physical results from the choice of $\ofa$ in $\opEa$ when discussing the dynamics, serving as a partial consistency check.
}
 $%\bea{def_opEa}
  \opEa \,\defeq \hmgId + \ofa^{-1} \inhId \WW
  %       \;=\; \ofA \WW + \hmgId - \ofa^{-1} \hmgId \WW
 %\eea
 $,
with a generic sign-definite function $\ofa$.
Our choice of $\ofa \teq \WW(\argsqrg)$ is a matter of convenience: when $\WW \tequiv \wwc \teq \const$, the operators $\opE$ and $\opinvE$ become the identity operator, recovering the {\wGUMG} results from Section \ref{Sect:AltAction_wGUMG}. In the quantum context, such choice is advantageous as their diagonals, ${E}(\sx,\sx)$ and ${E}^{-1}(\sx,\sx)$, reduce to the field-independent spatial $\delta(0)$, and their determinants acquire simpler forms\footnote{See Section \ref{SSect:Quantum_measure}.}.

The explicit spatial nonlocality encoded in the operator $\opE$ arises in the canonical constraint of the parameterized {\GUMG} theory with nonconstant $\WW(\argsqrg)$. This nonlocality stems from the inability to merge the functionally incomplete $\Hl$ contributions from the parameterization and secondary constraints into a local structure. Furthermore, one cannot redefine the phase space fields to achieve locality of the canonical action. The secondary constraint restricts the average-free part of the Hamiltonian density $\FF\Hl$ weighted by $\WW(\argsqrg)$, while the consistent parameterization promotes the Hamiltonian density to the unweighted homogeneous structure $\hmg{\FF\Hl}$. However, local gluing becomes possible in the {\wGUMG} models when $\WW(\argsqrg) \to \HIDE{\wwc =} \const $.

\newpar

A further simplification of the constraint basis is achieved by eliminating $\Hl$ via the substitution $ - \FF^{-1}\opinvE \ptf_\io$ in the secondary constraint $\CT_{,\zm}\HIDE{ \equiv (\WW\FF\Hl)_{,\zm}} \teq 0$. This leads to the cumbersome constraint $(\WW\opinvE \ptf_\io)_{,\zm} \teq 0$, which can be replaced with the simpler, yet equivalent, constraint constraint ${\ptf_\io}_{,\zm}=0$. This equivalence follows from the properties of the operator $\opE$ (see Corollary \ref{Corollary:var_invE_G_equations}).
The nontrivial equivalence between  $(\WW\opinvE \ptf_\io)_{,\zm} \teq 0$ and ${\ptf_\io}_{,\zm} \teq 0$ is significant, as it allows us to gather various auxiliary fields into the covariant spacetime divergence\HIDE{ of the Henneaux--Teitelboim auxiliary field},  $\partial_\Zm \HTf^\Zm$, in the alternative action (\ref{Alt_Action_GUMG}) instead of dealing with a noncovariant, spatially nonlocal metric-dependent differential structure.
\if{ %% version 2023-06
Thus, the complete constraint basis can be equivalently rearranged to
 \beq{constraint_basis_change_GUMG}
  \left\{
  \begin{array}{l}
    \po + \hmg{\FF \Hl} = 0,\; \\ \Hs_\zm = 0,\; \\ \inh{ \WW\FF\Hl } = 0,
  \end{array}
  \right.
 \qquad \Rightarrow \qquad
  \left\{
  \begin{array}{l}
    \po + \opE \FF \Hl = 0,\; \\ \Hs_\zm = 0,\; \\ \inh{ \WW \opinvE \po } = 0,
  \end{array}
  \right.
 \eeq
where nondegenerate operator $\opE$ is defined in (\ref{def_opE}).
}\fi

\newpar

To summarize the above discussion, the constraint basis of (\ref{paramAction_cons_GUMG}) is rearranged:
\beq{constraint_basis_change_GUMG}
  \left\{
  \begin{array}{l}
    \ptf_\io + \hmg{\FF \Hl} = 0\,,\; \\
    \Hs_\zm = 0 \,,\; \\
    { (\WW\FF\Hl)_{,\zm} } = 0 \,, \\
  \end{array}
  \right.
 \qquad \to \qquad
  \left\{
  \begin{array}{l}
    \ptf_\io + \opE \FF \Hl = 0 \,,\; \\
    \Hs_\zm = 0\,,\; \\
    {\ptf_\io}_{,\zm} = 0\,, \\
  \end{array}
  \right.
 \eeq
allowing the parameterized canonical action to be equivalently expressed as
   \bea{paramAction'_GUMG} %(\ref{parAction_BLRG})
    \SSS_{par'}[\g,\pg,\tf^\io\!,\ptf_\io,\replm,\repNs,\lmptf]
    \;= %\!\!&=& \!\!\!
    \int \!d\taux \hspace{1pt} \dsx %\int d\taux\, \dsx
     \,\LiB
     \pg^{\zm\zn}\tauxdot{\g}_{\zm\zn}
     \!+ \ptf_\io \tauxdot{\tf}^\io
     - \replm ( \ptf_{\io} {+} \opE {\FF \Hl} )
     - \repNs^\zm \Hs_\zm
     - \lmptf^\zm {\ptf_\io}_{,\zm}
     \RiB
     . \;\;
 \eea
\if{ %% version 2023-05
\bea{parAction3_GUMG} %(\ref{parAction_BLRG})
    &&\Spar_{3}[\g,\pg,\xo,\po,\replm_2,\repNs,\inh{\lmv}_2]
    \;=\; \int d\taux\, \dsx \, \LiB
     \pg^{\zm\zn}\tauxdot{\g}_{\zm\zn}
     + \po\,\tauxdot{\xo}
     - \replm_2(\underbrace{\po + \opE \FF \Hl}_{\CP_2})
     - \repNs^\zm \Hs_\zm
     - \Dvg{\lmv_2} %\inh{\lmv}_2
        %\inhId
      \WW \opinvE \po
     \RiB
  \nonumber\\
    &&\Spar_{3}
    \;=\; \int d\taux\, \LiB
     \BraKet{\tauxdot{\g}_{\zm\zn}} {\pg^{\zm\zn}}
     + \BraKet{\tauxdot{\xo}} {\po}
     - \BraKet{\replm_2} {\po + \opE \FF \Hl}
     - \BraKet {\repNs^\zm} {\Hs_\zm}
     - \BraOKet{ \inh{\lmv}_2} {\inhId \WW \opinvE} { \po}
     \RiB
   % \qquad
   % \nonumber\\
  \eea
}\fi
The\HIDE{ physical} equivalence of representations (\ref{paramAction_cons_GUMG}) and (\ref{paramAction'_GUMG}) is guaranteed by the freedom of choosing any equivalent constraint basis\HIDE{ in the canonical representations of the theory}. Alternatively, the change of the constraint basis (\ref{constraint_basis_change_GUMG})\HIDE{ in the action} corresponds to the nondegenerate redefinition of the Lagrange multipliers\footnote{
  Transformation (\ref{constraint_basis_change_GUMG}) differs from the Lagrange multiplier redefinition in {\wGUMG} (\ref{redef_lms_AMG}), due to the distinct  definition of the parametrization constraint in the naive and consistent parametrization schemes.
}
  $
   (\replm^\io \,, \replmCT^\zm)
    \;\to\;
   (\replm \,, \lmptf^\zm)
  $
in (\ref{paramAction_cons_GUMG}):
  \beq{basis_change_GUMG} %{def_lmv2_lmo2}
  \begin{array}{|cclc}
    \replm^\io \!\!
   % &=&  \replm + \Dvg{\replmCT'} \WW \opinvE    % &=&  \replm + \Dvg{\lmptf} \ofA^{-1} \tFProj{\ofA^{-1}}{1}
   &=&  \replm - \Dvg[\zk]{\lmptf} \,, %%  - \Dvg{\lmptf} \to
   \\
    \Dvg[\zk]{\replmCT} \!\!
   &=& - \inh{ \replm  \WW^{-1}\HIDE{\ofA} } \,. %% \ofA\to\WW^{-1}
   \\
  \end{array}
  %\label{def_lmv2_lmo2}
  %% checked 2024-10-25, 2024-11-05
 \eeq
%where $\Dvg{\lmptf}$ and $\Dvg{\replmCT}$ are divergences of

In the action (\ref{paramAction'_GUMG}), all appearances of the gravitational momenta $\pg^{\zn\zm}$ within the constraints are confined to the {\GR} structures $\Hl$ and $\Hs_{\zn}$\HIDE{ in functionally complete constraint} (\ref{BGUMGHamiltonianStructure}, \ref{BGUMGMomentaConstraints}).
This\HIDE{ simplification} comes at the cost of spatial delocalization, which is fully encoded in the operator $\opE$. However, this nonlocality is not an artifact of parameterization. Rather, it was implicitly present in the interplay between the dynamical and constraint structures of the original canonical theory (\ref{extAction_GUMG}). Consistent parameterization simply made this nonlocality explicit.

\if{  %% ADD HIDDEN (version 2023-12)
Inverse coordinate change is
  \beq{}
    \begin{array}{|crclc}
      & \replm
      &=&  \replm^\io \hmgId \opEA^{-1}
        - \Dvg{\replmCT} \WW \opEA^{-1}
      \\
      &&=&  \replm^\io \hmgId \opEA^{-1}
        - \Dvg{\replmCT} \ofA^{-1} \! \tFProj{\ofA^{-1}\!\!}{1}
      \\
      &&=& \hmg{\replm^\io}  \frac{ \ofA^{-1} }{ {\,\hmg{\!\ofA^{-1}\!}\,\vphantom{{I^{|^I}}}} }  %% \ofA^{-1}/\hmg{\!\ofA^{-1}\!}
        - \Dvg{\replmCT}\ofA^{-1}
        + \hmg{\Dvg{\replmCT}\ofA^{-1}}   \frac{ \ofA^{-1} }{ {\,\hmg{\!\ofA^{-1}\!}\,\vphantom{{I^{|^I}}}} } ,
      \\
       & \Dvg{\lmptf}
       &=&  \replm - \replm^\io
      \\
       &&=&  \inh{\replm} - \inh{\replm^\io} ,
       %\TODO{Check equivalence conditions for $\CP_2,\Hs_\zm,\inh{\CT}$ system (2023-04)}
     \\
    \end{array}
%% \quad \Leftrightarrow \qquad
  \eeq
where operators $\opinvE$ and projector $\! \tFProj{\ofA^{-1}\!\!}{1}$ act on functions to the left.

One can check that this is the inverse explicitly by substitution (\ref{basis_change_GUMG}) into this expression. For $\Dvg{\lmptf}$ it is evident. For $\replm$ one gets
  \bea{}
    \replm
    &=& \replm^\io \hmgId \opEA^{-1}
        + \inh{ \replm \ofA } \ofA^{-1} \! \tFProj{\ofA^{-1}\!\!}{1}
        \nonumber\\
    &=& \hmg{\replm} \;\tfrac{ \ofA^{-1} }{ {\,\hmg{\!\ofA^{-1}\!}\,\vphantom{{I^{|^I}}}} }
        \;+\;  \replm \;\cancelto{1}{\vphantom{\Big|}\ofA \ofA^{-1}}  \tFProj{\ofA^{-1}\!\!}{1}
        \;-\; \hmg{ \replm \ofA }\; \cancelto{0}{\ofA^{-1} \! \tFProj{\ofA^{-1}\!\!}{1} }
        \nonumber\\
    &=& \replm \lFProj{\ofA^{-1}\!\!}{1}
        \;+\;  \replm \Id
        -  \replm \lFProj{\ofA^{-1}\!\!}{1}
        \;-\; 0
        \nonumber\\
    &=&  \replm,
  \eea
  where we used $\hmg{\replm^\io}=\hmg{\replm}$.
}\fi

\subsubsection*{Lagrangian form of the parameterized {\GUMG} action}
 \hspace{\parindent}
The remaining steps toward the alternative {\GUMG} action in the  Henneaux--Teitelboim-like form (\ref{Alt_Action_GUMG}) are analogous to those discussed in Section \ref{SSect:AltAction_AMG} for {\wGUMG} models.
We rescale the Lagrange multiplier \,$\replm \to \repNl$\, by introducing \,$\repNl \teq \replm \opE \FF$\, to purify the $\Hl$ term in the constraint and rename \,$\ptf_\io \to \CCof$\, to highlight its role as the cosmological-constant field:
 \beq{replm_to_repNl_and_ccf0} %was {lmo2_to_Nl}
   \replm \:=\: \repNl \FF^{-1} \opinvE   
   \,,
   \qquad
   \ptf_\io \:=\: \CCof
   \,,
 \eeq
where operators $\opE$ and $\opinvE$ act on functions to the left.
In terms of the new variables, the action (\ref{paramAction_AMG_fin})
becomes
   \beq{paramAction_GUMG_fin}  %% (\ref{paramAction_GUMG_fin})
    %&&
    \SSS_{par''}[\g,\pg,\tf^\io\!,\CCof,\repNl\!,\repNs,\lmptf]
     \;=\; \!\!\int \!d\taux \hspace{1pt} \dsx  %\int d\taux\, \dsx
     \,\LiB
       \pg^{\zm\zn}\tauxdot{\g}_{\zm\zn}
       \!- \repNl ( \Hl {+}  \FF^{-1}\! \opinvE \! \CCof  )
       - \repNs^\zm \Hs_\zm\!
       + \big(\tauxdot{\tf}^\io {+} \Dvg[\zk]{\lmptf}\big)  \CCof\,
     \RiB
     . \;\;
     \vspace{-4mm}
   \if{
   \qquad \\
    &&\SSS_{par''}[\g,\pg,\tf^\io,\CCf,\repNl\!,\repNs,\lmptf]
    \,= \int d\taux\, \LiB
     \BraKet{\tauxdot{\g}_{\zm\zn}} {\pg^{\zm\zn}}
     - \BraKet{\repNl} {\Hl {+} \sqrg\CCf }
     - \BraKet {\repNs^\zm} {\Hs_\zm}
     + \BraKet{ \tauxdot{\tf}^\io {+} \Dvg{\lmptf}}  {\opE \FF\sqrg\CCf }
     \RiB .
   \qquad
   }\fi
    \eeq
 \if{ %% version 2023-06
 \bea{parAction_Halt_GUMG} %(\ref{parAction_BLRG})
    &&\Spar_{Halt}[\g,\pg,\xo,\CCf,\repNl,\repNs,\inh{\lmv}_2]
    \;=\; \int d\taux\, \dsx \, \LiB
     \pg^{\zm\zn}\tauxdot{\g}_{\zm\zn}
     - \repNl(\sqrg\CCf + \Hl)
     - \repNs^\zm \Hs_\zm
     + \tauxdot{\xo} \opE \FF\sqrg\CCf
     - \Dvg{\lmv_2} %\inh{\lmv}_2
      %\inhId
      \WW\FF\sqrg \CCf
     \RiB
  \nonumber\\
    &&\Spar_{Halt}
    \;=\; \int d\taux\, \LiB
     \BraKet{\tauxdot{\g}_{\zm\zn}} {\pg^{\zm\zn}}
     - \BraKet{\repNl} {\sqrg\CCf + \Hl}
     - \BraKet {\repNs^\zm} {\Hs_\zm}
     + \BraKet{\tauxdot{\xo}} {\opE \FF\sqrg\CCf }
     - \BraOKet{ \inh{\lmv}_2} {\inhId \WW\FF\sqrg } { \CCf }
     \RiB
   \qquad
    \nonumber\\
  \eea
 }\fi

Since all dependence on the\HIDE{ gravitational} momenta $\pg^{\zm\zn}$ is confined within the {\GR} canonical structure
   $\,
       \pg^{\zm\zn}\tauxdot{\g}_{\zm\zn}
       \!- \repNl \Hl
       - \repNs^\zm \Hs_\zm\!\,
   $,\,
they can be expressed through their own variational equations\HIDE{ of motion},
  \bea{momenta_ecK_GR}
    \pg^{\zm\zn} \eomeq {\sqrg}(\ecK^{\zm\zn}-\g^{\zm\zn} \trecK)
    \,, \qquad
    \ecK_{\zm\zn} = \tfrac{1}{2\repNl} ( \tauxdot{\g}_{\zm\zn} - \repNs_{\zm ; \zn} - \repNs_{\zn ; \zm})
    \,,
  \eea
and removed from the action\HIDE{ in the same manner} as in general relativity, converting the canonical structure into the Einstein-Hilbert Lagrangian. The Lagrangian {\GUMG} action thus takes the form
 \bea{ActionHTlike_GUMG_ccf0}
    \SSS_{\ialt_0}[\Gaux,\CCof,\HTf]
    &=& \!\! \int \!d\taux \hspace{1pt} \dsx %\int \!d\taux\, \dsx \,
     \sqrt{|\Gaux|} \Big( \stR (\Gaux) - \sqrg^{-1}\!\FF^{-1}\!\opinvE \!\CCof \Big)
    + \int \!d\taux \hspace{1pt} \dsx %\!\int \!d\taux\, \dsx
     \,\partial_\Zm \HTf^\Zm  \CCof
    \;, \quad
   % \nonumber
  \eea
where $\Gaux_{\Zm\Zn}$ is the\HIDE{ covariant} spacetime metric corresponding to the ADM variables $(\g_{\zm\zn},\repNs^\zm,\repNl)$  via
 $ % \beq{GBulk_rep_metric_ADM_COPY}% was {Gaux_correspondence}
  \Gaux_{\Zm\Zn}d{X}^{\Zm}d{X}^{\Zn}
   \teq  \g_{\zm\zn} (d\sx^\zm {+} {\repNs}^\zm d\taux)(d\sx^\zn {+} {\repNs}^\zn d\taux) - {\repNl}^2 d\taux \hspace{1pt} d\taux
   \,,
 $ % \eeq
the volume density $ \sqrt{|\Gaux|} $ emerged from $\teq \repNl\! \sqrg$, and $\stR(\Gaux)$ is the spacetime scalar curvature.
The cosmological matter terms in (\ref{paramAction_GUMG_fin}), proportional to $\CCof$, remain unaffected by the reduction.
For $\WW \tneq {-}1$, these terms\HIDE{ introduce a perfect-fluid degree of freedom and} break spacetime covariance.
The  Henneaux--Teitelboim\HIDE{ auxiliary} auxiliary spacetime vector field
  $ % \bea{def_tef_COPY}
    \HTf^\Zm \equiv \big(\HTf^0, \HTf^\zm \big)
    \teq
    \big(\tf^\io, \lmptf^\zm \big)
    \,,
  $ % \eea
(\ref{def_tef}),
aggregates in the covariant divergence $ \partial_\Zm \HTf^\Zm \teq \tauxdot{\tf}^\io {\,+\,} \Dvg[\zm]{\lmptf} $. This combination, appearing in the last term of (\ref{ActionHTlike_GUMG_ccf0}), ensures that $\CCof$ is a spacetime constant on shell.

\newpar

Another convenient form of the alternative representation arises from redefinition $\CCof \to \CCf$:
  \beq{ccf_GUMG}
  %%  \CCf \HIDE{(\taux,\sx)} \defeq \sqrg^{-1}\!\FF^{-1}\! \opinvE \! \CCof \,,
    \CCof \defeq \opE \FF \sqrg \CCf  \,,
  \eeq
%%where $\opinvE$, (\ref{invE_kernel}), acts to the right on $\CCof$,
where $\opE$, (\ref{opE_kernel}), acts to the right,
leading to the alternative action (\ref{Alt_Action_GUMG}),
 \bea{ActionHTlike_GUMG} %(\ref{ActionHTlike_AMG})
    \SSS_{\ialt}[\Gaux,\CCf,\HTf]
    &=& \!\! \int \!d\taux \hspace{1pt} \dsx  %\int \!d\taux\, \dsx \,
     \sqrt{|\Gaux|} \,\big( \stR(\Gaux) - \CCf \big)
    +  \int \!d\taux \hspace{1pt} \dsx  %\!\int \!d\taux\, \dsx
     \,\partial_\Zm \HTf^\Zm  \opE  \FF \! \sqrg \CCf
    \,.
 \eea
The first term takes the form of the Einstein-Hilbert action with the dynamical cosmological constant $\CCf(\taux,\sx)$. The apparent noncovariance shifts to the second term, which now confines model specific characteristic functions $\FF(\argsqrg)$ and $\WW(\argsqrg)$ (within $\opE$).
For the {\wGUMG} subfamily (\ref{AMG_part_case}), the correspondence with (\ref{ActionHTlike_AMG}) is achieved with $\opE \to 1$. In the {\UMG} exceptional case (\ref{UMG_part_case}), $\opE  \FF  \sqrg \to 1$ restoring spacetime covariance of the action.

All intermediate transformations of the {\GUMG} theory actions leading to (\ref{ActionHTlike_GUMG_ccf0}) and (\ref{ActionHTlike_GUMG}) preserve classical equivalence, ensuring that the alternative representation is\HIDE{, by construction,} equivalent to the original formulation.

%%------------------------=%       ******        %=------------------------%%
%%------------------------=%       ******        %=------------------------%%

%\newpage
   \subsection{Dynamical properties}
    \label{SSect:DynamicalProperties_GUMG}
     \hspace{\parindent}
The structure of\HIDE{ the} equations of motion is similar to that in {\wGUMG}, except for the presence of spatial nonlocality.
The variational equations with respect to the auxiliary fields  $\HTf^\Zm$\HIDE{, which enter the action linearly,}
  \beq{EoM_tef_GUMG}
    \VDer{\SSS_{\ialt}}{\HTf^\Zm} \eomeq 0
    \qquad\vareqmapsto\qquad
    \partial_\Zm (  \opE  \FF  \sqrg \CCf )
    \;\equiv\; \partial_\Zm  \CCof \eomeq 0
    \quad\Rightarrow\quad
    \CCof \,\eomeq\, 0
    \, ,
  \eeq
generalize the behavior of the cosmological constant, implying that the combination $\CCof \equiv  \opE  \FF  \sqrg \CCf$, (\ref{ccf_GUMG}), is a spacetime constant on shell.
The\HIDE{ presence of the} operator $\opE$, (\ref{opE_left-right_action}), introduces spatial nonlocality into the classical relation between the spatial measure factor $\sqrg$ and the cosmological-constant field $\CCf$ in the\HIDE{ action}{ representation} (\ref{ActionHTlike_GUMG}):
  \beq{ccf_onshell_GUMG}
    \CCf
    \eomeq \sqrg^{-1} \FF^{-1} \opinvE \cco
    = \sqrg^{-1} \FF^{-1} \frac{\WW^{-1}}{\,\hmg{\!\WW^{-1}\!}\,\vphantom{I^{|^I}}} \, \cco
    \, ,
  \eeq
where $\cco$ is an on-shell constant value of $\CCof$, \,and\, $\hmg{\!\WW^{-1}\!}(\taux)$ denotes the spatial average of the $\WW^{-1}(\taux,\sx)$ over a $\taux \teq \const$ hypersurface.
This result follows from the properties\HIDE{ \ref{Corollary:opE_1}}\footnote{
 For the more general delocalization operator $\opEa \,\equiv \hmgId\! + \ofa^{-1} \inhId \WW$, mentioned in Section \ref{SSect:AltAction_GUMG}, the on-shell expression would be $\CCf {\,\eomeq\,} \sqrg^{-1} \FF^{-1} \opEa^{-1} \cco $. However, according to Corollary \ref{Corollary:EE-1_1}, this again simplifies to $\sqrg^{-1} \FF^{-1} \frac{\WW^{-1}}{\,\hmg{\!\WW^{-1}\!}\,\vphantom{I^{|^I}}} \, \cco$ from the right-hand side of (\ref{ccf_onshell_GUMG}), verifying the independence of the choice of $\ofa$.
} (\ref{opinvE_properties}) of $\opinvE$.
%% v.2023
%% for general operator $\opEE^{-1}$ acting (to the right) on unit function. It is noteworthy that the result does not depend on arbitrary functional parameter of $\ofA$ in $\opEA$.
 %
\if{ $%% ADD in version 2023-06
  This result could be alternatively obtained by noticing that in initial minimally parameterized action (\ref{paramAction_cons_GUMG}) variational equations of motion with respect to inhomogeneous part of Lagrange multiplier $\replm$ imply $\po \eomeq \hmg{\po}(\taux)$ whereas variational equation with respect to $\xo$ finally leads to on-shell constraint $\po(\taux,\sx) \eomeq \cco$ where $\cco$ is arbitrary integration constant.\footnote{This constant is defined from the initial conditions on physical fields from variational equation of motion with respect to homogeneous part of $\replm$, $\cco \eomeq - \hmg{\FF\Hl}$.} From (\ref{ptf_to_ccf}) by definition $\CCf \equiv  \sqrg^{-1}\FF^{-1}\opinvE \po$ so on shell $\CCf \eomeq  \sqrg^{-1}\FF^{-1}\opinvE \cco$. By the property (\ref{opE_properties}) spatially homogeneous functions are (right) eigenfunctions of operator $\opE \WW^{-1}$ with eigenvalue $\hmg{\WW^{-1}}$. Due to nondegenerateness of $\opE$ (see Lemma \ref{Lemma:InvE}) (and thus of $\opE \WW^{-1}$), the inverse operator $\WW\opinvE$ exists and spatially homogeneous functions are its eigenvectors with eigenvalue $\big(\hmg{\WW^{-1}}\big)^{-1}$. This finally leads to $\CCf  \eomeq \sqrg^{-1} \FF^{-1} \WW^{-1} \big(\hmg{\WW^{-1}}\big)^{-1} \cco $ (\ref{varLambda_sol}).
}\fi
For the {\wGUMG} subfamily with $\WW=\wwc \teq \const$, the nonlocal factor ${\WW^{-1}} / {\,\hmg{\!\WW^{-1}\!} }$ reduces to unity, reproducing the results of Section \ref{SSect:Dynamical properties_AMG}.
In the unimodular case, where $\FF(\argsqrg) \teq \sqrg^{-1}$ and $\WW \teq {-}1$, the cosmological-constant field $\CCf$ remains equal to constant $\cco$ on shell.

The variational equation with respect to $\CCf$,
 \beq{EoM_ccf_GUMG}
    \VDer{\SSS_{\ialt}}{ \CCf }
   \eomeq 0
   \qquad\vareqmapsto\qquad
    \sqrt{|\Gaux|} \eomeq \partial_\Zm \HTf^\Zm  \opE  \FF  \sqrg
    \,,
 \eeq
relates the determinant of the spacetime metric to the nonlocally weighted divergence of the Henneaux--Teitelboim auxiliary field (the operator $\opE$ acts to the left on $\partial_\Zm \HTf^\Zm$), where the weight depends on the induced metric determinant $\sqrg$.
\if{
 \footnote{
  Note that the dynamical nature of cosmological-constant field $\CCf$ and its role of the Lagrange multiplier ensures that gauge-breaking term compensates cosmological-constant term in Hamilton principal function of the theory\HIDE{, like it happens for the Henneaux--Teitelboim representation of unimodular gravity}.
    %%https://en.wikipedia.org/wiki/Hamilton%E2%80%93Jacobi_equation#Hamilton's_principal_function
  }
}\fi

In terms of the ADM parameterization of the metric,\HIDE{ from this} one can express the lapse function of the covariant metric\HIDE{ of the parameterized action} as $\repNl \teq \partial_\Zm \HTf^\Zm  \opE \FF$, or relate the delocalized analog of the initial restriction condition to the divergence of the auxiliary field: $ \repNl \FF^{-1}  \opinvE \teq \partial_\Zm \HTf^\Zm $, where the operator $\opE$ acts to the left. In view of the redefinition (\ref{replm_to_repNl_and_ccf0}), the left-hand side of the latter equality is nothing but the Lagrange multiplier $\replm$,
reinstating the on-shell relation
 % $\replm \eomeq \tauxdot \tf^\io {\,+\,} \Dvg[\zm]{\lmptf} \equiv \partial_\Zm \HTf^\Zm$
  \beq{replm_onshell}
    \replm {\;\eomeq\;} \tauxdot \tf^\io + \Dvg[\zm]{\lmptf}
    {\;\equiv\;} \partial_\Zm \HTf^\Zm
    \,,
  \eeq
from the canonical action (\ref{paramAction'_GUMG}). However, $\replm$ is a Lagrange multiplier associated with the mixed-class constraint\footnote{
 Only the spatially homogeneous component of $\replm$ remains undetermined, corresponding to the first-class constraint (\ref{first_class_param_constraint_GUMG}).  The\HIDE{ weighted inhomogeneous} non-homogeneous part\HIDE{ of $\replm$} vanishes on shell: $\inh{ \replm \WW^{-1} } \eomeq 0$, as required by the Dirac consistency conditions for the second-class constraints in the canonical system (\ref{paramAction'_GUMG}).}.
By taking into account the on-shell constraints on the Lagrange multipliers associated with second-class constraints, one can further refine the relation.

A systematic approach to tracking the conditions imposed on canonical Lagrange multipliers is to start from the emerging conditions in the consistently parameterized action (\ref{paramAction_cons_GUMG}) and follow the redefinitions of the Lagrange multipliers. From the structure relations (\ref{GUMG_Algebra_offshell}) and (\ref{GUMG_addpar_offshell}), it follows that the longitudinal spatial diffeomorphisms and the secondary constraint are independent second-class constraints. The consistency conditions then imply

\if{ %% v.2025-02
  Equation (\ref{EoM_ccf_GUMG}) together with the on-shell equation\footnote{
   Only the spatially homogeneous part of $\replm$ turns to be the Lagrange multiplier for the first-class constraint (\ref{first_class_param_constraint_GUMG}). The\HIDE{ weighted inhomogeneous} non-homogeneous part of $\replm$ will vanish on shell (\ref{replm_onshell}), $\inh{ \replm \WW^{-1} } \eomeq 0$,  due to the Dirac consistency conditions in the canonical system (\ref{paramAction'_GUMG}) as the Lagrange multiplier for a second-class constraint.}
  on the non-homogeneous part of Lagrange multiplier $\replm$  reinstates the analog of {\GUMG} restriction condition\HIDE{ $\Nl=\FF(\argsqrg)$} (\ref{GUMG_restriction}) in the initial representation of the theory. In terms of the ADM decomposition of spacetime metric, equation (\ref{EoM_ccf_GUMG}) expresses the lapse function
  $\repNl\eomeq \partial_\Zm \HTf^\Zm  \opE  \FF$ in terms of auxiliary field and the determinant of the spatial induced metric. %\footnote{
  Note also that in view of the redefinition (\ref{replm_to_repNl_and_ccf0})\HIDE{ $\repNl = \replm\, \opEA\, \FF$} one gets the on-shell relation
    \beq{replm_onshell}
      \replm \eomeq \partial_\Zm \HTf^\Zm \equiv \tauxdot \tf^\io + \Dvg[\zm]{\lmptf}
      \,,
    \eeq
  which is the obvious consequence of $\VDer{\SSS_{par'}}{\ptf_\io}=0$  for the canonical action (\ref{paramAction'_GUMG}).
  %}

  The consistency conditions for the constraints of the parameterized action (\ref{paramAction_cons_GUMG}) lead to the constraints for the on-shell values of the Lagrange multipliers for the second-class constraints.\footnote{
   In contrast to the case of the unimodular gravity theory \cite{Henneaux:1989zc} where all constraints are first class, in {\GUMG} case there are two inhomogeneous-scalar second-class constraints on the non-GR branch (see discussion in Section \ref{SSect:GUMG_Constraint_Structure}). These second-class constraints are inherited by the consistently parameterized action.
   %Thus, there appear two on-shell constraints for the correspondent canonical Lagrange multipliers.
  }
   The Lagrange multipliers for the \emph{secondary} second-class constraints vanish on shell which here gives $\Dvg{ \replmCT} \eomeq 0$. One can check it directly by using the Poisson brackets of constraints in table \ref{table:Par_PB_WW} that for the minimally parameterized system $\SSS_{par}$, (\ref{paramAction_cons_GUMG}), the consistency conditions for constraints imply
}\fi

  \beq{IIClm_onshell}
   \begin{array}{|lcl}
%\repNs^\zn &\eomeq&\!
   %\hmg{\replm^{\io}} \, U_0^\zn + \repNs_{\isst}^\zn
   \Dvg{\repNs}
    \!\!&\eomeq&\!
     \hmg{\replm^{\io}} \, \inh{U}_0 \HIDE{(\g,\pg)}
    \,\equiv\,
     \hmg{\replm^{\io}}
      \big( \OOmega^{-1}\TTheta \trpg
      - \OOmega^{-1}\tfrac{\sint \OOmega^{-1}\TTheta \trpg}{\sint \OOmega^{-1}}
      \big)
   \,, \\
   \Dvg{\replmCT} \!\!&\eomeq&\! 0
   \,, \\
   \end{array}
  \eeq
where $\inh{U}_0(\g,\pg)$ is defined in (\ref{DvgNs_from_CS_GUMG}), and the structures $\OOmega(\argsqrg)$ and $\TTheta(\argsqrg)$ are introduced in (\ref{def_OOmega_TTheta}).
Under the basis transformation
  $
   (\replm^\io \,, \replmCT^\zm)
    \;\to\;
   (\replm \,, \lmptf^\zm)
   \,,
  $
(\ref{basis_change_GUMG}), leading to the representation (\ref{paramAction'_GUMG}),
%When passing  (\ref{basis_change_GUMG}) to the representation (\ref{paramAction'_GUMG}),
the second relation of (\ref{IIClm_onshell}) imposes the constraint
  \beq{inhreplm_onshell}
   \inh{ \replm \WW^{-1} }
   {\;\eomeq\;} 0 \,,
  \eeq
on the Lagrange multiplier $\replm$.
This ensures that, on shell, the action of the operator $\opE$, (\ref{opE_left-right_action}), on $\replm$ simplifies: $\replm\opE \eomeq \hmg{\replm}$, leaving the homogeneous component $\hmg{\replm}$\HIDE{ of the Lagrange multiplier} undetermined. This homogeneous component corresponds to the first-class part of the mixed-class constraint.
Upon further redefinition $\replm \to \repNl$,
(\ref{replm_to_repNl_and_ccf0}), we obtain
 $\inh{(\repNl \FF^{-1} \opinvE \WW^{-1})} \eomeq 0$,
where the inverse operator acts to the left on $\repNl \FF^{-1}$ and then the average-free component is taken of the whole expression. Using the $\opinvE$ property (\ref{def_opGmp}),\HIDE{ from Corollary \ref{Corollary:var_invE_G_properties}} this simplifies to $\inh{\repNl \FF^{-1}} \eomeq 0$. At the same time, from\HIDE{ the first$ property from} (\ref{opinvE_properties}), the undetermined component satisfies $\hmg{\replm} \teq \tfrac{\hmg{\repNl \FF^{-1}\WW^{-1}}}{\hmg{\WW^{-1}}\vphantom{|^{I^I}}}$, which applying $\inh{\repNl \FF^{-1}} \eomeq 0$, on shell can be expressed as $\hmg{\repNl \FF^{-1}}$ or $\repNl \FF^{-1}$.
Finally, from (\ref{replm_onshell}),
the homogeneous part of $\replm$ satisfies $\hmg{\replm} \eomeq \tauxdot{\hmg{\tf}}^{\io}$.

Summarizing the\HIDE{ above} discussion of on-shell constraints on\HIDE{ the second-class} Lagrange multipliers one obtains
 \beq{repNl_repNs_onshell}
   \begin{array}{|lcl}
   \inh{\repNl \FF^{-1}}
     \! &\! \eomeq \!& 0
   \,, \\
   \Dvg{\repNs}
    &\! \eomeq \!&
   % \repNl \FF^{-1} \inh{U}_0
   % \; \eomeq \;
   % \hmg{\repNl \FF^{-1}} \inh{U}_0
   % \; \eomeq \;
    \tauxdot{\hmg{\tf}}^{\io} \,\inh{U}_0
   \,, \\
   \end{array}
 \eeq
where $\inh{U}_0(\g,\pg)$ is defined in (\ref{DvgNs_from_CS_GUMG})\footnote{The general solution for\HIDE{ the Lagrange multipliers} $\repNs^\zn$ is given by $\repNs^\zn \teq \hmg{\replm^{\io}} \, U_0^\zn {\,+\,} \repNs_{\isst}^\zn$, where $U_0^\zm(\g,\pg)$ is defined in (\ref{lm_solutions_nonGR})\HIDE{ up to an arbitrary transverse vector}, and $\repNs_{\isst}^\zn$ is an arbitrary transverse vector field, $\Dvg{\repNs_{\isst}} \teq 0$, left undetermined by the equations of motion.}.
Additionally, from the\HIDE{ homogeneous} relation $\hmg{\repNl\FF^{-1}} \eomeq \hmg{\replm} \eomeq \hmg{\partial_\Zm \HTf^\Zm}$, in view of $\inh{\repNl\FF^{-1}} \eomeq 0$ and 
  $ % \beq{replm_onshell_COPY}
    \hmg{\partial_\Zm \HTf^\Zm }
    \teq \tauxdot{\hmg{\tf}}^{\io} ,
  $ % \eeq
one recovers the parameterized version of the restriction condition (\ref{GUMG_restriction}) from the initial restricted {\GUMG} setup\footnote{
  For the more general delocalization operator, $\opEa = \hmgId {+\,} \ofa^{-1} \inhId \WW$,
   %% mentioned in Section \ref{SSect:AltAction_GUMG},
  equation (\ref{inhreplm_onshell}) modifies to
   $ %  \beq{replm_onshell_ofa}
     \inh{ \replm \ofa^{-1} } \eomeq 0
     \, .
   $ %  \eeq
  This further implies $\replm\opEa \eomeq \hmg{\replm}$ and
   $ % \beq{replm_to_repNl_onshell_ofa}
     \repNl
     %= \replm \opEA \FF
     \eomeq \hmg{\replm} \,\FF
     \eomeq \tauxdot{\hmg{\tf}}^{\io}\, \FF
     \, .
   $ %\eeq
  Thus, the resulting expression retains the same form as (\ref{repNl_onshell}).
}:
 \beq{repNl_onshell}
   \repNl
   \,\eomeq\, \tauxdot{\hmg{\tf}}^{\io}\, \FF(\argsqrg)
   \,.
 \eeq
The presence of the free factor $\tauxdot{\hmg{\tf}}^{\io}(\taux) $ reflects the time-reparametrization gauge freedom in the parameterized {\GUMG} theory (\ref{ActionHTlike_GUMG}), which acts on metric fields only through its spatially homogeneous component. In the correspondence gauge (\ref{correspondence_gauge}), this factor is set to unity.

Using (\ref{repNl_onshell}), equality (\ref{ccf_onshell_GUMG}) can be rewritten as
 $ % \beq{ccf_onshell_GUMG_COPY}
    \CCf
    \,\eomeq \sqrt{|\Gaux|}^{-1} \frac{\WW^{-1}}{\,\hmg{\!\WW^{-1}\!}\,\vphantom{I^{|^I}}} \, \tauxdot{\hmg{\tf}}^{\io} \cco
    \, ,
 $ % \eeq
while the{ on-shell} relation (\ref{EoM_ccf_GUMG}) takes the form
 $ % \beq{EoM_ccf_GUMG_COPY}
    \sqrt{|\Gaux|} \,\eomeq\, \tauxdot{\hmg{\tf}}^{\io} \FF  \sqrg
    \,. %% rechecked 2025-04
 $ %\eeq
Also, from (\ref{EoM_ccf_GUMG}), the divergence of the Henneaux--Teitelboim auxiliary field can be expressed in terms of the metric fields as $\partial_\Zm \HTf^\Zm \,\eomeq\, \hmg{\Nl\FF^{-1}} \frac{\WW}{\,\hmg{\WW}\,\vphantom{I^{|^a}}} $.

\newpar

The variational derivative of the action $\SSS_{\ialt}[\Gaux,\CCf,\HTf]$, (\ref{ActionHTlike_GUMG}), with respect to the spacetime metric $\Gaux_{\Zm\Zn}$ takes the form of the Einstein-Hilbert equations with matter
 \beq{EoM_GG_GUMG}
   \VDer{\SSS_{\ialt}}{\Gaux_{\Zm\Zn}}
  \eomeq 0
  \qquad\vareqmapsto\qquad
    \stR^{\Zm\Zn} - \tfrac12 \stR \Gaux^{\Zm\Zn} 
    \eomeq    
    \tfrac12 \repTSEt^{\Zm\Zn} 
    \,,
 \eeq
where all terms proportional to $\CCf$  are attributed to the matter sector, $\SSS_{mat}=\!\int \!d\taux\, \dsx \, (-\sqrt{|\Gaux|}\CCf + \partial_\Zm \HTf^\Zm  \opE  \FF  \sqrg \CCf)$, whose contribution is encoded in the stress-energy tensor on the right-hand side\HIDE{ of the equations}. Taking into account  (\ref{EoM_tef_GUMG}) and (\ref{EoM_ccf_GUMG}), the stress-energy tensor on shell is given by
 \beq{Stress-Energy_Tensor_GUMG}
    \repTSEt^{\Zm\Zn}
   \defeq \frac{2}{\sqrt{|\Gaux|}} \VDer{\SSS_{mat}}{\Gaux_{\Zm\Zn}}
     %%=
     %% - \Gaux^{\Zm\Zn} \CCf
     %% +(\wwc+1)\frac{\sqrg^{\wwc+1}}{\sqrt{|\Gaux|}}
     %% \big(\Gaux^{\Zm\Zn} + n^{\Zm}n^{\Zn}\big)
     %% \CCf \big(\tauxdot{\tf}^\io - \wwc \Dvg{\replmCT}\big)
   \eomeq     \CCf \, n^{\Zm}n^{\Zn}
    + \WW
    \CCf  \big(\Gaux^{\Zm\Zn} \!+ n^{\Zm}n^{\Zn}\big)
    \,,
 \eeq
where $n^\Zm$ is a unit vector field orthogonal to constant-time hypersurfaces. In (\ref{Stress-Energy_Tensor_GUMG}) the energy density is $\repeSEt \teq \CCf $ and the pressure is $\reppSEt \eomeq \WW \CCf$. Thus, the perfect fluid equation of state (\ref{PerfectFluid_Eq_of_state}),
 \beq{}
   \reppSEt \eomeq \WW \repeSEt
  \,,
 \eeq
with a dynamic parameter $\WW=\WW(\argsqrg)$, is recovered. %\footnote{
 %The cosmological-constant term $\sqrt{|\Gaux|}\CCf$ generates $-\CCf\Gaux^{\Zm\Zn}$ contribution to $T_{\wwc}^{\Zm\Zn}$ and gives equal contribution to the energy and pressure. Whereas gauge-breaking term from the second integral (\ref{ActionHTlike_AMG}), which feels only spatial induced metric, gives contribution only to pressure component: $(\wwc+1)\CCf \big(\Gaux^{\Zm\Zn} + n^{\Zm}n^{\Zn}\big)$, in which by virtue of (\ref{EoM_ccf_GUMG}) we equated the overall factor $\sqrg^{\wwc+1}(\tauxdot{\tf}^\io - \wwc \Dvg{\lmv}/\sqrt{|\Gaux|}$ to unity. Extracting contribution to energy $\sim -n^\Zm n^\Zn$ from the first contribution and transferring the rest to pressure $\sim \Gaux^{\Zm\Zn} + n^\Zm n^\Zn$ gives the claimed balance $\reppSEt=\wwc\repeSEt$.
%}
 %
The dynamical behavior of the energy (and pressure) is more transparent when\HIDE{ the $\CCf$ is} expressed in terms of the\HIDE{ spatially} homogeneous constant of motion $\CCof \eomeq \cco$, (\ref{EoM_tef_GUMG}):
\vspace{-2mm}
  \bea{energy_pressure_GUMG}
  % \begin{array}{rcl}
    \repeSEt
     \!\!&\eomeq&\!\!
      \sqrt{|\Gaux|}^{-1} \frac{\WW^{-1}}{\,\hmg{\!\WW^{-1}\!}\,\vphantom{I^{|^I}}} \,\tauxdot{\hmg{\tf}}^\io \, \cco
     \,\eomeq\,
      \sqrg^{-1} \FF^{-1} \frac{\WW^{-1}}{\,\hmg{\!\WW^{-1}\!}\,\vphantom{I^{|^I}}}  \, \cco
   \,.
  %  \end{array}
  \eea
In the {\wGUMG} subfamily (\ref{AMG_part_case}), the nonlocal factor disappears, ${\WW^{-1}}\!/\,\hmg{\!\WW^{-1}\!} \to 1$.  In the unimodular gravity case, $\FF(\argsqrg) \teq \sqrg^{-1}$, $\WW \teq {-}1$,
 %%in the correspondence gauge (\ref{correspondence_gauge})
one finds $ \repeSEt \eomeq {-} \reppSEt \eomeq \HIDE{ \CCf = \CCof \eomeq } \cco$, as expected.

Notably, the metric dependence of the nonlocal operator $\opE(\g)$, (\ref{opE_kernel}, \ref{opE_left-right_action}), does not affect the on-shell
result of the metric variation.
The terms arising from the metric variation of $\opE$ vanish due to the equations of motion for the auxiliary fields\footnote{
 For the more general\HIDE{ delocalization} operator $\opEa = \hmgId {+\,} \ofa^{-1} \inhId \WW$, mentioned in Section \ref{SSect:AltAction_GUMG},
a similar result holds:
  \beq{}
    \sint % \int \!d\taux\, \dsx \,\,
     \partial_\Zm \HTf^\Zm (\var_{\!\Gaux} \opEa ) \FF \sqrg \CCf
    =
    \sint % \!\int \!d\taux\, \dsx \,
    % \LiB
   \big(\inh{\partial_\Zm \HTf^\Zm  \var_{\!\Gaux} \ofa^{-1}}\big) \,
   \big({\inh{\WW \FF  \sqrg \CCf }}\HIDE{_{\eomeq 0}}\big)
   +
    \sint % \!\int \!d\taux\, \dsx \,
     \big({\inh{\partial_\Zm \HTf^\Zm \ofa^{-1}\,}\!}\HIDE{_{\eomeq 0}}\big) \,
     \big(\inh{ (\var_{\!\Gaux}\WW) \FF\! \sqrg \CCf }\big)
    % \RiB
    \eomeq 0
    \,.
   \nonumber
  \eeq
This vanishes on shell due to
  $ % \beq{}
    \inh{\big(\WW \FF  \sqrg \CCf\big)}
    \eomeq \inh{ \big(\frac{\,\cco}{\,\hmg{\!\WW^{-1}\!} {\vphantom{I^{|^I}}}}\big) }
    \teq
     0
    \,
  $ % \eeq
and
$\inh{\big(\partial_\Zm \HTf^\Zm \ofa^{-1}\big)} \eomeq \inh{\big(\replm \ofa^{-1}\big)} \eomeq 0 $.
}:
  \bea{var_g_opE}
    \sint % \int \!d\taux\, \dsx \,\,
    \partial_\Zm \HTf^\Zm (\var_{\!\Gaux} \opE ) \FF  \sqrg \CCf
    \,\,=\,
    \sint %\!\int \!d\taux\, \dsx \,
    %\LiB\!
   \big(\inh{\partial_\Zm \HTf^\Zm  \var_{\!\Gaux} \WW^{-1}\!}\big)
   \big(\HIDE{\!\underbrace}{\inh{\WW \FF  \sqrg \CCf }}\HIDE{_{\eomeq 0}}\big)
   +
    \sint %\!\int \!d\taux\, \dsx \,
     \big(\HIDE{\!\underbrace}{\inh{\partial_\Zm \HTf^\Zm \WW^{-1}}}\HIDE{_{\eomeq 0}}\!\big)
     \big(\inh{ (\var_{\!\Gaux}\WW) \FF\!  \sqrg \CCf }\big)
    % \!\RiB
     \,\eomeq\, 0
    \,.\;\;
  \eea
\vspace{-4mm}\\
The first term on the right-hand side vanishes\HIDE{ on shell} due to\HIDE{ the on-shell behavior of $\CCf$,} (\ref{ccf_onshell_GUMG}):\,
   $ % \beq{}
    \big( \inh{ \WW \FF\!  \sqrg \CCf } \big)
    \eomeq \big( \inh{ \frac{\, \cco}{\,\hmg{\!\WW^{-1}\!}\vphantom{I^{|^I}}} } \big)
    {\,=\,} 0
    \,.
  $  % \eeq
The second term vanishes\HIDE{ on shell} due to the on-shell conditions $\inh{\partial_\Zm \HTf^\Zm \WW^{-1}} \eomeq \inh{\replm \WW^{-1}} \eomeq 0 $, (\ref{replm_onshell}, \ref{inhreplm_onshell}).

%\vspace{1cm}
\if{ %% version 2023-06
%\paragraph{Reminiscence of the {\GUMG} restriction condition.}
The other Lagrange multiplier on-shell constraint reinstates the {\GUMG} restriction condition. In Lagrangian parameterized action one can obtain this relation by varying the action (\ref{ActionHTlike_GUMG}) with respect to the ``cosmological constant'' field
 \beq{EoM_xo}
   \VDer{S}{ \CCf }
  =
   - \repNl\sqrg
   + \tauxdot{\xo} \opE \FF\sqrg
   - \Dvg{\lmv}_2 \WW \FF \sqrg
  \eomeq 0
 \eeq
where
 \beq{}
   \Dvg{\lmv}_2
  = \Dvg{\lmv}
   + \inh{\replm \WW^{-1}}
  \eomeq   \inh{\replm \WW^{-1}}
 \eeq
where we took into account on-shell equality $\Dvg{\lmv} \eomeq 0$ (\ref{lm_solutions_nonGR}).
Variational equation of parameterized action (\ref{paramAction_cons_GUMG})  with respect to $\po$\HIDE{, $\VDer{\SSS_{par}}{\po} \eomeq 0$,} implies $\replm \eomeq \tauxdot{\xo}$.

So for (\ref{EoM_xo}) one gets
 \beq{}
   - \repNl\sqrg
   + \tauxdot{\xo} \big( \opE - \WW^{-1} \inhId \WW \big) \FF\sqrg
  =
   - \repNl\sqrg
   + \tauxdot{\xo} \hmgId  \FF\sqrg
  \eomeq 0
 \eeq
where operators are acting to the left.

Thus, finally one gets
 \beq{}
  \repNl
  \eomeq   \tauxdot{\hmg{\xo}} \FF
  %\eomeq   %\hmg{\replm} \FF
 \eeq
which is parameterized version of the generalized unimodular restriction condition (\ref{GUMG_restriction}).
}\fi

%%-------------------------=%        ******         %=-------------------------%%
%%-------------------------=%        ******         %=-------------------------%%

  \subsection{Gauge structure\HIDE{ on the non-GR branch}}
   \label{SSect:GaugeStructure_GUMG}
    \hspace{\parindent}
The canonical treatment reveals the following first-class constraints on the non-GR branch for the consistently parameterized {\GUMG} action (\ref{paramAction_cons_GUMG}),
 \beq{I-class_constr}
  \begin{array}{|lll}
  \sint \geps (\ptf_\io + \hmg{\FF\Hl} + \hmg{U_0^\zn \Hs_\zn})
  \,, \\
  \sint \gzetat^\zm \Hs_\zm
  \,, \\
  \end{array}
 \eeq
where $U_0^\zn$ is the particular on-shell solution for the Lagrange multipliers $\Ns^\zn$ (\ref{def_Uo}), and $\gzetat^\zm$ is the transverse vector gauge parameter, satisfying $\Dvgi[\zk]{\gzetat} \teq 0$.
The spatial nonlocality of the gauge structure makes the analysis more intricate than in {\wGUMG}, particularly due to the nonlocal nature of $U_0^\zn$. However, this is mainly a technical and presentational complication; structurally, the gauge transformations remain similar to those in {\wGUMG}.

Rather than following the laborious approach of Section \ref{SSect:GaugeStructure_AMG}, which tracks the gauge transformations under field redefinitions (some of which are spatially nonlocal), we adopt a kind of gauge-restriction approach from the outset. Assuming canonical symmetries, we perturb the\HIDE{ canonical} action (\ref{paramAction_GUMG_fin})\HIDE{, whose configuration space includes ADM metric fields,} under unrestricted canonical transformations generated by \emph{all} constraints. We then examine the anomaly to determine the constraints on transformation parameters and identify the residual free parameters corresponding to gauge symmetry\footnote{Restricting general canonical transformations by anomaly-free condition here\HIDE{ turns out to be equivalent to} can be interpreted as examining gauge restrictions on the {\GR} gauge transformations, motivated by the restricted-theory approach.}.

\if{ %% DRAFT 2024-09
  Gauge transformations generated by the first first-class constraint in the case of nontrivial $U_0^\zn$ (\ref{def_Uo}) for nonconstant $\WW$ is rather involved for phase space fields. It can be simplified by subtracting trivial gauge transformation so that
   \beq{gauge_transf_noncanon}
    \begin{array}{|lll}
     \gvar \g_{\zn\zm}
   %  &=& \PB{\g_{\zn\zm}}{\sint (\po + \hmg{\FF\Hl}) \geps }
   %  + \PB{\g_{\zn\zm}}{\sint  \Hs_\zk } U_0^\zk \hmg{\geps}
   %  + \PB{\g_{\zn\zm}}{\sint  \Hs_\zk } \gzeta^\zk_{\isst}
     &=&
     \gvar^{canon} \g_{\zn\zm}
     + \mut_{(\zn\zm)}^{\;\zk} \VDer{S}{\repNs^\zk}
     \,, \vphantom{\big|^I} \\
     \gvar \pg^{\zn\zm}
   %  &=& \PB{\pg^{\zn\zm}}{\sint  (\po + \hmg{\FF\Hl}) \geps }
   %  + \PB{\pg^{\zn\zm}}{\sint  \Hs_\zk } U_0^\zk \hmg{\geps}
   %  + \PB{\pg^{\zn\zm}}{\sint  \Hs_\zk } \gzeta^\zk_{\isst}
     &=&
     \gvar^{canon} \pg^{\zn\zm}
     + \mut^{(\zn\zm)\zk} \VDer{S}{\repNs^\zk}
     \,, \vphantom{\big|^{I^I}} \\
     \gvar \Ns^\zn
       &=&
     \gvar^{canon} \Ns^\zn
      - \VDer{S}{\g_{\zn\zm}} \mut_{(\zn\zm)}^{\;\zk}
      - \VDer{S}{\pg^{\zn\zm}} \mut^{(\zn\zm)\zk}
     \,, \vphantom{\big|^I} \\
    \end{array}
    \quad \text{ where} \quad
    \begin{array}{|lll}
    \mut_{(\zn\zm)}^{\;\zk}
    &\equiv& \PB{\g_{\zn\zm}}{U_0^\zk} \hmg{\geps}\, , \vphantom{\big|^I}\\
    \mut^{(\zn\zm)\zk}
    &\equiv& \PB{\pg^{\zn\zm}}{U_0^\zk} \hmg{\geps}\, .\vphantom{\big|^I}\\
    \end{array}
   \eeq
  These infinitesimal gauge transformations (\ref{gauge_transf_noncanon}) acquire the form
   \beq{gauge_transfs_ham}
    \begin{array}{lll}
     \gvar \g_{\zn\zm}
     &=& \big( \g_{\zk\zm} \partial_\zn %\gzeta_{\isst}^{\zk}
     + \g_{\zn\zk} \partial_\zm %\gzeta_{\isst}^{\zk}
     + (\partial_\zk \g_{\zn\zm}) %\gzeta_{\isst}^{\zk}
     \big) ( \gzeta_{\isst}^{\zk} + U_0^\zk \hmg{\geps})
     + \tfrac{\FF}{\sqrt{\g}}\left(2\pg_{\zm\zn} -\frac{2}{\Ddim{-}2}\,\trpg\,\g_{\zm\zn} \right)
     \hmg{\geps} , \\
     \gvar \tf^\io &=& \geps, \\
     \gvar \ptf_\io &=& 0, \\
     \gvar \repNs^\zn
       &=& \TDer{}{\taux} (\gzeta_{\isst}^{\zn} + U_0^\zn \hmg{\geps})
           - \LieB{\zn}{\Ns}{\gzeta_{\isst} + U_0 \hmg{\geps}}
           - \FF^2 \g^{\zn\zm} \partial_\zm \big(\WW\Dvg{\lmv}\big) \hmg{\geps},
     \\
     \gvar \Dvg{{\lmv}} &=& % 0 +
                          \hmg{\replm} \Dvg{U_0} \hmg{\geps}
                           - \Dvg{\Ns} \hmg{\geps}
                          + \partial_\zk \Dvg{\lmv}\, ( \gzeta_{\isst}^{\zk} + U_0^\zk \hmg{\geps})
                          +  \Dvg{\lmv}  \Dvg{U_0}\hmg{\geps}
                          + \Dvg{\lmv}  \tfrac{\sint \OOmega^{-1}\TTheta \trpg}{\sint \OOmega^{-1}} \hmg{\geps} ,
     \\
     \gvar \replm &=& \TDer{}{\taux} \geps %\tauxdot{\geps}
     \;,
     \\
    \end{array}
   \eeq
  where $\LieB{\zn}{\xi}{\eta}$ is the Lie bracket of two vector fields, $\LieB{\zn}{\xi}{\eta} = {\xi}^\zm\partial_\zm {\eta}^\zm - {\eta}^\zm\partial_\zm {\xi}^\zn $.
}\fi

Consider the variation of the action (\ref{paramAction_GUMG_fin}), decomposed into {\GR} and auxiliary matter\HIDE{ gauge-breaking} parts:
   \bea{paramAction_GUMG_fin_COPY2} %(\ref{parAction_BLRG})
    %&&
    && \SSS_{par''}[\g,\pg,\tf^\io\!,\CCof,\repNl\!,\repNs,\lmptf]
     \;=\; \!\! \SSS^{\iGR}[\g,\pg,\repNl\!,\repNs] + \SSS_{mat}[\g,\repNl\!,\repNs,\CCof,\tf^\io\!,\lmptf]
     \nonumber\\
     && \;= \;\;
     \!\! \int \!d\taux \hspace{1pt} \dsx  %\int d\taux\, \dsx
     \,\Lib
       \pg^{\zm\zn}\tauxdot{\g}_{\zm\zn}
       \!- \repNl  \Hl
       - \repNs^\zm \Hs_\zm\!
     \Rib
     +
     \!\! \int \!d\taux \hspace{1pt} \dsx  %\int d\taux\, \dsx
     \,\Lib\,
      \tauxdot{\tf}^{\io} + \Dvg{\lmptf}
     - \repNl \FF^{-1}\opinvE
     \Rib \CCf_\io
     \,, \;\;
     \nonumber
   \eea
under the \emph{unrestricted} canonical transformations of the phase space variables,
  \beq{ugvar_canon_GUMG}
   \left|
    \begin{array}{lll}
     \ugvar \g_{\zn\zm}
     &\!=& \big( \g_{\zk\zm} \partial_\zn
     + \g_{\zn\zk} \partial_\zm
     + (\partial_\zk \g_{\zn\zm})
     \big)
     \gxs^{\zk} %\HIDE{\gzetat^{\zk}}
     + \tfrac{2 }{\sqrt{\g}}\left(\pg_{\zm\zn} -\frac{1}{\Ddim{-}2}\,\trpg\,\g_{\zm\zn} \right)
     \gxl %\HIDE{\hmg{\geps}}
     \, ,
     \\
     \ugvar \pg^{\zn\zm}
     &\!=&
     \gvarGR[\gxl,\gxs] \pg^{\zn\zm} +\PB{ \pg^{\zn\zm}}{\sint \gxl \FF^{-1}\!\opinvE \, \ptf_\io} \,,
    % \;=\;
    %  ... \;\;\text{ see (\ref{qieDkzkFfbs}) }
    % \\
     \\
     \ugvar \tf^\io
     &\!=&
      \gxl  \FF^{-1} \opinvE       - \Dvg{\gchi} \,,
      %\equiv \geps,
     \\
     \ugvar \CCof
     &\!\equiv&
     \gvar \ptf_\io
     \:=\;
     0 \,,
     \\
    \end{array}
   \right.
  \eeq
generated by the constraints,
%(\ref{gauge_transfs_ham_AMG'})
  \beq{canon_transf_GUMG''}
    \ugvar {(\,.\,)} = \PB{\,.\,}{\sint \gxl ( \Hl {+} \FF^{-1}\opinvE \ptf_{\io} )} + \PB{\,.\,}{\sint \gxs^\zm \Hs_\zm} + \PB{\,.\,}{\sint \gchi^\zm \ptf_{\io,\zm}}
    \,,
  \eeq
where gauge parameters $\gxl(\taux,\sx)$, $\gxs^\zm(\taux,\sx)$, $\gchi^\zm(\taux,\sx)$ are arbitrary functions, and under the corresponding compensating transformations of the canonical Lagrange multipliers.
From general principles --- from the requirement to cancel terms proportional to constraints in the phase-space\HIDE{ canonical} variations (\ref{canon_transf_GUMG''}), it follows that the gauge transformations of the metric Lagrange multipliers
  %%$(\repNl,\repNs^{\zn})$
are given by the general-relativistic diffeomorphism transformations\footnote{
 The transformations of the Lagrangian metric fields $\g_{\zm\zn},\Nl,\Ns^\zn$\HIDE{ are manifestly \emph{parental}, meaning they} coincide with the general-relativistic gauge diffeomorphisms, generated by the {\GR} canonical constraints $\Hl$ and $\Hs_\zn$.
  %%Transformations of the momenta generically acquire corrections, when restriction condition (or gauge-breaking condition here) is $\g_{\zm\zn}$-dependent. Of course it can be guessed in advance from the canonical structure of the action. Gauge transformations of auxiliary variables also could be guessed from the canonical rules, however it is often laborious. It is often easier to restore these transformations from direct calculation of gauge invariance. Here, however, we expect gauge anomaly for the unrestricted transformations, and we have auxiliary variables, absent in the parental theory configuration space. Thus one should be more cautious when presuming some gauge transformation properties from the parental gauge theory. However, this discussion goes beyond the focus of this paper. So one can consider this guess as the calculational tool for easier recovering canonically induced gauge transformations.
 Transformations of the gravitational momenta\HIDE{ generally} acquire corrections due to the $\g$-dependent extension of the parametrization constraint. However, up to a trivial gauge transformation\HIDE{ (i.e., those vanishing on the restricted equations of motion)}, these can be attributed to a general-relativistic gauge transformations, generated by the rescaled constraint $\opE \FF\Hl$ with correspondingly rescaled gauge parameter $\gxl\FF^{-1}\!\opinvE$.
  %Of course it can be guessed in advance from the canonical structure of the action. Gauge transformations of auxiliary variables also could be guessed from the canonical rules, however it is often laborious. It is often easier to restore these transformations from direct calculation of gauge invariance. Here, however, we expect gauge anomaly for the unrestricted transformations, and we have auxiliary variables, absent in the parental theory configuration space. Thus one should be more cautious when presuming some gauge transformation properties from the parental gauge theory. However, this discussion goes beyond the focus of this paper. So one can consider this guess as the calculational tool for easier recovering canonically induced gauge transformations.
 %% FOR THE FUTURE:
  %% The gauge transformations generated from  $\opE \FF\Hl$ in fact coincide for all metric fields in {\GR} and parameterised {\GUMG} theories (since the constraints are shifted by the gauge invariant canonical auxiliary parameter! However in this constraint basis the (unrestricted) algebras become open. And the theories do not satisfy the conditions of the theorem in \cite{Barvinsky:2022guw} which predicted the certain restricted correspondence between the (closed) algebras of parental and restricted theories. In GR closed are algebras in canonical or covariant gauge bases (which are related by the %\g%-independent transformations).
}
(\ref{GR_canon_gauge_transfs}),
  \beq{ugvar_NsNl_GUMG}
   \left|
    \begin{array}{lll}
     \ugvar \repNl
      &\!=& \tTDer{}{\taux} {\gxl} + \gxs^\zn \repNl_{\!\!\!\!,\zn} - \repNs^\zn\gxl_{,\zn}
      \,,
     \\
     \ugvar \repNs^{\zn}
      &\!=&  \tTDer{}{\taux} {\gxs}^{\zn} + \LieB{\zn}{\gxs^{}}{\repNs^{}} + (\gxl {\ader_\zm} \repNl) \g^{\zn\zm}
      \,.
     \\
     \end{array}
   \right.
  % \nonumber
  \eeq
The corresponding variation of $\Dvg{\lmptf}$\HIDE{ induced by the canonical gauge transformations} could be computed explicitly in terms of nontrivial structure functions (operators), but we leave it undefined for now to highlight the emergence of the anomaly and its structure. The action variation (\ref{paramAction_GUMG_fin}) under the above canonical transformations is given by:
   \bea{ugvar_paramAction_GUMG''} %(\ref{parAction_BLRG})
    \ugvar \SSS_{par''} %% [\g,\pg,\tf^\io,\CCof,\repNl\!,\repNs,\lmptf]
     &\!=\!&
    - \int \!d\taux\,\HIDE{ \dsx \,} \LiB
      \sint \dotuline{\gxl} \tTDer{}{\taux} (\FF^{-1}\opinvE) \CCof
   % \nonumber\\
   % &&
    \HIDE{-}+ \PB{\sint (\repNl \Hl {+} \repNs^\zn \Hs_\zn )}{\sint \gxl \FF^{-1}\opinvE \CCof}
    \RiB
   \nonumber\\
    &&
    - \int \!d\taux\,\HIDE{ \dsx \,} \LiB
     \sint
     ( \dotuline{\tTDer{}{\taux} {\gxl}} {\,+\,} \gxs^\zn \repNl_{\!\!\!\!,\zn} {-} \repNs^\zn\gxl_{,\zn} )
     \FF^{-1}\opinvE \CCof
 %  \nonumber\\
 %   &&
    + \PB{\sint \repNl \FF^{-1}\opinvE \CCof}{\sint (\gxl \Hl {+} \gxs^\zn \Hs_\zn)}
    \RiB
    \nonumber\\
    &&
     +    \!\! \int \!d\taux\, \HIDE{\dsx \,} \LiB
     %%\sint \ugvar \big(\tauxdot{\tf}^\io {+} \Dvg{\lmptf}\big) \, \CCof
     \sint \tTDer{}{\taux} \big( \dotuline{\gxl} \FF^{-1}\opinvE       {\,-\,} \Dvg{\gchi} \big) \CCof
       + \sint \ugvar \Dvg{\lmptf} \, \CCof
     \RiB
     \, .
   \eea
The first integral in this expression is the nonvanishing contribution from the {\GR} part of the action, which is due to the extension of $\ugvar\pg^{\zm\zn}$ proportional to $\CCof$.
The second integral originates from metric variations of the matter part: the first term incorporates $\ugvar \Nl$, while the second accounts for phase-space fields variations. The third integral results from the canonical transformations of auxiliary variables, reducing to $\sint \ugvar \big(\tauxdot{\tf}^\io {+} \Dvg{\lmptf}\big) \, \CCof$, given that $\ugvar{\CCof} \teq 0$. The underlined terms with time derivatives proportional to $\gxl$ cancel.

All terms on the right-hand side of (\ref{ugvar_paramAction_GUMG''}) are linear in $\CCof$, as expected\footnote{
 Pure {\GR} variation\HIDE{ in the canonical constraint basis of the metric fields} of the {\GR} part of the action vanishes due to the parental gauge symmetry, while the extension of\HIDE{ the momenta canonical transformations} $\ugvar \pg^{\zn\zm}$ beyond canonical $\GR$ transformations is proportional to $\CCof$.  The\HIDE{ unrestricted} variation of the matter part\HIDE{ of the action} is proportional to $\CCof$ because its integrand is proportional to $\CCof$, and the variation of $\CCof$ itself is trivial\HIDE{, $\ugvar \CCof=0$}.
}.
The role of $\sint \ugvar{\Dvg{\lmptf}} \CCof$ is to cancel other nonvanishing terms proportional to $\CCof$. If $\Dvg{\lmptf}$ were a functionally complete object\HIDE{ (unconstrained yet arbitrary function)}, its variation could fully compensate for all other terms\HIDE{ in arbitrary off-shell configurations of fields}, preventing any anomaly. However, due to the average-free nature of  $\Dvg{\lmptf}$, only contributions proportional to average-free component of $\CCof$ can be compensated. Consequently, for
  \beq{gvar_Dvg_lmptf_GUMG}
   %\left|
    \begin{array}{|lll}
    \ugvar{\Dvg{\lmptf}}
   &\!=\!& \tTDer{}{\taux}\Dvg{{\gchi}} +
   \Big(\!
  % \PB{\sint (\repNl \Hl {+} \repNs^\zn \Hs_\zn )}{ \HIDE{\sint}\inh{(\gxl \FF^{-1}\opinvE)}}
  %  +
     \inh{( \gxs^\zn {\repNl}_{\!\!\!\!,\zn}\FF^{-1}\opinvE )}
  %   {-}
  %   \inh{(\repNs^\zn\gxl_{,\zn}   \FF^{-1} \opinvE )}
%   \nonumber\\
%    &&
    +
    \PB{ \inh{( \repNl \FF^{-1} \opinvE) } }{\sint (\gxl \Hl {+} \gxs^\zn \Hs_\zn)}
    \,-\,
    \big(\repNl\!,\repNs^\zn \leftrightarrow  \gxl\!, \gxs^\zn\big)
    \!\Big)
   ,
    \end{array}
    %\right.
  \eeq
the anomaly $\!A$ of the unrestricted\HIDE{ canonical} transformations, given by
  $ \ugvar \SSS_{par''}  
    \!\!\!\teq\!
   \int \!d\taux\, A
  $,
can be expressed as
  \beq{anomaly_GUMG_init}
    \quad
    - A
    \;=\;
  %%  \big(
  %%{-\,}
   \sint
      %(
    \gxs^\zn {\repNl}_{\!\!\!\!,\zn}
   %  {\,-\,}
   %  \sint\repNs^\zn\gxl_{,\zn}
     %)
     \FF^{-1}\opinvE \hmgs{\CCof}
  % \nonumber\\
  %  && \quad
    +  %%  -
     \PB{\sint \repNl \FF^{-1}\opinvE \hmgs{\CCof}\,}{\sint (\gxl \Hl {\,+\,} \gxs^\zn \Hs_\zn)}
 %%  \big)
    \,-\,
    \big(\repNl\!,\repNs^\zn \leftrightarrow  \gxl\!, \gxs^\zn\big)
    \,.
    \quad
   \eeq
\if{ %% COMMENT
  %  Since $\opinvE$ inside the spatial integrals act to the right on the spatially-homogeneous factor $\hmgs{\CCof}(\taux)$, one may simplify: $\opinvE \hmgs{\CCof}= (\WW^{-1}\!/\hmg{\WW^{-1}\!}\,)\hmgs{\CCof}$, (\ref{opinvE_properties}).
  %  Anyway, the spatially homogeneous part of the cosmological-constant field, $\hmgs{\CCof}(\taux)$, enters the anomaly as the overall time-dependent factor.
}\fi
Evaluating the Poisson brackets in (\ref{anomaly_GUMG_init}) requires careful treatment of spatial nonlocality in the operator $\opinvE$. After processing the Poisson brackets, grouping similar terms, and using the properties (\ref{opinvE_properties}\,--\,\ref{opG_properties2_kernels}) of\HIDE{ the operator} $\opinvE$ (we relegate all technicalities to Appendix \ref{ASSect:gauge_calcs_GUMG}) the anomaly simplifies to
  \beq{anomaly_GUMG_fin}
   \quad
   A
    \;=\;
    \sint  \gxl \FF^{-1}\hspace{-1pt} \opGm
    \big( \Dvg{\repNs} \,\OOmega
     - \repNl \FF^{-1} \TTheta \trpg \big)
    \frac{\hmgs{\CCof}}{\,\hmg{\WW^{-1}}\,}
    -
    \sint \repNl \FF^{-1}\hspace{-1pt} \opGm
    \big(\Dvg{\gxs} \,\OOmega - \gxl \FF^{-1} \TTheta \trpg \big)
      %%\underbrace{ \WW\opinvE \hmg{\ptf_{\io}} }_{ =
    \frac{\hmgs{\CCof}}{\,\hmg{\WW^{-1}}\,}
      %%}
%\nonumber\\
    \,,
    \quad
   \eeq
where $\OOmega(\argsqrg)$ and $\TTheta(\argsqrg)$ are the {\GUMG} derivative characteristic functions (\ref{def_OOmega_TTheta}), and the
operator $\opGm \defeq \opinvE \WW^{-1} \!\inhId$ is defined by the two-point integral kernel
  \beq{}
   G_{\scriptscriptstyle{-}}(\taux|\sx,\sx')
   \;=\;
   \WW^{-1}(\taux,\sx) \delta(\sx,\sx') - \WW^{-1}(\taux,\sx) \frac1{\sint \WW^{-1} (\taux) } \WW^{-1}(\taux,\sx')
   \,,
  \eeq
with $\WW(\taux,\sx) \defeq \WW(\sqrg {\scalebox{0.8}{\mbox{\ensuremath{(\taux,\sx)}}}})$.
$\opGm$ is one of the operators that, according to Lemma \ref{Lemma:invE_variations}, appear in the variation of $\opinvE$\,: \,$\var \opinvE \teq\, {-\,} \opinvE \! \var\WW^{-1} \opGp  -\,  \opGm \var \WW \opinvE$\HIDE{, (\ref{Lemma:var_invE_statement})}. Both $\opGm$ and $\opGp$ are symmetric and one-mode degenerate: their left and right kernels are (spatially) constant functions $\hmg{f}$ (see Corollary \ref{Corollary:var_invE_G_properties}). Since all $\opinvE$ in (\ref{anomaly_GUMG_init}) act to the right on the spatially constant function $\hmgs{\CCof}$, the second projector form, $\opGp$, disappears from the result, (\ref{Corollary:invE_1_variations_statement}).

\if{ %% v. 2025-02
%% - - -
Note that all terms are linear in $\ptf_\io \tequiv \CCof$, which is expected since pure parental transformations of metric fields of \GR\ action cancelled due to parental gauge symmetry. Also note that if we pretend that we do not know the transformation of $\gvar \ptf_\io$ and add correspondent variation terms $\HIDE{\int d\taux\,} \HIDE{\dsx \,}
\sint (\tauxdot{\tf}^\io {+} \Dvg{\lmptf}) \, \gvar \ptf_\io$ one immediately concludes that $\gvar \ptf_\io$ should vanish, since no other terms (\ref{ugvar_paramAction_GUMG''}) contain neither $\tauxdot{\tf}^\io$ nor $\Dvg{\lmptf}$.
Now lets examine what terms could be compensated by the term $\HIDE{\int d\taux\,} \HIDE{\dsx \,}     \sint (\gvar \tauxdot{\tf}^\io) \ptf_\io$ (pretending that we do not know $\gvar{\tf^\io}$ from canonical structure) and reinstate $\gvar{\tf^\io}$ from time locality considerations. The very first terms in the first and the second line fulfill the total time derivative as a factor at $\ptf_\io$.
There are no other terms with time derivative in the action. The only exception is yet underfined $\gvar{\Dvg{\lmptf}}$ which multiplies only $\ptf_\io$ and thus may generate total derivative term $\Dvg{\tauxdot{\gchi}}$, which should also be compensated by  $\gvar{\tauxdot{\tf}^\io}$. Consequently $\gvar{\tf^\io}$ acquires the form (\ref{ugvar_canon_GUMG})\footnote{The complete gauge transformation of $\gvar{\tf^\io}$ besides $   \gxl  \FF^{-1}\opinvE - \Dvg{\gchi}$ from (\ref{ugvar_canon_GUMG}) also contains shift with arbitrary time-independent parameter, which is not fully covered by $\Dvg{\gchi}$ due to the functional incompleteness of the latter.} which we expected from the canonical structure of the action.

All other terms from the first two lines of (\ref{ugvar_paramAction_GUMG''}) should be compensated by $\gvar{\Dvg{\lmptf}}$. If there were no assumption of the compact topology without boundaries for spatial sections of spacetime and corresponding  boundary conditions, which make a vector divergence average-free, then it would be possible to compensate all yet noncancelled terms in (\ref{ugvar_paramAction_GUMG''}) by appropriate (spatially nonlocal) gauge variation of $\Dvg{\lmptf}$. However, due to average-free property the variation of the divergence it can compensate only terms, multiplied by spatially inhomogeneous part of  $\ptf_\io$. Lengthy explicit expression is not important for us, so we leave  $\gvar{\Dvg{\lmptf}}$ in the following\HIDE{ implicit} form\footnote{The right-hang side of (\ref{gvar_Dvg_lmptf_GUMG}) is nothing more than variation of the terms (\ref{ugvar_paramAction_GUMG''}), left after using $\gvar{\tf^\io}$, (\ref{ugvar_canon_GUMG}). Or, equivalently, the factor in the variation integral (\ref{ugvar_paramAction_GUMG''}) (after using $\gvar{\tf^\io}$) at independent equation-of-motion factor $\inh{\ptf_\io}=0$.}
  \bea{gvar_Dvg_lmptf_GUMG}
   \Big| \quad \gvar{\Dvg{\lmptf}}
   &=& \Dvg{\gchi} + %\sVDer{}{\inh{\ptf_\io}}
   \PB{\sint (\repNl \Hl {+} \repNs^\zn \Hs_\zn )}{ \HIDE{\sint}\inh{(\gxl \FF^{-1}\opinvE)} \HIDE{\ptf_{\io}}}
    +
    %\stVDer{}{\inh{\ptf_\io}}
    \HIDE{\sint}
     \inh{( \gxs^\zn \repNl_{\!\!\!\!,\zn}\FF^{-1}\opinvE )} \HIDE{\ptf_{\io}}
     {-}
     \inh{(\repNs^\zn\gxl_{,\zn}
     \FF^{-1}\opinvE )} \HIDE{\ptf_{\io}}
   \nonumber\\
    &&
    + %\sVDer{}{\inh{\ptf_\io}}
    \PB{\HIDE{\sint} \inh{( \repNl \FF^{-1}\opinvE) } \HIDE{\ptf_{\io}}}{\sint (\gxl \Hl {+} \gxs^\zn \Hs_\zn)}
    \big)
   \,,
  \eea
which defines gauge transformation of $\gvar{\lmptf^\zm}$ up to arbitrary transverse spatial vector parameter. Operators $\opinvE$ act on functions, which stand from the left.\HIDE{ The third and the fourth term in the right-hand side are made average-free. by correspondent projector.}

Finally, the gauge anomaly
  $ \gvarGR \SSS_{\ialt_0}
    =
   % 0 + \gvarGR \SSS_{mat_0}
   % \equiv
     \int \!d\taux\, A
  $
for the unrestricted gauge transformations (\ref{ugvar_canon_GUMG}), (\ref{ugvar_NsNl_GUMG}), (\ref{gvar_Dvg_lmptf_GUMG}) reads
  \bea{anomaly_GUMG_init}
    A
    &=&
   %%  \int d\taux\, \HIDE{\dsx  }
   %%  \LiB \,
      \PB{\sint (\repNl \Hl+\repNs^\zn\Hs_\zn)\,}{\sint \gxl \FF^{-1}\opinvE \hmg{\ptf_{\io}}}
   %\nonumber\\
   % &&
    + \sint ( \gxs^\zn \repNl_{\!\!\!\!,\zn} - \repNs^\zn\gxl_{,\zn} )  \FF^{-1}\opinvE \hmg{\ptf_{\io}}
   \nonumber\\
    &&
    \quad +\, \PB{\sint \repNl \FF^{-1}\opinvE \hmg{\ptf_{\io}}\,}{\sint (\gxl \Hl+\gxs^\zn \Hs_\zn)}
   %% \,\RiB
    \,,
   \eea
where only spatially homogeneous part of the cosmological constant invariant, $\hmg{\ptf_{\io}}(\taux)$, enters the anomaly as the overall time-dependent factor. Calculation of the Poisson brackets in (\ref{anomaly_GUMG_init}) implies careful handling of spatial nonlocality in operator $\opinvE$. Calculational details we took out to Appendix \ref{ASSect:gauge_calcs_GUMG}. The result of
processing the Poisson brackets, grouping similar terms and using properties (\ref{opinvE_properties}--\ref{opG_properties2_kernels}) of\HIDE{ the inverse operator} $\opinvE$, (\ref{invE}), reduces the anomaly to:
  \bea{anomaly_GUMG_fin}
    A
    &\!=\!&\!
    \sint  \gxl \FF^{-1} \opGm %\HIDE{\WW^{-1} \tFProj{\WW^{-1}}{1}}
    \,
    \big( \Dvg{\repNs} \,\OOmega
     - \repNl \FF^{-1} \TTheta \trpg \big)
    \frac{\hmg{\ptf_{\io}}}{\,\hmg{\WW^{-1}}\,}
    -
    \sint \repNl \FF^{-1} \opGm % \HIDE{\WW^{-1} \,\tFProj{\WW^{-1}}{1}}
    \,
    \big(\Dvg{\gxs} \,\OOmega - \gxl \FF^{-1} \TTheta \trpg \big)
      %%\underbrace{ \WW\opinvE \hmg{\ptf_{\io}} }_{ =
    \frac{\hmg{\ptf_{\io}}}{\,\hmg{\WW^{-1}}\,}
      %%}
%\nonumber\\
    , \quad\;\;
   \eea
 %%$\WW^{-1} \,\tFProj{\WW^{-1}}{1}$
where $\OOmega(\argsqrg)$ and $\TTheta(\argsqrg)$ are the {\GUMG} characteristic functions (\ref{def_OOmega_TTheta}), and the
operator $\opGm \equiv \opinvE \WW^{-1} \inhId$ with the two-point integral kernel
  \beq{}
   G_{\scriptscriptstyle{-}}(\taux|\sx,\sx') = \WW^{-1}(\taux,\sx) \delta(\sx,\sx') - \WW^{-1}(\taux,\sx) \frac1{\sint \WW^{-1} (\taux) } \WW^{-1}(\taux,\sx')
   \,,
  \eeq
where $\WW(\taux,\sx) \tequiv \WW(\sqrg {\scalebox{0.8}{\mbox{\ensuremath{(\taux,\sx)}}}})$,
is one of operators, whose properties were noted in Corollaries \ref{Corollary:var_invE_G_properties} and \ref{Corollary:var_invE_G_equations}. This symmetric operator is one-mode degenerate. Its left and right kernels are constant functions, so that when acting on function to the right, $\opGm \,\const = 0$, as well as the right action of operator implies $ \const \, \opGm = 0$. In Corollary \ref{Corollary:var_invE_G_equations} it was proved that equation $f\opGm = 0$ is equivalent to $f \inhId = 0$ and imply $\inh{f}=0$. Due to symmetry of the operator, the same is true for the equation $\opGm f = 0$, which also implies $\inh{f}=0$.

%%- - -
}\fi

Corollary \ref{Corollary:var_invE_G_equations} shows that the equation $f\,\opGm \teq 0$ is equivalent to $f \inhId \teq 0$, implying $\inh{f} \teq 0$. Owing to the symmetry of the operator, the equation $\opGm f \teq 0$  also implies $\inh{f} \teq 0$. Thus, the requirement for the anomaly (\ref{anomaly_GUMG_fin}) to vanish imposes the following constraints on\HIDE{ transformation} parameters:
 \beq{gauge_restriction_GUMG}
  %\boxed{
  \begin{array}{|ll|ll}
   %\overrightarrow
   \gxl \FF^{-1} %\overleftarrow
    {\opGm}
    = 0
    \,,
   & \quad\;\Rightarrow\; \quad &
    \;\gxl = \hmg{\geps}(\taux) \FF
    \,,
    \\
   {\opGm}
    \big( \Dvg{\gxs} \,\OOmega
     - \gxl \FF^{-1} \TTheta \trpg \big)
    = 0
    \,,
   & \quad\;\Rightarrow\; \quad &
   \;\Dvg{\gxs}
    = \hmg{\geps}(\taux)
      \big( \OOmega^{-1}\TTheta \trpg
      - \OOmega^{-1}\tfrac{\sint \OOmega^{-1}\TTheta \trpg}{\sint \OOmega^{-1}}
      \big)
   % \equiv \fProj_{(1,\OOmega^{-1})} \OOmega^{-1}\TTheta \trpg
   % = \OOmega^{-1} \fProj_{(\OOmega^{-1},1)} \TTheta \trpg
   \,.
   \\
  \end{array}
  %}
 \eeq
\HIDE{Arrows above the operators indicate the direction of operator's action on functions.}

These are two \emph{independent} constraints on gauge parameters because, in the unrestricted canonical variation of the action, they multiply \emph{independent} equations of motion\footnote{
 In DeWitt's compact notation, the anomaly structure for any field variation $\ugvar[\gxi] \phi^\zi \teq r^\zi_{\alpha} \gxi^\alpha$ in the action is given by
 $
  \int dt\, A
  \defeq  \ugvar[\xi] S[\phi^\zi]
 % \teq  \VDer{S}{\phi^\zi} \ugvar[\gxi] \phi^\zi
 % \teq  \VDer{S}{\phi^\zi} r^\zi_{\,\alpha} \gxi^\alpha
  \teq  \SSS_{,\zi} r^\zi_{\,\alpha} \xi^\alpha
 $.
 Thus, the anomaly form
  % $\VDer{S}{\phi^\zi} r^\zi_{\,\alpha}$
 $\SSS_{,\zi} r^\zi_{\,\alpha}$, which acts on transformation parameters $\xi^\alpha$, is proportional to the left-hand sides of the equations of motion, $\SSS_{,\zi}$, which may have linear dependencies when a gauge symmetry is present in the theory. Nontrivial anomaly-free conditions on $\xi^\alpha$ determine the\HIDE{ number of} independent directions in the space spanned by $ r^\zi_{\,\alpha}$ that are not tangent to the extremal surface defined by $\SSS_{,\zi} \teq 0$ (i.e., that are not\HIDE{ true} gauge directions). The number of these directions is\HIDE{ formally} given by the rank of
  % $\frac{\var S}{\var \phi^\zj \var \phi^\zi} r^\zi_{\,\alpha} \big|_{\VDer{S}{\phi^\zi}=0}$.
 $\SSS_{,\zj\zi} r^\zi_{\,\alpha} \big|_{\SSS_{,\zi}=0}$.
 Dually, this rank corresponds to the number of independent normal directions to the extremal surface (i.e., on-shell independent rows of $\SSS_{,\zj\zi}$) that are transverse to the vectors $r^\zi_{\,\alpha}$. These independent normal directions correspond to the independent on-shell conditions multiplying $\xi^\alpha$ in the anomaly.
}.
\if{
  \footnote{
   The general structure of anomaly for any variation of the fields $\ugvar[\gxi] \phi^\zi \teq r^\zi_{\,\alpha} \gxi^\alpha$ in the action is
   $
    \int dt A
    \defeq  \ugvar[\gxi] S[\phi^\zi]
    \teq  \VDer{S}{\phi^\zi} \ugvar[\gxi] \phi^\zi
    \teq  \VDer{S}{\phi^\zi} r^\zi_{\,\alpha} \gxi^\alpha
   $.
   Thus the coefficients of the gauge-restriction operator imposing restrictions on gauge parameters are always proportional to equations of motion.
   }
}\fi
The second term on the right-hand side of (\ref{anomaly_GUMG_fin}), proportional to $\repNl \FF^{-1}\hspace{-1pt} \opGm $, vanishes under the consistency condition $\inh{\repNl  \FF^{-1} } {\,\eomeq\,} 0$, (\ref{repNl_repNs_onshell}).
The first anomalous term in (\ref{anomaly_GUMG_fin}) is proportional to the  $\opGm \big(\Dvg{\repNs} \,\OOmega   {\,-\,} \repNl \FF^{-1} \TTheta \trpg \big)\HIDE{ \teq 0}$, which vanish on shell due to\HIDE{ the equation of motion} $\Dvg{\repNs} {\,\eomeq\,} \repNl \FF^{-1} \inh{U}_0$, where $\inh{U}_0(\g,\pg)$\HIDE{ first} appeared in (\ref{DvgNs_from_CS_GUMG}). The latter equation, besides $\inh{\repNl \FF^{-1} } {\,\eomeq\,} 0$, relies on the other consistency condition $\Dvg{\repNs} {\,\eomeq\,} \tauxdot{\hmg{\tf}}^\io \,\inh{U}_0 \HIDE{\Dvgi{U_0}}$, (\ref{repNl_repNs_onshell}). Thus, on-shell vanishing factors in the integrands of (\ref{anomaly_GUMG_fin}) encode two independent consistency conditions, constraining the pair of the Lagrange multipliers associated with the second-class constraints on the non-GR branch of the theory.

The degeneracy of operator $\opGm$ in the first equation of (\ref{gauge_restriction_GUMG}) ensures the emergence of a homogeneous gauge parameter $\hmg{\geps}(\taux)$, which parameterizes the residual gauge freedom
associated with homogeneous time reparametrization. In the second equation, the one-mode degeneracy of $\opGm$ guarantees the existence of a solution for the\HIDE{ one-mode incomplete} transformation parameter $\Dvg{\gxs}$. At the same time, since only $\Dvg{\gxs}$ is constrained, the transverse part of $\gxs^\zn$, associated with the spatial volume-preserving diffeomorphisms, remains arbitrary. Finally, there are no constraints on the symmetry parameterized by $\gchi^\zm$, which is internal to the auxiliary sector and does not affect the metric-sector fields.

Thus, in its alternative\HIDE{ parameterized} formulation, the {\GUMG} model possesses a gauge symmetry parameterized by the canonical transformations (\ref{ugvar_canon_GUMG}, \ref{ugvar_NsNl_GUMG}, \ref{gvar_Dvg_lmptf_GUMG}), subject to restrictions:
  \beq{gauge_param_restrict_GUMG}
    \begin{array}{|ll}
     \gxl = \FF \, \hmg{\geps} \,,
     \qquad
     &\text{$\hmg{\geps}(\taux)$ --- arbitrary homogeneous: $\hmg{\geps} \teq \hmg{\geps}(\taux)$;}
     \\
     \gxs^\zn = \hmg{\geps} \,{U^\zn_0} + \gzetat^\zn,
     \qquad
     &\text{$\gzetat^\zn$ --- arbitrary transverse: $\,\Dvg{\gzetat} \teq 0 $\,;}
     \\
     \gchi^\zn\,,
     \qquad
     &\text{$\gchi^\zn$ --- arbitrary (internal to the auxiliary sector),}
    \end{array}
  \eeq
where ${U^\zn_0}(\g,\pg)$ is defined in (\ref{lm_solutions_nonGR}).
  %%The homogeneity of the free parameter $\hmg{\gxs}^\zn$ reflects the homogeneity of time reparametrization in the metric sector, while the transversality of the free spatial diffeomorphism parameter $\gzetat^\zn$ indicates that the allowed gauge transformations are restricted to volume-preserving spatial reparametrizations. Additionally,
As a check,\HIDE{ let us note that} the appearance of $\hmg{\geps} \,{U^\zn_0}$ in the constraint on $\gxs^\zn$ confirms the conjectured form of the homogeneous first-class Hamiltonian constraint in (\ref{I-class_constr}).

\newpar

The first two gauge parameters, $\hmg{\geps}$ and $\gzetat^\zn$, in (\ref{gauge_param_restrict_GUMG}) fully parameterize the gauge freedom of the metric fields, which turns out to be the residual gauge symmetry inherited from the parental theory --- the Einstein general relativity. Corresponding gauge restriction conditions on the left side of (\ref{gauge_param_restrict_GUMG}) effectively restrict the parameters $\gxl$, $\gxs^\zn$ of the parental gauge symmetry (\ref{GR_canon_gauge_transfs}) --- spacetime diffeomorphisms in the canonical gauge basis\footnote{
 This is strictly correct for the gauge transformations of the Lagrangian metric fields. The modified transformation law for $\pg^{\zm\zn}$ can be matched to a parental counterpart by additional trivial gauge transformations\HIDE{ trivial with respect to both theories}.
}.
Hypothetical\HIDE{ noncanonical} modifications of the transformations for the auxiliary sector fields $\CCof$ and $\lmptf^\zn$ obviously cannot eliminate the anomaly\HIDE{ or qualitatively alter the imposed restrictions}. Similarly, local-in-time modifications of $\ugvar\tf^\io$ also can not compensate it. As previously observed in Section \ref{SSect:Gauge_Inv_Check_AMG} for {\wGUMG} models, the anomalous contributions can only be formally compensated by metric-dependent, time-nonlocal transformations of this field (in fact, only of its spatially averaged mode $\hmg{\tf}^{\io}$).
Thus, it was demonstrated that, within the parameterized framework, the residual gauge transformations in the metric sector on the non-GR branch are either canonical or unhealthy due to time nonlocality. Conversely, all canonically-allowed, time-local gauge transformations correspond to restricted gauge transformations of the parental theory. This supports the analogous conclusions drawn within the original formulation \cite{Barvinsky:2019agh}.

The gauge structure analysis above was performed for the\HIDE{ parameterized} canonical action $\SSS_{par''}$, (\ref{paramAction_GUMG_fin}).
This choice was made for convenience and to clarify the origin and meaning of the gauge restrictions, which are more transparent in canonical representations.
Of course, the same gauge-restriction procedure could be applied directly in the Lagrangian representation (\ref{ActionHTlike_GUMG_ccf0}), yielding the same result provided that the field-dependent coefficients of the constraints are tied to the restriction surface, $\pg^{\zm\zn} \teq {\sqrg}(\ecK^{\zm\zn} {\,-\,} \g^{\zm\zn} \trecK)$, (\ref{momenta_ecK_GR}). The latter affects only the function $U_0^\zm(\g,\pg)$, (\ref{lm_solutions_nonGR}), or its divergence $\inh{U}_0 \teq \Dvgi{U_0}$, (\ref{DvgNs_from_CS_GUMG}), which depend on the trace of the metric momenta, $\trpg$.

This concludes the analysis of the gauge structure in the parameterized alternative representation of {\GUMG}. The gauge symmetry was already spatially nonlocal in local {\wGUMG} models, where only the homogeneous part of time reparametrizations remained a symmetry, while spatial diffeomorphisms were restricted to their transverse components. In general models, explicit nonlocality in the action and equations of motion further severes the nonlocalities by feeling the spatial inhomogeneity of $\WW(\argsqrg)$. However, we showed that the overall gauge structure remains analogous to that of {\wGUMG}. This consistency supports the validity of our approach to handling such nonlocalities.

%%-------------------------=%        ******         %=-------------------------%%
%%-------------------------=%        ******         %=-------------------------%%

    \subsection{Note on the effect of nonlocality on quantum measure}
     \label{SSect:Quantum_measure}
      \hspace{\parindent}
To conclude, we briefly address the impact of $\opE$-type spatial nonlocalities on the quantum measure. These nonlocalities may enter through Faddeev-Popov determinants and Jacobians of the $\opE$-dependent field redefinitions beyond any canonical field set.
A correct path integral measure follows from the canonical Liouville measure. In the Lagrangian representation (\ref{ActionHTlike_GUMG_ccf0}), the non-ghost field sector retains the standard ultralocal {\GR} measure, including the DeWitt determinant factor generated by integrating out gravitational momenta\footnote{Additionally, an ultralocal Jacobian may appear when passing from ADM metric variables to the\HIDE{ actual} Lagrangian metric field representation. However, auxiliary sector fields remain in their\HIDE{ original} canonical representation, keeping their measure trivial.}.
However, in the representation (\ref{ActionHTlike_GUMG}), where the effective cosmological constant field $\CCf$ is introduced, the measure acquires an nonultralocal factor $\mathop{\mathrm{Det}} \opE$ due to the nonlocal field redefinition (\ref{ccf_GUMG}).
According to Lemma \ref{Lemma:DetE}, the determinant of $\opE$\HIDE{ acting} as a spatial operator with the two-point kernel ${E}(\taux|\sx;\sx')$, (\ref{opE_kernel}), is
  \beq{} % (\ref{sDet_opE})
    {\mathop{\mathrm{Det}}}{_{\sx}} \opE (\taux)
    \;=\;
    \hmg{\WW}(\taux)\,\hmg{\WW^{-1}}(\taux)
    \,\geq\, 1
    \,.
  \eeq
The full spacetime determinant of the operator\HIDE{ $\opE$}
with the kernel $\HIDE{{E}(\taux,\sx;\taux',\sx') \teq } {E}(\taux|\sx;\sx') \delta(\taux{-}\taux')$
follows as\footnote{
 Regularizing $\delta(0)$  as $(\int\! d \taux\, 1)^{-1} $ ensures that this determinant remains $\geq1$, with equality holding only when $\WW$ is spatially homogeneous. For example, for  $\WW \tequiv \wwc \teq \const$ \,or\, $\det \g_{\zm\zn} (\tx,\sx) = \hmg{\det \g_{\zm\zn}} (\taux)$.
}
  \beq{} % (\ref{Det_opE_t})
    {\mathop{\mathrm{Det}}}{_{\taux,\sx}} \opE
    \;=\;
    e^{ \delta(0) \int\! d\taux\, \ln \big( \hmg{\WW}\,\hmg{\WW^{-1}} \big) }
    %   \,\geq\,1
    \,.
  \eeq
Since $\opE$ is a purely spatial nonlocality of a non-differential type, it does not introduce derivative expansions\HIDE{ via the heat kernel method}. Furthermore, as $\opE$ arises solely from the spatial inhomogeneity of $\WW(\argsqrg)$, only a limited set of structures is expected to appear in its determinant.
This suggests that determinant contributions from such nonlocal structures remain relatively simple and manageable, at least for those arising from\HIDE{ nonderivative} field redefinitions.

Computing the Faddeev-Popov determinants may be a more involved problem due to the complexity of the gauge structure. However, a well-chosen gauge fermion could simplify or even eliminate $\opE$-dependent nonlocalities. In any case, this lies beyond the scope of this work\HIDE{, which is focused on classical properties of the alternative representation}.

%%-------------------------=%        ******         %=-------------------------%%
%%-------------------------=%        ******         %=-------------------------%%
%%-------------------------=%        ******         %=-------------------------%%

\section{Conclusions}
 \hspace{\parindent}
In this paper, we derived and analyzed the alternative representation of generalized unimodular gravity (\GUMG), constructed by generalizing the Henneaux–Teitelboim (HT) approach to unimodular gravity.
We addressed the problem of consistent parameterization and examined the key dynamical and gauge features, including the emergence of explicit spatial\HIDE{ $\opE$-type} nonlocality and its implications for the equations of motion and the gauge structure. We also confirmed that the gauge symmetry derived from canonical analysis matches the local-in-time part of the residual general-relativistic diffeomorphisms implied by the configuration-space restriction\HIDE{(\ref{GUMG_restriction})}.

The alternative formulation is physically equivalent to the original setup of \cite{Barvinsky:2017pmm,Barvinsky:2019agh}, inheriting its non-regular constraint structure with two dynamical branches, each featuring a different number of physical degrees of freedom. On the non-GR branch --- characterized here by a nonzero cosmological-constant field --- an additional degree of freedom emerges, behaving as a cosmological perfect fluid with an equation of state governed by the barotropic parameter $\WW(\argsqrg)$, in agreement with results from the original formulation.

A distinctive feature of the alternative representation is the presence of the {HT}-like auxiliary field entering the action through a total spacetime divergence, $\partial_\Zm\HTf^\Zm$. The latter dynamically selects a spacetime-constant observable $\cco$ --- the only free component of the cosmological-constant field in $\CCof$ parameterization (\ref{ActionHTlike_GUMG_ccf0}). Consequently, the on-shell behavior (\ref{ccf_onshell_GUMG}) of the effective cosmological constant $\CCf$ emerges as a spatially nonlocal function, sensitive to normalized spatial variations of $\WW(\argsqrg)$.
The constant $\cco$ parameterizes a direction along which the space of solutions\HIDE{ ( of equations of motion)} extends relative to general relativity (\GR).
 \if{
  \footnote{
    The extension reflects the fact that the {\GUMG} restriction is not merely a gauge fixing but also implies a physical constraint on a conjugate global observable, which is a gauge-invariant mode involving the Lagrange multiplier $\Nl$, and this physical constraint\HIDE{ cannot be confined to phase space} lies outside phase space.
   %This is analogously to what happens in the unimodular gravity \cite{Barvinsky:2022guw}.
  }.
 }\fi
In {\GR}, physical trajectories are confined to the constraint surface $\Hl \teq 0$, $\Hs_\zn \teq 0$. In contrast, the {\GUMG} constraint surface extends to:
    $ \Hl
      \teq {-}\FF^{-1}\opinvE \cco
      \,\teq {-}\FF^{-1} \frac{\WW^{-1}}{\,\hmg{\!\WW^{-1}\!}\,\vphantom{I^{|^I}}} \, \cco
    $,\,
   $\Hs_\zn \teq 0$\,,
where the first equation refines the secondary constraint $\CT_{,\zm} \teq 0$ using dynamical input. This explicit encoding of such an invariant in the action is advantageous within the restricted theory framework \cite{Barvinsky:2022guw}, which suggests an interrelation between the effective actions of the parental and restricted gauge theories.

Another notable feature of the alternative formulation is the manifest appearance of spatial nonlocality in the action (\ref{ActionHTlike_GUMG_ccf0}, \ref{ActionHTlike_GUMG}). Though arising in the course of consistent parameterization, this nonlocality is not an artifact of the formalism or an ill-chosen extension of configuration space but a genuine dynamical feature of {\GUMG} models with nonconstant $\WW$. It was implicitly present already in the original canonical formulation, where the spatial average $\hmg{\FF \Hl}$ defines a gauge-invariant constant of motion (taking on-shell equal to $- \cco$), while the inhomogeneous secondary constraint $\inh{\WW \FF \Hl}$ is also gauge-invariant and vanishes on shell. Hence, their linear combination $\opE \FF \Hl$ is a functionally complete (albeit spatially nonlocal) gauge-invariant integral of motion --- the property, which the alternative\HIDE{ parameterized} formulation just makes explicit.

The {\wGUMG} subfamily (\ref{AMG_part_case}), where the alternative parameterized formulation remains local, provides a useful particular case for understanding the dynamics and gauge structure of general {\GUMG}. Its equations of motion are local, while gauge transformations exhibit the simple spatial nonlocality: spatially homogeneous time reparameterizations and transverse spatial diffeomorphisms.
In the general case with spatially varying $\WW(\argsqrg)$, the additional $\opE$-type nonlocality does not qualitatively alter the theory's dynamics or gauge structure, making {\wGUMG} a full-fledged yet simpler setting for studying the {\GUMG} behavior.

The framework also includes unimodular gravity (\UMG) as the exceptional case (\ref{UMG_part_case}), reproducing the Henneaux–Teitelboim action with restored diffeomorphism invariance \cite{Henneaux:1989zc}.
Interestingly, the time reparameterization symmetry restores not due to local parameterization alone, but rather from the secondary constraint becoming first-class in {\UMG}. This constraint generates inhomogeneous time reparameterizations, while the missing homogeneous mode is effectively recovered via the homogeneous part of parameterization (the inhomogeneous auxiliary fields can be decoupled).
By contrast, in {\GUMG}, the secondary constraint $(\WW \FF \Hl)_{,\zm} \teq 0$ is second-class, and time reparameterization remains broken, regardless of the applicability of the local parameterization scheme. As seen in {\wGUMG} case, where the local parameterization is consistent, yet full diffeomorphism symmetry is not restored. The new constraint introduced by the parameterization does generate a homogeneous gauge symmetry in the metric sector governed by the parameter $\hmg{\geps}(\taux)$, while its inhomogeneous counterpart with $\Hl$ is fixed to zero by the secondary constraint and remains second-class, along with the longitudinal part of spatial diffeomorphisms.

\newpar

Our constructive derivation confirms the equivalence between the original and alternative formulations of generalized unimodular gravity, both of which admit generalizations to a broader class of restrictions\HIDE{ such as $\Nl = \FF(\g_{\zn\zm})$ without essential complications}. The analysis presented provides an overview of the dynamical features and gauge structure of the theory. Nevertheless, various issues are still are still open.
While we have justified the applicability of the parameterized formalism in the case of manifolds with noncompact spatial sections from general grounds and specific properties of the $\opE$-type nonlocalities, a complete treatment of the theory still requires further development, including a more precise specification of asymptotic behavior of the fields.
 %\footnote{See \cite{Barvinsky:2019qzx} for some hints in $\Ddim=4$.}. 
In the broader {\GUMG} framework, it would be valuable to explore cosmological solutions beyond the inflationary scenarios studied in \cite{Barvinsky:2019qzx}, including static configurations and other nontrivial spacetime geometries. 
Potential global limitations on the solution space which may be hidden in subtle properties of {\GUMG}-defining restriction also warrant further investigation.
In the context of Batalin-Vilkovisky quantization, it is of interest to clarify the relationship between the HT-like consistent parameterization scheme and the conversion-based approach, which is nontrivial even in the {\wGUMG} case. The homogeneous\HIDE{ time-reparametrization} constraint $\hmg{\replm}^\io (\hmg{\ptf_{\io} {+} \Fw\Hl})$ can be restored via the homogeneous parameterization by extending the phase space with $\hmg{\ptf}\io$, $\hmg{\tf}^\io$. Then, following the conversion method \cite{ConversionClassics}, one could attempt to convert the system of constraints with the two average-free second-class constraints: $\partial^\zn \hspace{-1pt}\Hs_\zn$ and $\inh{\wwc \Fw \Hl}$, into a first-class set by introducing average-free auxiliary canonical pair: $\inh{\tf}^\io$ and $\inh{\ptf}_\io$, and extending the constraints in powers of these conversion variables \cite{ConversionClassics,Batalin:2005df,Batalin:2018wxh}. This approach is more difficult, facing modification of local conversion scheme in view of a functional incompleteness of constraints, which probably would result in the spatially nonlocal series for the extended constraints.

We expect that the techniques developed here\HIDE{ for generalized unimodular gravity} will provide a useful basis for further studying of {\GUMG}-type theories\HIDE{  both at the classical and quantum levels}. In any case, the detailed description of their dynamical and gauge properties makes them valuable as nontrivial test cases for the restricted gauge theories framework\HIDE{, parameterized field theories} and general theory of constrained systems due to their functionally incomplete constraints, and potentially as modified gravity candidates for cosmological modeling.

%%%============================================================================%%%
%%%============================================================================%%%
%%%============================================================================%%%

%\newpage
\section*{Acknowledgements}
 \hspace{\parindent}
We wish to thank Andrei Barvinsky and Ksenia Lyamkina for stimulating discussions\HIDE{ on the early stages of this work}. The research was supported by the Russian Science Foundation grant No. \href{https://rscf.ru/en/project/23-12-00051/}{23-12-00051}.

%%-------------------------=%        ******         %=-------------------------%%
%%-------------------------=%        ******         %=-------------------------%%
%%-------------------------=%        ******         %=-------------------------%%
%\newpage
\appendix

%%------------------------=%        ******         %=------------------------%%
%%------------------------=%        ******         %=------------------------%%
%%------------------------=%        ******         %=------------------------%%

%\newpage
\section{Properties of Delocalization Operator}
   \label{ASect:opE_properties}
    %\subsection{Checks and properties}
    %\subsubsection{Derivation}
     \hspace{\parindent}
In this section, we use the bra-ket notation to represent the linear space of functions on {compact} $\tx \teq \const$ hypersurfaces of the spacetime manifold, their inner and outer products, and the linear operators acting on them. Bra- and ket-vectors denote local functions; for instance, $\Ket{f}$ corresponds to the local function $f(\sx)$, and $\Ket{1}$ represents the unit constant function ($f(\sx) \teq 1$). Inner product implies spatial integration, such that
  \beq{}
  \Bra{f}\op{A}\Ket{g} \quad \leftrightarrow \quad \sint f \op{A}\, g \,.\,
 \eeq
A local function appearing outside of bra- and ket- expressions denotes a diagonal operator, which acts on vector via pointwise multiplication. For example, $\Ket{\WW}$ corresponds to the local function $\WW(\sx) \defeq\WW\big(\!\sqrg(\sx)\big)$, while
$\WW$ outside the vector acts as a symmetric diagonal operator with the integral two-point kernel
 $\WW(\sx') \delta(\sx',\sx)
   %\teq \delta(\sx',\sx) \WW\big(\!\sqrg(\sx)\big)
 $,
so that $\WW\Ket{f}=\Ket{\WW f} \HIDE{$ corresponds to the function $} \; \leftrightarrow \; \WW(\sx) f(\sx)$.
The \emph{identity operator} $\Id$ corresponds to the operator with the two-point kernel $\delta(\sx',\sx)$ and satisfies $\Id\hspace{1pt} \Ket{f} \teq \Ket{f}$ for any continuous\HIDE{ function} $f(\sx)$.
We also define symmetric projectors $\hmgId$\, and \,$\inhId \teq \Id {\,-}\hmgId$ onto average\HIDE{ (homogeneous)} and average-free\HIDE{ (inhomogeneous)} components, (\ref{def_inh_hmg_sVol}), as follows:
 %%\,$\hmgId$\, and \,$\inhId\equiv\Id {\,-}\hmgId$\, are defined as
  \beq{def_hmgId_inhId}
    \hmgId \Ket{f} \,=\, \Ket{\hmg{f}}
    \,=\, \Ket{1} \HIDE{{\cdot}} \hmg{f}
    \,,
    \qquad
    \inhId \Ket{f} \,=\, \Ket{\inh{f}}
    \,=\, \Ket{f} - \Ket{1} \HIDE{{\cdot}} \hmg{f}
    \,.
  \eeq

 \subsection{Properties of delocalization operator and its inverse}
  \label{ASSect:opE_properties}
  \hspace{\parindent}
The operator $\opE$, with the two-point kernel (\ref{opE_kernel}), can be represented as
 \bea{def_opE_COPY}  %%(\ref{def_opE}),
  \opE  &\!=\!&  {\hmgId} + \WW^{-1} \inhId \WW
   % \nonumber\\
   % &=&
   \:=\: \Id + \Ket{1}\sVol^{-1\!}\Bra{1} - \Ket{\WW^{-1}}\sVol^{-1\!}\Bra{\WW}
   \,,
  %\nonumber
 \eea
where $\sVol\equiv \BraKet{1}{1}$ is the background spatial volume. Its featured actions are:
 %\underline{Basic properties}
\if{
 \beq{opE_properties}
  \begin{array}{|ll|l}
   \Bra{1} \WW \opE = \Bra{\WW} \opE = \BraKet{\WW}{1} \sVol^{-1} \Bra{1} ,
   &  &
   \opE \WW^{-1} \Ket{1} = \opE \Ket{\WW^{-1}} = \Ket{1}  \sVol^{-1} \BraKet{1}{\WW^{-1}}, \\
   \Bra{1} \opE =  \Bra{1} + \Bra{\WW^{-1}} \inhId \WW,
   &  &
   \opE \Ket{1} = \Ket{1}  + \WW^{-1}\inhId \Ket{\WW}. \\
  \end{array}
 \eeq
}\fi
 \bea{opE_properties}
  \begin{array}{lllll}
   \Bra{1} \opE &\!=\!&  \Bra{1} + \Bra{\WW^{-1}} \inhId \WW
   \,,
   \\ %& \quad &
   \opE \Ket{1} &\!=\!& \Ket{1}  + \WW^{-1}\inhId \Ket{\WW}
   \,,
  \end{array}
 % \\
 \qquad \qquad
 % \\
  \begin{array}{lllll}
   %\Bra{1} \WW \opE &=&
   \Bra{\WW} \opE   &\!\!=\!& \BraKet{\WW}{1} \sVol^{-1\!} \Bra{1}
    &\!\!=\!& \hmg{\WW} \Bra{1}
    \,,
    \\ % & \quad &
   %\opE \WW^{-1} \Ket{1} &=&
   \opE \Ket{\WW^{-1}}
   &\!\!=\!& \Ket{1}  \sVol^{-1\!} \BraKet{1}{\WW^{-1}}
   &\!\!=\!& \Ket{1} \hmg{\WW^{-1}} \vphantom{\big|^I}
   \,.
  \end{array}
 \eea

 % \underline{Homogeneity limit.}

For $\WW=\wwc \teq \const$, the operator $\opE$ reduces to the identity operator:
 \beq{EProp:homogeneity}
  \opE \to\Id \qquad
  \text{ for }  \;\WW\to\,\wwc = \const.
 \eeq
This also holds for\HIDE{ configurations of} the spatial metric field with a homogeneous volume density, $\sqrg \teq \hmgw{\sqrg}(\tx)$.
 %% $\WW(\sqrg(\tx))$
\newpar

 % \underline{Nondegenerateness of $\opE$.}
The operator $\opE$ is \emph{nondegenerate}, and their exists a unique inverse operator.

\begin{lemma}[Inverse of $\opE$] \label{Lemma:InvE}
 \!The operator $\opE$\! (\ref{def_opE_COPY})\HIDE{, defined on the space of functions over a manifold with finite coordinate volume,} is nondegenerate,\! and its inverse is given by
 \bea{invE}
   \opinvE    &\!=\!& \Id
         \;-\;  \Ket{\WW^{-1}} \tfrac{1}{\BraKet{1}{\WW^{-1}}} \Bra{1}
         \;-\;  \Ket{1} \tfrac{1}{\BraKet{\WW}{1}} \Bra{\WW}
         \;+\; 2 \, \Ket{\WW^{-1}} \tfrac{1}{\BraKet{1}{\WW^{-1}}} \sVol \tfrac{1}{\BraKet{\WW}{1}} \Bra{\WW}
         \,.
 \if{
   \nonumber\\
   &\equiv& \Id
         \;-\; \lFProj{\WW}{1}
         \;\;-\;\; \lFProj{\WW^{-1\!}}{1}
         \;\;+\;\; \lFProj{\WW^{-1\!}}{1} \;\;\; \lFProj{\WW}{1}
   \nonumber\\
   &=& \Id
         \;-\; \Big( \underbrace{\Ket{1} - \Ket{\WW^{-1}} \tfrac{1}{\BraKet{1}{\WW^{-1}}} \sVol}_{\tFProj{1}{\WW^{-1}}\Ket{1}} \Big)\tfrac{1}{\BraKet{\WW}{1}} \Bra{\WW}
         \;-\; \Ket{\WW^{-1}} \tfrac{1}{\BraKet{1}{\WW^{-1}}} \Big( \underbrace{\Bra{1} - \sVol \tfrac{1}{\BraKet{\WW}{1}} \Bra{\WW}}_{\Bra{1}\tFProj{\WW}{1}} \Big)
   \nonumber\\
   &\equiv& \Id
         \;-\; \tFProj{1}{\WW^{-1}} \;\;\; \lFProj{\WW}{1}
         \;-\; \lFProj{1}{\WW^{-1}} \;\;\; \tFProj{\WW}{1}
   }\fi
 \eea
\end{lemma}

\begin{proof}\footnote{We thank Andrei Barvinsky for suggesting this\HIDE{ elegant} constructive proof.}
 {\small
 We begin with the defining relation
 $ % \beq{}
  \opE \opinvE = \Id
 % \,.
 $ %\eeq
as an equation on $\opinvE$.
 %% Using the second representation for $\opE$ (\ref{def_opE_COPY}) one gets the equation on $\opinvE$
Using (\ref{def_opE_COPY}),
 \beq{invE1}
  \opinvE + \Ket{1} \tfrac{1}{\sVol} { \Bra{1}\opinvE} \HIDE{\Bra{X}} \;-\; \Ket{\WW^{-1}} \tfrac{1}{\sVol} {\Bra{\WW}\opinvE} \HIDE{\Bra{Y}}
  \;=\; \Id
  \,.
 \eeq
Let ${\Bra{X}} \defeq \Bra{1}\opinvE$ and ${\Bra{Y}} \defeq \Bra{\WW}\opinvE$. To determine an explicit form of $\opinvE$, we solve for $\Bra{X}$ and $\Bra{Y}$.
 %\\%
Contracting (\ref{invE1}) with $\Bra{\WW}$ gives\;
 $ % \beq{contr1_invE}
  { \Bra{Y} } + \BraKet{\WW}{1} \tfrac{1}{\sVol} \Bra{X} - { \sVol \tfrac{1}{\sVol} \Bra{Y} }
  \;=\; \Bra{\WW} ,
 $ % \eeq
from which we obtain\HIDE{ the resolution for $\Bra{X}$}
 \beq{BraX}
    \Bra{X}
  \;=\;  \sVol \tfrac{1}{\BraKet{\WW}{1}}\Bra{\WW}
  \,.
 \eeq
Contracting (\ref{invE1}) with $\Bra{1}$ yields
 $ % \beq{contr2_invE}
  \Bra{X}  + \sVol \tfrac{1}{\sVol} \Bra{X} -  \BraKet{1}{\WW^{-1}} \tfrac{1}{\sVol} \Bra{Y}
  \;=\; \Bra{1} ,
 $ % \eeq
leading to \HIDE{the resolution for $\Bra{Y}$}
 \beq{BraY}
   \Bra{Y}
  \;=\;  - \sVol \tfrac{1}{\BraKet{1}{\WW^{-1}}} \big( \Bra{1} - 2 \Bra{X} \big)
  \,.
 \eeq
Substituting (\ref{BraX}) and (\ref{BraY}) into (\ref{invE1}) gives{{ the explicit form of $\opinvE$ in}} (\ref{invE}).

The nondegeneracy of $\opE$\HIDE{ then} follows from the existence of its inverse\HIDE{{ operator}} $\opinvE$, satisfying $\opinvE \opE \teq \opE \opinvE \teq \Id$.
%%%, which was obtained by regular procedure of the explicit construction.
}
\end{proof}

The nondegeneracy of $\opE$ is also intuitively expected. In the limiting case $\WW \teq \wwc \teq \const$ and for metric configurations with homogeneous spatial volume density $\sqrg \teq \hmgw{\sqrg}(\tx)$, the operator\HIDE{ $\opE$} becomes the identity, $\opE \to\Id$, (\ref{EProp:homogeneity}), and is therefore\HIDE{ \nondegenerate} full-rank. By continuity of eigenvalues, nondegeneracy holds at least in a neighborhood of such configurations.
The constructive proof above confirms that $\opE$ is nondegenerate for any sign-definite, bounded function $\WW$.

\newpar

%\underline{Basic properties}
%\begin{corollary}[Basic contractions] \label{Corollary:var_invE_G_properties}

For the operator $\opE$ given by (\ref{def_opE_COPY}),
 % $ % \bea{def_opE_COPY2}
 %  \opE \teq \hmgId + \WW^{-1}\inhId \WW ,
 % $ % \eea
the basic actions of its inverse $\opinvE$ from (\ref{invE}) are:
  \bea{opinvE_properties} % (\ref{opinvE_properties})
    \begin{array}{lllll}
       \opinvE \Ket{1}
        &=&
     \if{
         \Ket{1}
         -  \Ket{1}
         -  \Ket{\WW^{-1}} \tfrac{1}{\BraKet{1}{\WW^{-1}}} \sVol
         + 2 \, \Ket{\WW^{-1}} \tfrac{1}{\BraKet{1}{\WW^{-1}}} \sVol
        \;=\; \Ket{\WW^{-1}} \tfrac{1}{\BraKet{\WW^{-1}}{1}} \sVol
       &=&
     }\fi
       \frac{\Ket{\WW^{-1}}}{\hmg{\WW^{-1}}\vphantom{|^{I^I}}}
       \,,
     % \nonumber
     \\
      \Bra{1} \opinvE        &=& \Bra{1}
     \if{
         \;-\;  \sVol \tfrac{1}{\BraKet{\WW}{1}} \Bra{\WW}
         \;-\;  \Bra{1}
         \;+\; 2 \, \sVol \tfrac{1}{\BraKet{\WW}{1}} \Bra{\WW}
        \;=\; \tfrac{1}{\BraKet{\WW}{1}} \sVol \Bra{\WW}
       &=&
     }\fi
       \frac{ \Bra{\WW}}{\hmg{\WW}\vphantom{|^{I^I}}}
       \,,
      %\nonumber
    \end{array}
   %   \\
   \qquad\qquad
   %   \\
    \begin{array}{lllll}
    \opinvE  \Ket{\WW^{-1}}\!\!
      &=&
     \if{
         \Ket{\WW^{-1}}
         \;-\;  \Ket{1} \tfrac{1}{\BraKet{\WW}{1}} \sVol
         \;-\;  \Ket{\WW^{-1}}
         \;+\; 2 \, \Ket{\WW^{-1}} \tfrac{1}{\BraKet{1}{\WW^{-1}}} \sVol \tfrac{1}{\BraKet{\WW}{1}} \sVol
       &\hspace{-15mm}=\hspace{15mm}& \hspace{-15mm}
     }\fi
       \big({-} \Ket{1} + 2 \frac{\Ket{\WW^{-1}}}{\hmg{\WW^{-1}}\vphantom{|^{I^I}}}  \big) \frac{1}{\,\hmg{\WW}\,\vphantom{|^{I^I}}}
       \,,
     % \nonumber
     \\
      \Bra{\WW} \opinvE       &=&
     \if{
      \Bra{\WW}
         \;-\;  \Bra{\WW}
         \;-\;  \sVol \tfrac{1}{\BraKet{1}{\WW^{-1}}} \Bra{1}
         \;+\; 2 \, \sVol \tfrac{1}{\BraKet{1}{\WW^{-1}}} \sVol \tfrac{1}{\BraKet{\WW}{1}} \Bra{\WW}
       &\hspace{-15mm}=\hspace{15mm}& \hspace{-15mm}
     }\fi
         \frac{1}{\,\hmg{\WW^{-1}}\,\vphantom{|^{I^I}}} \big({-} \Bra{1} + 2 \frac{\Bra{\WW}}{\hmg{\WW}\vphantom{|^{I^I}}}  \big) \vphantom{\big|^{I^I}}
         \,.
      %\nonumber
    \end{array}
  \eea
% \end{corollary}

 \subsubsection*{Projectors $\opG$ and their properties}
  %\label{ASSect:opG_properties}
   \hspace{\parindent}
In the dynamic and gauge structures of {\GUMG}, where $\opinvE$ is varied, two symmetric projectors arise, $\opGm$ and $\opGp$, whose basic properties we prove below. %%built of $\opinvE$, (\ref{invE}).

  \begin{corollary}[Projectors $\opG$] \label{Corollary:var_invE_G_properties}  %%\label{Corollary:inhIWE-1f}
  Define two operators:
  \bea{def_opGmp} %{opinvE_properties2a}
    \begin{array}{lllllll}
      \opGp
      \!&\!\defeq\!&
      \inhId \WW \opinvE       &=&
      \inhId \big(\, \WW \HIDE{\Id} - \Ket{\WW} \tfrac1{\BraKet{1}{\WW}} \Bra{\WW} \,\big)
    %  \!&=& \!
      %%  \underbrace{
    %     \WW - \Ket{\WW} \tfrac1{\BraKet{1}{\WW}} \Bra{\WW}
      %%   }_{=\; \tFProj{1}{\WW} \WW  \;=\; \WW \tFProj{\WW}{1}}
    %  &=&\!
    %  \big(\,\WW - \Ket{\WW} \tfrac1{\BraKet{1}{\WW}} \Bra{\WW}\,\big) \inhId
    %  \;=\;
    %   \WW \tFProj{\WW}{1} \inhId
       \,,
      \\
      \opGm
      \!&\!\defeq\!&
      \opinvE \WW^{-1}\! \inhId \!
      &=&\!\!
      \big( \HIDE{\Id} \WW^{-1}\! - \Ket{\WW^{-1}} \tfrac1{\BraKet{1}{\WW^{-1}}} \Bra{\WW^{-1}} \big)\inhId
    %  \!&=& \!\!
      %%  \underbrace{
    %     \WW^{-1}\! - \Ket{\WW^{-1}} \tfrac1{\BraKet{1}{\WW^{-1}}} \Bra{\WW^{-1}}
      %%   }_{=\; \tFProj{1}{\WW^{-1}} \WW^{-1}  \;=\; \WW^{-1} \tFProj{\WW^{-1}}{1}}
    %  \\
    %  &=&
    %  \inhId \big(\,\WW^{-1}\! - \Ket{\WW^{-1}} \tfrac1{\BraKet{1}{\WW^{-1}}} \Bra{\WW^{-1}}\,\big)
    %  \;=\;
    %   \inhId \tFProj{1}{\WW^{-1}} \WW^{-1}
       \,.
      \\
    \end{array}
   \eea
%%where on the right-hand side we expanded\HIDE{ operators} $\opinvE$, (\ref{invE}) and used \; $\inhId {\WW} \Ket{\WW^{-1}} = \Bra{\WW} {\WW^{-1}} \inhId = 0$.

%\newcommand{\opd}[1]{\op{#1}} %% diagonal ultralocal operator
%\newcommand{\dop}[1]{\op{#1}} %% diagonal utralocal operator

The operators $\opGm$ and $\opGp$ are \emph{symmetric} and \emph{degenerate} operators
  \bea{opG_properties1_symmetry}
    \begin{array}{lllllll}
      \opGp
  %    &=&
  %    \inhId \big(\, \WW - \Ket{\WW} \tfrac1{\BraKet{1}{\WW}} \Bra{\WW} \,\big)
      \!&\!=\!&
         {\WW}  - \Ket{\WW} \tfrac1{\BraKet{1}{\WW}} \Bra{\WW}
       \,,
      \\
      \opGm
  %    &=&\!\!
  %    \big( \WW^{-1}\! - \Ket{\WW^{-1}} \tfrac1{\BraKet{1}{\WW^{-1}}} \Bra{\WW^{-1}} \big)\inhId
      \!&\!=\!&\!
         {\WW^{-1}}\! - \Ket{\WW^{-1}} \tfrac1{\BraKet{1}{\WW^{-1}}} \Bra{\WW^{-1}}
       \,,
      \\
    \end{array}
   \eea
with left and right kernels consisting of constant functions, so that
  \bea{opG_properties2_kernels}%{opinvE_properties2b}
    \begin{array}{lllllll}
      \opGp \Ket{1}
      = \Bra{1} \opGp
      = 0
  %    \,,
  %  \qquad
  %    \opGp
  %    = \inhId \opGp
  %    = \opGp \inhId
       \,;
      \\
      \opGm \Ket{1}
      = \Bra{1} \opGm
      = 0
  %    \,,
  %  \qquad
  %    \opGm
  %    = \inhId \opGm
  %    = \opGm \inhId
       \,,
      \\
    \end{array}
   \qquad
   \qquad
    \begin{array}{lll}
      \opGp
      = \inhId \opGp
      = \opGp \inhId
       \,;
      \\
      \opGm
      = \inhId \opGm
      = \opGm \inhId
       \,.
      \\
    \end{array}
   \eea

\if{
  The following representations are also true
  \bea{opinvE_properties2b}
    \begin{array}{lllll}
      \inhId \WW \opinvE       &=&\!
      \big(\,\WW - \Ket{\WW} \tfrac1{\BraKet{1}{\WW}} \Bra{\WW}\,\big) \inhId
    %  \;=\;
    %   \WW \tFProj{\WW}{1} \inhId
       \,,
      \\
      \opinvE \WW^{-1} \inhId \!
      &=&\!
      \inhId \big(\,\WW^{-1}\! - \Ket{\WW^{-1}} \tfrac1{\BraKet{1}{\WW^{-1}}} \Bra{\WW^{-1}}\,\big)
    %  \;=\;
    %   \inhId \tFProj{1}{\WW^{-1}} \WW^{-1}
       \,.
      \\
    \end{array}
   \eea
 }\fi
\end{corollary}
\begin{proof}
{\small
 The forms on the right-hand side of (\ref{def_opGmp}) directly follow from the expanded\HIDE{ explicit} representation of $\opinvE$ in (\ref{invE}), simplified using the identities\: $\inhId {\WW} \Ket{\WW^{-1}} \teq \Bra{\WW} {\WW^{-1}} \inhId \teq 0$\,.
 %% and the fact that $\inhId\Ket{1} = \Bra{1}\inhId =0$.

 The symmetric representations (\ref{opG_properties1_symmetry}) and properties\HIDE{ on the right of} (\ref{opG_properties2_kernels}) follow from the chain of identities
 \beq{}
  \inhId \big( f - \Ket{f} \tfrac1{\BraKet{1}{f}} \Bra{f} \big)
   \;=\;
  \Id \big( f - \Ket{f} \tfrac1{\BraKet{1}{f}} \Bra{f} \big)
   \;=\;
   \,   f - \Ket{f} \tfrac1{\BraKet{1}{f}} \Bra{f}   \,
   \;=\;
  \big( f - \Ket{f} \tfrac1{\BraKet{1}{f}} \Bra{f} \big) \Id
   \;=\;
  \big( f - \Ket{f} \tfrac1{\BraKet{1}{f}} \Bra{f} \big) \inhId
  \,,
  \nonumber
 \eeq
 provided $\hmgId \, \big( f - \Ket{f} \tfrac1{\BraKet{1}{f}} \Bra{f} \big) \:=\: \big( f - \Ket{f} \tfrac1{\BraKet{1}{f}} \Bra{f} \big) \, \hmgId =0$.
 These in turn follow from the kernel conditions (\ref{opG_properties2_kernels}):\, $\Bra{1} \, \big( f - \Ket{f} \tfrac1{\BraKet{1}{f}} \Bra{f} \big) \:=\: \big( f - \Ket{f} \tfrac1{\BraKet{1}{f}} \Bra{f} \big) \, \Ket{1} = 0$, which are\HIDE{ easily} verified by direct check.
}
\end{proof}

Also, from the definitions of the projector forms $\opG$, it follows that
 \beq{}
   \opGp \opE \, \WW^{-1} \,=\; \inhId \;=\: \WW \opE \opGm
   \: .
 \eeq

\begin{corollary}[] \label{Corollary:var_invE_G_equations} %%\label{Corollary:inhIWE-1f=0}
The equation  $\opGp \Ket{h} \teq 0$ on $\Ket{h}$ is equivalent to $\inhId \Ket{h} \teq 0\,$:
 \beq{inhIWE-1f=0}
   \opGp \Ket{h} \;\equiv\;
   \inhId \WW \opinvE \Ket{h}
   \;=\: 0
   \quad\;\; \Leftrightarrow  \;\;\quad
   \inhId \Ket{h} = 0
   \,.
 \eeq

Analogously, the equation $ \Bra{g}\, \opGm \teq 0$ on $\Bra{g}$ is equivalent to $ \Bra{g} \inhId \teq 0\,$:
 \beq{fE-1W-1inhI=0}
   \Bra{g}\, \opGm \;\equiv\;
   \Bra{g}\, \opinvE \WW^{-1} \!\inhId
   \;=\: 0
   \quad\;\;  \Leftrightarrow  \;\;\quad
   \Bra{g} \inhId = 0
   \,.
 \eeq

\end{corollary}

\begin{proof}
{\small
 The identities $\opG  \teq  \opG \inhId  \teq  \inhId \opG$ from
  (\ref{opG_properties2_kernels})
 imply that  $\inhId \Ket{h} \teq 0$ and $\Bra{g} \inhId \teq 0$ lead to $\opG \Ket{h}=0$ and $\Bra{g} \opG =0$, respectively.

 The converse follows from the rank properties: since (symmetric) operators $\inhId$ and $\inhId \opG  \teq \opG  \teq  \opG \inhId$ have the same rank and kernels of one operator contain kernels of another, their kernels coincide.

 To clarify, note that\HIDE{
 the structure} $\opG \defeq \big( f - \Ket{f} \frac1{\BraKet{1}{f}} \Bra{f} \big)$ for a non-zero $f$ is a one-mode-degenerate operator, and it can be extended to the nondegenerate operator with the quantity of rank $1$ from the kernel of $\inhId\,$:
  \beq{}
   \opG \;\to\; \op{G}' \defeq \opG {} + \hmgId \;=\;  f - \Ket{f} \tfrac1{\BraKet{1}{f}} \Bra{f} \,+\, \Ket{1} \tfrac1{\BraKet{1}{1}} \Bra{1}
   \,.
  \eeq
 Crucially, $\opG \inhId  \teq  \op{G}' \inhId $ and $\inhId \opG  \teq  \inhId \op{G}'$, where $\op{G}'$ is nondegenerate and invertible.
 Then: $\,\inhId  \teq  \op{G}'^{-1} \op{G}' \inhId  =  \op{G}'^{-1} \opG \inhId  \teq  \op{G}'^{-1} \inhId \opG $ and, analogously, $\inhId  \teq  \inhId \op{G}' \op{G}'^{-1}  =   \inhId \opG \op{G}'^{-1}  \teq  \opG \inhId \op{G}'^{-1} $,
 where we used\HIDE{ the property} $ \inhId \opG  \teq  \opG \inhId$ from Corollary \ref{Corollary:var_invE_G_properties}.
  % $ \opG = \inhId \big( f - \Ket{f} \frac1{\BraKet{1}{f}} \Bra{f} \big) =  \big( f - \Ket{f} \frac1{\BraKet{1}{f}} \Bra{f} \big) =  \big( f - \Ket{f} \frac1{\BraKet{1}{f}} \Bra{f} \big) \inhId$.
  %
 Therefore,  $\opG \Ket{h} \teq 0$ and $ \Bra{g} \opG \teq 0$ imply $\inhId \Ket{h} \teq 0$ and $\Bra{g} \inhId \teq 0$ respectively.

 Together with the first obvious statement of the proof, this implies in each pair that two equations follow from each other, which finally proves their equivalence: (\ref{inhIWE-1f=0}) and (\ref{fE-1W-1inhI=0}).
 %
 %%The second statement, (\ref{fE-1W-1inhI=0}), is proved analogously.
}
\end{proof}

 This expresses the simple observation that the kernels of both operators are spanned by\HIDE{ constant} homogeneous functions, while in the average-free\HIDE{ functional} subspace, both operators are nondegenerate.

\newpar

 %%We use representations (\ref{opG_properties1_symmetry}) to simplify constraints and equations of motion.
These properties are used\HIDE{ implicitly} to simplify the constraint $\inhId \WW \opinvE \Ket{\ptf_\io} \teq 0\: \to \:\inhId \Ket{\ptf_\io} \teq 0$  in Section \ref{SSect:AltAction_GUMG},
and\HIDE{ to simplify} the equation of motion
 $ \Bra{\repNl \FF^{-1}} \opinvE \WW^{-1} \! \inhId \teq 0\: \to \:
   \Bra{\repNl \FF^{-1}} \inhId \teq 0
 $\, in Section \ref{SSect:DynamicalProperties_GUMG}.
In Section \ref{SSect:GaugeStructure_GUMG}, the operators $\opG$ and their properties are utilized to derive the gauge structure.

\if{
 %\underline{Symmetricity properties}

 Also note that $\WW \opinvE$ and  $\opinvE \WW^{-1}$ are the sums of the three symmetric terms and the last nonsymmetric term. Operators $\inhId \WW \opinvE \inhId$ and  $\inhId\opinvE \WW^{-1}\inhId$, as well as $\tFProj{X}{1} \WW \opinvE \tFProj{1}{X}$ and  $\tFProj{X}{1}\opinvE \WW^{-1}\tFProj{1}{X}$, \;\;$\tFProj{X}{\WW} \WW \opinvE \tFProj{\WW}{X}$ and  $\tFProj{X}{\WW^{-1}}\opinvE \WW^{-1}\tFProj{\WW^{-1}}{X}$ are symmetric.
}\fi

 \subsubsection*{Variations of $\opinvE$}
  %\label{ASSect:opG_properties}
   \hspace{\parindent}
Although the operators $\opE$ and $\opinvE$ are nondegenerate, their variations with respect to the\HIDE{ spatial} metric acquire a degenerate character.

\begin{lemma}[Variation of $\opinvE$] \label{Lemma:invE_variations}
Under a general metric variation, %% $\var \g_\zm\zn$
 \beq{Lemma:var_invE_statement} % (\ref{Lemma:var_invE_statement})
  \begin{array}{lll}
 %  \var \opE  %   & = &
 %   \HIDE{+} \var\WW^{-1} \opGp \opE \,+\, \opE \opGm \var \WW \,,
 %  \\
   \var \opinvE     &\! = \!&
    - \opinvE \! \var\WW^{-1} \opGp  -\,  \opGm \var \WW \opinvE \,.
   \\
  \end{array}
 \eeq
\end{lemma}

\begin{proof}
{\small
The variation of $\opE$ is
 $ % \beq{}
   \var \opE    \,=\,
   \HIDE{+} \var\WW^{-1} \inhId \WW
   + \WW^{-1} \inhId \var \WW
  % =
  % - \var\WW^{-1} \hmgId \WW
  % - \WW^{-1} \hmgId \var \WW
  % =
  %  \WW^{-1}\var\WW \WW^{-1} \hmgId \WW
  % - \WW^{-1} \hmgId \var \WW
  % \nonumber\\
   \,=\,
   \HIDE{+} \var\WW^{-1} \opGp \opE    \,+\, \opE \opGm  \var \WW
   ,
 $ % \eeq
\,from which it follows that
 $ %  \beq{}
   \var \opinvE    \,=\,
   - \opinvE  (\var \opE )\, \opinvE   % =
  % + \opinvE \WW^{-1}\var\WW \WW^{-1} \inhId \WW \opinvE   % -\opinvE \WW^{-1} \inhId \var \WW \opinvE    =
   - \opinvE \! \var\WW^{-1} \opGp
   -\, \opGm \var \WW \opinvE    \,.
 $ %  \eeq
}
\end{proof}

\begin{corollary}[] \label{Corollary:invE_1_variations}
  \beq{Corollary:invE_1_variations_statement}
   \var \opinvE \Ket{1}
   \;=\;
   - \opGm \var \WW \opinvE \Ket{1}
   \,.
  \eeq
\end{corollary}

\begin{proof}
{\small
Follows from Lemma \ref{Lemma:invE_variations} and the property $ \opGp \Ket{1} = 0$, (\ref{opG_properties2_kernels}).
}
\end{proof}

This result allows expressing the variation
 $\var {\big( \Ket{\WW^{-1}}/\hmg{\WW^{-1}} \big)}  \teq  - \opGm \var \WW \Ket{\WW^{-1}}/\hmg{\WW^{-1}}$
in Section \ref{SSect:GaugeStructure_GUMG}, clarifying
 %% together with property (\ref{fE-1W-1inhI=0})
the emergence of $\opGm$ in the constraints on unrestricted gauge parameters.

\begin{lemma}[Poisson brackets of $\opinvE$] \label{Lemma:invE_PB}
 \beq{invE_PB_Hl_Hs} % (\ref{sDet_opE})
  \begin{array}{lll}
   \PB{\Bra{f} \opinvE \Ket{e}}{\BraKet{h}{\Hl}}
    & = &
    - \hmg{{f}\,\opinvE{\WW^{-1}}}\,
    %%- \Bra{f}\opinvE \Ket{\WW^{-1}} \sVol^{-1}
    \BraKet{h \tfrac{1}{\FF}\TTheta \trpg  }{\opGp e}
    +
    \BraKet{f \opGm}{ h  \tfrac{1}{\FF}\TTheta \trpg }
    %% \sVol^{-1}\Bra{\WW} \opinvE \Ket{e}
    \,\hmg{{\WW} \opinvE {e}}
    \,, \vphantom{\big|}
   \\
   \PB{\Bra{f} \opinvE \Ket{e}}{\BraKet{\eta^\zn}{\Hs_\zn}}
    & = &
    \HIDE{+} \hmg{{f}\,\opinvE{\WW^{-1}}}\,
    %%+ \Bra{f}\opinvE \Ket{\WW^{-1}} \sVol^{-1}
    \BraKet{ \eta^\zn_{\,;\zn} {} \tTDer{\ln\WW}{\ln\sqrg} }{\opGp e}
    -
    \BraKet{f \opGm}{ \eta^\zn_{\,;\zn} {} \tTDer{\ln\WW}{\ln\sqrg} }
    \,\hmg{{\WW} \opinvE {e}}
    %% \sVol^{-1}\Bra{\WW} \opinvE \Ket{e}
    \,, \vphantom{\big|}
   \\
  \end{array}
 \eeq
where\, $\TTheta(\argsqrg) {\,\defeq\,}  \tfrac{1}{\Ddim{-}2} \tTDer{\ln{\WW}}{\ln{\!\sqrg}}\,\tfrac{\FF}{\sqrg}$
 %%, (\ref{def_OOmega_TTheta}),
\,and\, $\eta^\zn_{\,;\zn} {} \tTDer{\ln\WW}{\ln\sqrg} = \eta^\zn\partial_\zn \ln \WW +\Dvg{\eta} \tTDer{\ln\WW}{\ln\sqrg}$.
\end{lemma}

\begin{proof}
{\small
From (\ref{Lemma:var_invE_statement}),
 \beq{invE_PB_Hl} % (\ref{sDet_opE})
  \begin{array}{lll}
   \PB{ \Bra{f} \opinvE \Ket{e}}{\BraKet{h}{\Hl}}
    & = &
    -\Bra{f \opinvE} \PB{\WW^{-1}}{\BraKet{h}{\Hl}} \Ket{\opGp e}
    -\Bra{f \opGm} \PB{\WW}{\BraKet{h}{\Hl}} \Ket{\opinvE e}
    \\
    & = &
    -\Bra{f \opinvE} h  \WW^{-1} \,\tfrac{1}{\FF}\TTheta \trpg  \Ket{\opGp g}
    +\Bra{f \opGm} h \WW \,\tfrac{1}{\FF}\TTheta \trpg  \Ket{\opinvE e}
    \\
    & = &
   % - \Bra{f \opGm} \cancel{ h \tfrac{1}{\FF}\TTheta \trpg }  \Ket{\opGp e}
    - \Bra{f}\opinvE \Ket{\WW^{-1}} \frac{1}{\sVol} \BraKet{h \tfrac{1}{\FF}\TTheta \trpg  }{\opGp e}
   % \\
   % &&
   % + \Bra{f \opGm} \cancel{ h \tfrac{1}{\FF}\TTheta \trpg }  \Ket{\opGp e}
    + \BraKet{f \opGm}{ h  \tfrac{1}{\FF}\TTheta \trpg }\frac{1}{\sVol}\Bra{\WW} \opinvE \Ket{e} \,,
  \end{array}
  \nonumber
 \eeq
where we used (\ref{PB_sqrg_Hl}) and\HIDE{ in the last equality} cancelled the opposite contributions of $\Bra{f \opGm} { h \tfrac{1}{\FF}\TTheta \trpg }  \Ket{\opGp e}$, extracted from each term in the middle expression.
\if{
 \bea{}%{PB_sqrg_Hl}
  \PB{\sqrg}{\sint h B(\g) \Hl}
   %& =
   \,=\, - h\,\tfrac{1}{\Ddim{-}2} B(\g) \trpg %\,,
   \quad \Rightarrow \;
   \begin{array}{|rcl}
     \PB{ \FF}{\sint h B(\g) \Hl} &\!\! =\!&
     -  h \tfrac{1}{\Ddim{-}2} B(\g) \FF\, \WW \tfrac{\trpg}{\sqrg}
     \,, \\
     \PB{ \WW}{\sint h B(\g) \Hl} &\!\! =\!&
     -  h B(\g) \WW \,\tfrac{1}{\FF}\TTheta \trpg
     \,; \\
  %   \PB{ \WW\FF\sqrg}{\sint h B(\g) \Hl} &\!\! =\!&
  %   -  h \tfrac{1}{\Ddim{-}2} B(\g) \WW\FF\sqrg \,\OOmega \tfrac{\trpg}{\sqrg}
  %   \,, \\
   \end{array}
   \hspace{-5mm}
 \eea
}\fi

Similarly, from (\ref{Lemma:var_invE_statement}),
 \beq{invE_PB_Hs} % (\ref{sDet_opE})
  \begin{array}{lll}
   \PB{\Bra{f} \opinvE \Ket{e}}{\BraKet{\eta^\zn}{\Hs_\zn}}
    & = &
    -\Bra{f \opinvE} \PB{\WW^{-1}}{\BraKet{\eta^\zn}{\Hs_\zn}} \Ket{\opGp e}
    -\Bra{f \opGm} \PB{\WW}{\BraKet{\eta^\zn}{\Hs_\zn}} \Ket{\opinvE e}
 %   \\
 %   & = &
 %   +\Bra{f \opinvE \WW^{-1}} \PB{\WW}{\BraKet{\eta^\zn}{\Hs_\zn}} \Ket{\WW^{-1}\opGp e}
 %   -\Bra{f \opGm} \PB{\WW}{\BraKet{\eta^\zn}{\Hs_\zn}} \Ket{\opinvE e}
  %  \\
  %  & = &
  %  +\Bra{f \opinvE \WW^{-1}}
  %  %%\PB{\WW}{\BraKet{\eta^\zn}{\Hs_\zn}}
  %  \big( \eta^\zn \partial_\zn \WW
  %   + \Dvg{\eta} \WW \tTDer{\ln\WW}{\ln\sqrg} \big)
  %   \Ket{\WW^{-1}\opGp e}
  %  \\
  %  &  &
  %  -\Bra{f \opGm}
  %  %%\PB{\WW}{\BraKet{\eta^\zn}{\Hs_\zn}}
  %  \big( \eta^\zn \partial_\zn \WW
  %   + \Dvg{\eta} \WW \tTDer{\ln\WW}{\ln\sqrg} \big)
  %  \Ket{\opinvE e}
    \\
    & = &
    \HIDE{+} \Bra{f \opinvE \WW^{-1}}
    %%\PB{\WW}{\BraKet{\eta^\zn}{\Hs_\zn}}
     %\big( \eta^\zn \partial_\zn \ln{\WW}
     % {+} \Dvg{\eta} \tTDer{\ln\WW}{\ln\sqrg} \big)
     \eta^\zn_{\,;\zn} {} \tTDer{\ln\WW}{\ln\sqrg}
     \Ket{\opGp e}
  %  \\
  %  & &
    -\Bra{f \opGm}
    %%\PB{\WW}{\BraKet{\eta^\zn}{\Hs_\zn}}
     %\big( \eta^\zn \partial_\zn \ln{\WW}
     % {+} \Dvg{\eta} \tTDer{\ln\WW}{\ln\sqrg} \big)
     \eta^\zn_{\,;\zn} {} \tTDer{\ln\WW}{\ln\sqrg}
    \Ket{\WW\opinvE e}
    \\
    & = &
  %  +\Bra{f \opGm}
  %  \cancel{ \eta^\zn_{\,;\zn} {} \tTDer{\ln\WW}{\ln\sqrg} }
  %   \Ket{\opGp e}
    \HIDE{+} \Bra{f} \opinvE \Ket{\WW^{-1}} \frac{1}{\sVol}
    \Bra{1}
     \eta^\zn_{\,;\zn} {} \tTDer{\ln\WW}{\ln\sqrg}
     \Ket{\opGp e}
  %  \\
  %  & &
  %  -\Bra{f \opGm}
  %  \cancel{ \eta^\zn_{\,;\zn} {} \tTDer{\ln\WW}{\ln\sqrg} }
  %  \Ket{\opGp e}
    -\Bra{f \opGm}
     \eta^\zn_{\,;\zn} {} \tTDer{\ln\WW}{\ln\sqrg}
     \Ket{1}
    \frac{1}{\sVol} \Bra{\WW}\opinvE \Ket{e}
    \,,
  \end{array}
  \nonumber
 \eeq
where we used (\ref{PB_sqrg_Hs}) and\HIDE{ in the last equality} cancelled the opposite contributions of
    $\Bra{f \opGm}
    { \eta^\zn_{\,;\zn} {} \tTDer{\ln\WW}{\ln\sqrg} }
    \Ket{\opGp e}$,
extracted from each term in the middle expression.
}
\end{proof}

\begin{corollary}[] \label{Corollary:invE_PB}
When\HIDE{ the operator} $\opinvE$ acts to the right on a spatially homogeneous function, the Poisson brackets (\ref{invE_PB_Hl_Hs}) simplify:
 \beq{Corollary:invE_PB_Hl_Hs_statement} % (\ref{sDet_opE})
  \begin{array}{lll}
   \PB{\Bra{f} \opinvE \Ket{1}}{\BraKet{h}{\Hl}}
    &\! = \!&
    \BraKet{f \opGm}{ h  \tfrac{1}{\FF}\TTheta \trpg }
    \, \frac{1}{\hmg{\WW^{-1}}\vphantom{|^{I^I}}}
    %%\hmg{{\WW} \opinvE {1}}
    \,, \vphantom{\big|}
   \\
   \PB{\Bra{f} \opinvE \Ket{1}}{\BraKet{\eta^\zn}{\Hs_\zn}}
    &\! = \!&
    -
    \BraKet{f \opGm}{ \eta^\zn_{\,;\zn} {} \tTDer{\ln\WW}{\ln\sqrg} }
    \, \frac{1}{\hmg{\WW^{-1}}\vphantom{|^{I^I}}}
    %%\hmg{{\WW} \opinvE {1}}
    \,. \vphantom{\big|}
   \\
  \end{array}
 \eeq
\end{corollary}

\begin{proof}
{\small
This follows from Lemma \ref{Lemma:invE_PB}, along with\HIDE{ the properties} $ \opGp \Ket{1} \teq 0$, (\ref{opG_properties2_kernels}), \,and\, $\hmg{{\WW} \opinvE {1}} \teq {1}/\,{\hmg{\WW^{-1}}\vphantom{|^{I^I}}}$, (\ref{opinvE_properties}).
}
\end{proof}

 \subsubsection*{Determinant of operator $\opE$}
  %\label{ASSect:opG_properties}
  % \hspace{\parindent}
 \newcommand{\Det}{\mathop{\mathrm{Det}}}
 \newcommand{\sDet}{\mathop{\color{DarkGreen}\mathrm{Det}}_{\!\sx}}
 \newcommand{\sTr}{\mathop{\color{DarkGreen}\mathrm{Tr}}_{\!\sx}}
\begin{lemma}[Determinant of $\opE$] \label{Lemma:DetE}
The determinant of the spatial integral operator $\opE$, (\ref{def_opE_COPY}), with the two-point kernel ${E}(\taux|\sx;\sx')$, (\ref{opE_kernel}), reads:
 \beq{sDet_opE} % (\ref{sDet_opE})
  \sDet \opE \,(\taux)
  \;=\;
  \hmg{\WW}(\taux)\,\hmg{\WW^{-1}}(\taux)
  \geq 1
  \,,
 \eeq
 The equality is satisfied only for $\WW(\taux,\sx) \teq \hmg{\WW}(\taux)$, which happen either when $\WW \teq \wwc \teq \const$ or for metric configurations with $\sqrg(\taux,\sx) \teq \hmg{\sqrg}(\taux)$.
\end{lemma}

\begin{proof}
{\small
The variation of the determinant gives:
 \bea{}
  \var \ln \sDet {\opE}
  &\!=\!&
  \var \sTr \ln \opE   \;=\; \sTr {\opinvE \,\var\opE}
  \;=\;
    - \sTr { \big(\opinvE {\WW}^{-1}\hmgId \var{\WW}\, \big)}
    - \sTr { \big(\var{\WW^{-1}}\hmgId\WW\,\opinvE \big)}
  \nonumber
  \\
  &\!=\!& - \big[ \tfrac{1}{\sVol} \Bra{\var{\WW}} \opinvE \Ket{\WW^{-1}} \big]
  - \big[ \tfrac{1}{\sVol} \Bra{\WW} \opinvE \Ket{\var{\WW^{-1}}} \big]
  % \nonumber \\
  % &=&
  \;=\;
  \tfrac{\BraKet{\var \WW}{1}}{\BraKet{1}{\WW}}
  + \tfrac{\BraKet{1}{\var \WW^{-1}}}{\BraKet{\WW^{-1}}{1}}
  %\nonumber \\
  \;=\;
   \var{ \ln{ \big( \,\hmg{\WW}\, \hmg{\WW^{-1}}\, \big) } }
  \nonumber
  \,.
 \eea
This implies $\sDet {\opE} \teq \const\, \hmg{\WW}\, \hmg{\WW^{-1}}$. The\HIDE{ normalization} constant is fixed by the limit $\WW \to \wwc$, where $\opE$ is the identity operator\HIDE{ with unit determinant}.
}
\end{proof}

The result (\ref{sDet_opE}) %% for $\ln\sDet_{\!\sx\,} (\Id + \Delta)$
can also be checked via Taylor expansion of $\,\HIDE{\ln\sDet (\Id{+}\Delta) \teq } \sTr\ln (\Id + \Delta)\,$ in powers of $\,\Delta \tdefeq \hmgId {-\,} \WW^{-1}\! \hmgId \WW$, which after taking the trace gives the expansion of $\ln\big(1- (1{-}\hmg{\WW}\,\hmg{\WW^{-1}}) \big)$.

  \subsection{Properties of a generalized\HIDE{ $\opE$-type} delocalization operator and its inverse}
   \label{SSect:generic_op_E}  %% 2023-12-12
    \hspace{\parindent}
More general operators
 %% of the form \,$ \hmgId + \Af^{-1} \inhId \Bf^{-1} $\,
  \bea{def_opEE}
  \opEE
   &\!=\!& \hmgId + \Af^{-1\!} \inhId \Bf^{-1}
  % \nonumber\\
  % &=&
   \;=\;
   \Af^{-1\!} \Id \Bf^{-1} + \Ket{1}\invBraKet{1}{1}\Bra{1} - \Ket{\Af^{-1}}\invBraKet{1}{1}\Bra{\Bf^{-1}}
   \,,
  %\nonumber
 \eea
defined on compact spaces for arbitrary sign-definite functions $\Af(\sx)$ and $\Bf(\sx)$ possess similar properties, except for the homogeneity limit property (\ref{EProp:homogeneity}), which for $\Af(\sx)\to\hmg{\Af}$ and $\Bf(\sx)\to\hmg{\Bf}$ with $\hmg{\Bf}\neq\hmg{\Af}^{-1}$ leads to a rescaling in the inhomogeneous sector. Its featured actions are
 \bea{opEE_properties}
  \begin{array}{lllll}
   \Bra{1} \opEE &\!=\!&  \Bra{1} + \Bra{\hspace{0.5pt}\inh{\Af^{-1}} \Bf^{-1}}
   \,,
   \\ %& \quad &
   \opEE \Ket{1} &\!=\!& \Ket{1}  + \Ket{\Af^{-1} \inh{\Bf^{-1}}}
   \,,
  \end{array}
 % \\
 \qquad \qquad
 % \\
  \begin{array}{lllll}
   %\Bra{1} \WW \opE &=&
   \Bra{\Af} \opEE
   %% &\!\!=\!& \BraKet{\Af}{1} \sVol^{-1\!} \Bra{1}
    &\!\!=\!& \hmg{\Af}\, \Bra{1}
    \,,
    \\ % & \quad &
   %\opE \WW^{-1} \Ket{1} &=&
   \opEE \Ket{\Bf}
   %% &\!\!=\!& \Ket{1}  \sVol^{-1\!} \BraKet{1}{\Bf}
   &\!\!=\!& \Ket{1} \,\hmg{\Bf} \vphantom{\big|^I}
   \,.
  \end{array}
 \eea

\begin{lemma}[Nondegenerateness and inverse of $\opEE$] \label{Lemma:InvEE}
 Operator $\opEE \teq \hmgId + \Af^{-1} \inhId \Bf^{-1}$, (\ref{def_opEE}),
 where $\Af(\sx)$ and $\Bf(\sx)$ are bounded, sign-definite functions\HIDE{ on a compact manifold},
  %%(manifold with the finite coordinate volume),
is a \emph{nondegenerate} operator
 %%on the $L_2$ \TODO{(?)} space of scalar functions on the manifold.
with an \emph{inverse} of the form
 \bea{invEE}
   \opEE^{-1}
   &\!=\!& \Bf \Id \Af
         \;-\;  \Ket{\Bf\Af} \invBraKet{\Af}{1} \Bra{\Af}
         \;-\;  \Ket{\Bf} \invBraKet{\Bf}{1} \Bra{\Bf\Af}
         \;+\;  \Ket{\Bf} \invBraKet{\Bf}{1} \big( \BraKet{1}{1} {+} \BraKet{\Bf\Af}{1}\big) \invBraKet{\Af}{1} \Bra{\Af}
 \if{
  \nonumber\\
   &=& \Bf \Id \Af
         \;-\;  \Bf\Af\Ket{1} \invBraKet{\Af}{1} \Bra{\Af}
         \;-\;  \Bf\Ket{1} \invBraKet{\Bf}{1} \Bra{\Bf}\Af
         \;+\;  \Bf\Ket{1} \invBraKet{\Bf}{1} \big( \BraOKet{\Bf}{\Bf^{-1}}{1} {+} \BraOKet{\Bf}{\Af}{1}\big) \invBraKet{\Af}{1} \Bra{\Af}
 }\fi
 \,.
 \eea
\end{lemma}

From sign-definiteness and boundedness, it follows that $\Af, \Bf \neq0,\infty$, and that the integrals $\BraKet{\Af}{1}\HIDE{\neq0}$ and $\BraKet{\Bf}{1}\HIDE{\neq0}$ exist and do not vanish.
The nondegeneracy of $\opEE$ is established by the explicit construction of its inverse, $\opEE^{-1}$, mirroring the approach used in Lemma \ref{Lemma:InvE}.
The correctness of the expression (\ref{invEE}) for the inverse operator can also be verified directly.

\if{
\begin{proof}
{\small
  We start from the defining relation
   $ % \beq{}
    \opEE \, \opEE^{-1} = \Id \,,
   $ %\eeq
  as an equation on $\opEE^{-1}$.
  From\HIDE{ the right representation of} (\ref{def_opEE}):
   \beq{invEE1}
    \Af^{-1} \Id \Bf^{-1} \opEE^{-1} + \Ket{1} \invBraKet{1}{1} \underbrace{ \Bra{1}\opEE^{-1}}_{{\Bra{\Xf}}} \;-\; \Ket{\Af^{-1}} \invBraKet{1}{1} \underbrace{\Bra{\Bf^{-1}} \opEE^{-1}}_{\Bra{\Yf}}
    \;=\; \Id
    \,.
   \eeq
  To find $\opEE^{-1}$, we solve for $\Bra{\Xf}$ and $\Bra{\Yf}$.

  Contracting (\ref{invEE1}) with $\Bra{\Af}$ yields:
   %\if{
   \beq{contr1_invEE}
    \cancel{ \Bra{\Yf} } + \BraKet{\Af}{1} \invBraKet{1}{1}  \Bra{\Xf} - \cancel{ \BraKet{\Af}{\Af^{-1}} \invBraKet{1}{1}  \Bra{\Yf} }
    \;=\; \Bra{\Af}
   \eeq
   %}\fi
  %%one gets resolution for $\Bra{\Xf}$
   \beq{BraXX}
      \Bra{\Xf}
    \;=\;   \BraKet{1}{1} \invBraKet{\Af}{1} \Bra{\Af}
    \,.
   \eeq

  Next, contracting (\ref{invEE1}) with $\Bra{\Bf\Af}$
   %\if{
   \beq{contr2_invEE}
    \Bra{\Xf}  + \BraKet{\Bf\Af}{1} \invBraKet{1}{1} \Bra{\Xf} -  \BraKet{\Bf}{1} \invBraKet{1}{1} \Bra{\Yf}
    \;=\; \Bra{\Bf\Af}
   \eeq
   %}\fi
  and using (\ref{BraXX}) gives\HIDE{ the resolution for $\Bra{\Yf}$}
   \beq{BraYY}
    \Bra{\Yf}
    \;=\;  - \BraKet{1}{1} \invBraKet{\Bf}{1} \big( \Bra{\Bf\Af} - \BraKet{1}{1} \invBraKet{\Af}{1} \Bra{\Af} - \BraKet{\Bf\Af}{1} \invBraKet{\Af}{1} \Bra{\Af} \big)
    \,.
   \eeq
  Substituting (\ref{BraXX}) and (\ref{BraYY}) into (\ref{invEE1}) and left-multiplying by $\Bf\Id\Af$ yields the result (\ref{invEE}).
}
\end{proof}
}\fi

\begin{corollary}[] \label{Corollary:EE-1_1}
 For the operator % $\opEE$ of the form
 $ % \beq{def_opEE_COPY2}
  \,\opEE
   = \hmgId {+\,} \Af^{-1\!}\inhId \Bf^{-1}
 $, (\ref{def_opEE}), % \eeq
the following properties hold:
\begin{itemize}
  \item
    $\opEE^{-1} \Ket{1}$ is independent of $\,\Af$ and takes the form
     $ %\beq{} \
       \; \displaystyle
       \opEE^{-1} \Ket{1}
       = \Ket{\Bf}\tfrac{1}{\BraKet{\Bf}{1}} \BraKet{1}{1}
       = \tfrac{\Ket{\Bf}}{\hmg{\Bf}\vphantom{I^l} } \,.
    $ % \eeq
  %%The essence of the mechanism of the on-shell independence on the arbitrary choice of $\Af$

  \item
    $\Bra{1} \opEE^{-1} $ is independent of $\,\Bf$ and takes the form
     $ %\beq{} \
       \; \displaystyle
       \Bra{1} \opEE^{-1}
       = \BraKet{1}{1} \tfrac{1}{\BraKet{\Af}{1}}  \Bra{\Af}
       = \tfrac{\Bra{\Af}}{\hmg{\Af}} \,.
    $ % \eeq

  \item
    $\Bra{1} \opEE^{-1} \Ket{1}$ is independent of both $\,\Af$ and $\,\Bf$ and equals
     $ %\beq{} \
       \; \displaystyle
       \Bra{1} \opEE^{-1} \Ket{1}
       = \BraKet{1}{1} \,.
    $ % \eeq

\end{itemize}

\end{corollary}

\begin{corollary}[] \label{Corollary:inhIbE-1}
 For the operator % $\opEE$ of the form
  $ % \bea{def_opEE_COPY2}
  \opEE
   = \hmgId {+\,} \Af^{-1\!}\inhId \Bf^{-1}
 $, (\ref{def_opEE}), % \eeq
%where $\Af(\sx)$, $\Bf(\sx)$ --- sign-definite (finite) functions on compact manifold without boundaries,
the following holds:
 \bea{inhIbE-1}
   \inhId \Bf^{-1} \opEE^{-1}
   &\!=\!&   \inhId \Af
         -  \inhId\Ket{\Af} \invBraKet{\Af}{1} \Bra{\Af}
        % \:-\:  \inhId \Ket{1} \invBraKet{\Bf}{1} \Bra{\Bf\Af}
        % \:+\:  \inhId \Ket{1} \invBraKet{\Bf}{1} \big( \BraKet{1}{1} {+} \BraKet{\Bf\Af}{1}\big) \invBraKet{\Af}{1} \Bra{\Af}.
   \;=\;   \Af
         -  \Ket{\Af} \invBraKet{\Af}{1} \Bra{\Af}
   \;=\;   \Af \inhId{}
         -  \Ket{\Af} \invBraKet{\Af}{1} \Bra{\Af} \inhId
   \,.
 \eea

\end{corollary}

This is analogous to the properties of operators $\opG$ discussed in Corollary \ref{Corollary:var_invE_G_properties}. The first form arises because the last two terms in $\Bf^{-1} \opEE^{-1}$, (\ref{invEE}), are of the form $\Ket{1} \,..\, \Bra{\,.\,}$ and are annihilated by the{ inhomogeneous} projector $\inhId$ acting from the left. The subsequent two forms follow from the fact that the rank and the left and right kernels of $\Af -  \Ket{\Af} \invBraKet{\Af}{1} \Bra{\Af}$ coincide with those of its product with the inhomogeneous projector $\inhId$.
 %%(see the proof of Lemma \ref{Corollary:var_invE_G_properties}).

\begin{corollary}[] \label{Corollary:inhIbE-1f=0}
 The equation \,$\inhId\Bf^{-1}\opEE^{-1} \Ket{f} \teq 0\,$ for $\,\Ket{f}\,$ is equivalent to $\inhId \Ket{f} \teq 0\,$:
 \bea{inhIbE-1f=0}
   \inhId\Bf^{-1}\opEE^{-1} \Ket{f} = 0
    \quad\;\; \Leftrightarrow  \;\;\quad
   \inhId \Ket{f}=0
   \,.
 \eea
\end{corollary}

\if{
\begin{proof}
  Since $\inhId \Bf^{-1} \opEE^{-1}\Ket{1} = 0$ from (\ref{inhIbE-1}), inserting the partition $\Id = \hmgId + \inhId$ on the right implies $\inhId \Bf^{-1} \opEE^{-1} = \inhId \Bf^{-1} \opEE^{-1} \inhId$.

  Right kernels  of operators $\op{A}\op{B}$ and $\op{B}$ coincide if $\rank \op{A}\op{B} = \rank \op{B}$. Thus, equivalence of equations $\inhId\Bf^{-1}\opEE^{-1} \inhId \Ket{f}=0$ and $ \inhId \Ket{f}=0$ follows from validity of the maximal rank condition: $\rank \inhId\Bf^{-1}\opEE^{-1} \inhId = \rank \inhId$.
  This holds since $\Bf^{-1} \opEE^{-1}$ is nondegenerate, hence $\rank \inhId = \rank \inhId \Bf^{-1} \opEE^{-1} $. Due to $\inhId \Bf^{-1} \opEE^{-1} = \inhId \Bf^{-1} \opEE^{-1} \inhId$ one gets  $\rank \inhId = \rank \inhId\Bf^{-1}\opEE^{-1} = \rank \inhId \Bf^{-1} \opEE^{-1} \inhId $.
\end{proof}
}\fi

This result is analogous to Corollary \ref{Corollary:var_invE_G_equations}, and the proof follows similarly. The right kernels of both operators are spanned by constant (homogeneous) functions, while in the average-free functional subspace, both operators are nondegenerate.

%%-------------------------=%        ******         %=-------------------------%%
%%-------------------------=%        ******         %=-------------------------%%
%%-------------------------=%        ******         %=-------------------------%%

\newpage
\section{Calculational Notes on Gauge Symmetries}
 \label{ASect:gauge_calcs}
  %\hspace{\parindent}

  \subsection{{\wGUMG} canonical gauge transformations in \GR-like\HIDE{ gauge} basis}
   \label{ASSect:gauge_calcs_AMG}
    \hspace{\parindent}
To prepare for the Lagrangian reduction along the lines of general relativity in Section \ref{SSect:ConstraintBasisChange} we redefine the configuration space of the canonical system (\ref{paramAction_AMG'}) via ${\replm} \teq \repNl\Fw^{-1}$ and $\ptf_\io \teq \CCof $. This manifests all ADM components of the spacetime metric. Starting from gauge transformations (\ref{gauge_transfs_ham_AMG'}), we derive $\gvar[\geps,\gzeta] \repNl$ and substitute ${\replm} \to \repNl$, $\ptf_\io \to \CCof \teq \ptf_\io$. Thus we obtain the gauge transformations induced by canonical generators (\ref{I_class_constraints_AMG}) for the metric fields of the action  (\ref{paramAction_AMG_fin}):
 \beq{gauge_transfs_ham_AMG''_metric}
 \left|
  \begin{array}{lll}
   \gvar[\geps,\gzeta] \g_{\zn\zm} %\vphantom{\big|}
   &\!=&
   \big( \g_{\zk\zm} \partial_\zn %\gzeta_{\isst}^{\zk}
   + \g_{\zn\zk} \partial_\zm %\gzeta_{\isst}^{\zk}
   + \partial_\zk \g_{\zn\zm} %\gzeta_{\isst}^{\zk}
   \big) %(
   \gzetat^{\zk}
   + \tfrac{2}{\sqrg}\left(\pg_{\zm\zn} -\frac{1}{\Ddim{-}2}\,\trpg\,\g_{\zm\zn} \right)
   \Fw\hmg{\geps} \,,
   \\
   \gvar[\geps,\gzeta] \pg^{\zn\zm} %\vphantom{\big|}
   &\!=&
    \PB{\pg^{\zn\zm}}{\sint \gzetat^{\zn} \Hs_\zn} +
    \PB{\pg^{\zn\zm}}{\sint \Fw \Hl} \hmg{\geps}
    \,,
   \\
  \if{
   \gvar[\geps,\gzeta] \CCf \vphantom{\big|}
   %% used (\ref{gvar_sqrt})
  % &=&
  % \gvar[\geps,\gzeta] (\Fw^{-1} \sqrg^{-1}) \ptf_\io + \Fw^{-1} \sqrg^{-1} \gvar \ptf_\io
  % \\
   &=&
      - % \HIDE{ \sqrg^{-1}\Fw^{-1} (\sqrg\Fw)_{,\zk}}
       (\wwc{+}1)\tfrac{1}{\sqrg} \partial_\zk \sqrg  %%% \sqrg_{,\zk}
        \CCf\, \gzeta_{\isst}^{\zk}
      +  \tfrac{1}{\Ddim{-}2} (\wwc{+}1) \tfrac{\trpg}{\sqrg} \CCf\, \Fw\hmg{\geps}
    \,,
   \\
  }\fi
  \gvar[\geps,\gzeta] \repNl \vphantom{\big|}
  %% used (\ref{gvar_sqrt})
 % &\!=&
 % \gvar (\hmg{\replm} {+} \inh{\replm}) \Fw + \replm \gvar{\Fw}
 % \;=\;
 %   \Fw \TDer{}{\taux} \hmg{\geps} %\tauxdot{\geps}
 %  +
 %   (\Fw^{-1}\repNl)_{,\zn} \Fw \gzetat^\zn
 %  + \wwc \Dvg{\repNs} \Fw \hmg{\geps}
 %  +
 % \wwc \repNl \tfrac{1}{\sqrg} \sqrg_{,\zn} \gzeta_{\isst}^\zn
 %  - \wwc \repNl \tfrac{1}{\Ddim{-}2} \tfrac{\trpg}{\sqrg} \,\Fw\hmg{\geps}
 % \\
   &\!=&
     \Fw \TDer{}{\taux} \hmg{\geps} %\tauxdot{\geps}
   +
    \partial_\zn \repNl %%%\repNl_{\!\!\!\!,\zn}
    \gzetat^\zn
   + \wwc \Dvg{\repNs} \Fw \hmg{\geps}
   - \wwc \repNl \tfrac{1}{\Ddim{-}2} \tfrac{\trpg}{\sqrg} \,\Fw\hmg{\geps}
   \,,
   \\
   \gvar[\geps,\gzeta] \repNs^\zn
     &\!=& \vphantom{\big|}
         \TDer{}{\taux} \gzetat^{\zn}
       %% + \TDer{}{\taux} (U_0^\zn \hmg{\geps})
         - \LieB{\zn}{\repNs}{\gzetat \HIDE{+ U_0 \hmg{\geps}}}
         + \Fw \g^{\zn\zm} \partial_\zm (\Fw^{-1}\repNl) %%%(\Fw^{-1}\repNl)_{,\zm}
         \Fw\hmg{\geps}
         \,,
   \\
  \end{array}
  \right.
  \eeq
supplemented by the transformations of the auxiliary-sector fields:
 \beq{gauge_transfs_ham_AMG''_aux}
 \left|
  \begin{array}{lll}
   \gvar[\geps,\gzeta] \tf^\io \vphantom{\big|}
   &\!\!=\!&
   %%% \hmg{\geps} - \Dvg{\gchi}
   %%% % In the basis of canonical transformations (\ref{canon_transf_AMG'_App})
   %%% \equiv
  % \geps
  % \;=\;
   \hmg{\geps} + \inh{\geps}
   \,,
   %\\
   \qquad
   \gvar[\geps,\gzeta] \CCof
   % &=&
   \;=\; 0 \,,
  % \\
   \qquad
   \gvar[\geps,\gzeta] \Dvg{\lmptf} \vphantom{\big|}
   %&\!=&
   \;=\;
   %%% + \TDer{}{\taux} \Dvg{\gchi} +  \replm_{,\zm} \gzetat^\zm + \Dvg{\repNs}\wwc \hmg{\geps}
   %%% % In the basis of canonical transformations (\ref{canon_transf_AMG'_App})
   %%% =
    - \TDer{}{\taux} \inh{\geps} +
    \partial_\zn (\Fw^{-1}\repNl) %%%(\Fw^{-1}\repNl)_{,\zn}
    \gzetat^\zn + \Dvg{\repNs}\wwc \hmg{\geps}
   \,.
   \\
  \end{array}
  \right.
  \eeq

Gauge transformations in the Lagrangian theory follow from restricting right-hand sides onto the reduction surface (\ref{momenta_ecK_GR}). Before doing so, we express the\HIDE{ canonical} gauge transformations (\ref{gauge_transfs_ham_AMG''_metric}) in a {\GR}-like basis of canonical generators\footnote{
 In canonical general relativity, metric transformations are standardly generated by $\sint \gxl \Hl$ and $\sint \gxs^\zn \Hs_\zn$, (\ref{BGUMGHamiltonianStructure}), (\ref{BGUMGMomentaConstraints}). Here, transformations (\ref{gauge_transfs_ham_AMG''_metric}) are induced by $\sint \geps \Fw \Hl$ and $\sint \gxs^\zn \Hs_\zn$, constrained as $\geps \to \hmg{\geps}$ and $\gxs^\zn \to \gzetat^\zn$.
},
to provide a direct transition to (constrained) covariant gauge transformations of the spacetime metric. Naturally, such {\GR}-like transformations arise from the canonical generators associated with the constraints of the action (\ref{paramAction_AMG_fin}), subject to specific parameter constraints. Nevertheless, it is instructive to derive them directly from (\ref{gauge_transfs_ham_AMG''_metric}).

\HIDE{ (with appropriate constraints on their gauge parameters)}
\if{
 (\ref{replm_ptf_io_rescaling_AMG})
 $ % \bea{}
  \repNl = {\replm} \Fw \HIDE{\sqrg^{\wwc}}\,, \;%\quad
  \CCf = \HIDE{\sqrg^{-\wwc-1}} \Fw^{-1} \sqrg^{-1} \ptf_\io
 $ % \eea
}\fi
\if{
 \beq{gauge_transfs_ham_AMG'_COPY}
 \left|
  \begin{array}{lll}
   \gvar[\geps,\gzeta] \g_{\zn\zm}
   &\!=&
   \tfrac{2\Fw}{\sqrg} \left(\pg_{\zn\zm} -\frac{1}{\Ddim{-}2}\,\trpg\,\g_{\zn\zm} \right)\!
   \hmg{\geps}
  +
   \big( \g_{\zk\zm} {\gzeta_{\isst}^{\zk}}_{,\zn}
   + \g_{\zn\zk}  {\gzeta_{\isst}^{\zk}}_{,\zm}
   + \g_{\zn\zm,\zk} \gzeta_{\isst}^{\zk}
   \big)
   \,, \\
   \gvar[\geps,\gzeta]  \tf^\io
   &\!=&
   \geps
      \,, \qquad  \quad
      \gvar[\geps,\gzeta]  \ptf_\io \;=\; 0
  % \,, \\
  % \gvar[\geps,\gzeta]  \ptf_\io &=& 0
   \,, \\
   \gvar[\geps,\gzeta]  \repNs^\zn
     &\!=&
       \Fw^2 \g^{\zn\zm} \replm_{,\zm} \hmg{\geps}
      + \TDer{}{\taux} \gzetat^{\zn}
      - \LieB{\zn}{\repNs}{\gzetat}
   \,, \\
   \gvar[\geps,\gzeta]  \hmg{\replm}
    &\!=& \TDer{}{\taux} \hmg{\geps}
      \,, \qquad
      \gvar[\geps,\gzeta]  \inhw{\replm}
      \;=\;
      \wwc\, \Dvg[\zk]{\repNs}\hmg{\geps}
      + \replm_{,\zk} \gzetat^\zk
 %  \,, \\
 %  \gvar[\geps,\gzeta]  \inhw{\replm} &=&
 %   \wwc\, \Dvg[\zn]{\repNs}\hmg{\geps}
 %   + \replm_{,\zm} \gzetat^\zm
   \,, \\
   \gvar[\geps,\gzeta]  \Dvg[\zk]{\lmptf}
   \!\! &\!=&
    - \TDer{}{\taux} \inhw{\geps}
    + \wwc \,\Dvg[\zk]{\repNs} \hmg{\geps}
    + \replm_{,\zk} \gzetat^\zk
   \,, \\
  \end{array}
  \right.
 %% Note that in 2024 sign convention for $\gchi$ changed to $-gchi$
 %% Checked by calcs in the Section \ref{ASSect:GaugeInvAMG'}
 \eeq
where partial derivatives act to the right on the next factor or the expression in brackets, and $\Dvg{v}\defeq\Dvgp[\zn]{v}$ stands for a divergence of a spatial vector field.
}\fi

Examining (\ref{gauge_transfs_ham_AMG''_metric}, \ref{gauge_transfs_ham_AMG''_aux}) reveals that the average-free part of $\geps$ affects \emph{only} $\gvar[\geps,\gzeta]\tf^\io$ and $\gvar[\geps,\gzeta]\Dvg{\lmptf}$, canceling in the variation of $\tauxdot{\tf}^\io {\,+\,} \Dvg{\lmptf}$. In the metric-sector transformations (\ref{gauge_transfs_ham_AMG''_metric}), $\hmg{\geps}$\HIDE{ consistently} appears multiplied by $\Fw$, except in $\gvar[\geps,\gzeta] \repNl\!$.
However, the last two terms in $\gvar[\geps,\gzeta] \repNl\!$ can be reorganized:
 \beq{fasfasrcsrg}
    \wwc \big(\Dvg{\repNs} - \tfrac{1}{\Ddim{-}2}
     \HIDE{\Fw\replm} \repNl
     \tfrac{1}{\sqrg} \trpg \big) \Fw \hmg{\geps}
    \,=\,
    \big(\tTDer{}{\taux}{\Fw}
        -  \repNs^\zn
            \partial_\zn \Fw %%%{\Fw}_{,\zm}
    \big) \hmg{\geps}
    - \big(\,\dotuline{{ \tTDer{}{\taux}{\Fw} -   \PB{\Fw}{\sint \Hl}\repNl
    -  \PB{\Fw}{\sint \Hs_\zn} \repNs^\zn \big)\,\hmg{\geps} }} %%^{e.o.m.}\!
    \,,
    \quad
    \vspace{-1mm}
  \eeq
which factors out $\Fw\hmg{\geps}$ combinations in $\gvar[\geps,\gzeta] \repNl\,$, while simultaneously extracting a \emph{trivial} gauge transformation \cite{Henneaux:1992ig},\HIDE{ the underlined term,} proportional to equations of motion:
  \bea{}
    \gvar[\geps,\gzeta]{\repNl}
  %   \equiv \gvar ( \Fw \replm )
     &=& \tTDer{}{\taux} (\Fw \hmg{\geps})
     + \partial_\zn \repNl %%%\repNl_{\!\!\!\!,\zn}
     \gzetat^\zn %% {inh.}
     - \repNs^\zn \partial_\zn \Fw %%%{\Fw}_{,\zm}
     \hmg{\geps}
     -  %%\dotuline{e.o.m.}
     \dotuline{{ \tfrac12 \wwc\Fw\g^{\zm\zn} \tVDer{\SSS^{\iwwc}_{par''}}{\pg^{\zm\zn}}\, \hmg{\geps}}}
     \,.
   \label{gvarNl_with_triv_g.t.}
 \eea
The trivial term reflects a change in the basis of gauge generators (with a dual change in parameters). To consistently discard it from $\gvar[\geps,\gzeta] \repNl$, one must also subtract the corresponding trivial term from $\gvar[\geps,\gzeta] \pg^{\zn\zm}$:
 %The new transformations turns out to be canonical transformations in some other constraint basis.
  \bea{gvarpg_with_triv_g.t.}
     \gvar[{\geps},{\gzeta}] \pg^{\zn\zm}
     &\!=&
     %%\gvar[\FF{\geps},{\gzeta}]^{^{(\GR)}} \pg^{\zn\zm} +
     \PB{\pg^{\zn\zm}}{\sint \gzetat^{\zn} \Hs_\zn} +
     \PB{\pg^{\zn\zm}}{\sint \Hl} \Fw \hmg{\geps} +
     \dashuline{\PB{\pg^{\zn\zm}}{\sint \Fw} \Hl \hmg{\geps}}
     + \PB{\pg^{\zn\zm}}{\sint \dashuline{\,\Fw\,} \Fw^{-1}\CCof} \hmg{\geps}
    \nonumber \\
     &&
     \;=\;
      \PB{\pg^{\zn\zm}}{\sint (\Hl {+} \Fw^{-1}\CCof)} \Fw \hmg{\geps} + \PB{\pg^{\zn\zm}}{\sint \Hs_\zn} \gzetat^\zn
       + \HIDE{e.o.m.'}
       \dashuline{{ \tfrac12\wwc\Fw \g^{\zm\zn}\, \tVDer{\SSS^{\iwwc}_{par''}}{\repNl}\,\hmg{\geps} }} \,.
      \nonumber
  \eea
  %% This leads to a canonical form of the\HIDE{ gauge} transformation in which, instead of the generator $\ptf_\io + \hmg{\Fw \Hl}$, one uses the generator $\Hl + \Fw^{-1} \ptf_\io$, with the (a posteriori) constrained gauge parameter $\gxl \to \Fw \hmg{\geps}$.
We added the identity $0 = \PB{\pg^{\zn\zm}}{\sint \CCof } \hmg{\geps} = \PB{\pg^{\zn\zm}}{\sint \Fw^{-1} \CCof}{\,\Fw\,}  \hmg{\geps} + \PB{\pg^{\zn\zm}}{\sint \Fw } \Fw^{-1} \CCof \,\hmg{\geps}$
to extract a required trivial gauge transformation $\propto \tVDer{\SSS^{\iwwc}_{par''}}{\repNl}$. This leads to a\HIDE{ canonical} transformation generated by $\sint \gxl(\Hl {+\,} \Fw^{-1} \CCof)$
 %instead of $\ptf_\io {\,+\,} \hmg{\Fw \Hl}$,
with the (a posteriori) constrained gauge parameter $\gxl \to \Fw \hmg{\geps}$.

Similarly, the transformation $\gvar[\geps,\gzeta] \repNs^\zn$ can be cast to the (restricted) {\GR} form (\ref{GR_canon_gauge_transfs}) by noticing that
 %Analogously one can cast $\gvar[\geps,\gzeta] \repNs^\zn$ to the form of the restricted {\GR} transformations (\ref{GR_canon_gauge_transfs})\HIDE{ (with the same constraints of \GR\ canonical gauge parameter)} if one notice
$
  \Fw \g^{\zn\zm} \partial_\zm (\Fw^{-1}\repNl) %%%(\Fw^{-1}\repNl)_{,\zm}
  \Fw\hmg{\geps}
  =
  - \g^{\zn\zm} \big( \repNl \ader_\zm \Fw\hmg{\geps}\big)
$, which gives
  \bea{}
    \gvar[\geps,\gzeta] \repNs^\zn
    &=& \tTDer{}{\taux} \gzetat^{\zn}
       - \LieB{\zn}{\repNs}{\gzetat}
       - (\repNl \ader_\zm \Fw\hmg{\geps}) \g^{\zm\zn}
       \,.
   \label{gvarNs}
  \eea

Thus, omitting trivial transformations, for the metric fields one gets the new canonical form of the symmetry with field-dependent gauge parameters $\Fw\hmg{\geps}$ and $\gzetat$:
 \beq{gauge_transfs_ham_AMG''_fin}
  \begin{array}{|lll}
   \gvarp[\geps,\gzeta] \g_{\zn\zm}
   &=& \big( \g_{\zk\zm} \partial_\zn %\gzeta_{\isst}^{\zk}
   + \g_{\zn\zk} \partial_\zm %\gzeta_{\isst}^{\zk}
   + (\partial_\zk \g_{\zn\zm}) %\gzeta_{\isst}^{\zk}
   \big) %(
   \gzetat^{\zk}
   + \tfrac{2}{\sqrg}\left(\pg_{\zm\zn} -\frac{1}{\Ddim{-}2}\,\trpg\,\g_{\zm\zn} \right)
   \Fw\hmg{\geps} ,
   \\
   \gvarp[\geps,\gzeta] \pg^{\zn\zm}  %% fin
          &=&
      \gvar[\Fw\hmg{\geps}]^{^{(\GR)}} \pg^{\zn\zm}
      +\PB{\pg^{\zn\zm}}{\sint \Fw^{-1}\ptf_\io} \Fw \hmg{\geps}
      \,, \vphantom{\big|_I^I}
   \\
   \if{
   \gvar[\geps,\gzeta] \CCf
   &=&
      -
       (\wwc{+}1)\tfrac{1}{\sqrg} \sqrg_{,\zk}
        \CCf\, \gzeta_{\isst}^{\zk}
      +  \tfrac{1}{\Ddim{-}2} (\wwc{+}1) \tfrac{\trpg}{\sqrg} \CCf\, \Fw\hmg{\geps}
    \,,
   \\
    }\fi
   \gvarp[\geps,\gzeta]{\repNl} %% fin
     &=& \tTDer{}{\taux} (\Fw \hmg{\geps})
     + \partial_\zn \repNl %%%\repNl_{\!\!\!\!,\zn}
        \gzetat^\zn %% {inh.}
     - \repNs^\zn \partial_\zn \Fw %%%{\Fw}_{,\zn}
     \hmg{\geps}
   \,,  \vphantom{\big|^I}
   \\
   \gvarp[\geps,\gzeta] \repNs^\zn %% fin
    &=& \tTDer{}{\taux} \gzetat^{\zn}
       - \LieB{\zn}{\repNs}{\gzetat}
       - (\repNl \ader_\zm \Fw\hmg{\geps}) \g^{\zm\zn}
    \,.
   \\
   \if{
   \gvar[\geps,\gzeta] \tf^\io &=&
   %%% \hmg{\geps} - \Dvg{\gchi}
   %%% % In the basis of canonical transformations (\ref{canon_transf_AMG'_App})
   %%% \equiv
   \geps,
   \\
   \gvar[\geps,\gzeta] \Dvg{\lmptf} &=&
   %%% + \TDer{}{\taux} \Dvg{\gchi} +  \replm_{,\zm} \gzetat^\zm + \Dvg{\repNs}\wwc \hmg{\geps}
   %%% % In the basis of canonical transformations (\ref{canon_transf_AMG'_App})
   %%% =
    - \TDer{}{\taux} \inh{\geps} +  (\Fw^{-1}\repNl)_{,\zn} \gzetat^\zn + \Dvg{\repNs}\wwc \hmg{\geps}
   \;,
   \\
   }\fi
  \end{array}
  \eeq
These modified\HIDE{ canonical} transformations are induced by the canonical generator (\ref{canon_generator_AMG''_fin}),
%%(\ref{gauge_transfs_ham_AMG'})
  \beq{canon_generator_AMG''_fin_App}
    \gvarp[\geps,\gzeta] {(\,.\,)}
    = \PB{\,.\,}{\sint ( \Hl {+} \Fw^{-1} \CCof )} \Fw\hmg{\geps}
    + \PB{\,.\,}{\sint \Hs_\zn} \gzetat^\zn
    + \PB{\,.\,}{\sint \inh{\CCof} } \inh{\geps}
    %% + \PB{\,.\,}{\sint \ptf_{\io,\zm}} \gchi^\zm
    \,.
  \eeq
which arises from the constraint set of the action (\ref{paramAction_AMG_fin}).
%where the third generator term\HIDE{ on the right-hand side} with parameter $\inh{\geps}\equiv\Dvg{\gchi}$ acts solely on the auxiliary-sector fields.
These transformations for the (Lagrangian) metric fields $\g_{\zn\zm},\repNl,\repNs^\zn$ coincide with those in general relativity, (\ref{GR_canon_gauge_transfs}), under the parameter restrictions $\gxl \to \Fw \hmg{\geps}$, $\,\gxs^\zn \to \gzetat^\zn$ (\ref{wGUMG_gauge_restriction_from_GR}), although $\gvarp[\geps,\gzeta] \pg^{\zn\zm}$ contains an extra non-GR term proportional to $\ptf_\io$.

 %%Additional motivation to explicate the transition from the transformations (\ref{gauge_transfs_ham_AMG''_metric}) based on generator with $\sint \Fw\Hl \geps$ to (\ref{gauge_transfs_ham_AMG''_fin}) induced by generator (\ref{canon_generator_AMG''_fin_App}) with $\PB{\,.\,}{\sint ( \Hl {+} \Fw^{-1} \ptf_{\io} )} \Fw\hmg{\geps}$ was to explicate the trivial gauge transformation, extracted in  (\ref{gvarNl_with_triv_g.t.}) and (\ref{gvarpg_with_triv_g.t.}), which differ these transformations. It is interesting from the perspective of the restricted theory approach \cite{Barvinsky:2022guw}, because the existence and form of the trivial transformation suggests the possible source of the breakdown of the applicability to the {\GUMG} case of the general theorem on  reduction of the closed gauge algebras for the restricted theories \cite{Barvinsky:2022guw}.

The discussion of the transition from transformations %% (\ref{gauge_transfs_ham_AMG''_metric})
with generator $\PB{\,.\,}{\sint(\Fw\Hl {+} \CCof)}\geps$, to symmetries with
 %%(\ref{gauge_transfs_ham_AMG''_fin}), induced by (\ref{canon_generator_AMG''_fin_App}) with
$\PB{\,.\,}{\sint ( \Hl {+} \Fw^{-1} \CCof )} \Fw{\geps}$, was additionally motivated by the need to clarify the form of the trivial gauge transformation isolated in (\ref{gvarNl_with_triv_g.t.}) and (\ref{gvarpg_with_triv_g.t.}). This distinction of the two transformations is significant in the context of the restricted-theory approach \cite{Barvinsky:2022guw}, as the existence and form of the  {\wGUMG} trivial transformation (which is not a trivial gauge transformation in the parental theory) suggests a potential source of the breakdown in applying the general theorem on restriction of the closed gauge algebras \cite{Barvinsky:2022guw} for {\GUMG} cases.

%These subtracted transformation is trivial with respect to {\wGUMG} set of equations of motion and its subtraction is legitimate in {\wGUMG}. This while from the perspective of the parental theory (\GR)

  \subsection{{\GUMG} unrestricted canonical variation}
   \label{ASSect:gauge_calcs_GUMG}
    \hspace{\parindent}
    \newcommand{\hmgd}[1]{{\hspace{1pt}\hmg{{#1}}\vphantom{1^{l^l}}\hspace{1pt}}}
\!\!The anomaly (\ref{anomaly_GUMG_init}) of the unrestricted canonical variation of the {\GUMG} action is given by
  \bea{anomaly_GUMG_init_COPY}
   \!\!- A
    &\!=\!&
    \HIDE{+}
   %  \PB{\sint (\repNl \Hl {\,+\,} \repNs^\zn\Hs_\zn)\,}{\sint \gxl \FF^{-1}\opinvE \hmgs{\CCof}}
    \!\sint
      %(
    \gxs^\zn {\repNl}_{\!\!\!\!,\zn}
   %  {\,-\,}
   %  \sint\repNs^\zn\gxl_{,\zn}
     %)
     \FF^{-1}\opinvE \hmgs{\CCof}
  % \nonumber\\
  %  && \quad
     +
     \PB{\sint \repNl \FF^{-1}\opinvE \hmgs{\CCof}\,}{\sint (\gxl \Hl {\,+\,} \gxs^\zn \Hs_\zn)}
   %% \,\RiB
    \,-\,
    \big(\repNl\!,\repNs^\zn \leftrightarrow  \gxl\!, \gxs^\zn\big)
    ,
    \quad
   \eea
where all gauge parameters are yet unconstrained functions.

 \newpar

Consider the first contribution in (\ref{anomaly_GUMG_init_COPY}) coming from the uncompensated part of unconstrained gauge variation of $\repNl$. Integrating by parts and using $\opinvE \!\hmgs{\CCof} \!\teq \WW^{-1}\!  \tfrac{\hmgs{\CCof}}{\hmg{\WW^{-1}}}$, (\ref{opinvE_properties}), imply
  \bea{qrwrwettrw_varNl}
     &&
     \HIDE{+}
     \sint  \gxs^\zn {\repNl}_{\!\!\!\!,\zn} \FF^{-1} \opinvE \hmgs{\CCof}
     %%
     %% - \sint \repNs^\zn\gxl_{,\zn} \FF^{-1}\opinvE \hmgs{\CCof}
     \,-\,
      \big(\repNl\!,\repNs^\zn \leftrightarrow  \gxl\!, \gxs^\zn\big)
    %\nonumber
    \\
  \if{
    &=&
    \HIDE{+} \sint \gxs^\zn \repNl_{,\zn} \FF^{-1}
    \WW^{-1}\tfrac{\hmgs{\CCof}}{\hmg{\WW^{-1}}}
    %%\opinvE \hmgs{\CCof}
    - \sint \repNs^\zn\gxl_{,\zn} \FF^{-1}
    \WW^{-1}\tfrac{\hmgs{\CCof}}{\hmg{\WW^{-1}}}
    %%\opinvE \hmgs{\CCof}
    \nonumber\\
  }\fi
  \if{
    &=&
    \HIDE{+}
    \sint \repNl_{,\zn} \FF^{-1}
    \WW^{-1} \gxs^\zn
    {\cdot}\tfrac{\hmgs{\CCof}}{\hmg{\WW^{-1}}}
    %%\opinvE \hmgs{\CCof}
    - \sint\gxl_{,\zn} \FF^{-1}
    \WW^{-1} \repNs^\zn
    {\cdot}\tfrac{\hmgs{\CCof}}{\hmg{\WW^{-1}}}
    %%\opinvE \hmgs{\CCof}
    \nonumber\\
    &=&
    \HIDE{+} \sint \partial_\zn \big( \repNl \FF^{-1}
    \WW^{-1} \gxs^\zn \big)
    {\cdot}\tfrac{\hmgs{\CCof}}{\hmg{\WW^{-1}}}
    - \sint \partial_\zn \big( \gxl \FF^{-1}
    \WW^{-1} \repNs^\zn \big)
    {\cdot}\tfrac{\hmgs{\CCof}}{\hmg{\WW^{-1}}}
    \nonumber\\
    &&
    - \sint \repNl \partial_\zn \big(\FF^{-1}
    \WW^{-1} \big) \gxs^\zn
    {\cdot}\tfrac{\hmgs{\CCof}}{\hmg{\WW^{-1}}}
    + \sint \gxl \partial_\zn \big(\FF^{-1}
    \WW^{-1} \big) \repNs^\zn
    {\cdot}\tfrac{\hmgs{\CCof}}{\hmg{\WW^{-1}}}
    \nonumber\\
    &&
    - \sint \repNl \FF^{-1}
    \WW^{-1} \Dvg{\gxs}
    {\cdot}\tfrac{\hmgs{\CCof}}{\hmg{\WW^{-1}}}
    + \sint \gxl \FF^{-1}
    \WW^{-1} \Dvg{\repNs}
    {\cdot}\tfrac{\hmgs{\CCof}}{\hmg{\WW^{-1}}}
    \nonumber\\
  }\fi
  %  &=&
  %  t.d. +
  %  \nonumber\\
    &&
    =\;
    \HIDE{+} \sint \repNl \FF^{-1}
    \WW^{-1}  \big( \partial_\zn {\ln\FF} {+} \partial_\zn {\ln\WW} \big)
      \gxs^\zn
      {\,} \tfrac{\hmgs{\CCof}}{\hmg{\WW^{-1}}}
   %   - \sint \gxl \FF^{-1}
   %   \WW^{-1} \big( \partial_\zn {\ln\FF} {+} \partial_\zn {\ln\WW} \big)
   % \repNs^\zn
   % {\,} \tfrac{\hmgs{\CCof}}{\hmg{\WW^{-1}}}
   % \nonumber\\
   % &&
    - \sint \repNl \FF^{-1}
    \WW^{-1} \Dvg{\gxs}
    {\,} \tfrac{\hmgs{\CCof}}{\hmg{\WW^{-1}}}
   % + \sint \gxl \FF^{-1}
   % \WW^{-1} \Dvg{\repNs}
   % {\,} \tfrac{\hmgs{\CCof}}{\hmg{\WW^{-1}}}
    %
     \,-\,
      \big(\repNl\!,\repNs^\zn \leftrightarrow  \gxl\!, \gxs^\zn\big)
    \,,
   \nonumber %\\
  \eea
where we used $\partial_\zn \big(\FF^{-1}\WW^{-1} \big) \teq \FF^{-1} \WW^{-1} \big({-} \partial_\zn {\ln\FF} {-} \partial_\zn {\ln\WW}\big)$.

\newpar

The Poisson bracket with $\Hs_\zn$ in (\ref{anomaly_GUMG_init_COPY}) varies the phase-space variables (here it is only $\g_{\zm\zn}$) under spatial diffeomorphisms, and yields two contributions, corresponding to varying $\FF^{-1}$ and $\opinvE$:
 $ % \bea{qrwrwettrw1_Hs}
   %  &&
    \HIDE{+}
     \PB{\sint \repNl \FF^{-1}\opinvE \hmgs{\CCof}\,}{\sint \gxs^\zn \Hs_\zn }
   % \,-\,
   % \big(\repNl \leftrightarrow  \gxl\big)
   %% \big(\repNl\!,\repNs^\zn \leftrightarrow  \gxl\!, \gxs^\zn\big)
   %% \nonumber\\
  \if{
    &=&
    - \BraKet{ \repNl \FF^{-1} \opGm}{ \gxs^\zn_{\,;\zn} {} \tTDer{\ln\WW}{\ln\sqrg} }
    \,\tfrac{\hmgs{\CCof}}{\hmg{\WW^{-1}}}
    + \BraKet{\gxl \FF^{-1} \opGm}{ \repNs^\zn_{\,;\zn} {} \tTDer{\ln\WW}{\ln\sqrg} }
    \,\tfrac{\hmgs{\CCof}}{\hmg{\WW^{-1}}}
    \nonumber\\
    &&
    - \BraKet{ \repNl \FF^{-1} \WW \WW^{-1}}{ \gxs^\zn_{\,;\zn} {} }
    \,\tfrac{\hmgs{\CCof}}{\hmg{\WW^{-1}}}
    + \BraKet{\gxl \FF^{-1} \WW \WW^{-1}}{ \repNs^\zn_{\,;\zn} {} }
    \,\tfrac{\hmgs{\CCof}}{\hmg{\WW^{-1}}}
    \nonumber\\
  }\fi
  %  &&
   =
    - \sint { \repNl \FF^{-1} \WW^{-1} \gxs^\zn_{\,;\zn} {} }
     \WW {} \tfrac{\hmgs{\CCof}}{\hmgd{\WW^{-1}}}
   %  + \sint { \gxl \FF^{-1} \WW \repNs^\zn_{\,;\zn} {} }
   %   \, \WW^{-1}\!  {} \tfrac{\hmgs{\CCof}}{\hmgd{\WW^{-1}}}
 %   \nonumber\\
 %   &&
    - \sint { \repNl \FF^{-1} \opGm}{\, \gxs^\zn_{\,;\zn} {} \tTDer{\ln\WW}{\ln\sqrg} }   \,{} \tfrac{\hmgs{\CCof}}{\hmgd{\WW^{-1}}}
   %  + \sint {\gxl \FF^{-1} \opGm}{\, \repNs^\zn_{\,;\zn} {} \tTDer{\ln\WW}{\ln\sqrg} }   \,{} \tfrac{\hmgs{\CCof}}{\hmgd{\WW^{-1}}}
    %\nonumber\\
  \,.
  $ % \eea
The canonical action of $\Hs_\zn$ on $\FF^{-1}$ is taken from (\ref{PB_sqrg_Hs}). The variation of the\HIDE{ nonlocal} operator $\opinvE$ gives rise to the operator $\opGm$, where we used (\ref{Corollary:invE_PB_Hl_Hs_statement}) from Corollary \ref{Corollary:invE_PB}.
 %The first line on the right-hand side comes from action of the Poisson brackets on $\FF^{-1}(\argsqrg)$ in the left argument, while the second line is the result of acting on $\opinvE$ in the left argument, which gives rise to $\opGm\,$.

Decomposing $\opGm$ into a difference of a local part and a ``longitudinal'' projector form: $\sint f \opGm h \teq \sint f \WW^{-1} h - \sint f \WW^{-1} \frac1{\sint \WW^{-1}} \sint \WW^{-1} h,$ and noting that $\eta^\zn_{\,;\zn} \teq \sqrg^{-1}\partial_{\zn}(\sqrg\eta^\zn)$, allows one to write  $\,\eta^\zn_{\,;\zn} {} \WW \teq \eta^\zn\partial_\zn \ln \FF + \Dvg{\eta} \WW\,$ and $\,\eta^\zn_{\,;\zn} {} \tTDer{\ln\WW}{\ln\sqrg}  \teq  \eta^\zn\partial_\zn \ln \WW {\,+\,} \Dvg{\eta} \tTDer{\ln\WW}{\ln\sqrg}$, so that
  \bea{qrwrwettrw2_Hs}
    &&
    \HIDE{+}
     \PB{\sint \repNl \FF^{-1}\opinvE \hmgs{\CCof}\,}{\sint \gxs^\zn \Hs_\zn }
    \,-\,
     \big(\repNl\!,\repNs^\zn \leftrightarrow  \gxl\!, \gxs^\zn\big)
    %\nonumber
    \\
    &&=\,
    - \sint { \repNl \FF^{-1} \WW^{-1} {\,}(\gxs^\zn\partial_\zn \ln \FF {\,+\,}\Dvg{\gxs} \WW) }
     \tfrac{\hmgs{\CCof}}{\hmgd{\WW^{-1}}}
     % + \sint { \gxl \FF^{-1} \WW^{-1} {\,} (\repNs^\zn\partial_\zn \ln \FF {\,+\,}\Dvg{\repNs} \WW) }
     % \tfrac{\hmgs{\CCof}}{\hmgd{\WW^{-1}}}
   % \nonumber\\
   % &&
    \,-\,
     \sint { \repNl \FF^{-1} \WW^{-1} {\,} ( \gxs^\zn\partial_\zn \ln \WW {\,+\,} \Dvg{\gxs} \tTDer{\ln\WW}{\ln\sqrg} )} \,{} \tfrac{\hmgs{\CCof}}{\hmgd{\WW^{-1}}}
     % + \sint {\gxl \FF^{-1} \WW^{-1} {\,} ( \repNs^\zn\partial_\zn \ln \WW {\,+\,} \Dvg{\repNs} \tTDer{\ln\WW}{\ln\sqrg} )} \,{}
     % \tfrac{\hmgs{\CCof}}{\hmgd{\WW^{-1}}}
    \nonumber\\
    &&
    \quad + \sint {\repNl \FF^{-1} \WW^{-1}} {\tfrac1{\sint \WW^{-1}}} {\sint \WW^{-1}
    ( \gxs^\zn\partial_\zn \ln \WW {\,+\,} \Dvg{\gxs} \tTDer{\ln\WW}{\ln\sqrg} ) }     \,{} \tfrac{\hmgs{\CCof}}{\hmgd{\WW^{-1}}}
     % - \sint {\gxl \FF^{-1} \WW^{-1}} {\tfrac1{\sint \WW^{-1}}} {\sint \WW^{-1} ( \repNs^\zn\partial_\zn \ln \WW {\,+\,} \Dvg{\repNs} \tTDer{\ln\WW}{\ln\sqrg} ) }  \,{}
     % \tfrac{\hmgs{\CCof}}{\hmgd{\WW^{-1}}}
   %
    \:-\:
     \big(\repNl\!,\repNs^\zn \leftrightarrow  \gxl\!, \gxs^\zn\big)
   \,.
   \nonumber
  \eea

\newpar
Summing the two contributions (\ref{qrwrwettrw_varNl}) and (\ref{qrwrwettrw2_Hs}), after cancellations, the residual terms are gathered in an integral structure with the $\opGm$ operator. In the local parts, terms proportional to $\partial_\zn \ln\FF$ and $\partial_\zn \ln\WW$ cancel, leaving only divergence terms with divergences, with coefficients summing to the combination $\OOmega$, (\ref{def_OOmega_TTheta}).
In the nonlocal part of (\ref{qrwrwettrw2_Hs}), integrating by parts in the first term yields:  $\sint \WW^{-1} \big( \eta^\zn\partial_\zn \ln \WW {\,+\,} \Dvg{\eta} \tTDer{\ln\WW}{\ln\sqrg} \big) = \sint \WW^{-1}\big(\Dvg{\eta} {\,+\,} \Dvg{\eta} \tTDer{\ln\WW}{\ln\sqrg} \big)$
to which we can add a total derivative $\sint \WW^{-1} \Dvg{\eta} \WW$ to factor out $\OOmega$: $\sint \WW^{-1} \Dvg{\eta} \OOmega$.
 % Summing\HIDE{ these} two contributions (\ref{qrwrwettrw_varNl}) and (\ref{qrwrwettrw2_Hs}) after certain cancellations the residual terms will gather in an integral structure with $\opGm$ operator. In local parts terms proportional to $ \partial_\zn {\ln\FF} \HIDE{\gxs^\zn}$ and $ \partial_\zn {\ln\WW} \HIDE{\gxs^\zn}$ cancel, leaving only terms, proportional to divergences\HIDE{ $\Dvg{\gxs}$}, coefficients at which sums to combination $\OOmega$, (\ref{def_OOmega_TTheta}).
 % In the nonlocal part of (\ref{qrwrwettrw2_Hs}), after integrating by parts of the first term, $\sint \WW^{-1} \big( \eta^\zn\partial_\zn \ln \WW {\,+\,} \Dvg{\eta} \tTDer{\ln\WW}{\ln\sqrg} \big)$ can be represented as $\sint \WW^{-1}\big(\Dvg{\eta} {\,+\,} \Dvg{\eta} \tTDer{\ln\WW}{\ln\sqrg} \big)$ to which one can add the total spatial derivative term $ \sint \WW^{-1} \Dvg{\eta} \WW$ to gather the $\OOmega$ factor: $\sint \WW^{-1} \Dvg{\eta} \OOmega$.
 %
Thus, the sum of (\ref{qrwrwettrw_varNl}) and (\ref{qrwrwettrw2_Hs}) becomes
  \bea{qrwrwettrw_varNl_Hs}
    &&
    \HIDE{+}
     \sint  \gxs^\zn {\repNl}_{\!\!\!\!,\zn} \FF^{-1} \opinvE \hmgs{\CCof}
    +
     \PB{\sint \repNl \FF^{-1}\opinvE \hmgs{\CCof}\,}{\sint \gxs^\zn \Hs_\zn }
    \,-\,
    \big(\repNl\!,\repNs^\zn \leftrightarrow  \gxl\!, \gxs^\zn\big)
    %\nonumber
    \\
  \if{
    &=&
    - \sint { \repNl \FF^{-1} \WW^{-1} {\,}( \cancel{\gxs^\zn\partial_\zn \ln \FF} {\,+\,}\Dvg{\gxs} \WW) }
     \tfrac{\hmgs{\CCof}}{\hmgd{\WW^{-1}}}
    + \sint { \gxl \FF^{-1} \WW^{-1} {\,} ( \bcancel{\repNs^\zn\partial_\zn \ln \FF} {\,+\,}\Dvg{\repNs} \WW) }
     \tfrac{\hmgs{\CCof}}{\hmgd{\WW^{-1}}}
    \nonumber\\
    &&
    - \sint { \repNl \FF^{-1} \WW^{-1} {\,} ( \cancel{\gxs^\zn\partial_\zn \ln \WW} {\,+\,} \Dvg{\gxs} \tTDer{\ln\WW}{\ln\sqrg} )} \,{} \tfrac{\hmgs{\CCof}}{\hmgd{\WW^{-1}}}
    + \sint {\gxl \FF^{-1} \WW^{-1} {\,} ( \bcancel{\repNs^\zn\partial_\zn \ln \WW} {\,+\,} \Dvg{\repNs} \tTDer{\ln\WW}{\ln\sqrg} )} \,{} \tfrac{\hmgs{\CCof}}{\hmgd{\WW^{-1}}}
    \nonumber\\
   &&
    + \sint {\repNl \FF^{-1} \WW^{-1}} {\tfrac1{\sint \WW^{-1}}} {\sint \WW^{-1}
      ( \gxs^\zn\partial_\zn \ln \WW {\,+\,} \Dvg{\gxs} \tTDer{\ln\WW}{\ln\sqrg} ) }     \,{} \tfrac{\hmgs{\CCof}}{\hmgd{\WW^{-1}}}
    - \sint {\gxl \FF^{-1} \WW^{-1}} {\tfrac1{\sint \WW^{-1}}} {\sint \WW^{-1}
      ( \repNs^\zn\partial_\zn \ln \WW {\,+\,} \Dvg{\repNs} \tTDer{\ln\WW}{\ln\sqrg} ) }    \,{} \tfrac{\hmgs{\CCof}}{\hmgd{\WW^{-1}}}
    \nonumber\\
    &+&
    t.d.1 +
    \nonumber\\
    &&
    {+} \sint \repNl \FF^{-1}
    \WW^{-1}  \big(\cancel{ \partial_\zn {\ln\FF} {+} \partial_\zn {\ln\WW} }\big)
    \gxs^\zn
    {} \tfrac{\hmgs{\CCof}}{\hmg{\WW^{-1}}}
    - \sint \gxl \FF^{-1}
    \WW^{-1} \big(\bcancel{ \partial_\zn {\ln\FF} {+} \partial_\zn {\ln\WW} }\big)
    \repNs^\zn
    {} \tfrac{\hmgs{\CCof}}{\hmg{\WW^{-1}}}
    \nonumber\\
    &&
    - \sint \repNl \FF^{-1}
    \WW^{-1} \Dvg{\gxs}
    {} \tfrac{\hmgs{\CCof}}{\hmg{\WW^{-1}}}
    + \sint \gxl \FF^{-1}
    \WW^{-1} \Dvg{\repNs}
    {} \tfrac{\hmgs{\CCof}}{\hmg{\WW^{-1}}}
    \nonumber\\
  }\fi
   && =\;
    - \sint \repNl \FF^{-1} \WW^{-1} \Dvg{\gxs}\, \OOmega
    {\,} \tfrac{\hmgs{\CCof}}{\hmg{\WW^{-1}}}
     %  + \sint \gxl \FF^{-1} \WW^{-1} \Dvg{\repNs} \OOmega
     %  {} \tfrac{\hmgs{\CCof}}{\hmg{\WW^{-1}}}
  %%  \nonumber\\
  %%  &&
    + \sint {\repNl \FF^{-1} \WW^{-1}} {\tfrac1{\sint \WW^{-1}}} {\sint \WW^{-1}
      \Dvg{\gxs} \,\OOmega} {\,} \tfrac{\hmgs{\CCof}}{\hmgd{\WW^{-1}}}
      %  - \sint {\gxl \FF^{-1} \WW^{-1}} {\tfrac1{\sint \WW^{-1}}} {\sint \WW^{-1}
      %  \Dvg{\repNs} \OOmega} \,{} \tfrac{\hmgs{\CCof}}{\hmgd{\WW^{-1}}}
   % \nonumber\\
   % &+&  t.d.1 + t.d.2
    \,-\,
    \big(\repNl\!,\repNs^\zn \leftrightarrow  \gxl\!, \gxs^\zn\big)
    \nonumber\\
   && =\;
    - \sint \repNl \FF^{-1}   \opGm \Dvg{\gxs}\, \OOmega
    {\,} \tfrac{\hmgs{\CCof}}{\hmg{\WW^{-1}}}
    %  + \sint \gxl \FF^{-1}   \opGm \Dvg{\repNs} \OOmega
    %  {} \tfrac{\hmgs{\CCof}}{\hmg{\WW^{-1}}}
    \,-\,
    \big(\repNl\!,\repNs^\zn \leftrightarrow  \gxl\!, \gxs^\zn\big)
    \nonumber
   % \\
   % &+&  t.d.1 + t.d.2
    \,.
  \eea

\newpar
The cancellations and appearance of a compact structure proportional to the divergence of the parameters of spatial diffeomorphism transformation is not accidental. These two anomaly terms correspond to the action of spatial Lie derivatives:
  \bea{Lie_varNl_Hs}
    %&&
    \HIDE{+}
     \sint  \gxs^\zn {\repNl}_{\!\!\!\!,\zn} \FF^{-1} \opinvE \HIDE{ \hmgs{\CCof} }
    +
     \PB{\sint \repNl \FF^{-1} \opinvE \HIDE{ \hmgs{\CCof}} {\,}}{\sint \gxs^\zn \Hs_\zn}
 %   \,-\,
 %   \big(\repNl\!,\repNs^\zn \leftrightarrow  \gxl\!, \gxs^\zn\big)
   % \nonumber
  %  \\
  %  &&
    \;=\;
    \HIDE{+}
     \sint \LieD{\gxs}{(\repNl)}\, \FF^{-1} \opinvE \HIDE{ \hmgs{\CCof} }
    +
     \sint \repNl\, \LieD{\gxs} {(\FF^{-1} \opinvE)} \HIDE{ \hmgs{\CCof} }
 %   \,-\,
 %   \big(\repNl\!,\repNs^\zn \leftrightarrow  \gxl\!, \gxs^\zn\big)
   % \nonumber\\
   % &=&
 %  \;=\;
 %   \HIDE{+}
 %    \sint \LieD{\gxs}{\big(\repNl \FF^{-1} \opinvE \big)} \HIDE{ \hmgs{\CCof} }
 %   \,-\,
 %   \big(\repNl\!,\repNs^\zn \leftrightarrow  \gxl\!, \gxs^\zn\big)
 % \nonumber
  \,,
  \eea
and the nonvanishing contribution\HIDE{ which is not a total derivative} of $\sint \LieD{\gxs}{\big(\repNl \FF^{-1} \opinvE \big)}$ just reflects the noncovariance of the $(\repNl)\, \FF^{-1} \opinvE $ structure under unrestricted spatial diffeomorphisms.
  % If the derivative acted on the spatial covariant object (scalar density) this would be total derivative. Possibly non-null term just measures the anomaly. \TODO{Check, the Lie derivative does not depend on connection (for tensors). However it feels the tensorial structure.}

\newpar

The third contribution --- the Poisson bracket with the $\Hl$ structure --- is the canonical transformation of $(\repNl)\, \FF^{-1} \opinvE $ generated by $\Hl$. It enters the anomaly in the antisymmetrized combination of  $\repNl$ and $\gxl$. This antisymmetrization cancels the canonical action of $\Hl$ on $\FF^{-1}$:\,  $\PB{ \sint \repNl \FF^{-1} f}{\sint \gxl \Hl}\big|_{f=...} \!\!\! -\, \big(\repNl \! \leftrightarrow  \gxl\big) = 0 $,
as $\g_{\zn\zm}$ in $\FF^{-1}$ and $\pg^{\zn\zm}$ in $\Hl$ enter ultralocally, so $\repNl$ and $\gxl$ appear in symmetric combination $\repNl\gxl$ and cancel upon antisymmetrizing.
Thus, only the canonical action of $\opinvE$ remains nontrivial:
 $ % \beq{qrwrwettrw1_Hl}
    \HIDE{+}
    \PB{\sint \repNl \FF^{-1}\opinvE \hmgs{\CCof}\,}{\sint \gxl \Hl }
    \,-\,
    \big(\repNl \!\leftrightarrow  \gxl\big)
    \nonumber\\
    = \HIDE{+} \PB{\sint f \opinvE \hmgs{\CCof}\,}{\sint \gxl \Hl }\big|_{f=\repNl \FF^{-1}} \!\!\!\!
     - \, \PB{\sint f \opinvE \hmgs{\CCof}\,}{\sint \repNl \Hl }\big|_{f=\gxl \FF^{-1}}
  $, %\eeq
which is given by (\ref{Corollary:invE_PB_Hl_Hs_statement}) from Corollary \ref{Corollary:invE_PB}. This leads to
 \if{
  \bea{}
    \PB{\Bra{f} \opinvE \Ket{1}}{\BraKet{h}{\Hl}}
     =
    \BraKet{f \opGm}{ h  \tfrac{1}{\FF}\TTheta \trpg }
    \, \frac{1}{\hmg{\WW^{-1}}}
  \eea
 }\fi
  \beq{qrwrwettrw2_Hl}
    \!\HIDE{+}
    \PB{\sint \repNl \FF^{-1}\opinvE \hmgs{\CCof}\,}{\sint \gxl \Hl }
    \,-\,
    \big(\repNl \!\leftrightarrow  \gxl\big)
    %\nonumber\\
    =
     \sint {\repNl \FF^{-1} \opGm}{ \gxl \FF^{-1}\TTheta \trpg }
     \tfrac{\hmgs{\CCof}}{\hmgd{\WW^{-1}}}
      % -
      % \sint{\gxl \FF^{-1} \opGm}{ \repNl \FF^{-1}\TTheta \trpg }
      %  \tfrac{\hmgs{\CCof}}{\hmgd{\WW^{-1}}}
    \:-\:
     \big(\repNl\!,\repNs^\zn \leftrightarrow  \gxl\!, \gxs^\zn\big)
     ,
  \eeq
which naturally gathers the nonlocality into the $\opGm$ structure.

\newpar
Putting together (\ref{qrwrwettrw_varNl_Hs}) and (\ref{qrwrwettrw2_Hl}), and flipping the overall sign, one obtains (\ref{anomaly_GUMG_fin}):
  \bea{qrwrwettrw_all}
   +A
   \if{
   &=&
    \HIDE{+}
     \sint  \gxs^\zn {\repNl}_{\!\!\!\!,\zn} \FF^{-1} \opinvE \hmgs{\CCof}
    +
     \PB{\sint \repNl \FF^{-1}\opinvE \hmgs{\CCof}\,}{\sint \gxl \Hl {+} \gxs^\zn \Hs_\zn }
    \,-\,
    \big(\repNl\!,\repNs^\zn \leftrightarrow  \gxl\!, \gxs^\zn\big)
    \nonumber\\
   }\fi
    &=&
    \HIDE{+} \big( \sint \repNl \FF^{-1}
    \opGm \Dvg{\gxs} \OOmega
    {\,} \tfrac{\hmgs{\CCof}}{\hmgd{\WW^{-1}}}
    -  \sint {\repNl \FF^{-1} \opGm}{ \gxl \FF^{-1}\TTheta \trpg }
     \tfrac{\hmgs{\CCof}}{\hmgd{\WW^{-1}}}
    \big)
    %%
   % +
   % \big(
   %  \sint {\repNl \FF^{-1} \opGm}{ \gxl \FF^{-1}\TTheta \trpg }
   % -
   %  \sint{\gxl \FF^{-1} \opGm}{ \repNl \FF^{-1}\TTheta \trpg }
   % \big)  \tfrac{\hmgs{\CCof}}{\hmgd{\WW^{-1}}}
    %%
   %  \nonumber\\
   %  &+&  t.d.1 + t.d.2
    \:-\:
     \big(\repNl\!,\repNs^\zn \leftrightarrow  \gxl\!, \gxs^\zn\big)
    \nonumber
    \\
    &=&
    \HIDE{+} \Big( \sint \repNl \FF^{-1}
    \opGm
    \big( \Dvg{\gxs} \OOmega -  \gxl \FF^{-1}\TTheta \trpg \big)
    - \sint \gxl \FF^{-1}
    \opGm
    \big( \Dvg{\repNs} \OOmega -  \repNl \FF^{-1}\TTheta \trpg \big)
    \Big)
    \tfrac{\hmgs{\CCof}}{\hmgd{\WW^{-1}}}
  %  \nonumber\\
  %  &+&  t.d.1 + t.d.2
    \,.
   \nonumber
  \eea

%%-------------------------=%        ******         %=-------------------------%%
%%-------------------------=%        ******         %=-------------------------%%
%%-------------------------=%        ******         %=-------------------------%%

\newpage

    \section{Homogeneous Parameterization}
     \label{ASect:Homogeneous_Parameterization}
      \hspace{\parindent}
 \newcommand{\hmgCCof}{\,\hmg{\!\CCof\!\!}\,\,}
In Section \ref{SSect:ParameterizedAction} we mentioned that homogeneously parameterized action
 %% in {\wGUMG} context it is (\ref{hmgparAction_AMG})
 \beq{hmg_param_Action}
  \SSS_{\,\hmg{\!par\!}}
   [\g,\pg,\hmg{\tf}^\io\!,\hmg{\ptf}_\io,\hmg{\replm}^\io\!,\repNs,\replmCT]
   = \int d\taux\hspace{1pt} \dsx \, \Lib
     \pg^{\zm\zn}\tauxdot{\g}_{\zm\zn}
     +  \hmg{\ptf}_\io\,\tauxdot{\hmg{\tf}}^\io
     -   \hmg{\replm}^\io ( \hmg{\ptf}_\io {+} \hmg{\FF \Hl} )
     - \repNs^\zm \Hs_\zm
     - \replmCT^\zm \CT_{,\zm}
     \Rib
  \eeq
is a self-contained equivalent extension of the canonical theory (\ref{extAction_GUMG}). Which at the same time allows to complete the absent homogeneous mode of $\Hl$ canonical structure and reinstate the Einstein-Hilbert term in the Lagrangian action. We noted that adding of the inhomogeneous gauge-trivial sector
  \beq{trivial_inh_sector}
     \sint \Lib
     \inh{\ptf}_\io\,\tauxdot{\inh{\tf}}^\io
     {\,-\,} \inh{\replm}^\io \,\inh{\ptf}_\io
     \Rib
  \eeq
is motivated by the common practice of dealing only with local (functionally-complete) fields which has certain advantages. The same question is how the alternative action would look like if one factors out the auxiliary gauge-trivial sector of inhomogeneous fields.

The answer is straightforward and could be obtained either by repeating the steps of parameterization, or by reducing the trivial sector from the alternative action (\ref{ActionHTlike_GUMG_ccf0}). Both these ways lead a model, defined on mixed-type configuration space of local spacetime metric $\Gaux_{\Zm\Zn}(\taux,\sx)$ and a pair of homogeneous fields $\hmg{\tf}^\io(\taux)$, $\hmgCCof(\taux)$ with the action
 \bea{ActionHTlike_hmg_GUMG_ccf0}
    \SSS_{\,\hmg{\!\ialt}_0}[\,\Gaux,\hmg{\tf}^\io\!,\hmgCCof]
    &\!=\!& \! \int \!d\taux \hspace{1pt} \dsx %\int \!d\taux\, \dsx \,
      \sqrt{|\Gaux|} \Big( \stR (\Gaux)
    - \sqrg^{-1}\!\FF^{-1}
      \tfrac{\,\WW^{-1}\vphantom{|}}{\,\hmg{\WW^{-1}\!}\,\vphantom{I^{|^I}}}
      \hmgCCof \Big)
    + \int \!d\taux \, %\!\int \!d\taux\, \dsx
     \,\tauxdot{\hmg{\tf}}^\io\, \sVol \hmgCCof
    \;. \quad
   % \nonumber
  \eea
Here $\frac{\WW^{-1}}{\,\hmg{\!\WW^{-1}\!}\,\vphantom{I^{|^I}}} \hmgCCof$ comes from $\opinvE \hmgCCof$,
and a nondynamical spatial volume factor $\sVol=\sint 1$ arise from the spatial integration of the homogeneous integrand\footnote{
 Here, as well as in Sections \ref{SSect:ParameterizedAction_AMG} and (\ref{SSect:ParameterizedAction}) when relating auxiliary homogeneous fields with the homogeneous modes of local fields we implicitly assumed that $\sx$-parameterization of $\taux \teq \const$ hypersurfaces is such that the background volume $\sint 1$ is preserves in time (this is mostly relevant for compact spatial manifolds). This, in particular allowed to interchange time differentiation and extracting the homogeneous mode $\tauxdot{\hmg{\tf}}^{\io} = \hmg{\tauxdot{\tf}}^{\io}$.
}.

Dynamical consequences on this reduced configuration space are those, described in Section \ref{SSect:DynamicalProperties_GUMG}. In particular variation with respect to homogeneous $\hmg{\tf}^\io$ implies $\hmgCCof \eomeq \cco$. Which for effective cosmological constant $\CCf$ reinstates spatially nonlocal on-shell relation
 $ % \beq{ccf_onshell_GUMG_COPY}
  \,  \CCf
    \eomeq
   % \sqrg^{-1} \FF^{-1} \opinvE \cco
   % =
   \sqrg^{-1} \FF^{-1} \frac{\WW^{-1}}{\,\hmg{\!\WW^{-1}\!}\,\vphantom{I^{|^I}}} \, \cco
    \,,
 $ % \eeq
(\ref{ccf_onshell_GUMG}).
Variation with respect to the homogeneous $\hmgCCof$ leads to spatially integral  relation between $\tauxdot{\hmg{\tf}}^\io$ and $\hmg{\repNl\FF^{-1}     {\,\WW^{-1}}} /{\,\hmg{\WW^{-1}\!}}$ which in view of additional dynamical input equals $\hmg{\repNl\FF^{-1}}$. Since $\tauxdot{\hmg{\tf}}^\io$ is not fixed from equations of motion, this relation reflects the homogeneous time-reparameterization ambiguity.

The gauge structure is the part of that described in Section \ref{SSect:GaugeStructure_GUMG}. Namely, in the alternative homogeneously parameterized formulation (\ref{ActionHTlike_hmg_GUMG_ccf0}), the {\GUMG} model has a gauge symmetry parameterized by the canonical transformations (\ref{ugvar_canon_GUMG}, \ref{ugvar_NsNl_GUMG}, \ref{gvar_Dvg_lmptf_GUMG}), subject to restrictions (\ref{gauge_param_restrict_GUMG}) with the gauge parameters $\hmg{\geps}(\taux)$ and $\gzetat^\zn(\taux,\sx)$, but now without $\gchi^\zn(\taux,\sx)$ gauge parameter, which parameterized the trivial symmetry of the localized inhomogeneous gauge sector. Remind, $\inh{\replm}^{\io}$ was promoted to $\Dvg{\lmptf}$, which extended the symmetry from $\inh{\geps}$ to that parameterized by spatial vector $\gchi^\zn$  with $\Dvg[\zn]{\gchi} = - \inh{\geps}$. These issues were\HIDE{ better} discussed in the context of the {\wGUMG} subfamily in Section \ref{SSect:GaugeStructure_AMG}.

%%-------------------------=%        ******         %=-------------------------%%
%%-------------------------=%        ******         %=-------------------------%%
%%-------------------------=%        ******         %=-------------------------%%

\newpage

\newcommand{\rR}{R}
\newcommand{\irR}{{\scriptscriptstyle{\rR}}}
\newcommand{\sintR}{\sint_{\!\!\irR\,}}
\newcommand{\sVolR}{\sVol_{\!\!\irR\,}}

\newcommand{\regR}{\mathop{\mathrm{reg}}{}_{\!\irR}}
\newcommand{\limR}{\lim_{\irR \to \infty}}

\newcommand{\invsVol}{\tfrac1{\sVol}}
\newcommand{\invsVVol}{\tfrac1{\sVol^2}}
\newcommand{\invsVolR}{\tfrac1{\sVolR}}
\newcommand{\invsVVolR}{\tfrac1{\sVolR^2}}

%limit structures
\newcommand{\OOO}{{\color{sColor}``O(1)"}}
\newcommand{\ooo}{{\color{sColor}o(1)}}
\newcommand{\ol}{{\color{sColor}o}}
\newcommand{\Ol}{{\color{sColor}O}}
\newcommand{\lsV}{\sVol}
\newcommand{\lsv}{1\mathrm{\color{sColor} v}}
\newcommand{\lsw}{1\mathrm{\color{sColor} w}}

\newcommand{\R}{\mathbb{R}} %% reals
\newcommand{\Sigmat}{\Sigma_t}
\newcommand{\asv}{*}%{\infty}
\newcommand{\pert}[1]{{\hspace{0pt}\vartriangle} #1}

\newcommand{\asvII}{\Id_{\!\asv}}
\newcommand{\pertII}{{\hspace{0.5pt}\vartriangle}}

\newcommand{\hmgII}{\hmgId} %{\mathop{\hmg{\op{I}}}}
\newcommand{\inhII}{\inhId}%{\mathop{\inh{\op{I}}}}

\newcommand{\tfA}{\mathcal{A}}
\newcommand{\tfB}{\mathcal{B}}

{
\def \tempfrac {\tfrac}
\section{Noncompact Spatial Sections and Delocalization}
 \label{ASect:NonCompact} %% 2025-04
  \hspace{\parindent}  
Spacetime manifolds with noncompact spatial sections $\Sigmat$ have important cosmological implications. Therefore, it is worth verifying that results obtained in the compact case remain applicable to noncompact settings. A potential challenge arises from the singular nature of spatial integration in noncompact spaces, particularly when dealing with integral operators. 
A full rigorous treatment would require a careful specification of asymptotic boundary conditions compatible with the theory’s dynamics and gauge structure --- an analysis that deserves a separate discussion. Nevertheless, we outline the key issues below and demonstrate the applicability of our framework with delocalization operators to a broad class of fields.

\newpar

On the non-GR branch, where the diffeomorphism symmetry is broken down to foliation-preserving transverse diffeomorphisms, the natural asymptotic conditions for field configurations on noncompact spatial backgrounds are asymptotically constant\HIDE{ ones}:
 $ \varphi^i(\tx,\sx) \big|_{|\sx|\to\infty}  =\, \varphi_\asv^i(\tx) $.
Since we are primarily concerned with infrared aspects of the formalism, we 
assume that such fields are bounded functions on $\Sigmat$\footnote{
 Formally, this assumption may exclude black hole–like solutions. However, boundedness is adopted here solely as a technical simplification.
}.
These asymptotic conditions capture physically interesting cosmological vacua and align with the residual local symmetry of the non-GR branch where spatial nonlocality becomes significant. This  motivates the following decomposition:
  \beq{NC_field_decomp}
    %\g_{\Zn\Zm}
    \varphi^i(\tx,\sx) = \varphi_\asv^i(\tx) + \pert{\varphi}^i(\tx,\sx)
    \,,  \qquad\varphi^i(\tx,\sx)\big|_{|\sx|\to\infty} \!\!\!\! \to \varphi_\asv^i(\tx)
    \,,  \qquad \pert{\varphi}^i\big|_{|\sx|\to\infty} \!\!\!\! = O\big(\tempfrac1{|\sx|^{c_i}}\big) \,\to\, 0
    \,,
  \eeq
where \emph{asymptotic backgrounds} $\varphi_\asv^i(\tx)$ are spatially constant functions, and ``\emph{perturbations}''\footnote{While $\pert{\varphi}^i$ are not\HIDE{ assumed to be} small in the bulk (only bounded), they\HIDE{ asymptotically} act as small corrections at large distances. Thus, although both components can be locally of the same order, they behave differently under integration.} $\pert{\varphi}^i$  decay at infinity with some falloff rates $c_i$.
Regular local functions of such fields inherit the same structure: asymptotic values can be decoupled allowing a similar decomposition:
 \beq{NC_func_decomp}
  \tfA(\tx,\sx)
  = \tfA_\asv(\tx) + \pert{\tfA}(\tx,\sx)
  \,, \qquad
  \tfA_\asv = \tfA(\varphi^i_\asv)
  \,, \qquad
  \pert{\tfA} = \Ol(\pert{\varphi^\zi})
  \,,
 \eeq
where the latter represents the asymptotically decaying component.

\newpar

Even when focusing on local properties of fields and observables: solutions of their equations of motion, gauge transformations etc. the presence of spatial nonlocalities in the action necessitates examining\HIDE{ the behavior of} integrated expressions. A key issue arises because integration over noncompact spaces is intrinsically singular for the bounded class of fields. While local theories handle this and derive local properties consistently, additional spatial nonlocalities\HIDE{ potentially} can compromise this consistency introducing divergent or ambiguous contributions to local expressions.
Below we argue that the nonlocal components of $\opE$ and $\opinvE$ act as\HIDE{ controlled, finite} perturbations of the ultralocal operator in the integrated quantities and unambiguously contribute to local quantities in the theory.

To handle integration over noncompact spatial manifolds and control possible singularities, we adopt a finite-volume regularization scheme. The constant-time spatial slice $\Sigmat$ is truncated to a large but finite subregion of characteristic radius $\rR$ --- for example, a $(\Ddim{-}1)$-dimensional ball $|\sx| < \rR$ --- and all spatial integrations are preformed over this finite domain:
  \beq{def_R_regularization}
    {\regR}:\qquad
    \sint \ldots  \;\to\; \sintR \ldots  \defeq \sint_{ |\sx| < \rR } \ldots
    \,,
    \qquad
    \sVol \;\to\; \sVolR \defeq \sintR 1
    \,,
  \eeq
with the regulator $\rR$ ultimately taken to infinity\footnote{
 Alternatively, one could use smooth smearing functions or regularize in Fourier space. However, since our analysis does not require repeated integration by parts, a simple cutoff using a characteristic function suffices. This approach also permits direct analogs of the averaging projectors used in the compact case.
}.
%\newpar

Integrals of bounded asymptotically constant \emph{local} functions, (\ref{NC_func_decomp}), decompose as
 \beq{sint_tfA}
   \sint \tfA
   \;=\; \tfA_\asv \sVol + \sint \pert{\tfA}
   \;=\; \tfA_\asv \sVol \,
      + \tfA_{1\asv} \sint \pert{\varphi}
      +  \tfA_{2\asv} \sint  \pert{\varphi} \pert{\varphi} + \tfA_{3\asv} \sint \pert{\varphi} \pert{\varphi} \pert{\varphi} + \ldots
   \,,
 \eeq
where the terms fall into different divergence classes with respect to the regularization. The integral over the decaying perturbation part integrates terms, proportional to at least one $\pert{\varphi}$, and is thus subleading compared to the volume-divergent contribution $\sint \tfA_\asv = \sVol \tfA_\asv $ originating from the asymptotic background, $\tfA_\asv=\tfA(\varphi_\asv)$. Integrals with more decaying factors have better convergence and eventually become finite.
Bilinear combinations decompose analogously:
 $ %\beq{sint_tfA_tfB}
    \sint \tfA \, \tfB
    = \tfA_{\asv} \sVol \tfB_{\asv}
    + \tfA_{\asv} \sint \!\pert{\tfB}
    + \sint \!\pert{\tfA} \cdot \tfB_{\asv}
    +  \sint \!\pert{\tfA} \pert{\tfB}
   \,,
 $ %\eeq
with the leading volume-divergent contribution $\sint (\tfA\tfB)_{\asv}=\tfA_{\asv} \sVol \tfB_{\asv}$ and terms with decaying components giving subleading contributions.

In parameterized {\GUMG}\HIDE{ models}, integrations involving spatial nonlocal operators 
take the form
 \bea{NC_integral_nonloc}
   \sint \tfA' \,\op{C}\, \tfA''
    &=& \sint \tfB^0 \,
      +
    \tempfrac{\sint\tfB^1\,\sint \tfB^3}{\sint \tfB^2}
     +
    \tempfrac{\sint\tfB^4\,\sint\tfB^6\, \sint \tfB^8}{\sint \tfB^5\,\sint \tfB^7}
     +
     \ldots
 \eea
with a finite number of nonlocal contributions, where $\tfB^a$ are bounded local integrands. The balanced number of integrations in each nonlocal term is expected since $\opE$-type operators appear as a sum of projectors at most multiplied by a local regular bounded function. This guarantees that the divergence rate of each term is at most that of a single integration of local structure.

More importantly, nonlocal terms in (\ref{NC_integral_nonloc}) give unambiguous finite contributions to variationally derived local expressions, such as equations of motion, gauge transformations, Noether identities etc. This is because the variational derivative acts only on a one integral factor at a time, while the remaining undifferentiated integral factors compensate the divergencies between numerator and denominator. Resulting nonlocal factors at local quantities  $b^a_{i}(\tx,\sx) \equiv \VDer{\tfB^a}{\varphi^i(\tx,\sx)}$ are ratios of asymptotic values of integrands. For example,
  \beq{}
   \tVDer{}{\varphi^\zi}\, \tempfrac{\sint\tfB^1\sint \tfB^3}{\sint \tfB^2}
   \;=\; b^1_{i} \, \tempfrac{\sint\tfB^3}{\sint \tfB^2}
    - b^2_{i}  \,\tempfrac{\sint\tfB^1\sint\tfB^3}{\sint \tfB^2 \sint \tfB^2}
    + b^3_{i}  \, \tempfrac{\sint\tfB^1}{\sint \tfB^2}
   \;=\;   b^1_{i} \, \tempfrac{\tfB^3_\asv}{\tfB^2_\asv}
    - b^2_{i} \, \tempfrac{\tfB^1_\asv\tfB^3_\asv}{ \tfB^2_\asv  \tfB^2_\asv}
    + b^3_{i} \, \tempfrac{\tfB^1_\asv}{ \tfB^2_\asv}
    + \ooo
   \eeq
up to irrelevant contributions $\ooo$ which vanish when regularization is removed. These terms are the same local quantities  $b^a_{i}$ multiplied by nonlocal factors of the type
  \beq{}
     \tempfrac{\sint \pert{\tfB^1}}{ \sint \tfB^2} \,,\;
     \tempfrac{\sint {\tfB^1}}{\sint (\tfB^2_\asv+\pert{\tfB^2})} - \tempfrac{\sint{\tfB^1}}{ \sVol\tfB^2_\asv} \,,\; ...
     \;\sim\; \Ol (\tempfrac{\sint \pert{\varphi}}{\sVol} )
     = \ooo \,.
  \eeq

 % Analogous considerations may be applied to higher variational structures.

\newpar
This observation confirms the validity of the analyses of the gauge structure and equations of motion in the alternative\HIDE{ parameterized} formulation of {\GUMG}, given in the main part, for noncompact spatial backgrounds. To further support this, we examine the specific properties of $\opE$ and $\opinvE$ in more detail.
For this purpose, we reconsider the notion of spatial averaging for local functions and define corresponding average and average-free projectors with well-defined properties within the adopted regularization.
\if{
Averages of local quantities
  \beq{}
    \hmg{\pert{\varphi^\zi}}
    \equiv \tempfrac{\sint\pert{\varphi^\zi}}{\sVol} \,,\;
    \quad
    \hmg{\pert{\tfA}}
    \equiv \tempfrac{\sint\pert{\tfA}}{\sVol} \,,\;
    \quad
    \hmg{{\tfA}}
    \equiv \tempfrac{\sint{\tfA}}{\sVol} = \tfA_\asv + \hmg{\pert{\tfA}} \,,\;
  \eeq
play the specific role here: averages of perturbations measure the difference between the asymptotic and average values of local functions, and are irrelevant in final finite local expressions, being of order $\Ol (\tempfrac{\sint \pert{\varphi}}{\sVol} )$. However, before variations are taken such contributions should be tracked in integrated expressions.
}\fi
 %Analogs of average and average-free projectors, which we used in compact case, can also be defined for the noncompact case. They naturally respect the chosen regularization procedure.
The symmetric averaging projector $\hmgII$: \;$\hmgII \hmgII \tfA = \hmgII \tfA$ is defined via 
  \beq{}
    \sint \tfA \hmgII \tfB
    \;\defeq\; \sint \tfA \,\invsVol \sint \tfB
    \;\defeq\;  \!\lim_{\irR \to \infty} \sintR \tfA\, {\invsVolR} \sintR \tfB
    \,.
  \eeq
%Limits should be taken after regularizing all terms and extracting singular $\Ol(\sVol)$ and subleading nonvanishing contributions. 
and acts on local bounded functions of asymptotically constant fields as
  \beq{}
    \hmg{\tfA} (\taux)
    \;\defeq\;
    \hmgII \tfA  (\taux)
    \;=\; 1 {\cdot}  \tempfrac{1}{\sVol} \sint \tfA
    \defeq
    \tfA_\asv (\taux) +
    \hmg{\pert{\tfA}} (\taux) %\tempfrac{\sintR \pert{\tfA} }{\sVolR}
 %   + \Ol(\invsVVol) %(\tempfrac{1}{\lsV^2})
    \,,
  \eeq
producing the normalized integrated value. In the noncompact case, the average of the perturbative part, $\hmg{\pert{\tfA}} \defeq 1 {\cdot} \!\limR  \tempfrac{\sintR \pert{\tfA} }{\sVolR}$, give the difference between the $\hmg{\tfA}(\tx)$ and the background value $\tfA_\asv(\tx)$ and vanish as $\rR\to\infty$. Although it is irrelevant locally, it should be retained until after integrations or functional variations are taken.
 
The complementary projector $\inhII \defeq \Id - \hmgII$\,: \;$\inhII \inhII \tfA = \inhII \tfA$, is defined by
 \beq{}
   \sint \tfA \inhII \tfB
    \;\defeq\; \sint \tfA\, \tfB - \sint \tfA \,\invsVol \sint \tfB
    \,,
 \eeq
and extracts the integral-free (average-free) part of a function:
  \beq{InhII_A}
    \inh{\tfA } (\taux,\sx)
    \;\defeq\;
    {\inhII} \tfA (\taux,\sx)
    \;\defeq\; \tfA (\taux,\sx) - \hmg{\tfA} (\taux)
    \,,
    \qquad
    \sint \inhII \tfA = 0
    \,.
  \eeq
The integral-free components  necessarily decay at spatial infinity and can be represented as divergences of bounded vector fields with appropriate asymptotic behavior.
 %\footnote{With a sufficiently fast decay rate to neglect possible boundary contributions when integrating by parts.}

 %These analogs of average and average-free projectors naturally respect the chosen regularization procedure being symmetric with respect to integration (in contrast to projectors, extracting the asymptotic or perturbation parts of a function).

\newpar

We are now ready to specify and analyze properties of the delocalization operators $\opE$ and $\opinvE$, (\ref{opE_kernel}), (\ref{invE_kernel}), in the context of noncompact $\Sigmat$.
The parameterization procedure of Section~\ref{SSect:ParameterizedAction} applies analogously, yielding a nonlocal operator at $\Hl$ with the same integral kernel (\ref{opE_kernel})
   \beq{opE_kernel_COPY}
    E(\taux|\sx;\sx')
    \;=\;
    \delta^{\Ddim{-}1}(\sx;\sx')
    - \WW^{-1} (\taux,\sx) \tempfrac1{\sVol{\HIDE{(\taux)}}}  \WW (\taux,\sx')
    + \tempfrac1{\sVol\HIDE{(\taux)}}
  \,.
  \eeq
This operator is the sum of the identity operator and the nonlocal correction by two projector-type terms, each of order $O(\invsVol)$. When acting on a bounded asymptotically constant function to either side (i.e. after integrating the kernel with such a function) local and nonlocal contributions become locally of order $O(1)$.
However, decomposing $\WW$ and $\WW^{-1}$ into asymptotic and decaying parts, cf.~(\ref{NC_func_decomp}), reveals the cancellation of the \HIDE{leading} asymptotic contributions in the nonlocal terms, leaving them with at least one decaying factor. The kernel thus takes the form:
  \beq{}
   %\begin{array}{lllll}
    E(\taux|\sx;\sx')
    \:=\: \delta^{\Ddim{-}1}(\sx;\sx')
    %&& {\color{Gray} \sim O(\lsV)}
    - \pert{\WW^{-1}}(\taux,\sx) \tempfrac{1}{\sVol} \WW_\asv
    - \WW^{-1}_{\asv} \tempfrac{1}{\sVol} \pert{\WW} (\taux,\sx')
    - \pert{\WW^{-1}}(\taux,\sx)  \tempfrac{1}{\sVol} \pert{\WW} (\taux,\sx')
    \,,
   %\end{array}
  \eeq
suggesting that the nonlocal part is the perturbative correction to the unit operator
 \beq{}
  \opE \defeq\: \Id  + \pert{\opE}
  \,.
 \eeq   %
When acting on a bounded asymptotically constant function\HIDE{ to the right}, $\opE \,\tfB (\taux,\sx) \equiv \sint_{\!\sx'} E(\taux|\sx,\sx') \tfB(\taux,\sx')$,
  \beq{}
   \opE \, \tfB (\taux,\sx)
           \;=\;
              \tfB(\taux,\sx)
           \,-\, \pert{\WW^{-1}}(\taux,\sx)\, \tempfrac{\sint \WW_\asv \tfB}{\sVol}
            - \WW^{-1}(\taux,\sx) \tempfrac{\sint \pert{\WW} \tfB}{\sVol}
        %   \;\defeq\;
        %    \Id \tfB + \pert{\opE} \, \tfB
        \,,
  \eeq
the nonlocal correction yields two types of contributions a decaying $O(\pert{\varphi})$ term, and a suppressed $O(\sint\pert{\varphi}/\sVol)$ term, which can be neglected in local expressions (e.g., equations of motion).
 
As noted in Section~\ref{SSect:AltAction_GUMG}, there is freedom in defining $\opE$ within a family of invertible operators $\opEa \,\defeq \hmgII + \ofa^{-1} \inhII \WW$. In nonlocal case, our choice, (\ref{opE_kernel}),
 %% \footnote{As well as\HIDE{ allowable} similar operators $\opE_\asv= \Id - \WW^{-1}_\asv \hmgII \WW + \hmgII$.}
is convenient due to effective disentangling of the leading local operator (here --- identical operator) from the nonlocal correction $\pert{\opE}$, facilitating the analysis of regularity and asymptotic behavior.

Moreover, the nonlocal part of $\opE$ exhibits even more improved behavior, which is easier seen when integrated with two arbitrary functions.
In  $\sint \tfA \,\pert{\opE}\, \tfB$ the subleading but potentially diverging terms\footnote{  Small perturbations of the spatial volume density are related to {\GUMG} massive extra degree of freedom \cite{Barvinsky:2019agh}, which suggest that $\pert{\sqrg}$ acquires a short-range behavior making integrations with $\pert{\WW}$ and $\pert{\WW^{-1}}$ finite.
} 
$O(\sint\pert{\sqrg})$ come from $\sint \tfA \pert{\WW^{-1}} $ and $\sint \pert{\WW} \tfB$ (multiplied by bounded factors) 
  \beq{AEB_pert}
   \begin{array}{lllll}
    \sint \tfA \,\pert{\opE}\, \tfB
    &\!\!=\!\!&
     {-} \tfA_\asv  \sint \left(\pert{\WW^{-1}} \WW_\asv \!+\! \pert{\WW} \,\WW^{-1}_{\asv} \right)  \tfB_\asv
     - \sint \pert{\tfA} \pert{\WW^{-1}} \HIDE{\cdot} \WW_\asv \tfB_\asv
     - \tfA_\asv \WW^{-1}_\asv \sint \pert{\WW} \pert{\tfB}
    + ...
    \,,
   \end{array}
  \eeq
where dots denote irrelevant terms vanishing in local variational expressions in  $\rR{\to}\infty$ limit.
 % which are at most $\Ol({\tempfrac{\sint \pert\sqrg \sint \pert\sqrg}{\sVol}})$. 
Asymptotical behavior of $\pert{\WW^{-1}}$:
 $\pert{\WW^{-1}}
% = \WW^{-1} - \WW^{-1}_\asv
 = \tempfrac{1}{\WW_\asv + \pert{\WW}} - \WW^{-1}_\asv
 = - \tempfrac{\pert{\WW}}{\WW^2_\asv} + \tempfrac{\pert{\WW}\pert{\WW}}{\WW^3_\asv} + O(\pert{\WW}^3)
 $, is defined by that of $\pert{\WW}$, which improves the first term in the right hand side of (\ref{AEB_pert}):
 %of the two terms in $\sint \left(\pert{\WW^{-1}} \WW_\asv \!+\! \pert{\WW} \,\WW^{-1}_{\asv} \right)$  cancel:
 \beq{NC_opE_decay_improvement}
    \sint \big(\pert{\WW^{-1}} \WW_\asv \!+\! \pert{\WW} \,\WW^{-1}_{\asv} \big)
   \;=\; - \sint \big(\tempfrac{\pert{\WW}\pert{\WW}}{\WW^3_\asv} +  O(\pert{\WW}^3) \big)
   \,.
 \eeq
This property holds for general asymptotically constant functions $\tfA$ and $\tfB$, 
and reflects a general property of $\opE$. This makes the delocalization correction to the identity operator in (\ref{AEB_pert}) safer even for slow-decaying field perturbations.

\newpar

Similar properties hold for the inverse operator $\opinvE$, (\ref{invE}), whose kernel\HIDE{ representation} (\ref{invE_kernel}),
 \beq{invE_kernel_COPY}
  E^{-1}(\taux|\sx;\sx')
 % \if{
  \;=\; \delta^{\Ddim{-}1}(\sx;\sx')
    \,-\, \Big( \tempfrac{\WW^{-1}(\taux,\sx) }{\sint {\WW^{-1}} (\taux)\vphantom{\big{|}}}
    \,+\, \tempfrac{\WW(\taux,\sx')}{\sint{\WW}(\taux)\vphantom{\big{|}}\,} \Big)
    \,+\, 2
    \tempfrac{\WW^{-1}(\taux,\sx) }{\sint {\WW^{-1}} (\taux)\vphantom{\big{|}}}
    \,\sVol{\HIDE{(\taux)}}\,
    \tempfrac{\WW(\taux,\sx')}{\sint{\WW}(\taux)\vphantom{\big{|}}\,}
 % }\fi
  \if{
   = \delta^{\Ddim{-}1}(\sx;\sx')
    - \WW^{-1}(\taux,\sx) \tempfrac{1}{\sint {\WW^{-1}}(\taux)}
    -\tempfrac{1}{\sint{\WW}(\taux)} \WW(\taux,\sx')
    + 2 \WW^{-1}(\taux,\sx)
        \tempfrac{\sVol(\taux)}{\sint{\WW^{-1}}(\taux)  \sint{\WW}(\taux) }
        \WW(\taux,\sx')
  }\fi
  \,,
 \eeq
can be derived from (\ref{opE_kernel_COPY}) following the steps of\HIDE{ the constructive proof of} Lemma \ref{Lemma:InvE} implying the regularization. Again,
  \beq{}
    \opinvE \defeq \Id + \pert{\opinvE}
    \,,
  \eeq
since in the nonlocal correction the leading asymptotic terms cancel\footnote{
The sign at the first contribution in brackets is inverted by contributions from the latter nonlocal term.}:
  \bea{invopE_kernel_pert_Noncompact}
  % \begin{array}{lllll}
    \,E^{-1}(\taux|\sx;\sx')\,
    &\!\!=\!\!& \delta^{\Ddim{-}1}(\sx;\sx')
    \,+\, \Big(  { \HIDE{\big(} \tempfrac{\inhII\WW^{-1}(\taux,\sx) }{\,\sint{\WW^{-1}} (\taux)\vphantom{|}} \HIDE{\big)} }
      + { \HIDE{\big(} \tempfrac{\inhII\WW(\taux,\sx')}{\, \sint{\WW\vphantom{|}}(\taux)\,}  \HIDE{\big)} } \Big)
    \,+\, 2\,
     { \HIDE{\big(} \tempfrac{\inhII\WW^{-1}(\taux,\sx) }{\,\sint{\WW^{-1}} (\taux)\vphantom{|}} \HIDE{\big)} }
     \sVol\HIDE{(\taux)}
     {\HIDE{\big(} \tempfrac{\inhII\WW(\taux,\sx')}{\, \sint{\WW\vphantom{|}}(\taux)\vphantom{|}\,}  \HIDE{\big)} }
    \,. \quad\;\;
  % \end{array}
  \eea
All local factors in nonlocal correction are affected by the \emph{average-free} projectors $\inhII$, (\ref{InhII_A}),
which ensure that nonlocal contributions decay faster than generic perturbations. For instance,
  $ { \HIDE{\big(} \tempfrac{\inhII\WW^{-1}(\taux,\sx) }{\,\sint{\WW^{-1}} (\taux)\vphantom{|}} \HIDE{\big)} }
 \defeq  \tempfrac{\WW^{-1}(\taux,\sx) }{\,\sint{\WW^{-1}} (\taux)\vphantom{|}} - \tempfrac1{\sVol}
 =  \tempfrac{\pert{\WW^{-1}}(\taux,\sx) - \hmg{\pert{\WW^{-1}}}(\taux)}{\sint{\WW^{-1}}(\taux)\vphantom{|}} $
(and similarly for $\WW$).
The first term in the right hand side is of order $O(\tempfrac{\pert{\sqrg}}{\sVol})$, while the subtracted term is $O(\tempfrac{\sint\pert{\sqrg}}{\sVol^2})$\HIDE{ and is irrelevant in local expressions}. %% but may give slowly divergent contributions in $\sint \tfA \pert{\opinvE} \tfB$ for effectively massless $\pert{\sqrg}$.
However, the subtraction remove the leading potentially divergent integrated parts from the first term, so that when integrated with bounded asymptotically constant functions, they yield contributions quadratic in perturbations.
This reflects a general rule, that for bounded asymptotically constant functions,
  \beq{} 
   \sint \tfA \inhII \tfB \;=\; \sint \pert\tfA \pert\tfB  + \ooo \,,
  \eeq 
showing improved convergence behavior of terms with average-free factors\HIDE{ upon integration}.
Accordingly, the kernel form (\ref{invopE_kernel_pert_Noncompact}), with average-free local factors, ensures improved decay of the integrated nonlocal correction:
  \beq{}
     \sint \tfA \pert{\opinvE} \tfB
     \;\sim\;
     \Ol{(\sint \pert{\tfA} \,\pert{\WW^{-1}} )} + \Ol{(\sint  \pert{\WW} \pert{\tfB} )}  + \ooo
    % \;\sim\;
    % \Ol{(\sint (\pert{\varphi})^s (\pert{\sqrg})^r )}  + \ooo
    \,,
  \eeq
in full analogy with the corresponding property (\ref{NC_opE_decay_improvement}) of $\opE$.

\newpar

Finally, we verify the\HIDE{ nonlocal} on-shell behavior of the effective cosmological constant $\CCf$\, in the noncompact case. The action of $\opinvE$, (\ref{invE_kernel_COPY}),\HIDE{ to the right} on a spatially-homogeneous function yields
  \beq{}
    \sint_{\!\sx'\,} E^{-1}(t|\sx;\sx') \,\hmg{f}(\taux)
  \if{
    \;=\;
     \cco
    - \tempfrac{\WW^{-1}(\taux,\sx) }{ \hmg{\WW^{-1}} (\taux)\vphantom{\big{|}}} \hmg{f}(\taux)
    \,-\, \,\hmg{f}(\taux)
    \,+\, 2
     \tempfrac{\WW^{-1}(\taux,\sx) }{ \hmg{\WW^{-1}} (\taux)\vphantom{\big{|}}} \,\hmg{f}(\taux)
  }\fi
    \:=\:
     \tempfrac{\WW^{-1}(\taux,\sx) }{ \hmg{\WW^{-1}} (\taux)\vphantom{\big{|}}} \,\hmg{f}(\taux)
    \:=\:
     \tempfrac{\WW^{-1}(\taux,\sx) }{ (\WW^{-1})_\asv (\taux) }  \,\hmg{f}(\taux) + \ooo
     \,,
  \eeq
where the contribution from the third term in (\ref{invE_kernel_COPY}) compensates the contribution from the identity operator, while the contribution from the fourth term inverts that from the second term. 
In finite local expressions, spatial averages such as $\hmg{\WW^{-1}}$ can be replaced by their asymptotic values, since the difference $\hmg{\pert{\WW^{-1}}} \sim \Ol ( \frac{\sint \pert{\sqrg}}{\sVol} )$ vanishes in the limit $\rR {\,\to\,} \infty$. Using $\frac{1}{(\WW^{-1})_\asv} = \WW_\asv$, this simplifies further.
Thus, on-shell behavior of the effective cosmological constant takes the form
  \beq{ccf_onshell_GUMG_nonloc}
    \CCf
    \:\eomeq\: \sqrg^{-1} \FF^{-1} \opinvE \cco
    \:=\: \sqrg^{-1} \FF^{-1} \tempfrac{\WW^{-1}}{\,\hmg{\WW^{-1}\!}\vphantom{{|}}\,} \, \cco
    \:=\: \sqrg^{-1} \FF^{-1} \tempfrac{\WW_\asv}{\WW} \, \cco
    \,,
  \eeq
(cf.~(\ref{ccf_onshell_GUMG})).
This confirms that the nonlocal behavior of $\CCf$, which feels the asymptotic behavior of the perfect-fluid equation-of-state parameter $\WW(\argsqrg)$\HIDE{of the cosmological perfect fluid}, survives in the noncompact case.

%%-------------------------=%        ******         %=-------------------------%%
%%-------------------------=%        ******         %=-------------------------%%
%%-------------------------=%        ******         %=-------------------------%%

\newpage

    \section{Poisson Brackets of Constraint Structures}
     \label{ASect:PB_GUMG_Tables}
      \hspace{\parindent}
Below we gather the basic Poisson brackets of the constraint structures.
%% For generic {\GUMG} model with generic $\WW(\argsqrg)$
\vspace{-3mm}
\begin{table}[h]
\hspace{-11mm}
\renewcommand{\arraystretch}{1.2}
\newcommand{\tempsint}{\sint}
%\begin{center}
\begin{tabular}{||l||c|c|c||}
 \hline
 \hline
  % after \\: \hline or \cline{col1-col2} \cline{col3-col4} ...
    \!\!\!\!\!
    $\boxed{\! \PB{\mathbf{A}}{\mathbf{B}} \!}$
    %% \;\;\;
     $\mathbf{B} \!\rightarrow \!\!$
    %% &   %%\!\!$\rightarrow$\!\!
 & $\sint h (\ptf_\io {+} \FF\Hl)$
 & $\sint \eta^\zn\Hs_\zn$
 & $\sint \eta^\zn \CT_{,\zn}$
     \phI
  \\
%    \hline
    \;\; $\mathbf{A} \! \downarrow$ \quad $\searrow$\!
       &   &   &
    %% &  %% \phantom{i}
 \\
 \hline
 \hline
 %&&&
 %\\
 $\sint f (\ptf_\io {+} \FF\Hl)$
     & $0$
     & $- \tempsint f_{,\zn} \eta^\zn \tfrac1\WW\, \CT
        + \tempsint f \Dvg{\eta} \,\CT$
     & $ - \tempsint f \Dvg{\eta} \,\TTheta \,\trpg\,\CT $
     %&
     \phI
 \\
  \color{gray}
    $\phantom{.}$
     &
   %% $+O(\Hs_\zm)$
     {\color{gray}
       $ + \tempsint \big(f \ader_{\zn} h\big) \,\FF^2 \g^{\zn\zk} \Hs_\zk$
     }
     &
     {\color{gray}
       $0$
     }

     & %% $+O(\Hs_\zm)$
     {\color{gray}
       $ - \tempsint \big(f \ader_{\zn} \Dvg{\eta}\WW\big) \,\FF^2 \g^{\zn\zk} \Hs_\zk $
     }
     %&
     %\phI
 \\
 \hline
 %&&&&
 % \vspace{-4.5mm}
 %\\
 %\hline
 % &&&&
 %\\
 $\sint \xi^\zm\Hs_\zm$
     &  $ \tempsint h_{,\zm} \xi^\zm \tfrac1\WW \,\CT
        - \tempsint h \Dvg{\xi} \,\CT $
     & $0$
      %$({\xi}^{\zl}{\eta}^\zk_{\;,\zl}{-}{\eta}^{\zl}{\xi}^\zk_{\;,\zl}) {\cdot}\Hs_\zk$
     & $\HIDE{+} \tempsint \Dvg{\xi}\Dvg{\eta} {\,} \OOmega\, \CT $
       %%$+ \xi^{\zm} \eta^\zn_{\z,\zn} {\cdot} \CT_\zm$
     %&
     \phI
 \\ $\phantom{.}$
     &
     {\color{gray}
       $0$
     }
     &
     %% $+O(\Hs_\zm)$
     {\color{gray}
       $+\tempsint \LieB{\zm}{\xi}{\eta}\Hs_\zm$
     }
     & %%$+O(\CT_{,\zm})$
     {\color{gray}
       +$ \tempsint \Dvg{\eta} \,\xi^\zm\CT_{,\zm}$
     }
     %&
     %\phI
 \\
 \hline
 %&&&
 %\\
 $ \sint \xi^\zm \CT_{,\zm}$
     & $ \tempsint \Dvg{\xi} h\,\TTheta \,\trpg\,\CT $
     & ${-} \tempsint \Dvg{\xi}\Dvg{\eta} {\,} \OOmega\, \CT $
      %% $ - \xi^{\zm}_{\z,\zm} \eta^\zn {\cdot} T_\zn$
     & $0$
      %$(\xi^{\zm}_{\z,\zm}\partial_\zk\eta^\zn_{\z,\zn}
      % {-} \eta^\zn_{\z,\zn}\partial_\zk\xi^{\zm}_{\z,\zm})
      %  {\cdot} \FF^2\WW^2\g^{\zk\zl}\Hs_\zl$
     %&
     \phI
 \\
 $\phantom{.}$
     & %% $+O(\Hs_\zm)$
     {\color{gray}
       $- \tempsint \big(\Dvg{\xi}\WW \ader_{\zn} h \big) \,\FF^2 \g^{\zn\zk} \Hs_\zk$
     }
     & %% $+O(\CT_{,\zm})$
     {\color{gray}
       $- \tempsint \Dvg{\xi} \,\eta^\zn\CT_{,\zn}$
     }
     & %% $+O(\Hs_\zm)$
     {\color{gray}
       $+ \tempsint \big(\Dvg{\xi} \ader_{\zn} \Dvg{\eta}\big) \,\FF^2 \WW^2 \g^{\zn\zk} \Hs_\zk$
     }\hspace{-4mm}
     %&
     %\phI
 \\
  \hline
   %&&&& \HIDE{\phantom{i}} \\
  \hline
\end{tabular}
%\beq{PBtableGUMGs0Equiv}
%\text{\textit{Poisson brackets of constraints of the {\GUMG} parameterized action (\ref{parAction_BGUMG}).}}
%\eeq
\caption{Poisson brackets of constraint structures\HIDE{ of the parameterized theory} ($\WW(\argsqrg)$ of general form)}
 \label{table:Par_PB_WW}
  %{\raggedright <Text> \par}
  %{\small \it   }
\renewcommand{\arraystretch}{1.0}
%  \beq{PBtableGUMGs0Equiv}
%   \text{\textit{Poisson brackets of \underline{nonextended} constraints of the {\GUMG} parameterized action (\ref{parAction_BGUMG}).}}
%  \eeq
%\end{center}
\end{table}
\vspace{-3mm}

We use the notations: $\Dvg{\eta}\defeq\Dvgp[\zn]{\eta}\,$, $\,(f \ader_{\zn} h) \defeq f h_{,\zn} {-} f_{,\zn} h$ \,and\, $\LieB{\zm}{\xi}{\eta} {\defeq\,} \xi^\zn \eta^\zm_{,\zn} {\,-\,} \eta^\zn \xi^\zm_{,\zn}$.
Derivative characteristic structures\, $\OOmega(\argsqrg) \,{\defeq}\,  \tTDer{\ln{\WW}}{\ln\!\sqrg} {\,+\,} \WW {\,+\,} 1$,  $\TTheta(\argsqrg) {\,\defeq\,}  \tfrac{1}{\Ddim{-}2} \tTDer{\ln{\WW}}{\ln{\!\sqrg}}\,\tfrac{\FF}{\sqrg}$\, are defined in (\ref{def_OOmega_TTheta}).
\if{
  \bea{def_OOmega_TTheta}
    \OOmega(\argsqrg)=  \tTDer{\ln{\WW}}{\ln\sqrg} + \WW + 1\;,
    \qquad
    \TTheta(\argsqrg)=  \tfrac{1}{\Ddim{-}2} \tTDer{\ln{\WW}}{\ln{\sqrg}}\,\tfrac{\FF}{\sqrg}
    \,.
    %\nonumber\\
  \eea
}\fi
Weakly vanishing terms are marked in grey. Also useful are the relations:
\vspace{-1mm}
 \bea{PB_sqrg_Hl}
  \!\! \PB{\sqrg}{\sint h B(\g) \Hl}
   %& =
   \,=\, - h\,\tfrac{1}{\Ddim{-}2} B(\g) \trpg %\,,
   \quad\!\! \Rightarrow \;
   \begin{array}{|rcl}
     \PB{ \FF}{\sint h B(\g) \Hl} &\!\!\! =\!\!&
     -  h \tfrac{1}{\Ddim{-}2} B(\g) \FF\, \WW \tfrac{\trpg}{\sqrg}
     \,, \\
     \PB{ \WW}{\sint h B(\g) \Hl} &\!\!\! =\!\!&
     -  h B(\g) \WW \,\tfrac{1}{\FF}\TTheta \trpg
     \;\;
     \,, \\
  %   \PB{ \WW\FF\sqrg}{\sint h B(\g) \Hl} &\!\! =\!&
  %   -  h \tfrac{1}{\Ddim{-}2} B(\g) \WW\FF\sqrg \,\OOmega \tfrac{\trpg}{\sqrg}
  %   \,, \\
   \end{array}
   %\hspace{-4mm}
 \eea
 %where   (\ref{TThetaFromPB_BGUMG}) %% GUMG_NDV.tex
\vspace{-4mm}
 \bea{PB_sqrg_Hs}
   \PB{\sqrg}{\sint \eta^\zn \Hs_\zn}
   %& =
   \,=\, \partial_\zn \big( \eta^\zn \sqrg \big) %\,,
   \quad\!\!  \Rightarrow
   \begin{array}{|rcl}
     \PB{ \FF}{\sint  \eta^\zn \Hs_\zn}
     &\!\!\! =\!\!&
     \eta^\zn \partial_\zn \FF
     +
     \Dvg{\eta} \FF \WW  \vphantom{\tTDer{\ln\FF}{\ln\sqrg} }
    \:=\:
     \eta^\zn_{\,;\zn} \FF \WW
     \,, \\
     \PB{ \WW}{\sint  \eta^\zn \Hs_\zn}
     &\!\!\! =\!\!&
      \eta^\zn \partial_\zn \WW
     + \Dvg{\eta} \WW \tTDer{\ln\WW}{\ln\sqrg}
    \:=\:
     \eta^\zn_{\,;\zn} \WW \tTDer{\ln\WW}{\ln\sqrg}
     \,, \;\;\\
  %   \PB{ \WW\FF\sqrg}{\sint f B(\g) \Hl} &\!\! =\!&
  %   -  f \tfrac{1}{\Ddim{-}2} B(\g) \WW\FF\sqrg \,\OOmega \tfrac{\trpg}{\sqrg}
  %   \,, \\
   \end{array}
   \hspace{-5mm}
  \eea
with arbitrary test functions $f(\sx)$, $\eta^\zn(\sx)$, and any{ function} $B(\g)$ depending only on $\g_{\zm\zn}$.
\if{
 $  \PB{\sint f \FF}{\sint g B(\g) \Hl} = -\sint f g \tfrac{1}{\Ddim{-}2} B(\g) \WW\FF \tfrac{\trpg}{\sqrg}$,
 $  \PB{\sint f \WW}{\sint g B(\g) \Hl} = - \sint f g B(\g) \tfrac{\WW}{\FF} \,\TTheta \trpg$,
 $\PB{\sint f \WW\FF\sqrg}{\sint g B(\g) \Hl}
   =
   - \sint f g \tfrac{1}{\Ddim{-}2} B(\g) \WW\FF\sqrg \OOmega \tfrac{\trpg}{\sqrg}
 $
 %\;with\;
 %$ \TTheta \pg = \tfrac{1}{\Ddim{-}2}  \FF \tTDer{\ln\WW}{\ln\sqrg} \tfrac{1}{\sqrg} \pg$
}\fi
The constraint ${\ptf_\io}_{,\zm}$ of the parameterized canonical action is an abelian constraint and has strictly vanishing Poisson brackets with all structures that do not depend on $\tf^\io$.

\newpar

For {\wGUMG} subfamily $\WW \teq \wwc \teq \const$ with $\Fw \defeq \sqrg^\wwc$ and $\CTw \defeq \wwc\Fw\Hl$. This imply:
%%the Poisson brackets of generic {\GUMG} simplify to
\vspace{-3mm}
\begin{table}[h]
\hspace{-12mm}
%\begin{center}
%\setlength\extrarowheight{5pt}
\renewcommand{\arraystretch}{1.2}
\newcommand{\tempsint}{\sint}
\begin{tabular}{||l||c|c|c||}
 \hline
 \hline
  % after \\: \hline or \cline{col1-col2} \cline{col3-col4} ...
    \!\!\!\!\!
    $\boxed{\! \PB{\mathbf{A}}{\mathbf{B}} \!}$
    %% \;\;\;
     $\mathbf{B} \!\rightarrow \!\!$
    %% &   %%\!\!$\rightarrow$\!\!
 & $\sint h (\ptf_\io {+} \Fw\Hl\!)$
 & $\sint \eta^\zn\Hs_\zn$
 & $\sint \eta^\zn %%%%\CTw{}_{,\zn} %\to
     {(\wwc\Fw\Hl)}_{,\zn} $
    %% & \phI
    \\
%    \hline
    \;\; $\mathbf{A} \! \downarrow$ \quad $\searrow$\!
       &   &   &
    %% &  %% \phantom{i}
  \\
  \hline
  \hline
  %% &&&&&
  %% \\
  $\sint f (\ptf_\io {+} \Fw\Hl\!)$$\vphantom{\ader}$\!\!\!  %% &
     & $0$
     & $ %%\!\! {-} \tempsint f_{,\zn} \eta^\zn \tfrac1\WW\, \CTw {+}
         %%%% \tfrac{\wwc+1}{\wwc} \tempsint f \Dvg{\eta} \, \CTw  %\to
            (\wwc{+}1) \tempsint f \Dvg{\eta} \, \Fw\Hl $
        & $ 0 $
     %% & %% \phI
  \\
  \color{gray}
     $\vphantom{\ader}$  %% &
     &
   %% $+O(\Hs_\zm)$
     {\color{gray}
       $ {+} \tempsint\big(f \ader_{\zn} h \big) \sqrg^{2\wwc} \g^{\zn\zk} \Hs_\zk$
     }
     &
     {\color{gray}
       ${+} %%%% \tfrac1\wwc \tempsint f \eta^\zn \, \CTw{}_{,\zn} %\to
         \tempsint f \eta^\zn \, {(\Fw\Hl)}_{,\zn} $
     }

     & %% $+O(\Hs_\zm)$
     {\color{gray}
       $ {-} \wwc \tempsint\big(f \ader_{\zn} \Dvg{\eta}\big) \sqrg^{2\wwc} \g^{\zn\zk} \Hs_\zk $
     }
     \phI
  \\
  \hline
 %% %% &
 %% &&&
 %% %% &
 %% \\
 $\sint \xi^\zm\Hs_\zm$  $\vphantom{\ader}$  %% &
     & $ %%\tempsint g_{,\zm} \xi^\zm \tfrac1\WW \,\CT
        {-} %%%% \tfrac{\wwc+1}{\wwc} \tempsint g \Dvg{\xi} \,\CTw %\to
          (\wwc{+}1) \tempsint h \Dvg{\xi} \,\Fw\Hl $
     & $ 0 $
      %$({\xi}^{\zl}{\eta}^\zk_{\;,\zl}{-}{\eta}^{\zl}{\xi}^\zk_{\;,\zl}) {\cdot}\Hs_\zk$
     & $ %%%% (\wwc{+}1) \tempsint \Dvg{\xi}\Dvg{\eta} \, \CTw  \to
         \wwc (\wwc{+}1) \tempsint \Dvg{\xi}\Dvg{\eta} \, \Fw\Hl  $
       %%$+ \xi^{\zm} \eta^\zn_{\z,\zn} {\cdot} \CT_\zm$
     \phI
  \\
 $\vphantom{\ader}$  %% &
     &
     {\color{gray}
       ${-} %%%% \tfrac1\wwc \tempsint g \,\xi^\zm \,\CTw{}_{,\zm} %\to
         \tempsint h \,\xi^\zm \,{(\Fw\Hl)}_{,\zm}  $
     }
     &
     %% $+O(\Hs_\zm)$
     {\color{gray}
       $ {+} \tempsint \LieB{\zm}{\xi}{\eta}\Hs_\zm$
     }
     & %%$+O(\CT_{,\zm})$
     {\color{gray}
       $ {+} %%%% \tempsint \Dvg{\eta} \,\xi^\zn \CTw{}_{,\zn} %\to
        \wwc \tempsint \Dvg{\eta} \,\xi^\zm {(\Fw\Hl)}_{,\zm} $
     }
     %% & %% \phI
  \\
  \hline %%&
 %%    &&&
 %%   %% &
 %% \\
 $\!\sint \xi^\zm %%%%\CTw{}_{,\zm} %\to
     {(\wwc\Fw\Hl)}_{,\zm}$  $\vphantom{\ader}$  \!\!\!\!\!\!
     %% &
     & $ 0 $
     & $\!\!{-} %%%% (\wwc{+}1) \tempsint \Dvg{\xi}\Dvg{\eta} \, {\CTw} \! \%
        \wwc (\wwc{+}1) \tempsint \Dvg{\xi}\Dvg{\eta} \, {\Fw\Hl} \! $
      %% $ - \xi^{\zm}_{\z,\zm} \eta^\zn {\cdot} T_\zn$
     & $0$
      %$(\xi^{\zm}_{\z,\zm}\partial_\zk\eta^\zn_{\z,\zn}
      % {-} \eta^\zn_{\z,\zn}\partial_\zk\xi^{\zm}_{\z,\zm})
      %  {\cdot} \FF^2\WW^2\g^{\zk\zl}\Hs_\zl$
     \phI
   \\
   $\vphantom{\ader}$  %% &
     & %% $+O(\Hs_\zm)$
     {\color{gray}
       $ \!\! {-} \wwc \tempsint \big(\Dvg{\xi}\ader_{\zn} h\big) \!\sqrg^{2\wwc} \g^{\zn\zk} \Hs_\zk \!\!\!\! $
     }
     & %% $+O(\CT_{,\zm})$
     {\color{gray}
       $ {-} %%%% \tempsint \Dvg{\xi} \,\eta^\zn \CTw{}_{,\zn} %\to
         \wwc \tempsint \Dvg{\xi} \,\eta^\zn {(\FF\Hl)}_{,\zn}  $
     }
     & %% $+O(\Hs_\zm)$
     {\color{gray}
       $ \!\! {+} \wwc^2 \! \tempsint \big(\Dvg{\xi} \ader_{\zn} \Dvg{\eta}\big) \!\sqrg^{2\wwc}  \g^{\zn\zk} \Hs_\zk \!\!\!\! $
     }
     %% & %% \phI
     \\
  \hline
  %% &&&&
   %% & %% \phantom{i}
  %% \\
  \hline
\end{tabular}
 \caption{Poisson brackets of constraint structures\HIDE{ of the parameterized theory} ($\WW \tequiv \wwc \teq \const$)}
 \label{table:Par_PB_ww}
  %{\raggedright <Text> \par}
  %{\small \it }
\renewcommand{\arraystretch}{1.0}
%  \beq{PBtableGUMGs0Equiv}
%   \text{\textit{Poisson brackets of \underline{nonextended} constraints of the {\GUMG} parameterized action (\ref{parAction_BGUMG}).}}
%  \eeq
%\end{center}
\end{table}
\vspace{-3mm}

To obtain the Poisson brackets for the first-class constraint $\CPI$ defined in (\ref{CPI_AMG}), use the substitution $\sint f(\ptf_\io {+} \Fw\Hl) \to  \sint (f\ptf_\io {+} \hmg{f}\Fw\Hl)$, which implies $f\to\hmg{f}$ and $f_{,\zn}\to 0$. Also note that: $\wwc\sint \,\hmg{f} \Dvg{\eta} \,{\Fw\Hl} = O\big(({\wwc\Fw\Hl})_{,\zn}\big) {\,\weq\,} 0$.

%%-------------------------=%        ******         %=-------------------------%%
%%-------------------------=%        ******         %=-------------------------%%
%%-------------------------=%        ******         %=-------------------------%%

\if{

\newpage
\section{(Add): Parameterization of the Action in Field Theory}
 \label{ASect:parametrization} %% 2024-10
  \hspace{\parindent}
\newcommand{\ii}{\zi}
\newcommand{\jj}{\zj}
\newcommand{\qo}{\qq^{\io}}
\newcommand{\tauxdotqo}{\tauxdot{\qq}^{\io}}
\renewcommand{\po}{\pp_{\io}}
\renewcommand{\lmo}{\lambda^{\io}}
\newcommand{\lm}{\lambda}
\renewcommand{\replm}{\rep{\lambda}}
\renewcommand{\Ho}{H_\io}
\newcommand{\ccm}{\Zm}
\newcommand{\ccn}{\Zn}

Discussion of parametrization in mechanics can be found in many textbooks on constrained systems, including \cite{Henneaux:1992ig}. Parameterization of the field theories is a more rare topic. In this section we briefly review the parametrization procedure and summarize discussion on parameterizing the field theory models. Parameterization converts the theory with a physical Hamiltonian to the theory with null Hamiltonian by converting physical Hamiltonian into a new ``parametrization'' constraint via introducing the triple of auxiliary variables: two conjugated phase-space variables $\qo$, $\po$ and a Lagrange multiplier $\lmo$.

Any mechanical model on phase space $(\qq^\ii,\pp_\ii)$ with Hamiltonian $\Ho(\qq,\pp)$ and the complete set of constraints $\CC_\ccm(\qq,\pp)$,
  \bea{extAction_MechWithConstr} % (\ref{initAction_Mechanics})
   \SSS_{\iE} [\qq,\pp,\lm]
    \,= \int \!d\tx\,  \LiB
      \pp_\ii \dot \qq^\ii - \Ho(\qq,\pp) - \lm^\ccm \CC_\ccm(\qq,\pp)
     \RiB \,,
  \eea
can be parameterized. The equivalent parameterized model is defined by the action
  \bea{parAction_MechWithConstr} % (\ref{parAction_Mechanics})
   \Spar[\qo,\po,\qq,\pp,\lmo,\rep{\lm}]
   = \int \!d\taux\, \LiB
    \pp_\ii\tauxdot\qq^\ii + \po\tauxdotqo - \lmo\big(\po {+} \Ho\big)- \replm^\ccm \CC_\ccm(\qq,\pp)
   \RiB \,.
  \eea
One can interpret the parameterized model as the model defined with respect to the new time parameter $\taux$, so that initial time $t$ becomes a dynamical variable $\qo(\taux)$. In this interpretation one often assumes the correspondence of the new $\qq^\ii(\taux)$ to initial $\qq^\ii\big(\qo(\taux)\big)$ and implicitly assumes redefinition of the Lagrange multipliers to new ones $\rep{\lm}^\ccm(\taux)$, which are expressed in terms of the initial ones as $\lmo(\taux) {\lm}^\ccm\big(\qo(\taux)\big)$ or $\tauxdot\qo(\taux) {\lm}^\ccm\big(\qo(\taux)\big)$. Alternatively one can consider the parameterized model (\ref{parAction_MechWithConstr}) as model defined with respect to the same time variable, without redefinition of configuration variables.

Both interpretations are possible due to physical equivalence of the models, which is bases on two facts. The first is that addition of constraint $\po {+} \Ho(\qq,\pp)$ to the system of constraints introduces to the model the new first-class constraint, generically expressible as
 $ % \beq{def_par_Iclass}
   \CPI(\po,\qq,\pp) \equiv \po {+} \Ho(\qq,\pp) + a^\ccm(\qq,\pp) \CC_\ccm(\qq,\pp)
 $ % \eeq
with some phase-space functions  $a^\ccm(\qq,\pp)$.
which generate time-reparametrization gauge symmetry. The correspondence between two theories is most easily established by partial gauge fixing of the parameterized model, so that gauge-fixing condition together with the parametrization constraint form a conjugated second-class pair and thus can be reduced ``in the action'' by reduction with respect to the set of auxiliary variables  $\qo,\po,\lmo$ and the auxiliary Lagrange multiplier for a gauge-fixing constraint.
Thus, expression of these four auxiliary variables from the system of their own variational equations of motions reduces the parameterized action (\ref{parAction_MechWithConstr}) to the initial extended action (\ref{extAction_MechWithConstr}).

There is a universal \emph{correspondence gauge},
 \beq{def_corr_gauge}
  \qo-\taux =0 \,.
 \eeq
Which is accessible by construction, since it is transversal to the parametrization constraint, $\PB{(\qo{-}\taux)}{\CPI} = 1$.

If all constraints $\CC_\Zm(\qq,\pp)$  are first class or if all the Lagrange multipliers for the second-class constraints\HIDE{, which are fixed by the consistency conditions on equations of motion,} are fixed to vanishing values, then the new reparametrization first-class constraint $\CPI$ can be chosen in the form $(\po {+} \Ho(\qq,\pp))$, up to arbitrary linear combination of the first-class constraints. This is guaranteed either by Dirac-Bergmann procedure of obtaining the complete set of constraint or alternatively by the direct check that dynamics, governed in (\ref{extAction_MechWithConstr}) by $\Ho$, preserves the constraint surface $\CC_\ccm=0$, which is equivalent that this set is complete and the action (\ref{extAction_MechWithConstr}) is the extended action of the model.
If in the constraint set of (\ref{extAction_MechWithConstr}) there are second-class constraints $\IIC_\rho (\qq,\pp)$ \footnote{Assume that the the complete constraint set $\CC_\ccm(\qq,\pp)$ of the  system (\ref{extAction_MechWithConstr}) is disentangled in to equivalent set of the first- and the second-class constraints $(\IC_\alpha, \IIC_\rho)$, whose Lagrange multipliers are correspondingly divided into two subsets $(v^\alpha,u^\rho)$, so that $\lm^\ccm \CC_\ccm =  v^\alpha \IC_\alpha {+} u^\rho \IIC_\rho$.}, and their Lagrange multipliers $u^\rho$ for some of them are nontrivially expressed in terms of the phase space variables $u^\rho \eomeq U^\rho(\qq,\pp)$ within the system of the equations of motion, then the $\Ho$ in $\CPI$ should be extended in particular way
 \beq{def_par_Iclass}
   \CPI(\po,\qq,\pp) \equiv \po {+} \Ho(\qq,\pp) + U^\rho(\qq,\pp) \IIC_\rho(\qq,\pp)
 \eeq
again up to arbitrary linear combination of the first-class constraints. Combination $H'(\qq,\pp) = \Ho(\qq,\pp) {+} U^\rho(\qq,\pp) \IIC_\rho(\qq,\pp)$ is known in the literature on constrained systems as ``primed'' or first-class Hamiltonian.\footnote{Generically any function of phase-space variables by appropriate shift with the second-class constraints can be made first class \cite{Henneaux:1992ig}.}

Finally, $H'(\qq,\pp)$ is a first-class function in the system (\ref{extAction_MechWithConstr}) and has on-shell vanishing Poisson brackets with all its constraints. Thats why the new constraint $\po{+}H'(\qq,\pp)$ in  (\ref{parAction_MechWithConstr}) automatically becomes in weak involution with all \emph{other} constraints of the parameterized system. And there is a special mechanical property that every phase-space function is in involution with itself, which implies $\PB{\po{+}H'}{\po{+}H'}=0$. Thus $\po{+}H'(\qq,\pp)$, \ref{def_par_Iclass}, is the first-class constraint in the parameterized system. \HIDE{It is by construction independent constraint since it is the only one dependent on new variable $\po$.}

\subsubsection*{Parameterization in field theory}
 %\label{ASSect:Parameterization} %% 2024-10 major revision
  \hspace{\parindent}
In the general field theory model
  \bea{extAction_FTWithConstr} % (\ref{initAction_Mechanics})
   \SSS_{\iE}[\qq,\pp,\lm]
    \,= \int \!d\taux \hspace{1pt} \dsx %\int \!d\tx \,\dsx
     \,\LiB
      \pp_\ii \dot \qq^\ii - \Ho(\qq,\pp) - \lm^\ccm \CC_\ccm(\qq,\pp)
     \RiB \,,
  \eea
where fields $(\qq^\ii, \pp_\ii, \lm^\ccm)$ and their local phase-space functions $f(\qq,\pp)$ depend besides time $t$ on spatial coordinates $\sx^\zm$,
the parametrization procedure analogous to (\ref{extAction_MechWithConstr})$\to$(\ref{parAction_MechWithConstr}) is not guaranteed.
Though many procedures in constrained field theory are analogous to that of mechanics, there are various differences, which in particular may lead to failure of the naively transferred parametrization prescription. The simplest reason is that local field constraints --- are infinite dimensional from the mechanical point of view. And if the phase-space functions objects are not algebraic in canonical fields and contain spatial derivatives then they do not generically commute even with themselves.

In field theory the Dirac consistency procedure as well as the condition of the dynamical completeness of the system of constraints guarantees only that \emph{integrated} Hamiltonian $\mathrm{H}'(\tx)=\sint H'(\tx,\sx)$ is the first-class nonlocal function and is in weak involution with the complete set of local constraints. However the naive prescription (\ref{extAction_MechWithConstr}) $\to$ (\ref{parAction_MechWithConstr}) turn it into \emph{local} constraint function $\po(\taux,\sx) {+}  H'(\taux,\sx)$, whose involution with local constraints of the extended system (\ref{extAction_FTWithConstr}) is not guaranteed. Even when $H'(\taux,\sx)$ contain terms with spatial derivatives it generically does not commute with itself, $\PB{H'(\taux,\sx)}{H'(\taux,\sx')}\neq0$.
\\

Lets determine the formal criteria when naive parametrization, analogous to (\ref{extAction_MechWithConstr}) $\to$ (\ref{parAction_MechWithConstr}), is valid. In other words when the action
  \bea{parAction_FTWithConstr}
   \Spar[\qo,\po,\qq,\pp,\lmo,\rep{\lm}]
    \,= \int \!d\taux \hspace{1pt} \dsx  %\int \!d\taux\, \dsx
    \,\LiB
    \pp_\ii\tauxdot\qq^\ii + \po\tauxdotqo - \lmo\big(\po {+} \Ho\big)- \replm^\ccm \CC_\ccm(\qq,\pp)
    \RiB \,.
  \eea
describes the system, which is physically equivalent to (\ref{extAction_FTWithConstr}).
 %
%The sufficient\footnote{For theories with regular structure of independent constraints this is also the necessary condition of equivalence.} criteria of validity of the naive parametrization prescription is the following.
 %
Theories (\ref{extAction_FTWithConstr}) and (\ref{parAction_FTWithConstr}) are physically equivalent when two following conditions are satisfied:
  \begin{enumerate}
    \item Equal rank condition: the on-shell ranks of Poisson brackets matrix of the complete constraint sets of extended action (\ref{extAction_FTWithConstr}) and parameterized action (\ref{parAction_FTWithConstr})  coincide,
      \bea{rank_validity_condition}
        \rank \PB{\CC_\ccm}{\CC_\ccn} \weq \rank \PB{\Psi_M}{\Psi_N}
        \,,
      \eea
  %  where $\psi_\mu\equiv\big(\Hs_\zm,\CT_{,\zm}\big)$ and $\Psi_M\equiv\big(\ptf_\io {+} {\FF \Hl},\Hs_\zm,\CT_{,\zm}\big)$.
  where  $\Psi_M \equiv \big( (\po {+} \Ho), \CC_\ccm \big)$.

    \item There exist an admissible gauge-fixing condition for the reparametrization symmetry so that after the partial gauge fixing of the parameterized action (\ref{parAction_FTWithConstr}) one obtains the initial action (\ref{extAction_FTWithConstr}) via partial reduction with respect\HIDE{ to auxiliary variables} to conjugated second-class constraints.
  \end{enumerate}
The rank condition (\ref{rank_validity_condition}) in particular demands that new constraint $(\po {+} \Ho )$ brings one new functionally independent \emph{first-class} constraint to the system of constraints. This new first-class constraint may be either the new constraint itself or some linear combination of the new constraint and the constraints $\CC_\ccm$ of the type
(\ref{def_par_Iclass}).%\footnote{The coefficient at the new constraint in the new first-class linear combination should be full-rank at least at constraint surface.}

As we mentioned earlier these conditions are too restrictive in field theory and the parametrization is not guaranteed. However there are various new options. First note that the Hamiltonian density $\Ho$, integrated over space at fixed time, is defined in (\ref{parAction_FTWithConstr}) modulo rather arbitrary total derivatives with respect to time spatial coordinates. In particular a total divergence of the local spatial-vector function may be added. And if we allow description with spatial nonlocality, then we may add any spatially nonlocal average-free functions or functionals.\footnote{We restrict ourselves to preserve only locality in time.} Such total derivatives and average-free functionals added to the hamiltonian density become nontrivial when incorporated in the new local constraint.

To restore validity of the parametrization procedure one may allow a sort of the partial conversion procedure by modifying Hamiltonian and constraints by the terms, proportional to powers of the new auxiliary variables $\qo,\po$ and their spatial derivatives. The objective is to satisfy the rank condition (\ref{rank_validity_condition}), and at the same time do not spoil the second, correspondence condition. One can add such terms with care, limiting oneself to the terms $ \sim\qo_{,\zm}$, which assure\HIDE{ help providing} the correspondence.
 %
%% %% INCORRECT! 2024-10
%% For example, an example of such scheme is parametrization the action of the relativistic scalar field, where perturbative procedure of extending the parameterized hamiltonian with such terms (to achieve the local involution with itself) ends in two steps. However, generically this procedure is infinite and may be spatially nonlocal.
 %
However, this procedure generically generates infinite series.
 %\footnote{Even the model of the free relativistic scalar field fails to be locally parameterized by naive parametrization procedure and correction with powers of auxiliary fields.}
The naive time parametrization of the general {\GUMG} theory with $\WW{\not\equiv} \const $ fails and cannot be by finite local expansion of such kind. \TBC{The same is true even for the free relativistic scalar field.}

Below we briefly outline an alternative manifestly spatially nonlocal approach. It may be unsatisfactory in problems when one wants to use an equivalent formulation with a true local reparametrization gauge symmetry, which acts locally on the initial configuration space, and at the same time one cannot extend it to a functionally complete gauge freedom by other incomplete constraints from the set of constraints before parametrization. However this scheme works in the nontrivial case of  {\GUMG} and is worth mentioning.

\paragraph{Homogeneous reparametrization}\vphantom{.}\\
%\vspace{5mm}

% \TODO{Edit} %% copied 2024-09 from main text

In general field theory one can always use the direct analog of the quantum mechanical reparametrization by introducing only spatially \emph{homogeneous}, id est $\sx$-independent\HIDE{ and only time-dependent}, auxiliary degrees of freedom $\hmg{\qo}=\hmg{\qo}(\taux)$, $\hmg{\po}=\hmg{\po}(\taux)$, $\hmg{\lmo} =\hmg{\lmo} (\taux)>0$, which account for a spatially-homogeneous time reparametrization. Homogeneously-parameterized action
 \bea{parhmgAction_FTWithConstr}
   \Spar_{hmg}[\hmg{\qo},\hmg{\po},\qq,\pp,\hmg{\lmo},\rep{\lm}]
    \,= \int \!d\taux \hspace{1pt} \dsx  %\int \!d\taux\, \dsx
    \,\LiB
    \pp_\ii\tauxdot\qq^\ii + \hmg{\po} \tauxdot{\hmg{\qo}} - \hmg{\lmo}\big(\hmg{\po} {+} \Ho\big)- \replm^\ccm \CC_\ccm(\qq,\pp)
    \RiB \,,
  \eea
 %%  %% DOUBLING
 %% where $\hmg{\qo}=\hmg{\qo}(\taux)$, $\hmg{\po}=\hmg{\po}(\taux)$, $\hmg{\lmo} =\hmg{\lmo} (\taux)>0$ are new spatially homogeneous fields which depend only on time and $\sVol\equiv \sint 1$  is the coordinate volume of $\taux = \const$ hypersurfaces. All other fields are reexpressed as functions of new time variable $\taux$ (and spatial coordinates $\sx$).\footnote{One could embed the original theory in parameterized theory in different ways. For example one can leave the time coordinate of the base spacetime manifold unchanged and just introduce new extended set of fields. This way does not assume reexpression of fields since the time coordinate is not changed. Also one can rescale Lagrange multipliers by $\replm$, or by $\tauxdot{\xo}$, or keep them as in original theory. \TBC{These all are physically equivalent ways of parameterizing the theory}.}
is classically equivalent to the initial action (\ref{extAction_FTWithConstr}) for the same reasons, as in mechanics, because the hamiltonian converted to homogeneous constrain is still integrated over space. The new homogeneous ``Hamiltonian'' constraint,
 \beq{def_parhmg_Iclass}
   \hmg{\CPI}(\po,\qq,\pp) \equiv \hmg{\po} {+} \hmg{\Ho}(\qq,\pp) + \hmg{ U^\rho(\qq,\pp) \IIC_\rho(\qq,\pp) }
   \,,
 \eeq
is in weak involution with other constraint, because integrated primed Hamiltonian $\sint ({\Ho} {+} { U^\rho\IIC_\rho})$ is the first-class function. And this action may be gauge-fixed with condition $\hmg{\qo}(\taux)-\taux=0$ so that the reduction with respect to the conjugated second-class constraints $\hmg{\CPI}, \hmg{\qo}(\taux)-\taux$ recovers the extended action (\ref{extAction_FTWithConstr}).

The homogeneous first-class (\ref{def_parhmg_Iclass}) generates the homogeneous transformations with the homogeneous gauge parameter $\hmg{\geps} (\taux)$. In the auxiliary sector this implies $\gvar[\hmg{\geps}] \qo = \hmg{\geps}$, $\gvar[\hmg{\geps}] \po = 0$, and induces a Lagrange multiplier gauge transformation $\gvar{\lmo}=\tauxdot{\hmg{\geps}}$.

The configuration space of the homogeneously parameterized theory (\ref{parhmgAction_FTWithConstr}) contains fields from different functional spaces\HIDE{ (i.e. sections of different bundles)}. This introduces (spatial) nonlocality to the theory. For many reasons it may be more convenient to work with configuration space of the same functional space with spacetime-local fields. Below we demonstrate that it is always possible to reformulate (\ref{parhmgAction_FTWithConstr}) using only local auxiliary fields. This can be achieved by incorporating the complementary inhomogeneous auxiliary sector, which is purely gauge, into (\ref{parhmgAction_FTWithConstr}).
%Below we show that it is always possible to reformulate (\ref{parhmgAction_FTWithConstr}) in terms of only local fields by adding to the auxiliary homogeneous sector the complementary inhomogeneous sector which is a pure gauge.
%\footnote{In the context of {\GUMG} the local parametrization will keep us closer to Henneaux--Teitelboim line of reasoning for {\UMG} theory \cite{Henneaux:1989zc}.}

\subsubsection*{Trivially localized homogeneous parametrization}%\vspace{5mm}
   \hspace{\parindent}
In field theory one can add to (\ref{parhmgAction_FTWithConstr}) an action of the pure gauge sector
 \bea{triv add_FT}
   \SSS_{triv} [\inh{\qo},\inh{\po},\inh{\lmo}]
    \,= \int \!d\taux \hspace{1pt} \dsx  %\int \!d\taux\, \dsx
    \,\LiB
       \inh{\po} \tauxdot{\inh{\qo}} - \inh{\lmo}\inh{\po}
    \RiB \,,
  \eea
with spatially average-free fields $\inh{\qo}$,$\inh{\po}$,$\inh{\lmo}$, which complement homogeneous auxiliary sector fields in (\ref{parhmgAction_FTWithConstr}) to the functionally complete quantities. Action $\SSS_{par} = \SSS_{\,\hmg{\!par\!}} + \SSS_{triv}$, which we name minimally parameterized, take almost local form
 \bea{parminAction_FTWithConstr}
   \SSS_{par}[\qo,\po,\qq,\pp,\lmo,\rep{\lm}]
    \,= \int \!d\taux \hspace{1pt} \dsx  %\int \!d\taux\, \dsx
    \,\LiB
    \pp_\ii\tauxdot\qq^\ii + {\po} \tauxdotqo - {\lmo}\big({\po} {+} \hmg{\Ho}\big)- \replm^\ccm \CC_\ccm(\qq,\pp)
    \RiB \,,
  \eea
with the only spatially nonlocal interaction term $\sint {\lmo} \hmg{\Ho} \equiv \sint \hmg{\lmo}{\Ho}$ between the configuration space of the initial extended theory (\ref{extAction_FTWithConstr}) and (now local) auxiliary fields.

The addition of (\ref{triv add_FT}) introduces in to (\ref{parminAction_FTWithConstr}) additional ``lazy'' gauge invariance in this sector, which is parameterized by some spatially-inhomogeneous gauge parameter $\inh{\geps}(\taux,\sx)$. Aggregating this symmetry with the homogeneous canonical symmetry, defined by the constraint (\ref{def_parhmg_Iclass}), one can write the additional first-class constraint of (\ref{parminAction_FTWithConstr}) in the form
 \beq{def_parhmg_Iclass}
   {\CPI}(\po,\qq,\pp) \,\equiv\, {\po} + \hmg{\Ho}(\qq,\pp) + \hmg{ U^\rho(\qq,\pp) \IIC_\rho(\qq,\pp) }
   \,,
 \eeq
which generates gauge transformation with local parameter $\geps(\taux,\sx)=\hmg{\geps}(\taux) + \inh{\geps}(\taux,\sx)$. This implies $\gvar[\geps] \qo = \geps$, $\gvar[\geps] \po = 0$, and induces a Lagrange multiplier gauge transformation $\gvar{\lmo}=\tauxdot{\geps}$. However, it is obvious that localized are the gauge transformations only in the auxiliary sector. At the same time (\ref{def_parhmg_Iclass}) acts on phase space variables and their functions,
 \beq{}
  \gvar[\geps]f(\g,\pg)
  \equiv \PB{f(\g,\pg)}{\sint \geps {\CPI}}
  = \PB{f(\g,\pg)}{\sint ({\Ho} {+}  U^\rho\IIC_\rho)} {\cdot}\hmg{\geps}
  \,,
 \eeq
with the  spatially-homogeneous part of the gauge parameter $\hmg{\geps}(\taux)$. Which induce homogeneous action on the Lagrange multipliers in the configuration space, inherited from (\ref{extAction_FTWithConstr}).

}\fi

%%%=============================================================================%%%
%%%=============================================================================%%%
%%%=============================================================================%%%

\makeatletter
\@ifundefined{INSERTION} %{\foo}
  {%
    % \foo not defined

%%%=============================================================================%%%
%%%================================ Bibliography ===============================%%%
%%%=============================================================================%%%
{}

%\newpage
 {
  \def \INSERTION {}
 }

\end{document}